\documentclass[a4paper,10pt,leqno]{amsart}
        \title{Conformal nets II: conformal blocks}       
       \author{Arthur Bartels}
      \address{Westf\"alische Wilhelms-Universit\"at M\"unster\\
               Mathematisches Institut\\
               Einsteinstr.~62,
               D-48149 M\"unster, Deutschland}
        \email{bartelsa@wwu.de}
      \urladdr{http://www.math.uni-muenster.de/u/bartelsa}
       \author{Christopher L. Douglas} 
      \address{Mathematical Institute\\ Radcliffe Observatory Quarter\\ Woodstock Road\\ Oxford\\ OX2 6GG\\ United Kingdom}
        \email{cdouglas@maths.ox.ac.uk}
      \urladdr{http://people.maths.ox.ac.uk/cdouglas}
       \author{Andr{\'e} Henriques}
      \address{Mathematisch Instituut\\
               Universiteit Utrecht, Postbus 80.010\\
               3508 TA Utrecht, The Netherlands}
        \email{a.g.henriques@uu.nl}
      \urladdr{http://www.staff.science.uu.nl/\!\raisebox{-1mm}{~}\!henri105}  
         \date{\today}

\usepackage{hyperref}
\usepackage{enumerate,amssymb}
\usepackage[arrow,curve,matrix,tips,2cell]{xy}
  \SelectTips{eu}{10} \UseTips
  \UseAllTwocells
\usepackage{tikz}                \usetikzlibrary{calc}               \usetikzlibrary{matrix}		\usetikzlibrary{patterns}
\usepackage{pdfcolmk}
\usepackage{calc}
\usepackage{multirow}
\usepackage{stmaryrd}

  \newcommand{\secDualizability}{Sec.~4}

  \newcommand{\subsecsectorsfornets}{Sec.~1B}
  \newcommand{\eqfunctorphi}{Eq.~(1.12)}
  \newcommand{\secconformalembeddings}{Sec.~1.E}
  \newcommand{\defimplements}{Def.~2.1}
  \newcommand{\lemLconformalimplementaionOP}{Lem.~2.9}
  \newcommand{\propprojectiveimplementationdiffeo}{Prop.~2.10}
  \newcommand{\thmVaccumSector}{Thm.~2.13}
    \newcommand{\propvacuumsectorcovariant}{(ii)}
    \newcommand{\propvacuumsectorL}{(iv)}
  \newcommand{\lemHS}{Lem.~3.5}
  \newcommand{\thmKLallirreduciblesectorsarefinite}{Thm.~3.14}
  \newcommand{\lemdualofHlambda}{Lem.~3.17}
  \newcommand{\thmKLM}{Thm.~3.24}     
  \newcommand{\lemNNN}{Lem.~3.30}
  \newcommand{\coritsapartialSigmasector}{Cor.~3.31}

       \newcommand{\cala}{\mathcal{A}}
       \newcommand{\calb}{\mathcal{B}}
  \newcommand{\IC}{\mathbb{C}}     
  \newcommand{\ID}{\mathbb{D}}

       \newcommand{\cali}{\mathcal{I}}
     \newcommand{\calj}{\mathcal{J}}

  \newcommand{\IP}{\mathbb{P}}     
       
  \newcommand{\IR}{\mathbb{R}}     \newcommand{\calr}{\mathcal{R}}

  \newcommand{\IU}{\mathbb{U}}

  \newcommand{\IX}{\mathbb{X}}     
       
  \newcommand{\IZ}{\mathbb{Z}}

  \newcommand{\bfB}{{\mathbf B}}

  \newcounter{
  counter}

  \definecolor{AHcolor}{rgb}{0.5,0.0,0.5}   
  \definecolor{CDcolor}{rgb}{0.7,0.0,0.3}   
  \definecolor{ABcolor}{rgb}{0.2,0.8,0.2}   
  
  \newcommand{\AB}[1]{\marginpar{\raggedright\tiny\color{ABcolor}{ #1}}}

  \newcommand{\CDcomm}[1]{}
  \newcommand{\ABcomm}[1]{}
  \newcommand{\AHcomm}[1]{}

  \newcommand{\tikzmath}[2][]
     {\vcenter{\hbox{\begin{tikzpicture}[#1]#2
                     \end{tikzpicture}}}
     }

  \definecolor{spacecolor}{gray}{.7}
  \definecolor{antispacecolor}{gray}{.45}

  \theoremstyle{plain}
  \newtheorem{theorem}{Theorem}[section]
  \newtheorem{maintheorem}[theorem]{Main Theorem}
  \newtheorem{lemma}[theorem]{Lemma}
  \newtheorem{corollary}[theorem]{Corollary}
  \newtheorem{proposition}[theorem]{Proposition}

  \newtheorem*{theorem*}{Theorem}
  \newtheorem*{theorem1*}{Theorem 1 (preliminary version: no boundary)}
  \newtheorem*{theorem2*}{Theorem 1}
  \newtheorem*{theorem3*}{Theorem 2 (preliminary version: gluing along closed manifolds)}
  \newtheorem*{theorem4*}{Theorem 2}

  \theoremstyle{definition}
  \newtheorem{definition}[theorem]{Definition}
  
  \newtheorem{warning}[theorem]{Warning}

  \theoremstyle{remark}
  
  \newtheorem{remark}[theorem]{Remark}
  
  \newtheorem*{example*}{Example}

  \makeatletter\let\c@equation=\c@theorem\makeatother

     {\end{list}}

  \DeclareMathOperator{\ad}{Ad}
  \DeclareMathOperator{\aut}{Aut}

  \DeclareMathOperator{\Diff}{Diff}

  \DeclareMathOperator{\PU}{PU}
  \DeclareMathOperator{\Rep}{Sect}
  \DeclareMathOperator{\U}{U}
  
  \DeclareMathOperator{\supp}{supp}

\DeclareMathOperator\bigboxtimes{\tikzmath{
\useasboundingbox (-.2,-.22) rectangle (.23,.2);
\node[scale=1.5] {$\boxtimes$};}}

\def\squared#1{\tikz{\useasboundingbox (-.13,-.11) rectangle (.13,.12);\node[draw, inner sep = 1.5]{\tiny #1};}}

  \newcommand{\INT}{{\mathsf{INT}}}
  \newcommand{\VN}{{\mathsf{VN}}}
  
  \newcommand{\modules}[1]{{#1\text{-}\mathsf{modules}}}

  \newcommand{\alg}{{\mathit{alg}}}
  
  \newcommand{\op}{{\mathit{op}}}

  \newcommand{\x}{{\times}}
  \newcommand{\ox}{{\otimes}}
  
  \newcommand{\dd}{{\partial}}

  \DeclareRobustCommand{\SkipTocEntry}[5]{}

\renewcommand{\thesubsection}{\arabic{section}.{\sc\alph{subsection}}}

\begin{document}

	\begin{abstract} 
Conformal nets provide a mathematical formalism for conformal field theory.  Associated to a conformal net with finite index, we give a construction of the `bundle of conformal blocks', a representation of the mapping class groupoid of closed topological surfaces into the category of finite-dimensional projective Hilbert spaces.  We also construct infinite-dimensional spaces of conformal blocks for topological surfaces with smooth boundary.  We prove that the conformal blocks satisfy a factorization formula for gluing surfaces along circles, and an analogous formula for gluing surfaces along intervals.  We use this interval factorization property to give a new proof of the modularity of the category of representations of a conformal net.
	\end{abstract}

\date{}
\maketitle
\tableofcontents

\newcommand{\comment}[1]{}
\comment{
}


\section*{Introduction}
\subsection*{Conformal field theory and conformal blocks}
\addtocontents{toc}{\SkipTocEntry}
Given a 2d conformal field theory (CFT), its \emph{partition function} $Z$ assigns a number to every Riemann surface~$\Sigma$.%
\footnote{In fact, a non-trivial CFT always has an anomaly and, despite the terminology, the partition function depends not only on a conformal structure on the surface but also on a metric.
However, the partition function transforms in a specified way under changes within a given conformal equivalence class, see~\cite[Sec.3]{Friedan-Shenker(The-analytic-geometry-of-two-dimensional-conformal-field-theory)}.}
The \emph{correlation functions} generalize the partition function:
they assign numbers to Riemann surfaces with marked points (and local coordinates) labelled by fields of the CFT---the fields form a vector space, also called the state space of the conformal field theory.

\emph{Conformal blocks} were first introduced in the famous paper of 
Belavin, Polyakov, and Zamolodchikov \cite{Belavin-Polyakov-Zamolodchikov} for the special case of the Virasoro conformal field theory.
They are certain holomorphic functions $\mathfrak{F}_a$ that enter in a formula for the correlation and partition functions.
In its simplest version, the formula reads
\begin{equation}\label{eq: Z=FFbar}
Z=\sum h^{ab} \mathfrak{F}_a\bar{\mathfrak{F}}_b
\end{equation}
where $h^{ab}$ is a certain matrix that depends only on the genus of the surface, and on the conformal field theory.
These functions $\mathfrak{F}_a$ were called conformal blocks ``because any correlation function is built up of these functions'' \cite{Belavin-Polyakov-Zamolodchikov}.
Unlike the correlation and partition functions which are single-valued functions on the moduli space, the conformal blocks typically depend on more data 
than just the surface and marked points, for instance on a pair-of-pants decomposition of the surface.

It was quickly realized (\cite{Friedan-Shenker(The-analytic-geometry-of-two-dimensional-conformal-field-theory)}, \cite{Moore-Seiberg(classical+quantum-conf-field-theory)})
that the formula \eqref{eq: Z=FFbar} can be rewritten in multiple ways (related by what Moore-Seiberg \cite{Moore-Seiberg(classical+quantum-conf-field-theory)} call ``duality'', that is, the move
$
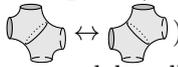
\tikzmath[scale=.7]{
\filldraw[fill=gray!20] (1,.2) to [in=90, out=180, looseness=1.2] (.8,0) arc (0:-180: .2 and .05) -- (.4,0) to [in=0, out=90] (0,.4) -- (0,.8) to [in=-90, out=0] (.2,1) -- (.6,1) to [in=180, out=-90] (1,.6) arc (90:-90:.05 and .2) -- cycle;
\filldraw[fill=gray!50] (.4,1) circle (.2 and .05) (0,.6) circle (.05 and .2);
\draw[densely dashed] (1,.4) + (0,.2) arc (90:270:.05 and .2) (.6,0) + (.2,0) arc (0:180:.2 and .05);
\pgftransformrotate{45}
\draw (.424,0) arc (-180:0:.283 and .05);
\draw[densely dotted] (.424,0) arc (180:0:.283 and .05);}
\leftrightarrow
\tikzmath[scale=.7]{\pgftransformrotate{-90}
\filldraw[fill=gray!20] (1,.2) to [in=90, out=180, looseness=1.2] (.8,0) -- (.4,0) to [in=0, out=90] (0,.4) -- (0,.8) to [in=-90, out=0] (.2,1) arc (180:0: .2 and .05) -- (.6,1) to [in=180, out=-90] (1,.6) arc (90:-90:.05 and .2) -- cycle;
\draw[densely dashed] (1,.4) + (0,.2) arc (90:270:.05 and .2) (.4,1) + (.2,0) arc (0:-180:.2 and .05);
\filldraw[fill=gray!50] (.6,0) circle (.2 and .05) (0,.6) circle (.05 and .2);
\pgftransformrotate{45}
\draw (.424,0) arc (-180:0:.283 and .05);
\draw[densely dotted] (.424,0) arc (180:0:.283 and .05);}
$~\!)
and that the only object that is truly invariant is the vector space spanned by all the conformal blocks for a given Riemann surface.
This is the so-called \emph{space of conformal blocks}.

\subsection*{The structure of the spaces of conformal blocks}
\addtocontents{toc}{\SkipTocEntry}
As it turns out, the spaces of conformal blocks only depend on an object that contains somewhat less information (but is more delicate to define) than a conformal field theory.
That object goes by the name \emph{chiral conformal field theory}.
To emphasize the difference between chiral conformal field theories and conformal field theories, the latter are also often called \emph{full conformal field theories}.
One may think very roughly of a chiral conformal field theory as being some kind of algebra.
For example (unlike for conformal field theories), there exists a notion of \emph{representation}, also called \emph{sector}, of a chiral conformal field theory.

The space of conformal blocks associated to a Riemann surface $\Sigma$ is denoted $V(\Sigma)$;
the choice of chiral conformal field theory is suppressed from the notation.
There is also a generalization where the Riemann surface is equipped with a finite collection of points $p_i$ (along with first order coordinates at those points), labelled by representations $\lambda_i$ of the chiral conformal field theory.
The space of conformal blocks is then denoted
\begin{equation}\label{eq: VOA conf blocks}
V(\Sigma,p_1,\ldots,p_n;\lambda_1,\ldots,\lambda_n).
\end{equation}

The expected algebraic and differential geometric structures of the spaces of conformal blocks were described in \cite{Friedan-Shenker(The-analytic-geometry-of-two-dimensional-conformal-field-theory)} and \cite{Moore-Seiberg(classical+quantum-conf-field-theory)}.
They should form holomorphic bundles over the moduli spaces of Riemann surfaces of genus $g$ with $n$ marked points (with first order coordinates).
Conformal blocks should certainly be functorial with respect to isomorphisms between Riemann surfaces, but they should also be functorial with respect to orientation reversing isomorphisms~\cite{Moore-Seiberg(Naturality-in-conformal-field-theory)}:
an antiholomorphic isomorphism ${f:\Sigma_1\to\Sigma_2}$ between Riemann surfaces should induce an antiunitary isomorphism
\[
f_*:\,V\big(\Sigma_1,p_1,\ldots,p_n;\lambda_1,\ldots,\lambda_n\big)\,\to\, V\big(\Sigma_2,f(p_1),\ldots,f(p_n);\bar\lambda_1,\ldots,\bar\lambda_n\big),
\]
where $\bar\lambda$ stands for the representation dual to $\lambda$.
These bundles should also be equipped with a hermitian inner product and a projectively flat connection, compatible with both the holomorphic and the unitary structure.

Most importantly, the behavior of these bundles under degeneration should be governed by the so-called factorization formula.
Informally, the factorization formula says that
\begin{equation}\label{eq:fact}
\lim_{\varepsilon\to 0}\,\,\,\,V\Bigg(\tikzmath[scale=1.2]{\useasboundingbox (-.05,-.2) rectangle (2.05,1.2);
\filldraw[fill=gray!25](.9,.6) to[out=90+45, in=0](.4,.85) to[out=180, in=90](0,.5) to[out=-90, in=180](.4,.15) to[out=0, in=-90-45](.9,.4) to[out=45, in=90+45]%
(1.1,.4) to[in=180, out=-45](1.6,.15) to[in=-90, out=0](2,.5) to[in=0, out=90](1.6,.85) to[in=45, out=180](1.1,.6) to[out=-90-45, in=-45](.9,.6);
\filldraw[fill=white](.27,.5) to[out=-50, in=-90-40](.6,.5) to[out=90+40, in=50] (.27,.5);\draw (.27,.5) -- +(-.04,.05) (.6,.5) -- +(.04,.05);
\filldraw[fill=white](2-.27,.5) to[out=-90-40, in=-50](2-.6,.5) to[out=50, in=90+40] (2-.27,.5);\draw (2-.27,.5) -- +(.04,.05) (2-.6,.5) -- +(-.04,.05);%
\draw (1,.5-.059) arc (-90:90:.02 and .059);\draw[densely dotted] (1,.5+.059) arc (90:90+180:.02 and .059);
\draw[stealth-] (1,.4) -- (1,.15) node[below, scale=.8]{\small neck of} node[below, scale=.8, yshift=-8]{\small size $\varepsilon$};
\filldraw[very thin](.27,.83-.05) circle (.02) -- +(0,.15) node[above, scale=.8]{$\scriptstyle \mu_1$}(.55,.832-.05) circle (.02) -- +(0,.15) node[above, scale=.8, xshift=1]{$\scriptstyle \mu_2$}
(2-.27,.83-.05) circle (.02) -- +(0,.15) node[above, scale=.8, xshift=1]{$\scriptstyle \nu_2$}(2-.55,.832-.05) circle (.02) -- +(0,.15) node[above, scale=.8]{$\scriptstyle \nu_1$};  }   
\Bigg) \,\,=\,\,\, \bigoplus_{\lambda}\,\,\,
V\Bigg(\tikzmath[scale=1.2]{\useasboundingbox (-.05,0) rectangle (1,1);
\filldraw[fill=gray!25](.9,.6) to[out=90+45, in=0](.4,.85) to[out=180, in=90](0,.5) to[out=-90, in=180](.4,.15) to[out=0, in=-90-45](.9,.4) to[out=45, in=-45](.9,.6);
\filldraw[fill=white](.27,.5) to[out=-50, in=-90-40](.6,.5) to[out=90+40, in=50] (.27,.5);\draw (.27,.5) -- +(-.04,.05) (.6,.5) -- +(.04,.05);
\filldraw[very thin](.27,.83-.05) circle (.02) -- +(0,.15) node[above, scale=.8]{$\scriptstyle \mu_1$}(.55,.832-.05) circle (.02) -- +(0,.15) node[above, scale=.8, xshift=1]{$\scriptstyle \mu_2$}
(.9,.5) circle (.02) -- +(.03,.17) node[above, scale=.8]{$\scriptstyle \lambda$};  }\Bigg)\otimes  
V\Bigg(\tikzmath[scale=1.2]{\useasboundingbox (1,0) rectangle (2.05,1);
\filldraw[fill=gray!25](1.1,.4) to[in=180, out=-45](1.6,.15) to[in=-90, out=0](2,.5) to[in=0, out=90](1.6,.85) to[in=45, out=180](1.1,.6) to[out=-90-45, in=90+45](1.1,.4);
\filldraw[fill=white](2-.27,.5) to[out=-90-40, in=-50](2-.6,.5) to[out=50, in=90+40] (2-.27,.5);\draw (2-.27,.5) -- +(.04,.05) (2-.6,.5) -- +(-.04,.05);
\filldraw[very thin](2-.27,.83-.05) circle (.02) -- +(0,.15) node[above, scale=.8, xshift=1]{$\scriptstyle \nu_2$}
(2-.55,.832-.05) circle (.02) -- +(0,.15) node[above, scale=.8]{$\scriptstyle \nu_1$}(1.1,.5) circle (.02) -- +(-.03,.17) node[above, scale=.8]{$\scriptstyle \bar\lambda$};  }\Bigg)  
\smallskip
\end{equation}
where the direct sum runs over all the irreducible representations $\lambda$ of the chiral conformal field theory.

\subsection*{Constructions of spaces of conformal blocks}
\addtocontents{toc}{\SkipTocEntry}
If one interprets ``chiral conformal field theory'' to mean ``vertex algebra'', then a construction of the spaces of conformal blocks can be found \cite[Chapter 10]{Frenkel--Ben-Zvi}.
However, as far as we know, there is no general proof of for instance the factorization property. 
The literature is somewhat confusing, and we refer the reader to the review paper \cite{Huang-Lepowsky(Tensor-categories-and-the-mathematics-of-rational-and-logarithmic-conformal-field-theory)} for a detailed history (restricted to genus zero and one) of what has been proved and when, along with an indication of various false claims in the literature and later retractions.

The first mathematical constructions of the spaces of conformal blocks are due to Tsuchiya-Ueno-Yamada \cite{Tsuchiya-Ueno-Yamada(CFT-universal-family)} and Hitchin \cite{Hitchin(Flat-connections-and-geometric-quantization)}, both restricted to the case of WZW models (the WZW models are a certain class of chiral conformal field theories that depend on a choice of compact, simple, connected, simply connected Lie group $G$, and a positive integer $k$).
An isomorphism between these two constructions was later constructed in \cite{Beauville-Laszlo(Conformal-blocks-and-generalized-theta-functions)},
\cite{Faltings(A-proof-for-the-Verlinde-formula)}, and \cite{Kumar-Narasimhan-Ramanathan(Infinite-Grassmannians-and-moduli-spaces-of-G-bundles)},
and the projectively flat connections connections present on both sides were identified in \cite{Laszlo(Hitchin's-and-WZW-connections-are-the-same)}.
Factorization for those models was proved in \cite{Tsuchiya-Ueno-Yamada(CFT-universal-family)}.

The unitarity of the spaces of conformal blocks has turned out to be more difficult than the other properties.
It is believable that one could combine the unitarity of the modular tensor categories coming from quantum groups (\cite{wenzl-cstar}, \cite[Sec.4]{Rowell(From-quantum-groups-to-UMTC)} and references therein), the fact that modular functors are determined by their genus zero data (\cite{Andersen-Ueno(Modular-functors-determined-genus-zero)}),
and the equivalence of the braided tensor categories associated to quantum groups and to affine Lie algebras (\cite[Sec.3]{Huang-Lepowsky(Tensor-categories-and-the-mathematics-of-rational-and-logarithmic-conformal-field-theory)} and references therein) to prove the unitarity of the spaces of conformal blocks for most WZW models.
However, at the moment, the only statement that is clearly available in the literature is the one for the $SU(n)$ WZW models, due to work of Andersen and Ueno \cite{Andersen-Ueno(Abelian-CFT+det-bundles)}, \cite{Andersen-Ueno(Geometric-constr-modular-functors)}, \cite{Andersen-Ueno(Reshetikhin-Turaev-from-CFT)} (see \cite{Ramadas(The-Harder-Narasimhan-trace-and-unitarity-of-the-KZ/Hitchin-connection:-genus-0)} and \cite{Looijenga(Unitarity-of-SL(2)-conformal-blocks-in-genus-zero)} for an alternative proof for the $SU(2)$ WZW models in genus zero).
In an early paper, Axelrod, Della Pietra, and Witten made an attempt at defining the hermitian inner product on the spaces of conformal blocks of the WZW models by writing down a formal expression in terms of functional integrals \cite{Axelrod-DellaPietra-Witten}.
They explained how to give a meaning to their formula in genus one, but were unable to deal with the case of higher genus Riemann surfaces.

\subsection*{Formalisms for weak CFT}
\addtocontents{toc}{\SkipTocEntry}
Positive energy conformal nets \cite[Def.~4.5]{BDH(nets)} (which is what people typically mean by the term ``conformal net'' in the literature, see e.g. \cite{Longo(Lectures-on-Nets-II)})
form, along with vertex algebras and Segal's chiral weak CFTs \cite[Def.~5.2]{Segal(Def-CFT)},
one of the three existing mathematical formalizations of the notion of a chiral conformal field theory. (The extent to which these formalisms are equivalent is not completely understood, but note the important advance by Carpi--Kawahigashi--Longo--Weiner constructing a bijection between strongly local unitary vertex algebras and a corresponding class of conformal nets~\cite{cklw}.)
Given the prominent role that the spaces of conformal blocks play in conformal field theory, 
it is surprising that nobody has defined them within the formalism of conformal nets.
Our paper aims to fill this gap.

The notion of \emph{conformal net} that we will be using here (introduced in~\cite{BDH(nets)}, see Definition~\ref{def:conformal-net} below)
is more general than the one commonly used in the literature.
Our notion 
is closely related to that of a conformally covariant net on 2-dimensional Minkowski space \cite{Brunetti-Guido-Longo(1993modular+duality-in-CQFT)}.
We believe that given such a 2-dimensional net (not necessarily modular invariant),
one can obtain a conformal net in our sense by the procedure of  time-zero restriction~\cite[Sec.4]{Kawahigashi-Longo-Mueger(2001multi-interval)}.
It is even plausible that the above procedure provides an equivalence between these two kinds of nets.


\subsection*{Unitary spaces of conformal blocks}
\addtocontents{toc}{\SkipTocEntry}
As explained above, it is expected that the spaces of conformal blocks form a finite-dimensional vector bundle with hermitian inner product and projectively flat unitary connection over the moduli space of Riemann surfaces.
(In the vertex algebra context, this would be under the assumption that the vertex algebra is unitary \cite{Dong-Lin(Unitary-vertex-operator-algebras)} and subject to appropriate finiteness conditions.)
Passing to the associated bundle of projective spaces, the connection becomes genuinely flat and
its monodromy yields a projective representation of the fundamental group of the moduli space, a.k.a. the mapping class group.

In the present article, we will be working with conformal nets instead of vertex algebras, and with topological instead of Riemann surfaces.
The price to pay for working with topological surfaces instead of Riemann surfaces is that we can only expect the space of conformal blocks $V(\Sigma)$ to form a projective space, not a vector space.
(One way to still get a vector space would be to pick a Lagrangian in $H_1(\Sigma;\IR)$, as explained in \cite{Bakalov-Kirillov(Lect-tens-cat+mod-func)}. We shall not pursue that approach here.)

We first describe the result of our construction in the case of surfaces without boundary.
The object that plays the role of the moduli spaces of Riemann surfaces is the following groupoid.
Its objects are connected oriented topological surfaces and its morphisms are isotopy classes of homeomorphisms.
We allow orientation reversing maps:
the automorphism group of a surface in that groupoid is thus an index two overgroup of the mapping class group.
The simplest form of our main construction is:

\begin{theorem1*}
Let $\cala$ be a conformal net of finite index (Definition \ref{def:finite-nets}), and let $\Sigma$ be a closed connected oriented topological surface.
We can associate to $\cala$ and $\Sigma$ a finite dimensional projective Hilbert space $\IP V(\Sigma)$ called the projective space of conformal blocks of $\Sigma$. 

For every isotopy class of homeomorphism $f:\Sigma_1\to \Sigma_2$ we have an induced map $f_*:\IP V(\Sigma_1)\to \IP V(\Sigma_2)$
which is projective unitary if $f$ is orientation preserving and projective antiunitary otherwise.
These maps satisfy $(f\circ g)_*=f_*\circ g_*$.
\end{theorem1*}

\begin{remark}
Among the conformal nets corresponding to WZW models, only the ones with gauge group $SU(n)$ are known to have finite index
(this was proven by Xu in \cite{Xu(Jones-Wassermann-subfactors)}, based on the work of Wassermann \cite{Wassermann(Operator-algebras-and-conformal-field-theory)}).
For the other WZW conformal nets, the finite index property is also expected to be true but has not yet been proven.
\end{remark}

\begin{remark}
We emphasize that every step of our construction is visibly unitary.
This is in stark contrast with the work of Andersen and Ueno~\cite{Andersen-Ueno(Abelian-CFT+det-bundles)}, \cite{Andersen-Ueno(Geometric-constr-modular-functors)}, \cite{Andersen-Ueno(Reshetikhin-Turaev-from-CFT)}, whose proof of unitarity for the $SU(n)$ WZW spaces of conformal blocks relies on a rather non-trivial identification (in~\cite{Andersen-Ueno(Reshetikhin-Turaev-from-CFT)}) of the $SU(n)$ WZW modular functor with a certain modular functor 
constructed by Blanchet~\cite{Blanchet(Hecke-algebras-modular-categories-and-3-manifolds-quantum-invariants)} using skein theory.
\end{remark}

Dealing with projective spaces tends to be notationally cumbersome.
We therefore prefer to state the above result by saying that for every topological surface $\Sigma$, we have a finite dimensional Hilbert space $V(\Sigma)$ that is ``well defined up to a phase''
(see Section \ref{sec: ``up to non-canonical isomorphism''} for a detailed discussion of this notion),
and that the maps $f_*:V(\Sigma_1)\to V(\Sigma_2)$ are well defined up to phase.

\subsection*{Conformal nets}
\addtocontents{toc}{\SkipTocEntry}
In our setup \cite{BDH(nets)}, a conformal net is a functor $\cala$ from a certain category of intervals (that is, 1-manifolds diffeomorphic to $[0,1]$) to the category of von Neumann algebras.
It sends orientation preserving embeddings to homomorphisms and orientation reversing embeddings to antihomomorphisms. 
We show in Theorem \ref{thm: extend A to hatA} that every conformal net admits a canonical extension to the larger category of all compact 1-manifolds.
In particular, given a closed 1-manifold $S$ (a disjoint union of circles), there is an associated von Neumann algebra $\cala(S)$.
If the conformal net has finite index, then $\cala(S)$ is a finite direct sum of type $I$ factors.
Moreover, if $S^1$ is the standard circle then, by Theorem \ref{thm:compute-bfB(net)}, there is a canonical isomorphism
\begin{equation}\label{eq: algebra A(S)}
\cala(S^1)\cong \bigoplus_\lambda \bfB(H_\lambda),
\end{equation}
where $H_\lambda$ are the irreducible representations (also called sectors) of the conformal net.  (Note that this isomorphism is similar to the decomposition of the reduced universal $C^\ast$-algebra in Carpi--Conti--Heillier--Weiner~\cite{cchw}.) 
For a disjoint union of circles $S=S_1\sqcup \ldots \sqcup S_n$, we have $\cala(S)\cong \cala(S_1)\,\bar\otimes\, \ldots \,\bar\otimes\,\cala(S_n)$.
The irreducible summands of $\cala(S)$ are indexed by the set of labelings of $\pi_0(S)$ by isomorphism classes of irreducible representations of $\cala$.

\subsection*{Topological surfaces with smooth boundary}
\addtocontents{toc}{\SkipTocEntry}
If $\Sigma$ is a compact oriented topological surface with boundary, and $\partial\Sigma$ is equipped with a smooth structure,
then our construction produces an infinite dimensional Hilbert-space-well-defined-up-to-phase $V(\Sigma)$, equipped with an action of the von Neumann algebra $\cala(\partial \Sigma)$.
The relationship between $V(\Sigma)$ and the finite dimensional vector spaces \eqref{eq: VOA conf blocks} is as follows, at least conjecturally.
Pick a complex structure on $\Sigma$ and a parametrization of each boundary component by the standard circle.
By gluing a copy of the standard disk on each boundary component of $\Sigma$, we obtain a closed Riemann surface denoted $\widehat \Sigma$.
Let $p_i$ be the centers of the discs.
Then we expect a unitary isomorphism 
\begin{equation}\label{V(Sigma) = VOA--conf--blocks}
V(\Sigma)\,\, \cong\!
\raisebox{0cm}{$\displaystyle\bigoplus_{\substack{\text{labelings $\lambda_{1},\ldots,\lambda_m$}\\\text{of $\pi_0(\partial \Sigma)$ by irreps of $\cala$}}}
\!V\big(\widehat\Sigma,p_1,\ldots,p_n;\lambda_1,\ldots,\lambda_n\big)\otimes H_{\lambda_1}\otimes \cdots \otimes H_{\lambda_n}
$},
\end{equation}
canonical up to phase. 
Note that the above Hilbert space $V(\Sigma)$ is quite natural from the point of view of conformal field theory.
For the chiral WZW models with target $G$, 
the right hand side of \eqref{V(Sigma) = VOA--conf--blocks} is the geometric quantization of the moduli space of holomorphic $G_\IC$-bundles over $\Sigma$ trivialized over $\partial \Sigma$.
That statement was proved by Posthuma \cite[eq. just below (3.2)]{Posthuma(The-Heisenberg-group-and-conformal-field-theory)} in the case $G$ is a torus,
and its algebro-geometric analog was shown by Teleman \cite[Thm.~4]{Teleman(Borel-Weil-Bott)} for $G$ semi-simple.

\begin{remark}
If one interprets $V\big(\widehat\Sigma,p_1,\ldots,p_n;\lambda_1,\ldots,\lambda_n\big)$ to be the vertex-algebra-theoretic 
spaces of conformal blocks, then the above equation is of course conjectural,
as the correspondence between conformal nets and vertex algebras has not yet been worked out to any satisfactory level of detail.
On the other hand, if one defines $V\big(\widehat\Sigma,p_1,\ldots,p_n;\lambda_1,\ldots,\lambda_n\big)$ via conformal nets, as in \eqref{eq: cfblk as multiplicity space} below,
then equation \eqref{V(Sigma) = VOA--conf--blocks} becomes trivially true.
\end{remark}

We can now state our main result for the case of surfaces with boundary.
The construction, which is the same as the one for surfaces without boundary, is outlined in Section \ref{sec: The Hilbert space associated to a surface}:

\begin{theorem2*}Let $\cala$ be a conformal net with finite index.
If $\Sigma$ is a topological surface whose boundary is equipped with a smooth structure,
then there is an associated Hilbert-space-well-defined-up-to-phase $V(\Sigma)$ equipped with an action of $\cala(\partial \Sigma)$.
Each irreducible representation of $\cala(\partial \Sigma)$ appears with finite multiplicity in $V(\Sigma)$.

A homeomorphism $f:\Sigma_1\to \Sigma_2$ that is smooth on the boundary induces a map, well defined up to phase $f_*:V(\Sigma_1)\to V(\Sigma_2)$
that is unitary if $f$ is orientation preserving and antiunitary otherwise.\footnote{If $f$ is neither orientation preserving nor orientation reversing (which can happen if the surfaces are disconnected), then $V(f)$ is not defined.}
These maps are compatible with the algebra actions, and satisfy $(f\circ g)_*=f_*\circ g_*$ up to phase.

Furthermore, if two maps from $\Sigma_1$ to $\Sigma_2$ are isotopic by an isotopy that is constant on $\partial \Sigma$,
then the induced maps from $V(\Sigma_1)$ to $V(\Sigma_2)$ are equal up to phase.

The spaces $V(\Sigma)$ that we assign to a surface $\Sigma$ are constrained by the existence of gluing isomorphisms, discussed below in Theorem 2.

\end{theorem2*}

\begin{example*} The Hilbert space associated to a pair-of-pants is given by
\[
V\Big(\,\tikzmath[scale=.4]{
\fill[gray!50] (.9,1.2) circle (.6 and .18);
\filldraw[fill=gray!20]
(-.6,0) arc (-180:0:.6 and .18) arc (180:0:.3 and .18) arc (-180:0:.6 and .18)
[rounded corners=3]-- ++(0,.3) -- ++(-.9,.6) [sharp corners]-- ++(0,.3) 
arc (0:-180:.6 and .18)
[rounded corners=3]-- ++(0,-.3) -- ++(-.9,-.6) [sharp corners]-- cycle;
\draw[dashed] (0,0) +(0:.6 and .18) arc (0:180:.6 and .18);
\draw[dashed] (1.8,0) +(0:.6 and .18) arc (0:180:.6 and .18);
\draw (.9,1.2) +(0:.6 and .18) arc (0:180:.6 and .18);
}\,\Big)
 = \bigoplus_{\lambda,\mu,\nu} N_0^{\lambda\mu\nu} H_\lambda\otimes H_\mu\otimes H_\nu
\]
as a representation of $\cala(\partial \tikzmath[scale=.25]{
\useasboundingbox (-.7,-.2) rectangle (2.5,1.1);
\fill[gray!50] (.9,1.2) circle (.6 and .18);
\filldraw[fill=gray!20]
(-.6,0) arc (-180:0:.6 and .18) arc (180:0:.3 and .18) arc (-180:0:.6 and .18)
[rounded corners=1.7]-- ++(0,.3) -- ++(-.9,.6) [sharp corners]-- ++(0,.3) 
arc (0:-180:.6 and .18)
[rounded corners=1.7]-- ++(0,-.3) -- ++(-.9,-.6) [sharp corners]-- cycle;
\draw[dashed] (0,0) +(0:.6 and .18) arc (0:180:.6 and .18);
\draw[dashed] (1.8,0) +(0:.6 and .18) arc (0:180:.6 and .18);
\draw (.9,1.2) +(0:.6 and .18) arc (0:180:.6 and .18);
}\;\!)=\cala(S^1)\,\bar\otimes\,\cala(S^1)\,\bar\otimes\,\cala(S^1)$,
see Proposition~\ref{prop: Conf Blocks 1}. 
Here, the direct sum runs over all triples of irreducible representations,
and $N_0^{\lambda\mu\nu}$ denotes the multiplicity of the unit object inside the fusion product of $H_\lambda$, $H_\mu$, and $H_\nu$.
\end{example*}

As a consequence of the functoriality of the construction $\Sigma\mapsto V(\Sigma)$, for every surface $\Sigma$ and every conformal net $\cala$ with finite index, there is a projective (anti)unitary representations of the following infinite dimensional topological group:
\[
G(\Sigma):=\big\{f\in\mathrm{Homeo}_+(\Sigma)\cup \mathrm{Homeo}_-(\Sigma): f|_{\partial\Sigma}\text{ is smooth}\big\}\big/ \,\text{isotopy rel }\partial\Sigma.
\]
If $\Sigma$ is connected, then the group $G(\Sigma)$ fits into a short exact sequence
\[
1\,\to\, \Gamma(\Sigma)\,\to\, G(\Sigma)\,\to\, D\,\to\, 1,
\]
where $D=\Diff_+(\partial \Sigma)\cup \Diff_-(\partial \Sigma)$, and $\Gamma(\Sigma)$ is the mapping class group of $\Sigma$ relative to its boundary.
Here, the subscripts \raisebox{1pt}{$\scriptstyle +/-$} refer to orientation preserving/reversing maps, respectively.
The above action of $G(\Sigma)$ on $V(\Sigma)$ has already been pointed out by Posthuma in the case of the chiral CFTs associated to lattices \cite[Thm.~2.11]{Posthuma(The-Heisenberg-group-and-conformal-field-theory)}.

Given the above theorem, we can turn \eqref{V(Sigma) = VOA--conf--blocks} around and define
the spaces of conformal blocks $V\big(\Sigma;\lambda_1,\ldots,\lambda_n\big)$ to be 
the multiplicity space of $H_{\lambda_1}\otimes \cdots \otimes H_{\lambda_n}$ in the Hilbert space $V(\Sigma)$, that is
\begin{equation}\label{eq: cfblk as multiplicity space}
V(\Sigma;\lambda_1,\ldots,\lambda_n):=\hom_{\cala(\partial \Sigma)} \big(H_{\lambda_1}\otimes \ldots \otimes H_{\lambda_n},V(\Sigma)\big).
\end{equation}
Note that for this definition, $\Sigma$ must have parametrized boundary, since otherwise it is not clear how to let $\cala(\partial \Sigma)$ act on $H_{\lambda_1}\otimes \ldots \otimes H_{\lambda_n}$.

\subsection*{Factorization along circles and along intervals}
\addtocontents{toc}{\SkipTocEntry}
The most notable property of the spaces of conformal blocks is factorization \eqref{eq:fact}:

\begin{theorem3*}
Let $\cala$ be a conformal net with finite index.
Let $\Sigma_1$ and $\Sigma_2$ be topological surfaces with smooth boundary, and
let $S$ be a closed $1$-manifold (a disjoint union of circles) equipped with an orientation reversing embedding $S\hookrightarrow \partial \Sigma_1$
and an orientation preserving embedding $S\hookrightarrow \partial \Sigma_2$.
Then there is a unitary isomorphism
\[
V(\Sigma_1\cup_S \Sigma_2) \,\,\cong\,\, V(\Sigma_1)\otimes_{\cala(S)} V(\Sigma_2)
\]
well defined up to phase and compatible with the actions of $\cala(\partial\Sigma_i\!\setminus\! S)$.

Moreover, the above isomorphisms satisfy an obvious associativity diagram.
\end{theorem3*}

\noindent
Here is an example illustrating the above result:\medskip
\[
V\bigg(\tikzmath[scale = .5]{
\useasboundingbox (.25,-1) rectangle (5.75,1);\draw[gray, rounded corners=5, fill=gray!20](.3,.4) -- (1,1) -- (2,1) --  (3,.75) -- (4,1) -- (5,1) -- (5.7,.4) -- (5.7,-.4) -- (5,-1) -- (4,-1) -- (3,-.75) -- (2,-1) -- (1,-1) -- (.3,-.4) -- cycle;\fill[white](1.19,-.01) .. controls (1.5,.25) and (2,.25) .. (2.31,-.01)(1.19,-.01) .. controls (1.5,-.21) and (2,-.21) .. (2.31,-.01)(3.69,-.01) .. controls (4,.25) and (4.5,.25) .. (4.81,-.01)(3.69,-.01) .. controls (4,-.21) and (4.5,-.21) .. (4.81,-.01);\draw[gray] (1,.1) .. controls (1.5,-.25) and (2,-.25) .. (2.5,.1)(1.2,0) .. controls (1.5,.25) and (2,.25) .. (2.3,0)(3.5,.1) .. controls (4,-.25) and (4.5,-.25) .. (5,.1)(3.7,0) .. controls (4,.25) and (4.5,.25) .. (4.8,0);} 
\bigg)\,\cong\,V\bigg(\tikzmath[scale = .5]{\useasboundingbox (.2,-1) rectangle (2,1);\draw[gray, rounded corners=5, fill=gray!20](.3,.4) -- (1,1) -- (2,1) --  (3,.75) -- (4,1) -- (5,1) -- (5.7,.4) -- (5.7,-.4) -- (5,-1) -- (4,-1) -- (3,-.75) -- (2,-1) -- (1,-1) -- (.3,-.4) -- cycle;\fill[white](1.19,-.01) .. controls (1.5,.25) and (2,.25) .. (2.31,-.01)(1.19,-.01) .. controls (1.5,-.21) and (2,-.21) .. (2.31,-.01)
(3.69,-.01) .. controls (4,.25) and (4.5,.25) .. (4.81,-.01)(3.69,-.01) .. controls (4,-.21) and (4.5,-.21) .. (4.81,-.01);\draw[gray] (1,.1) .. controls (1.5,-.25) and (2,-.25) .. (2.5,.1)(1.2,0) .. controls (1.5,.25) and (2,.25) .. (2.3,0)(3.5,.1) .. controls (4,-.25) and (4.5,-.25) .. (5,.1)(3.7,0) .. controls (4,.25) and (4.5,.25) .. (4.8,0);\fill[white] (1.7,-1.1) rectangle (5.8,1.1);\filldraw[fill = gray!50] (1.7,.592) ellipse(.09 and .408);\filldraw[fill = gray!50] (1.7,-.58) ellipse(.09 and .42);} 
\bigg)\otimes_{\cala\big(\,\,\tikzmath[scale=.3]{\draw (4.2,.59) ellipse(.1  and .41);\draw (4.2,-.58) ellipse(.09  and .41);} 
\,\,\big)}V\bigg(\tikzmath[scale = .5]{\useasboundingbox (1.3,-1) rectangle (5.8,1);\clip (1.58,-1.1) rectangle (5.8,1.1);\draw[gray, rounded corners=5, fill=gray!20](.3,.4) -- (1,1) -- (2,1) --  (3,.75) -- (4,1) -- (5,1) -- (5.7,.4) -- (5.7,-.4) -- (5,-1) -- (4,-1) -- (3,-.75) -- (2,-1) -- (1,-1) -- (.3,-.4) -- cycle;\fill[white](1.19,-.01) .. controls (1.5,.25) and (2,.25) .. (2.31,-.01)(1.19,-.01) .. controls (1.5,-.21) and (2,-.21) .. (2.31,-.01)(3.69,-.01) .. controls (4,.25) and (4.5,.25) .. (4.81,-.01)(3.69,-.01) .. controls (4,-.21) and (4.5,-.21) .. (4.81,-.01);\draw[gray] (1,.1) .. controls (1.5,-.25) and (2,-.25) .. (2.5,.1)(1.2,0) .. controls (1.5,.25) and (2,.25) .. (2.3,0)(3.5,.1) .. controls (4,-.25) and (4.5,-.25) .. (5,.1)(3.7,0) .. controls (4,.25) and (4.5,.25) .. (4.8,0);\fill[white] (.5,-1.1) rectangle (1.7,1.1);\filldraw[fill = gray!50] (1.7,.592) ellipse(.09 and .408);\filldraw[fill = gray!50] (1.7,-.58) ellipse(.09 and .42);} 
\bigg).\medskip
\]

Given the above result, the factorization formula
\[
V\big(\Sigma_1\cup_S\Sigma_2;\mu_1,\ldots\mu_m,\nu_1,\ldots\nu_n\big) \cong {\bigoplus}_\lambda V\big(\Sigma_1;\mu_1,\ldots\mu_m,\lambda\big)\otimes
V\big(\Sigma_2;\bar\lambda,\nu_1,\ldots\nu_n\big)
\]
follows easily from the definition of the spaces of conformal blocks (\ref{eq: cfblk as multiplicity space}) and from the computation \eqref{eq: algebra A(S)} of the algebra $\cala(S)$. 

The most novel aspect of our work, impossible to even formulate in the vertex algebraic setup, is a variant of factorization where circles are replaced by intervals.
Namely, in the above theorem, we can relax the condition that $S$ be a closed manifold, and allow it to be a manifold with boundary.
The algebras associated to 1-manifolds with boundary are no longer type $I$ von Neumann algebras, and so
the formulation of our result requires the use of a more elaborate notion of tensor product: the so-called \emph{relative tensor product} or \emph{Connes fusion} (\cite{BDH(dualizability)}, \cite{Bisch(Bimodules)}, \cite{Connes(Non-commutative-geometry)}, \cite{Sauvageot(Sur-le-produit-tensoriel-relatif)}, \cite{Wassermann(Operator-algebras-and-conformal-field-theory)}), denoted by the symbol~$\boxtimes$.

\begin{theorem4*}
Let $\cala$ be a conformal net with finite index.
Let $\Sigma_1$ and $\Sigma_2$ be topological surfaces with smooth boundary.
Let $M$ be a compact $1$-manifold (possibly disconnected) with boundary, equipped with two
embeddings $M\hookrightarrow \partial \Sigma_1$ and $M\hookrightarrow \partial \Sigma_2$,
the first one orientation reversing and the second orientation preserving.
Equip the boundary of $\Sigma_3:=\Sigma_1\cup_M \Sigma_2$ with a smooth structure such that
the smooth structures on $\partial\Sigma_i$, $i\in\{1,2,3\}$ assemble to a ``smooth structure'' on $\partial\Sigma_1\cup_M \partial\Sigma_2$
in the sense of Definition \ref{def: smooth structure on trivalent graph}.

Then there is a unitary isomorphism
\begin{equation}\label{eq: gluing iso intro}
V(\Sigma_1\cup_M \Sigma_2) \,\,\cong\,\, V(\Sigma_1)\boxtimes_{\cala(M)} V(\Sigma_2),
\end{equation}
canonical up to phase
and compatible\vspace{-.03cm} with actions of $\cala(\partial\Sigma_1\!\setminus\! \put(3.9,7.5){$\scriptscriptstyle \circ$} M)$
and $\cala(\partial\Sigma_2\!\setminus\! \put(3.9,7.5){$\scriptscriptstyle \circ$} M)$, where $\put(3.9,7.5){$\scriptscriptstyle \circ$} M$ denotes the interior of $M$.
Moreover, the above isomorphisms satisfy an obvious notion of associativity.

Finally, the assignment $\Sigma\mapsto V(\Sigma)$ from Theorem 1 is determined by the existence of the gluing law \eqref{eq: gluing iso intro}, and by the requirement that the Hilbert space associated to a disc should be the vacuum sector of $\cala$.
\end{theorem4*}

\noindent
For example, the following is an instance of the above isomorphism:\smallskip
\[
\def\x{.1}V\bigg(\tikzmath[scale = .5]{\useasboundingbox (.25-\x,-1) rectangle (5.75+\x,1);\draw[gray, rounded corners=5, fill=gray!20](.3-\x,.4) -- (1-\x,1) -- (2-\x,1) --  (3,.75) -- (4+\x,1) -- (5+\x,1) -- (5.7+\x,.4) -- (5.7+\x,-.4) -- (5+\x,-1) -- (4+\x,-1) -- (3,-.75) -- (2-\x,-1) -- (1-\x,-1) -- (.3-\x,-.4) -- cycle;\fill[white](1.19-\x,-.01) .. controls (1.5-\x,.25) and (2-\x,.25) .. (2.31-\x,-.01)(1.19-\x,-.01) .. controls (1.5-\x,-.21) and (2-\x,-.21) .. (2.31-\x,-.01)(3.69+\x,-.01) .. controls (4+\x,.25) and (4.5+\x,.25) .. (4.81+\x,-.01)(3.69+\x,-.01) .. controls (4+\x,-.21) and (4.5+\x,-.21) .. (4.81+\x,-.01);
\draw[gray] (1-\x,.1) .. controls (1.5-\x,-.25) and (2-\x,-.25) .. (2.5-\x,.1)(1.2-\x,0) .. controls (1.5-\x,.25) and (2-\x,.25) .. (2.3-\x,0)(3.5+\x,.1) .. controls (4+\x,-.25) and (4.5+\x,-.25) .. (5+\x,.1)(3.7+\x,0) .. controls (4+\x,.25) and (4.5+\x,.25) .. (4.8+\x,0);\fill[white] (2.75,.3) rectangle (3.25,1.1);\draw[fill=gray!50](2.6,0.25) to[rounded corners=2] (2.9,-.31) to[rounded corners=0] (3,-.45) to[rounded corners=2] (3.1,-.31) to[rounded corners=5] (3.4,0.25) to[rounded corners=5] (3.25,1.007) to[rounded corners=1] (3.06,.45)  to[rounded corners=0] (3,.35)
to[rounded corners=1] (2.94,.45) to[rounded corners=5] (2.75,1.007) --  cycle;} 
\bigg)\,\cong\,V\bigg(\tikzmath[scale = .5]{\useasboundingbox (.25-\x,-1) rectangle (3.2,1);\draw[gray, rounded corners=5, fill=gray!20](.3-\x,.4) -- (1-\x,1) -- (2-\x,1) --  (3,.75) -- (4+\x,1) -- (5+\x,1) -- (5.7+\x,.4) -- (5.7+\x,-.4) -- (5+\x,-1) -- (4+\x,-1) -- (3,-.75) -- (2-\x,-1) -- (1-\x,-1) -- (.3-\x,-.4) -- cycle;\fill[white](1.19-\x,-.01) .. controls (1.5-\x,.25) and (2-\x,.25) .. (2.31-\x,-.01)(1.19-\x,-.01) .. controls (1.5-\x,-.21) and (2-\x,-.21) .. (2.31-\x,-.01)(3.69+\x,-.01) .. controls (4+\x,.25) and (4.5+\x,.25) .. (4.81+\x,-.01)(3.69+\x,-.01) .. controls (4+\x,-.21) and (4.5+\x,-.21) .. (4.81+\x,-.01);\draw[gray] (1-\x,.1) .. controls (1.5-\x,-.25) and (2-\x,-.25) .. (2.5-\x,.1)(1.2-\x,0) .. controls (1.5-\x,.25) and (2-\x,.25) .. (2.3-\x,0)(3.5+\x,.1) .. controls (4+\x,-.25) and (4.5+\x,-.25) .. (5+\x,.1)(3.7+\x,0) .. controls (4+\x,.25) and (4.5+\x,.25) .. (4.8+\x,0);\fill[white] (2.75,.3) rectangle (6,1.1) (2.959,.-1.1) rectangle (6,.4);\draw[fill=gray!50](3.04,.37) to[rounded corners=1] (2.98,.45) to[rounded corners=5] (2.75,1.005) --  (2.6,0.3) to[rounded corners=2] (2.8,-.24) to[sharp corners] (2.9,-.4)arc (-142.7:9.8:.079 and 1);} 
\bigg)\,\boxtimes_{\cala\big(\hspace{.09cm}\tikzmath[scale=.3]{\useasboundingbox (0,0) -- +(0,0.9);\draw (.08,.2) arc (180+142.7:180-9.8:.09 and 1);} 
\,\,\big)}V\bigg(\tikzmath[scale = .5]{\useasboundingbox (2.8,-1) rectangle (5.75+\x,1);\clip (2.9,-1.1) rectangle (5.8+\x,1.1);\draw[gray, rounded corners=5, fill=gray!20](.3-\x,.4) -- (1-\x,1) -- (2-\x,1) --  (3,.75) -- (4+\x,1) -- (5+\x,1) -- (5.7+\x,.4) -- (5.7+\x,-.4) -- (5+\x,-1) -- (4+\x,-1) -- (3,-.75) -- (2-\x,-1) -- (1-\x,-1) -- (.3-\x,-.4) -- cycle;\fill[white](1.19-\x,-.01) .. controls (1.5-\x,.25) and (2-\x,.25) .. (2.31-\x,-.01)(1.19-\x,-.01) .. controls (1.5-\x,-.21) and (2-\x,-.21) .. (2.31-\x,-.01)(3.69+\x,-.01) .. controls (4+\x,.25) and (4.5+\x,.25) .. (4.81+\x,-.01)(3.69+\x,-.01) .. controls (4+\x,-.21) and (4.5+\x,-.21) .. (4.81+\x,-.01);\draw[gray] (1-\x,.1) .. controls (1.5-\x,-.25) and (2-\x,-.25) .. (2.5-\x,.1)(1.2-\x,0) .. controls (1.5-\x,.25) and (2-\x,.25) .. (2.3-\x,0)(3.5+\x,.1) .. controls (4+\x,-.25) and (4.5+\x,-.25) .. (5+\x,.1)(3.7+\x,0) .. controls (4+\x,.25) and (4.5+\x,.25) .. (4.8+\x,0);\fill[white] (3.25,.3) rectangle (2.75,1.1) (3.041,.-1.1) rectangle (2.75,.4);\draw[fill=gray!50](2.96,.37) to[rounded corners=1] (3.02,.45) to[rounded corners=5] (3.25,1.005) --  (3.4,0.3) to[rounded corners=2] (3.2,-.24) to[sharp corners] (3.1,-.4)arc (180+142.7:180-9.8:.079 and 1);} 
\bigg).\medskip
\]

Note that for von Neumann algebras that are direct sums of type $I$ factors, Connes fusion agrees with (the Hilbert space completion of) the usual algebraic tensor product.
The previous theorem, about gluing along closed 1-manifolds, is therefore a special case of this last result.

\subsection*{Modularity}
\addtocontents{toc}{\SkipTocEntry}
In the last section of our paper, we use our technology to revisit some aspects of the representation theory of conformal nets.
We reinterpret the monoidal structure on $\mathrm{Rep}(\cala)$ as the operation of tensoring 
over $\cala(\,\tikzmath[scale=.3]{\useasboundingbox (-.6,-.18) rectangle (2.4,.6);\draw (0,0) circle (.6 and .18);\draw (1.8,0) circle (.6 and .18);}\,)$
with the Hilbert space
$V\big(\,\tikzmath[scale=.3]{\fill[gray!50] (.9,1.2) circle (.6 and .18);\filldraw[fill=gray!20](-.6,0) arc (-180:0:.6 and .18) arc (180:0:.3 and .18) arc (-180:0:.6 and .18)[rounded corners=2]-- ++(0,.3) -- ++(-.9,.6) [sharp corners]-- ++(0,.3) arc (0:-180:.6 and .18)[rounded corners=2]-- ++(0,-.3) -- ++(-.9,-.6) [sharp corners]-- cycle;
\draw[dashed] (0,0) +(0:.6 and .18) arc (0:180:.6 and .18);\draw[dashed] (1.8,0) +(0:.6 and .18) arc (0:180:.6 and .18);\draw (.9,1.2) +(0:.6 and .18) arc (0:180:.6 and .18);}\,\big)$;
we also define the braiding of two representations via the action of the homeomorphism $\beta:
\tikzmath[scale=.25]{\fill[gray!50] (.9,1.2) circle (.6 and .18);\filldraw[fill=gray!20](-.6,0) arc (-180:0:.6 and .18) arc (180:0:.3 and .18) arc (-180:0:.6 and .18)[rounded corners=1.7]-- ++(0,.3) -- ++(-.9,.6) [sharp corners]-- ++(0,.3) arc (0:-180:.6 and .18)[rounded corners=1.7]-- ++(0,-.3) -- ++(-.9,-.6) [sharp corners]-- cycle;
\draw[dashed] (0,0) +(0:.6 and .18) arc (0:180:.6 and .18);\draw[dashed] (1.8,0) +(0:.6 and .18) arc (0:180:.6 and .18);\draw (.9,1.2) +(0:.6 and .18) arc (0:180:.6 and .18);}
\to
\tikzmath[scale=.25]{\fill[gray!50] (.9,1.2) circle (.6 and .18);\filldraw[fill=gray!20](-.6,0) arc (-180:0:.6 and .18) arc (180:0:.3 and .18) arc (-180:0:.6 and .18)[rounded corners=1.7]-- ++(0,.3) -- ++(-.9,.6) [sharp corners]-- ++(0,.3) arc (0:-180:.6 and .18)[rounded corners=1.7]-- ++(0,-.3) -- ++(-.9,-.6) [sharp corners]-- cycle;
\draw[dashed] (0,0) +(0:.6 and .18) arc (0:180:.6 and .18);\draw[dashed] (1.8,0) +(0:.6 and .18) arc (0:180:.6 and .18);\draw (.9,1.2) +(0:.6 and .18) arc (0:180:.6 and .18);}
$
that exchanges the two legs of the pair of pants:
the braiding of $H$ and $K$ is given by\smallskip
\[
\big(H\otimes K\big)\otimes_{\cala(\,\tikzmath[scale=.3]{\useasboundingbox (-.6,-.18) rectangle (2.4,.6);\draw (0,0) circle (.6 and .18);
\draw (1.8,0) circle (.6 and .18);
}\,)}
V\Big(\,\tikzmath[scale=.4]{
\fill[gray!50] (.9,1.2) circle (.6 and .18);
\filldraw[fill=gray!20]
(-.6,0) arc (-180:0:.6 and .18) arc (180:0:.3 and .18) arc (-180:0:.6 and .18)
[rounded corners=3]-- ++(0,.3) -- ++(-.9,.6) [sharp corners]-- ++(0,.3) 
arc (0:-180:.6 and .18)
[rounded corners=3]-- ++(0,-.3) -- ++(-.9,-.6) [sharp corners]-- cycle;
\draw[dashed] (0,0) +(0:.6 and .18) arc (0:180:.6 and .18);
\draw[dashed] (1.8,0) +(0:.6 and .18) arc (0:180:.6 and .18);
\draw (.9,1.2) +(0:.6 and .18) arc (0:180:.6 and .18);
}\,\Big)
\xrightarrow{\,\tau\,\otimes\,\beta_*\,}
\big(K\otimes H\big)\otimes_{\cala(\,\tikzmath[scale=.3]{\useasboundingbox (-.6,-.18) rectangle (2.4,.6);\draw (0,0) circle (.6 and .18);
\draw (1.8,0) circle (.6 and .18);
}\,)}V\Big(\,\tikzmath[scale=.4]{
\fill[gray!50] (.9,1.2) circle (.6 and .18);
\filldraw[fill=gray!20]
(-.6,0) arc (-180:0:.6 and .18) arc (180:0:.3 and .18) arc (-180:0:.6 and .18)
[rounded corners=3]-- ++(0,.3) -- ++(-.9,.6) [sharp corners]-- ++(0,.3) 
arc (0:-180:.6 and .18)
[rounded corners=3]-- ++(0,-.3) -- ++(-.9,-.6) [sharp corners]-- cycle;
\draw[dashed] (0,0) +(0:.6 and .18) arc (0:180:.6 and .18);
\draw[dashed] (1.8,0) +(0:.6 and .18) arc (0:180:.6 and .18);
\draw (.9,1.2) +(0:.6 and .18) arc (0:180:.6 and .18);
}\,\Big),
\smallskip
\]
where $\tau$ is the permutation operator.

Finally, we provide an alternative proof of a famous result of Kawahigashi--Longo--M\"uger~\cite{Kawahigashi-Longo-Mueger(2001multi-interval)} about the modularity of the representation category of conformal nets:

\begin{theorem*}
If $\cala$ is a conformal net with finite index, then its category of representations is a modular tensor category.
\end{theorem*}

Our proof is based on the property of factorization along intervals.
Here, the definition of modularity that we use to verify the above statement is that there are no transparent objects in $\mathrm{Rep}(\cala)$ aside from the unit object and its multiples 
(that is, no objects $T$ such that $\raisebox{5pt}{
$\!\!\tikzmath[scale=.5]{\draw (1,0) -- (0,1)node[above]{$\scriptstyle X$} (0,0) -- (.4,.4)(.6,.6)--(1,1)node[above]{$\scriptstyle T$};}$}
\!=\!
\raisebox{5pt}{$\tikzmath[scale=.5]{\draw (1,0) --(.6,.4)(.4,.6) -- (0,1)node[above]{$\scriptstyle X$} (0,0) --(1,1)node[above]{$\scriptstyle T$};}$}
\!$ \smallskip for every $X$ in the category).

\subsection*{Acknowledgements} 
\addtocontents{toc}{\SkipTocEntry}
We thank Stefan Stolz and Peter Teichner for their continuous support throughout this project.
We would also like to thank Yi-Zhi Huang, J\o rgen E. Andersen, and Michael M\"uger for their help with references.
The last author thanks Roberto Longo and Sebastiano Carpi for a pleasant stay in Rome, during which he was able to present this work.

\section{Conformal nets} \label{sec:nets}

In this paper, all $1$-manifolds are compact, smooth, and oriented.
The standard circle $S^1:=\{z\in \IC:|z|=1\}$ is the set of complex numbers of modulus one, equipped with the counter-clockwise orientation.
By a \emph{circle}, we shall mean a manifold that is diffeomorphic to the standard circle.
Similarly, by an \emph{interval}, we shall mean a manifold that is diffeomorphic to the standard interval $[0,1]$. 
For a 1-manifold $I$, we denote by $\bar I$ the same manifold equipped with the opposite orientation.
We denote by $\Diff(I)$ the group of diffeomorphisms of $I$, and by $\Diff_+(I)$ the subgroup of orientation preserving diffeomorphisms.
Given an interval $I$, we also let $\Diff_0(I)$ be the group of diffeomorphisms that restrict to the identity near the boundary of $I$.
We let $\INT$ denote the category whose objects are intervals and whose morphisms are embeddings, not necessarily orientation-preserving.

Let $\VN$ be the category whose objects are von Neumann algebras with separable preduals, and whose morphisms are $\IC$-linear $\ast$-homomorphisms, and $\IC$-linear $\ast$-antihomomorphisms.
A \emph{net} is a covariant functor $\cala \colon \INT \to \VN$ taking 
orientation-preserving embeddings to injective homomorphisms and orientation-reversing embeddings to injective antihomomorphisms.
It is said to be \emph{continuous} if the natural map $\mathrm{Hom}_{\INT}(I,J)\to \mathrm{Hom}_{\VN}(\cala(I),\cala(J))$ is continuous for the $\mathcal C^\infty$ topology on the source and Haagerup's $u$-topology on the target.
Given a subinterval $I \subseteq K$, we will often not distinguish between $\cala(I)$ and its image in $\cala(K)$.

\begin{definition}
\label{def:conformal-net}
A \emph{conformal net} is a continuous net subject to the following conditions.
Here, $I$ and $J$ are subintervals of some interval $K$. 
\begin{enumerate}
\item \emph{Locality:} If $I,J\subset K$ have disjoint interiors, then $\cala(I)$ and $\cala(J)$ are commuting subalgebras of $\cala(K)$.
\item \emph{Strong additivity:} If $K = I \cup J$, then $\cala(K) = \cala(I) \vee \cala(J)$.
\item \emph{Split property:} If $I,J\subset K$ are disjoint, then the map from the algebraic tensor product $\cala(I) \ox_{\alg} \cala(J) \to \cala(K)$ extends to the spatial tensor product
\[
\cala(I) \, \bar{\ox} \, \cala(J) \to \cala(K).
\]
\item \emph{Inner covariance:} If $\varphi\in\Diff_0(I)$, then $\cala(\varphi)$ is an inner automorphism.
\item \label{def:conformal-net:vacuum} 
\emph{Vacuum sector:} 
Suppose that $J \subsetneq I$ contains the boundary point $p \in \dd I$.
The algebra $\cala(J)$ acts on $L^2(\cala(I))$ via the left action of $\cala(I)$, and $\cala(\bar{J})\cong \cala(J)^\op$ acts on $L^2(\cala(I))$ via the right action of $\cala(I)$.
In that case, we require that the action of $\cala(J) \ox_{\alg} \cala( \bar{J} )$ on $L^2(\cala(I))$ extends to $\cala(J \cup_p \bar{J})$,
\begin{equation}\label{eq: Vaccum sector axiom for nets}
\qquad\tikzmath{
\matrix [matrix of math nodes,column sep=1cm,row sep=5mm]
{ 
|(a)| \cala(J) \ox_{\alg} \cala( \bar{J} ) \pgfmatrixnextcell |(b)| \bfB(L^2\cala(I))\\ 
|(c)| \cala(J \cup_p \bar{J}) \\ 
}; 
\draw[->] (a) -- (b);
\draw[->] (a) -- (c);
\draw[->,dashed] (c) -- (b);
}
\end{equation}
where $J \cup_p \bar{J}$ is equipped with any smooth structure that is compatible with the one on $J$,
and for which the involution that exchanges $J$ and $\bar J$ is smooth.
Here, one should picture the interval $J \cup_p \bar{J}$ as a submanifold of the double $I\cup_{\partial I}\bar I$ of $I$:
$\,\,
\tikzmath[scale=.06]
{\useasboundingbox (-20,-18) rectangle (20,18);
\draw[->] (-.7,15) -- (-.8,15);\draw[->] (.7,-15) -- (.8,-15);\draw (0,0) circle (15)(-17,0) -- (-13,0) (13.5,0) -- (17.5,0)
(16,0) arc (0:55:16) (16,0) arc (0:-55:16)(-6,9) node {$I$} (19.5,0) node {$p$} 
(17,9) node {$J$} (17,-9) node {$\bar{J}$};}$
\end{enumerate}
\end{definition}

\noindent Throughout this paper, we shall assume that our conformal nets are \emph{irreducible}, i.e., that the algebras $\cala(I)$ are factors.

Recall from \cite[Section 1.B]{BDH(nets)} that if $S$ is a circle, a Hilbert space is called an \emph{$S$-sector} of $\cala$
if it has compatible actions of $\cala(I)$ for every interval $I\subset S$.
The category of $S$-sectors is denoted $\Rep_S(\cala)$.
It is equipped with a distinguished object $H_0(S)=H_0(S,\cala)$, well defined up to non-canonical isomorphism, called the \emph{vacuum sector} of $\cala$ on $S$.
The vacuum sector is defined as follows.
If $j\in\Diff_-(S)$ is an orientation reversing involution that fixes the boundary of some interval $I\subset S$, 
then we set $H_0(S):=L^2(\cala(I))$.
It is equipped with:
\begin{list}{$\bullet$}{\leftmargin=0ex \rightmargin=0ex \labelsep=1ex \labelwidth=-1ex \itemsep=.1ex \topsep=.7ex}
\item for all $J\subset I$, an action $\cala(J)\hookrightarrow \cala(I)\xrightarrow{\text{left action}} \bfB(L^2(\cala(I)))$,
\item for all $J\subset j(I)$, an action $\cala(J)\hookrightarrow \cala(j(I))\xrightarrow{\!\cala(j)\!}\cala(I)^\op \xrightarrow{\!\text{right action}\!} \bfB(L^2(\cala(I)))$,
\end{list}\vspace{.7ex}
and those extend uniquely to the structure of an $S$-sector.
When we need to be specific, we will refer to the above construction of $H_0(S)$ as
the \emph{vacuum sector of $\cala$ associated to $S$, $I$, and $j$}.

More generally \cite[Section 3.C]{BDH(nets)}, given a closed 1-manifold $M$, we call a Hilbert space an $M$-sector of $\cala$
if it comes with compatible actions of $\cala(I)$ for every interval $I\subset M$.
The category of $M$-sectors is denoted $\Rep_M(\cala)$.

\subsection{Extending conformal nets to arbitrary 1-manifolds}\label{sec: Extend to 1-manifolds}
Let $\mathsf{1MAN}$ be the category whose objects are compact oriented 1-manifolds, possibly disconnected,
and whose morphisms are embeddings that are either orientation preserving, or orientation reversing.
The goal of this section is to show that every conformal net has a canonical extension to the larger category $\mathsf{1MAN}$.

\begin{theorem}\label{thm: extend A to hatA}
Every conformal net $\cala:\INT\to \VN$ has a canonical extension $\hat\cala:\mathsf{1MAN}\to \VN$.
That extension is symmetric monoidal\,\footnote{\label{footnote: BZ/2}Note that the symmetric monoidal structure on $\mathsf{1MAN}$ is only partially defined as one cannot take the disjoint union of an orientation preserving map with an orientation reversing one.
One way of making precise the sense in which $\hat \cala$ is symmetric monoidal is to 
let $B\IZ/2$ be the category with one object and $\IZ/2$ as automorphisms, then
note that both $\mathsf{1MAN}$ and $\VN$ are equipped with a functor to $B\IZ/2$ and that they are symmetric monoidal as categories over $B\IZ/2$.
The functor $\hat \cala$ is then symmetric monoidal over $B\IZ/2$.
} in the sense that it takes disjoint unions of 1-manifolds to spatial tensor products of von Neumann algebras.
\end{theorem}

\noindent The proof of this result will occupy this section.
We first present some useful technical definitions.

By a \emph{$Y$-graph} we mean any topological space that is homeomorphic to $\{z\in\IC:z^3\in [0,1]\}$,
and by a \emph{trivalent graph} any compact Hausdorff space that is locally homeomorphic to
a $Y$-graph.

\begin{definition}\label{def: smooth structure on trivalent graph}
A \emph{smooth structure} on a trivalent graph $\Gamma$ is the data of a smooth structure on every interval $I\subset \Gamma$, subject to the conditions:
\\({\it i}\hspace{.4mm}) Whenever $I_1\subset I_2\subset \Gamma$, the smooth structure on $I_1$ is the one inherited from~$I_2$.
\\({\it ii}\hspace{.4mm}) For every $Y$-graph $Y\subset \Gamma$, there exists a faithful action of the symmetric group $\mathfrak S_3\to\mathrm{Homeo}(Y)$,
such that for all $g\in \mathfrak S_3$ and $I\subset Y$, the corresponding map $I\to gI$ is smooth.
\end{definition}

Locally around a trivalent point, a smooth structure on
\(\tikzmath[scale=.3]{\useasboundingbox (-1,-.8) -- (1,.8);\draw (0,0) -- (-30:1)(0,0) -- (90:.95)(0,0) -- (210:1);;}\)\vspace{-.05cm}
is given by compatible smooth structures on
\(\tikzmath[scale=.3]{\useasboundingbox (-1,-.8) -- (1,.8);\draw (0,0) -- (90:.95)(0,0) -- (210:1);;}\),
\(\tikzmath[scale=.3]{\useasboundingbox (-1,-.8) -- (1.1,.8);\draw (0,0) -- (-30:1)(0,0) -- (210:1);}\), and
\(\tikzmath[scale=.3]{\useasboundingbox (-1,-.8) -- (1.1,.8);\draw (0,0) -- (-30:1)(0,0) -- (90:.95);;}\).
The compatibility condition says that there should exist a faithful action of $\mathfrak S_3$
that induces smooth maps when restricted to each one of the above pieces.\vspace{-.15cm}

Recall that given a smooth trivalent graph like this:
\tikz[scale=.3]{
\useasboundingbox (-1.15,-.6) rectangle (2.25,1.2); \draw (60:1) arc (60:300:1); \draw (1,0) +(120:1) arc (120:-120:1); \draw (0,0) +(60:1) -- +(300:1);},
with circle subgraphs
\[
\tikzmath[scale=.033]{ \useasboundingbox (-30,-18) rectangle (28,20); \draw[line width=.7] (60:14) arc (60:300:14); \draw[densely dotted] (14,0) +(120:14) arc (120:-120:14);
\draw[line width=.7] (0,0) +(60:14) -- +(300:14); \node at (-27,-1) {$S_1$}; 
\draw[->] (90:14) ++ (-.5,0) -- +(-.1,0); }
\,\,,\qquad 
\tikzmath[scale=.033]{ \useasboundingbox (-30,-18) rectangle (28,20); \draw[densely dotted] (60:14) arc (60:300:14); \draw[line width=.7] (14,0) +(120:14) arc (120:-120:14);
\draw[line width=.7] (0,0) +(60:14) -- +(300:14); \node at (-27,-1) {$S_2$}; 
\draw[->] (14,14) ++ (-.5,0) -- +(-.1,0); }
\,\,,\qquad 
\tikzmath[scale=.033]{ \useasboundingbox (-30,-18) rectangle (28,20); \draw[line width=.7] (60:14) arc (60:300:14); \draw[line width=.7] (14,0) +(120:14) arc (120:-120:14);
\draw[densely dotted] (0,0) +(60:14) -- +(300:14); \node at (-27,-1) {$S_3$}; 
\draw[->] (14,14) ++ (-.5,0) -- +(-.1,0); }
\]
there is a corresponding functor \cite[Section 1.C]{BDH(nets)}
\[
\boxtimes_I:\Rep_{S_1}(\cala)\times \Rep_{S_2}(\cala)\to \Rep_{S_3}(\cala)
\]
given by fusing along $\cala(I)$, where $I=S_1\cap S_2$.
Moreover, there is a non-canonical isomorphism of $S_3$-sectors
\begin{equation}\label{eq: [paper I, Corollary 1.33]}
H_0(S_1) \boxtimes_I H_0(S_2) \,\cong\, H_0(S_3).
\end{equation}

Given a 1-manifold $M$ and a conformal net $\cala$, we now describe the algebra $\hat\cala(M)$.
Pick a cover $\cali=\{I_1,\ldots,I_n\}$ of $M$ consisting of intervals with disjoint interiors.
This induces a cover of the 2-manifold $\Sigma:=M\times [0,1]$ by rectangles $I_i\times [0,1]$.
Orient the circles $S_i:=\partial(I_i\times [0,1])$ so that their orientation agrees with that of $M\times\{0\}$ on $I_i\times\{0\}$.
Set $\Gamma:=\bigcup S_i$.
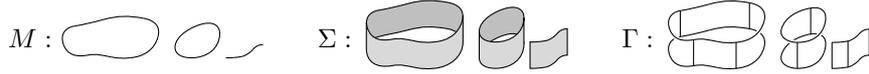
\bigskip
\[
\def\coords{\coordinate (a) at (.4,0.1);\coordinate (ab) at (1.3,0.03);\coordinate (b) at (2,0);\coordinate (bc) at (3.05,.8);\coordinate (c) at (2.5,1.5);\coordinate (cd) at (1,1.45);\coordinate (d) at (-0.1,1.2);\coordinate (da) at (-.5,.5);\coordinate (e) at (4.3,0);\coordinate (ef) at (5.3,.6);\coordinate (f) at (4.8,1.3);\coordinate (fe) at (3.7,.5);\coordinate (g) at (5.6,0);\coordinate (gh) at (5.9,0.01);\coordinate (h) at (6.2,.05);\coordinate (hi) at (6.6,0.3);\coordinate (i) at (7,.5);}
\def\AB{  \draw (a) to[out = 10, in = 170, looseness=1.1] (ab);\draw (ab) to[out = -10, in = 185, looseness=1.1] (b); }
\def\BC{  \draw (b) to[out = 5, in = -105, looseness=1] (bc);\draw (bc) to[out = 75, in = -10, looseness=1.2] (c);}
\def\CD{  \draw (c) to[out = 170, in = 15, looseness=1] (cd); \draw (cd) to[out = 195, in = 10, looseness=1] (d);}
\def\DA{  \draw (d) to[out = 190, in = 110, looseness=1.2] (da);  \draw (da) to[out = -70, in = 190, looseness=1.1] (a);}
\def\EF{  \draw (e) to[out = 0, in = -110, looseness=1] (ef); \draw (ef) to[out = 70, in = 10, looseness=1.5] (f);}
\def\FE{  \draw (f) to[out = 190, in = 75, looseness=.9] (fe); \draw (fe) to[out = -105, in = 180, looseness=1.1] (e);}
\def\GH{\draw (g) to[out = 0, in = 185, looseness=1] (gh); \draw (gh) to[out = 5, in = 190, looseness=1] (h);}
\def\HI {\draw (h) to[out = 10, in = -130, looseness=.7] (hi); \draw (hi) to[out = 50, in = 180, looseness=1] (i);}
M:\tikzmath[scale=.35]{\useasboundingbox (-.65,0.1) rectangle (7.05,1.2); \coords \AB\BC\CD\DA
\draw ($(fe)+(-.01,-.1)$) to[out = -90, in = 180, looseness=1.05] (e) to[out = 0, in = -90, looseness=1.05] ($(ef)+(.07,.32)$)
($(ef)+(.07,.32)$) to[out = 90, in = 10, looseness=1.2] (f) to[out = 190, in = 90, looseness=.8] ($(fe)+(-.01,-.1)$); \GH\HI} 
\qquad\Sigma:\,\tikzmath[scale=.35]
{\useasboundingbox (-.65,0.1) rectangle (7.05,2);
\coords 
\fill[gray!50, even odd rule] (a) to[out = 10, in = 170, looseness=1.1] (ab) to[out = -10, in = 185, looseness=1.1] (b) to[out = 5, in = -105, looseness=1] 
(bc) to[out = 75, in = -10, looseness=1.2] (c) to[out = 170, in = 15, looseness=1] (cd) to[out = 195, in = 10, looseness=1] (d) to[out = 190, in = 110, looseness=1.2] (da) to[out = -70, in = 190, looseness=1.1] (a) -- cycle ($(a)+(0,1)$) to[out = 10, in = 170, looseness=1.1] ($(ab)+(0,1)$) to[out = -10, in = 185, looseness=1.1] ($(b)+(0,1)$) to[out = 5, in = -105, looseness=1] 
($(bc)+(0,1)$) to[out = 75, in = -10, looseness=1.2] ($(c)+(0,1)$) to[out = 170, in = 15, looseness=1] ($(cd)+(0,1)$) to[out = 195, in = 10, looseness=1] ($(d)+(0,1)$) to[out = 190, in = 110, looseness=1.2] ($(da)+(0,1)$) to[out = -70, in = 190, looseness=1.1] ($(a)+(0,1)$) -- cycle;
\draw ($(bc)+(-.15,.53)$) to[out = 145, in = -10, looseness=1] (c) to[out = 170, in = 15, looseness=1] (cd) to[out = 195, in = 20, looseness=1] ($(d)+(-.3,-.1)$);
\fill[gray!30] ($(da)+(-.05,.25)$) to[out = -90, in = 190, looseness=1.2] (a) to[out = 10, in = 170, looseness=1.1] (ab) to[out = -10, in = 185, looseness=1.1] (b) to[out = 5, in = -90, looseness=.99] 
($(bc)+(.03,.2)$) -- ($(bc)+(.03,1.2)$) to[in = 5,out = -90, looseness=.99] ($(b)+(0,1)$) to[in = -10, out = 185, looseness=1.1] ($(ab)+(0,1)$) to[in = 10, out = 170, looseness=1.1] ($(a)+(0,1)$)
to[in = -90, out = 190, looseness=1.2] ($(da)+(-.05,1.25)$); \draw ($(da)+(-.05,1.25)$) -- ($(da)+(-.05,.25)$) to[out = -90, in = 190, looseness=1.2] (a) to[out = 10, in = 170, looseness=1.1] (ab) to[out = -10, in = 185, looseness=1.1] (b) to[out = 5, in = -90, looseness=.99] ($(bc)+(.03,.2)$) -- ($(bc)+(.03,1.2)$);  
\fill[gray!50]($(ef)+(.07,.32)$) to[out = 90, in = 10, looseness=1.2] (f) to[out = 190, in = 90, looseness=.8] ($(fe)+(-.01,-.1)$) -- ($(fe)+(-.01,.9)$) to[in = 190, out = 90, looseness=.8] ($(f)+(0,1)$) to[in = 90, out = 10, looseness=1.2]  ($(ef)+(.07,1.32)$); \draw($(ef)+(.07,.32)$) to[out = 90, in = 10, looseness=1.2] (f) to[out = 190, in = 90, looseness=.8] ($(fe)+(-.01,-.1)$);
\draw($(ef)+(.07,1.32)$) to[out = 90, in = 10, looseness=1.2] ($(f)+(0,1)$) to[out = 190, in = 90, looseness=.8] ($(fe)+(-.01,.9)$);
\filldraw[fill=gray!30] ($(fe)+(-.01,-.1)$) to[out = -90, in = 180, looseness=1.05] (e) to[out = 0, in = -90, looseness=1.05] ($(ef)+(.07,.32)$) -- ($(ef)+(.07,1.32)$) to[in = 0, out = -90, looseness=1.05] ($(e)+(0,1)$) to[in = -90, out = 180, looseness=1.05] ($(fe)+(-.01,.9)$) -- cycle;
\draw[fill=gray!30](g) to[out = 0, in = 185, looseness=1] (gh) to[out = 5, in = 190, looseness=1] (h) to[out = 10, in = -130, looseness=.7] (hi) to[out = 50, in = 180, looseness=1] (i) -- ($(i)+(0,1)$) to[in = 50, out = 180, looseness=1] ($(hi)+(0,1)$) to[in = 10, out = -130, looseness=.7] ($(h)+(0,1)$) to[in = 5, out = 190, looseness=1] ($(gh)+(0,1)$) to[in = 0, out = 185, looseness=1] ($(g)+(0,1)$) -- cycle;
\draw ($(a)+(0,1)$) to[out = 10, in = 170, looseness=1.1] ($(ab)+(0,1)$) to[out = -10, in = 185, looseness=1.1] ($(b)+(0,1)$) to[out = 5, in = -105, looseness=1] ($(bc)+(0,1)$) to[out = 75, in = -10, looseness=1.2] ($(c)+(0,1)$) to[out = 170, in = 15, looseness=1] ($(cd)+(0,1)$) to[out = 195, in = 10, looseness=1] ($(d)+(0,1)$) to[out = 190, in = 110, looseness=1.2] ($(da)+(0,1)$) to[out = -70, in = 190, looseness=1.1] ($(a)+(0,1)$) -- cycle;} 
\qquad\Gamma:\,\tikzmath[scale=.35]
{\useasboundingbox (-.65,0.1) rectangle (7.05,2); \coords \AB\BC\CD\DA\EF\FE\GH\HI \pgftransformyshift{28} \coords  \draw[white, ultra thick]
(da) to[out = -70, in = 190, looseness=1.1] (a) (b) to[out = 5, in = -105, looseness=1] (bc) (e) to[out = 0, in = -110, looseness=1] (ef) (fe) to[out = -105, in = 180, looseness=1.1] (e); 
\AB\BC\CD\DA\EF\FE\GH\HI \foreach \x in {a,b,c,d,e,f,g,h,i} {\draw (\x) -- +(0,-1);}  } 
\medskip
\]
Let us pick local coordinates on $M$ around every point of the finite set $\bigcup_j\partial I_j$.
Using those local coordinates (and the standard coordinates on the copies of $[0,1]$), the trivalent graph $\Gamma$ can then be endowed with a smooth structure in the sense of Definition \ref{def: smooth structure on trivalent graph}.
In particular, the circles $S_i$ have smooth structures and we can talk about their vacuum sectors~$H_0(S_i)$.

Let $p_1,\ldots,p_m\in M$ be the points where two of the intervals $I_i$ touch each other.
If $p_i\in I_j\cap I_k$, then the algebra $A_i:=\cala(\{p_i\}\times [0,1])$ acts on the left on $H_0(S_j)$ and on the right on $H_0(S_k)$
or the other way around, depending on orientations.
We can therefore form the fusion 
\begin{equation}\label{eq: def H-sigma}
\qquad H_\Sigma\,:=\,\bigboxtimes_{\{A_i\}}\big\{H_0(S_j)\big\}_{\scriptscriptstyle 1\le i\le m,\, 1\le j\le n}
\end{equation}
of all the sectors $H_0(S_j)$ along all the algebras $A_i$, as explained in Definition~\ref{def: graph fusion} in the appendix; this Hilbert space $H_\Sigma$ can in fact be well defined up to canonical unitary isomorphism, as established in Theorem~\ref{thm: any cover of the circle}, also in the appendix.
The result is a $\partial \Sigma$-sector 
by~\cite[\coritsapartialSigmasector]{BDH(nets)}.
We then define
\begin{equation}\label{eq: hata(M)}
\hat\cala(M)\,:=\, \bigvee_{\substack{\text{intervals} \\ {}^{\phantom{+}}I\subset M{}^{\phantom{+}}}} \cala(I\times \{0\})
\end{equation}
as a subalgebra of $\bfB(H_\Sigma)$.

In the next two lemmas, we will see that the Hilbert space \eqref{eq: def H-sigma} and the algebra \eqref{eq: hata(M)} are independent of the cover $\cali$,
and of the choice of local coordinates.
For the moment, we shall write $H_\Sigma^{(\cali)}$ and $\hat\cala(M)^{(\cali)}$ to stress the dependence on the cover 
(the dependence on the local coordinates is also there, but left implicit in the notation).
Let us first introduce some terminology to refer to the above structures:

\begin{definition}\label{def: c-cover}
A $c$-cover of a compact 1-manifold $M$ consists of a cover $\cali=\{I_1,\ldots,I_n\}$ by intervals with disjoint interiors,
along with germs of local coordinates at every point of the set $\bigcup_j\partial I_j$.
\end{definition}

Given two $c$-covers $\cali=\{I_1,\ldots,I_n\}$ and $\calj=\{J_1,\ldots,J_m\}$ of $M$,
we shall say that $\calj$ refines $\cali$, and write it $\calj\prec\cali$ if the underlying cover $\{J_1,\ldots,J_m\}$ refines $\{I_1,\ldots,I_n\}$ and the local coordinates of $\cali$ agree 
with those of $\calj$
at every point of $\bigcup_j\partial I_j$.

\begin{lemma}\label{lem: H_Sigma^(cali)= H_Sigma^(calj)}
Let $M$ be a compact $1$-manifold equipped with germs of local coordinates at each point of $\partial M$, and
let $\cali$ and $\calj$ be $c$-covers of $M$ that are compatible with the given local coordinates.
Then there is a non-canonical isomorphism of $\partial \Sigma$-sectors
\[
H_\Sigma^{(\cali)}\to H_\Sigma^{(\calj)}.
\]
\end{lemma}

\begin{proof}
We first treat the case when $\calj$ refines $\cali$.
The $c$-cover $\calj$ can be obtained from $\cali$ by successively subdividing intervals,
and so we may as well assume that $\cali=\{I_1,I_2,\ldots,I_n\}$ and $\calj=\{I_1',I_1'',I_2,\ldots,I_n\}$ with $I_1'\cup I_1''=I_1$.
Let\smallskip
\begin{gather*}
S_1':=\partial(I_1'\times [0,1])\,\,\,\,\quad S_1'':=\partial(I_1''\times [0,1])\,\,\,\,\quad S_j:=\partial(I_j\times [0,1])\\
p_0:=I_1'\cap I_1''\,\,\,\,\quad A_0:=\cala(\{p_0\}\times [0,1]).
\smallskip\end{gather*}
By \eqref{eq: [paper I, Corollary 1.33]}, there is an isomorphism of $S_1$-sectors 
\begin{equation}\label{hoAho=Ho}
H_0(S_1')\boxtimes_{A_0}H_0(S_1'')\cong H_0(S_1).
\end{equation}
It follows that 
\begin{equation}\label{eq: from cali to calj}
\begin{split}
H_\Sigma^{(\calj)}\,=\,\,\,&\bigboxtimes_{\{A_i\}}\Big\{H_0(S_1'),\,H_0(S_1''),\,H_0(S_2),\ldots, H_0(S_n)\Big\}_{0\le i\le m}\\
\cong\,\,\,&
\bigboxtimes_{\{A_i\}}\Big\{H_0(S_1')\boxtimes_{A_0} H_0(S_1''),\,H_0(S_2),\ldots, H_0(S_n)\Big\}_{1\le i\le m}\\
\cong\,\,\,&
\bigboxtimes_{\{A_i\}}\Big\{H_0(S_1),\,H_0(S_2),\ldots, H_0(S_n)\Big\}_{1\le i\le m}=\,\,H_\Sigma^{(\cali)}.
\end{split}
\end{equation}

The general case follows from the above special case by the following observation.
For any two $c$-covers $\cali$, $\cali'$ of $M$, there exist $c$-covers $\calj$, $\cali''$, $\calj'$ of $M$ such $\cali\succ \calj\prec \cali''\succ \calj'\prec \cali'$.
\end{proof}

We then have the following result.

\begin{lemma}\label{lem: two covers that agree at dM}
Given two $c$-covers $\cali$, $\calj$ of $M$ that induce the same local coordinates at the points of $\partial M$, there is a canonical algebra isomorphism
\[
\hat \cala(M)^{(\cali)}\to\hat \cala(M)^{(\calj)}.
\]
\end{lemma}

\begin{proof}
Consider the colimit $A := {\mathrm{colim}}_{I\subset M}\cala(I)$
in the category of $*$-algebras, indexed by the poset of subintervals of $M$.
By definition, $\cala(M)^{(\cali)}$ is the von Neumann algebra generated by (the image of) $A$ in $\bfB\big(H_\Sigma^{(\cali)}\big)$
and, similarly, $\cala(M)^{(\calj)}$ is generated by $A$ in $\bfB\big(H_\Sigma^{(\calj)}\big)$.
By Lemma \ref{lem: H_Sigma^(cali)= H_Sigma^(calj)}, the Hilbert spaces $H_\Sigma^{(\cali)}$ and $H_\Sigma^{(\calj)}$ are isomorphic as representations of $A$.
The two von Neumann algebras are therefore isomorphic, and there is a unique arrow that makes this diagram commute:
\[
\tikzmath{
\node (a) at (0,.8) {$A$};
\node (b) at (-1.5,0) {$\hat \cala(M)^{(\cali)}$};
\node (c) at (1.5,0) {$\hat \cala(M)^{(\calj)}$};
\draw[->] (a) -- (b);
\draw[->] (a) -- (c);
\draw[->,dashed]
(b) -- (c);
}
\qedhere
\]
\end{proof}

As a corollary, we have:

\begin{corollary}
Given two $c$-covers $\cali$ and $\calj$ of $M$ (without any conditions on the local coordinates at the endpoints of $M$), there is a canonical algebra isomorphism
\[
\hat \cala(M)^{(\cali)}\to\hat \cala(M)^{(\calj)}.
\]
\end{corollary}

\begin{proof}
Write $M$ as a disjoint union of connected component $M=M_1\sqcup\ldots\sqcup M_k$.
We then have canonical isomorphisms 
$\hat\cala(M)^{(\cali)}\cong \hat\cala(M_1)^{(\cali_1)}\,\bar\otimes\,\ldots\,\bar\otimes\,\hat\cala(M_k)^{(\cali_k)}$
and
$\hat\cala(M)^{(\calj)}\cong \hat\cala(M_1)^{(\calj_1)}\,\bar\otimes\,\ldots\,\bar\otimes\,\hat\cala(M_k)^{(\calj_k)}$,
where $\cali_i$ and $\calj_i$ are the restrictions of the $c$-covers $\cali$ and $\calj$ to the components $M_i$ of $M$.
It is therefore enough to treat the case when $M$ is connected.
If $M$ is a circle, then Lemma \ref{lem: two covers that agree at dM} yields the result.
Let us therefore assume that $M$ is an interval.

Let $\cali_0=\{M\}$ be the $c$-cover consisting of just $M$, along with the local coordinates at $\partial M$ induced by $\cali$.
Define $\calj_0$ similarly.
Then we have canonical isomorphisms $\hat\cala(M)^{(\cali)}\cong \hat\cala(M)^{(\cali_0)}\cong \cala(M) \cong \hat\cala(M)^{(\calj_0)}\cong \hat\cala(M)^{(\calj)}$.
\end{proof}

By the above results, we can see that $\hat \cala$ defines a functor $\mathsf{1MAN}\to \VN$,
that its restriction to the subcategory $\INT$ recovers $\cala$, and that
it is symmetric monoidal in the sense that $\hat \cala(M\sqcup N) \cong \hat \cala(M)\,\bar\otimes\,\hat\cala(N)$.
This completes the proof of Theorem~\ref{thm: extend A to hatA}.

\begin{remark} 
Given a natural transformation $\tau:\cala\to \calb$, it is natural to ask whether it extends to a natural transformation between functors $\mathsf{1MAN}\to \VN$.
We only know how to prove that it does if the conformal net $\cala$ has finite index (Definition \ref{def:finite-nets}), using Proposition \ref{prop: colim over S}.
\end{remark}

From now on, since there is no risk of confusion, we will drop the hat notation and denote the extension 
of a conformal net $\cala$ to $\mathsf{1MAN}$ simply by $\cala$:
\[
\tikzmath{ \matrix [matrix of math nodes,column sep=1.2cm,row sep=5mm, inner sep=5]
{ |(a)| \INT \pgfmatrixnextcell |(b)| \VN\\  |(c)| \mathsf{1MAN}\\ }; 
\draw[->] (a) --node[above, pos=.48, scale=.9]{$\cala$} (b);
\draw[->] ($(a.south) + (.1,-.05)$) arc (0:180:.05) -- (c.north);
\draw[->,dashed] (c) --node[below,xshift=1, scale=.9]{$\cala$} (b);}
\]

\subsection{The algebra associated to a circle}
Given a circle $S$, note that an $S$-sector of $\cala$ is 
\emph{not} the same thing as an $\cala(S)$-module.
Every $\cala(S)$-module admits the structure of an $S$-sector of $\cala$, but the converse does not always hold.
However, we will see later that for conformal nets with finite index those two notions do agree.
The same will also hold for any closed 1-manifold $M$.

\begin{definition}\label{def:finite-nets}
Let $S$ be a circle, and let $I_1, I_2, I_3, I_4\subset S$ be intervals that are arranged so that each $I_i\cap I_{i+1}$ (cyclic numbering) is a single point.
A conformal net \emph{has finite index} if the bimodule
\[
{}_{\cala(I_1\cup I_3)}H_0(S)_{\cala(I_2\cup I_4)^\op}
\]
is dualizable.
\end{definition}

\noindent For the definition of dualizability of bimodules 
see~\cite[\secDualizability]{BDH(Dualizability+Index-of-subfactors)}.

\begin{lemma}\cite[\lemHS]{BDH(nets)} \label{lem: H_0(-S)}
Let $S$, $I_1$, $I_2,\ldots,I_4$ be as above, and
let $\bar S$, $\bar I_1,\ldots,\bar I_4$ be the same manifolds with the reverse orientation.
Let $\cala$ be a conformal net with finite index.
Then the dual of the bimodule ${}_{\cala(I_1\cup I_3)}H_0(S)_{\cala(I_2\cup I_4)^\op}$
is
\[
{}_{\cala(\bar I_2\cup \bar I_4)}H_0(\bar S)_{\cala(\bar I_1\cup \bar I_3)^\op}.
\]
Here, we have used the canonical identifications
$\cala(\bar I_1\cup \bar I_3)\cong \cala(I_1\cup I_3)^\op$ and
$\cala(\bar I_2\cup \bar I_4)\cong \cala(I_2\cup I_4)^\op$.
\end{lemma}

Let us call a tensor category a \emph{fusion category} if its underlying linear category is equivalent to $\mathsf{Hilb}^n$
(the category whose objects are $n$-tuples of Hilbert spaces) for some finite $n$, and if all its irreducible objects are dualizable.
We have seen in~\cite[\thmKLallirreduciblesectorsarefinite]{BDH(nets)} that if $\cala$ is a conformal net with finite index, then $\Rep_S(\cala)$ is a fusion category.
Let $\Delta$ be the set of isomorphism classes of simple sectors,
and let $\Delta\to \Delta:\lambda\mapsto \bar \lambda$ be the involution that corresponds to taking the dual sector.
Note that the finite set $\Delta$ is independent, up to canonical isomorphism, of the choice of circle $S$.

\begin{lemma}\cite[\lemdualofHlambda]{BDH(nets)}\label{lem: dual of H_lambda}
Let $S$ be a circle, decomposed into two subintervals $I_1$ and~$I_2$.
Then the dual of the bimodule ${}_{\cala(I_1)}H_\lambda(S)_{\cala(I_2)^\op}$ is
\[
{}_{\cala(\bar I_2)}H_{\bar \lambda}(\bar S)_{\cala(\bar I_1)^\op}
\]
under the canonical identifications 
$\cala(\bar I_1)\cong \cala(I_1)^\op$ and $\cala(\bar I_2)\cong \cala(I_2)^\op$.
\end{lemma}

For each $\lambda\in\Delta$ and circle $S$, pick a representative $H_\lambda(S)\in \Rep_S(\cala)$ of the isomorphism class.
Let $\Sigma:=S\times[0,1]$, and let $H_\Sigma$ be as in \eqref{eq: def H-sigma}.
Writing $\partial\Sigma$ as $S\sqcup \bar S$, then (provided the net $\cala$ is finite index) there is a 
non-canonical unitary isomorphism of 
$S\sqcup \bar S$-sectors~\cite[\thmKLM]{BDH(nets)}
\begin{equation} \label{eq:   KLM  }
H_\Sigma \,\cong\,\, \bigoplus_{\lambda\in\Delta} H_\lambda\big(S\big) \otimes H_{\bar \lambda}\big(\bar S\big).
\end{equation}
In fact, this isomorphism can be chosen canonically, as established in Theorem~\ref{thm: H_Sigma == L^2 cala(S)} in the appendix.  In light of the next theorem, that isomorphism can be reexpressed as
\begin{equation} \label{eq:HSigmaL^2}
H_\Sigma \,\cong\, L^2 \cala(S).
\end{equation}

\begin{theorem} \label{thm:compute-bfB(net)}
Let $\cala$ be a finite irreducible conformal net, and let $S$ be a circle.
Then there is a canonical isomorphism
\begin{equation}\label{eq: hatcalaS}
\cala(S) \,\cong\, \bigoplus_{\lambda\in\Delta} \bfB(H_\lambda(S)).
\end{equation}
\end{theorem}

\begin{proof}
The von Neumann algebra $\cala(S)$ is the completion in $\bfB(H_\Sigma)$
of the algebra generated by all the $\cala(I)$ for $I\subset S$.
By \eqref{eq:   KLM  }, an operator commutes with all those algebras if and only if it is contained in 
\[
\bigoplus_{\lambda\in\Delta}  \IC\otimes \bfB(H_\lambda(\bar S))\,\subset\, \bfB(H_{\Sigma}).
\]
Taking the commutant, it follows that 
\begin{equation}\label{eq: canonical? -- 2}
\cala(S)\,=\bigvee_{I\subset S}\cala(I) \,=\, \bigoplus_{\lambda\in\Delta} \bfB(H_\lambda(S)) \otimes \IC\,\subset\, \bfB(H_{\Sigma}).
\end{equation}

The isomorphism \eqref{eq:   KLM  } is a priori non-canonical: it is only well defined up to a phase factor on each direct summand.
Those phase factors, however, do not affect the isomorphism \eqref{eq: hatcalaS}, and the latter is therefore canonical.
\end{proof}

\begin{corollary}\label{cor: S-sector == A(S)-module}
If $\cala$ is a conformal net with finite index, then the forgetful functor from $\cala(S)$-modules to $S$-sectors of $\cala$
is an equivalence of categories.
\end{corollary} 

\begin{proof}
The functor is clearly fully faithful.
To show essential surjectivity, we need to argue
that every $S$-sector of $\cala$ induces an $\cala(S)$-module structure on the same Hilbert space.

Since $\cala$ has finite index, the category of $S$-sectors is 
semisimple~\cite[\thmKLallirreduciblesectorsarefinite]{BDH(nets)},
and any sector $H$ can be decomposed as
\begin{equation}\label{eq: sum decomposition of arbitrary sector*}
H \,\,\cong\,\, \bigoplus_{\lambda\in \Delta}\, H_\lambda(S)\otimes M_\lambda,
\end{equation}
where the multiplicity spaces $M_\lambda=\hom_{\Rep_S(\cala)}(H_\lambda(S),H)$ are Hilbert spaces, and the tensor product is the completed tensor product of Hilbert spaces.
It then follows from \eqref{eq: hatcalaS} that $H$ is a module for $\cala(S)$.
\end{proof}

The category of von Neumann algebras is cocomplete.
Given a diagram $\{A_i\}$ of von Neumann algebras, its colimit is computed as follows.
Let $A_0:=\mathrm{colim}_{\text{$*$-alg}}A_i$ be the colimit in the category of $*$-algebras,
and let $\IU$ be a Grothendieck universe that contains $A_0$.
Then $\mathrm{colim}_{\text{vN-alg}}A_i$ is the ultraweak closure of $A_0$ inside $\bfB(\bigoplus_{H\in \mathrm{Rep}(A_0)\cap \IU} H)$, where
the direct sum is taken over all representations of $A_0$ that are elements of~$\IU$.

\begin{proposition}\label{prop: colim over S}
Let $\cala$ be a conformal net with finite index and let $S$ be a circle.
Then
\[
\cala(S) = \underset{I\subset S}{\mathrm{colim}}\,\,\cala(I),
\]
where the colimit is taken in the category of von Neumann algebras, and is indexed by the poset of subintervals of $S$.
\end{proposition}

\begin{proof}
Let $A_0:=\mathrm{colim}_{\text{$*$-alg}}\cala(I)$, and let $\IU$ be a Grothendieck universe that contains~$A_0$.
The Hilbert space $\bigoplus_{H\in \mathrm{Rep}(A_0)\cap \IU} H$ is an $S$-sector of $\cala$ and is therefore of the form \eqref{eq: sum decomposition of arbitrary sector*}.
The same argument as in the proof of Theorem \ref{thm:compute-bfB(net)} then applies,
from which it follows that $A:=\mathrm{colim}_{\text{vN-alg}}\cala(I)\cong \bigoplus_{\lambda\in\Delta} \bfB(H_\lambda(S))$,
and that the natural map $A\to\cala(S)$ is an isomorphism.
\end{proof}

The following result generalizes Corollary \ref{cor: S-sector == A(S)-module}.

\begin{proposition}\label{prop: M-sector == A(M)-module}
If $\cala$ is a conformal net with finite index, then
the forgetful functor from $\cala(M)$-modules to $M$-sectors of $\cala$
is an equivalence of categories.
\end{proposition}

\begin{proof}
Let us write $M=S_1\sqcup\ldots\sqcup S_n$ as a disjoint union of circles, and let $H$ be an $M$-sector.
By applying \eqref{eq: sum decomposition of arbitrary sector*} to $H$ viewed as an $S_1$-sector, we get a decomposition
\[
H \,\,\cong\,\, \bigoplus_{\lambda\in \Delta}\, H_\lambda(S_1)\otimes K_\lambda.
\]
The multiplicity spaces $K_\lambda$ are equipped with residual actions of the algebras $\cala(I)$ for $I\subset S_2\sqcup\ldots\sqcup S_n$, 
and are therefore $S_2\sqcup\ldots\sqcup S_n$-sectors.
Applying \eqref{eq: sum decomposition of arbitrary sector*} to the $K_\lambda$ viewed as $S_2$-sectors, we get a further decomposition
\(
H \cong \bigoplus_{\lambda_1,\lambda_2\in \Delta} H_{\lambda_1}(S_1)\otimes H_{\lambda_2}(S_2)\otimes K_{\lambda_1,\lambda_2},
\)
where $K_{\lambda_1,\lambda_2}$ are now $S_3\sqcup\ldots\sqcup S_n$-sectors.
After $n$ steps, this yields a decomposition
\[
H \,\,\cong \bigoplus_{\lambda_1,\lambda_2,\ldots,\lambda_n\in \Delta} H_{\lambda_1}(S_1)\otimes H_{\lambda_2}(S_2)\otimes\cdots\otimes H_{\lambda_n}(S_n)\,\otimes\, K_{\lambda_1,\lambda_2,\ldots,\lambda_n},
\]
where the multiplicity spaces $K_{\lambda_1,\lambda_2,\ldots,\lambda_n}$ are mere Hilbert spaces.
The space $H$ is now visibly an $\cala(S_1)\,\bar\otimes\,\ldots\,\bar\otimes\,\cala(S_n)$ module by Theorem \ref{thm:compute-bfB(net)},
that is, an $\cala(M)$ module.
\end{proof}

\section{The Hilbert space associated to a surface}

\subsection{Canonical up to phase}\label{sec: ``up to non-canonical isomorphism''}
The reader is probably familiar with what it means for a Hilbert space to be well defined up to canonical unitary isomorphism,
and what it means for one to be well defined up to non-canonical unitary isomorphism.
In this paper, we shall also need an intermediate notion: the notion of a Hilbert space that is well defined up to canonical-up-to-phase unitary isomorphism.
An equivalent way of saying that a Hilbert space $H$ is well defined up to canonical-up-to-phase unitary isomorphism,
is to say that the associated projective space $\IP H$ is well defined up to canonical projective unitary isomorphism.
Given that the difference between all those notions is rather subtle, we take the time to spell them out in detail.

One says that a Hilbert space $H$ is well defined up to canonical unitary isomorphism if, even though its construction depends on some choices,
there is a given coherent way of identifying all the Hilbert spaces constructed.
More precisely, letting $X$ be the set that parametrizes the choices,
then for every $x\in X$ there should be a Hilbert space $H_x$, for every $x,y\in X$, there should be a given unitary isomorphism $u_{xy}:H_x\to H_y$,
and for every $x,y,z\in X$, the diagram
\[
\tikzmath{\node (1) at (0,0) {$H_x$};\node (2) at (1.4,.7) {$H_y$};\node (3) at (2.8,0) {$H_z$};\draw[->] (1) tonode[above, yshift=-1, xshift=-4]{$\scriptstyle u_{xy}$} (2);\draw[->] (2) tonode[above, yshift=-1, xshift=3]{$\scriptstyle u_{yz}$} (3);\draw[->] (1) tonode[above, yshift=-1]{$\scriptstyle u_{xz}$} (3);}
\]
should commute.

By contrast, a Hilbert space $H$ is said to be \emph{well defined up to non-canonical (unitary) isomorphism} if its construction depends on some choices,
any two choices yield isomorphic results, but there are no preferred isomorphisms.
\begin{definition}\label{def: well defined up to canonical-up-to-phase unitary isomorphism}
We shall say that a Hilbert space $H$ is \emph{well defined up to canonical-up-to-phase unitary isomorphism} if:
\begin{list}{$\bullet$}{\leftmargin=.7cm \parsep=1pt \labelsep=.2cm \listparindent=0pt \labelwidth=.3cm}
\item The construction of $H$ depends on some set of choices; say $X$ is the set that parametrizes those choices.
For every $x\in X$, we therefore have a Hilbert space $H_x$.
\item For any two choices $x,x'\in X$, there exists a unitary isomorphism $\Phi :H_x\to H_{x'}$.
Once again, the construction of $\Phi $ depends on some set of choices,
let us call $Y(x,x')$ the set that parametrizes those choices:
given $y\in Y(x,x')$, we have a given isomorphism $\Phi_y:H_x\to H_{x'}$.
\item Given two choices $y,y'\in Y(x,x')$, the isomorphisms $\Phi_y$ and $\Phi_{y'}$ are equal up to a scalar.
\item Finally, for any $y\in Y(x,x')$, $y'\in Y(x',x'')$, $y''\in Y(x,x'')$, there exists a scalar $\lambda$ such that $\Phi_{y'}\circ \Phi_{y}=\lambda\, \Phi_{y''}$.
\end{list}
\end{definition}

\noindent Note that the above concepts apply not only to Hilbert spaces, but also to modules over von Neumann algebras, bimodules, sectors, etc.
They also apply to functors with values in a category such as $\mathsf{Hilb}$, $\modules{A}$, etc.

When the mathematical object under consideration is a number, a linear map, or a natural transformation between functors,
then the story is somewhat simpler: such an object is either canonical, canonical up to phase, or not canonical
(here, ``canonical'' is a synonym of ``well defined'').

For example, if $\cala$ is an irreducible conformal net, 
then by Schur's lemma, the vacuum sector $H_0(S)\in\Rep_S(\cala)$ is well defined up to canonical-up-to-phase unitary isomorphism.
By the same reasoning, the isomorphism $H_0(S_1) \boxtimes_{\cala(I)} H_0(S_2)\cong H_0(S_3)$ in \eqref{eq: [paper I, Corollary 1.33]} is canonical up to phase.

In practice, a convenient way of showing that a Hilbert space $H$ is well defined up to canonical-up-to-phase unitary isomorphism
is to exhibit a simply connected $2$-dimensional CW-compex $\IX$, and do the following:
\begin{list}{$\bullet$}{\leftmargin=.7cm \parsep=1pt \labelsep=.2cm \listparindent=0pt \labelwidth=.3cm}
\item For every vertex $x$ of $\IX$, construct a Hilbert space $H_x$.
\item For every edge $y$ of $\IX$ between vertices $x,x'\in \IX$, construct a unitary isomorphism $\Phi_y:H_x\to H_{x'}$.
\item For every 2-cell of $\IX$ with boundary $y_1y_2\ldots y_n$, check that the \hyphenation{auto-mor-phism}automorphism $\Phi_{y_1}\circ\Phi_{y_2}\circ\ldots\circ \Phi_{y_n}$ is a scalar.
\end{list}
We call such a CW-complex a \emph{definition complex} for $H$.

The relation to Definition~\ref{def: well defined up to canonical-up-to-phase unitary isomorphism} is as follows.
The set $X$ is the set of vertices of $\IX$.
Given $x,x'\in X$, the set $Y(x,x')$ that parametrizes isomorphisms between $H_x$ and $H_{x'}$ is the set of all sequences of edges in $\IX$ that altogether go from $x'$ to $x$. The isomorphism $\Phi_p:H_x\to H_{x'}$ that corresponds to a path $p=y_1\cdots y_n$ is the composite $\Phi_{y_1}\circ\ldots\circ\Phi_{y_n}$.
In order to check the two conditions in Definition \ref{def: well defined up to canonical-up-to-phase unitary isomorphism},
we need to show that, given a loop $\gamma$ based at a vertex $x\in \IX$, the corresponding automorphism $\Phi_\gamma\in\U(H_x)$ is a scalar.
Because the definition complex is simply connected, one can write $\gamma$ as a product $\prod p_i w_i p_i^{-1}$
where the $p_i$ are paths, and the $w_i$ are loops along which $2$-cells are attached.
Since $\Phi_{w_i}\in \U(1)$ is central, we conclude that $\Phi_\gamma=\prod\Phi_{p_i} \Phi_{w_i} \Phi_{p_i}^{-1}$ is also a scalar, as desired.

\subsection{The construction}\label{sec: The Hilbert space associated to a surface}

In this section, we give a construction of a Hilbert space $V(\Sigma)\in\Rep_{\partial\Sigma}(\cala)$
associated to a topological surface $\Sigma$ with smooth (possibly empty) boundary.
The construction depends on the auxiliary choice of a certain kind of cell decomposition of $\Sigma$.
We will show later, in Section \ref{sec: Independence of the cell decomposition}, that it is actually independent of that choice.

Let us call a cell decomposition \emph{regular} if for every 2-cell $\ID$ the corresponding map $S^1\cong\partial \ID\to \Sigma$ is injective,
and \emph{trivalent} if each vertex is incident to exactly three edges.
We call a cell decomposition $\Sigma=\ID_1\cup\ldots\cup\ID_n$ \emph{ordered} if the set $\{\ID_1,\ldots,\ID_n\}$ of 2-cells is ordered.
Finally, a cell decomposition is \emph{smooth} if the 1-skeleton is equipped with a smooth structure in the sense of Definition \ref{def: smooth structure on trivalent graph},
and this smooth structure restricts to the given smooth structure on $\partial \Sigma$.
Our construction of $V(\Sigma)$ will depend, a priori, on an ordered regular trivalent smooth cell decomposition of $\Sigma$.

The idea of the construction is to associate to each 2-cell $\ID\subset \Sigma$ the vacuum sector $H_0(\partial \ID)$, and to then glue them using Connes fusion.

Let $\Sigma_i$ be the 2-manifold $\ID_1\cup\ldots\cup\ID_i$, including the case $\Sigma_0=\emptyset$.
Let us give the 1-manifolds $M_i:=\Sigma_i\cap\ID_{i+1}$ the orientations coming from $\partial \ID_{i+1}$
(the orientation of $\partial\ID_{i+1}$ is itself induced from that of $\ID_{i+1}$).
We define $V(\Sigma_i)\in\Rep_{\partial \Sigma_i}(\cala)$ inductively by
\begin{equation}\label{eq: V(D_i) -- BIS}
\begin{cases} V(\emptyset):=\IC\\
V(\Sigma_{i+1}):= V(\Sigma_i)\boxtimes_{\cala(M_i)}H_0(\partial \ID_{i+1}).
\end{cases}
\end{equation}
To give a meaning to the above construction,
we need to explain how the algebra $\cala(M_i)$ acts on the Hilbert spaces $V(\Sigma_i)$ and $H_0(\partial \ID_{i+1})$.
The left action on $H_0(\partial \ID_{i+1})$ is obvious when $M_i\subsetneq \partial \ID_{i+1}$, and follows from Theorem \ref{thm:compute-bfB(net)} for $M_i= \partial \ID_{i+1}$.
Assuming by induction that $V(\Sigma_i)$ is a $\partial \Sigma_i$-sector,
the right action of $\cala(M_i)$ on $V(\Sigma_i)$ follows from Proposition \ref{prop: M-sector == A(M)-module}.
To finish the inductive definition, we still need to show that $V(\Sigma_{i+1})$ is a $\partial \Sigma_{i+1}$-sector.
This is obvious for $M_i= \partial \ID_{i+1}$, and follows 
from~\cite[\lemNNN]{BDH(nets)} for $M_i\subsetneq \partial \ID_{i+1}$.

\subsection{The Hilbert space associated to a disk}
The goal of this section is to show that if $\Sigma$ is a disk, then
regardless of the choice of cell decomposition of $\Sigma$, the resulting sector $V(\Sigma)\in\Rep_{\partial\Sigma}(\cala)$ is the vacuum $H_0(\partial \Sigma)$.

In order to compute the sector associated to a disc, we will need to address the seemingly more difficult problem of computing the isomorphism 
type of the sector associated to a sphere with $n$ holes.
Given $H_{\lambda_1},\ldots, H_{\lambda_n}\in\Rep(\cala)$,
let $N_0^{\lambda_1,\ldots,\lambda_n}$ denote the multiplicity of $H_0$ inside $H_{\lambda_1}\boxtimes\ldots\boxtimes H_{\lambda_n}$.

Recall from~\cite[\subsecsectorsfornets]{BDH(nets)} 
that given a sector $H_{\lambda}\in\Rep(\cala)$ and given a circle $S$, 
we write $H_\lambda(S)\in \Rep_S(\cala)$ for the associated $S$-sector of $\cala$.
It is given by $\varphi^* H_\lambda$ for some $\varphi\in\Diff_+(S, S^1)$, and is only well defined up to non-canonical isomorphism.

\begin{lemma}\label{lem: computation of V((000))}
Let $\cala$ be a conformal net with finite index and let $\Sigma$ be a genus zero surface (a sphere with holes).
Equip $\Sigma$ with an ordered regular trivalent smooth cell decomposition $\Sigma=\ID_1\cup\ldots\cup \ID_n$ and let $V(\Sigma)$ be the corresponding Hilbert space.

Decompose the boundary of $\Sigma$ as a union of circles $\partial \Sigma=S_1\cup\ldots\cup S_m$.
Then for every choice of $\lambda_1,\ldots, \lambda_m \in \Delta$, the vector space
\[
\hom_{\Rep_{\partial\Sigma}(\cala)}\big(H_{\lambda_1}(S_1)\otimes\ldots\otimes H_{\lambda_m}(S_m),V(\Sigma)\big)
\]
has dimension $N_0^{\lambda_1,\ldots,\lambda_m}$.
\end{lemma}

\begin{proof}
Letting $\Sigma':=\ID_1\cup\ldots\cup \ID_{n-1}$ and $M:=\partial\ID_n\cap \Sigma'$, we have by definition
\[
V(\Sigma)=V(\Sigma')\boxtimes_{\cala(M)}H_0(\partial\ID_n).
\]
The $1$-manifold $M$ is either a circle or a union of intervals.
If $M$ is a circle, then $\Sigma'$ is connected and we can write its boundary $\partial\Sigma'$ as $S_0\cup S_1\cup\ldots\cup S_m$ with $M=\bar S_0$.
Using that for any right $\cala(M)$-module $K$, there is a canonical isomorphism
\[
K\boxtimes_{\cala(M)}H_0(M) \,\cong\, \hom_{\cala(S_0)}(H_0(S_0),K),
\]
we see by induction on $n$ that
\[
\begin{split}
&\,\textstyle\hom_{\Rep_{\partial\Sigma}(\cala)}\big(H_{\lambda_1}(S_1)\otimes\ldots\otimes H_{\lambda_m}(S_m),V(\Sigma)\big)\\
\cong\,\,&\,\textstyle\hom_{\Rep_{\partial\Sigma}(\cala)}\big(H_{\lambda_1}(S_1)\otimes\ldots\otimes H_{\lambda_m}(S_m),V(\Sigma')\boxtimes_{\cala(M)}H_0(M)\big)\\
\cong\,\,&\textstyle\hom_{\Rep_{\partial\Sigma'}(\cala)}\big(H_0(S_0)\otimes H_{\lambda_1}(S_1)\otimes\ldots\otimes H_{\lambda_m}(S_m),V(\Sigma')\big)\\
\end{split}
\]
has dimension $N_0^{0,\lambda_1,\ldots,\lambda_m}=N_0^{\lambda_1,\ldots,\lambda_m}$, as desired.

We now assume that $M$ is a union of intervals.
Decompose $\Sigma'$ as a disjoint union of connected surfaces $\Sigma'=\Sigma'_1\cup\ldots\cup \Sigma'_l$,
and write $\partial\Sigma'_i=S'_{i1}\cup S'_{i2}\cup\ldots\cup S'_{ik_i}$.
Up to homeomorphism, the surface $\Sigma$ appears as follows:
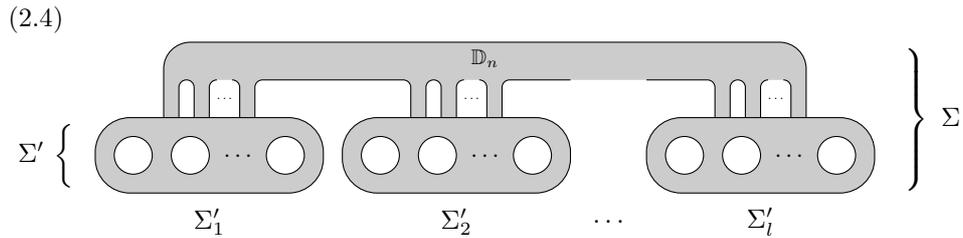
\begin{equation}\label{eq: big octopus}
\tikzmath[scale = .5]{
\def \rcA{2.8}\def \rcB{10}\def \rcC{13}
\fill[gray!40] (.8,1)-- +(.4,0) [rounded corners=\rcA]-- +(.4,1) -- +(.8,1) [sharp corners]-- +(.8,0) -- +(1.2,0) [rounded corners=\rcA]-- +(1.2,1) -- +(2,1) [sharp corners]-- +(2,0) -- +(2.4,0) [rounded corners=\rcA]-- +(2.4,1) -- (7.3,2) [sharp corners]-- (7.3,1)  -- +(.4,0) [rounded corners=\rcA]-- +(.4,1) -- +(.8,1) [sharp corners]-- +(.8,0) -- +(1.2,0) [rounded corners=\rcA]-- +(1.2,1) -- +(2,1) [sharp corners]-- +(2,0) -- +(2.4,0) [rounded corners=\rcA]-- +(2.4,1) -- (15.3,2) [sharp corners]-- (15.3,1) -- +(.4,0) [rounded corners=\rcA]-- +(.4,1) -- +(.8,1) [sharp corners]-- +(.8,0) -- +(1.2,0) [rounded corners=\rcA]-- +(1.2,1) -- +(2,1) [sharp corners]-- +(2,0) -- +(2.4,0) [rounded corners=\rcB]-- +(2.4,2) -- (.8,3) [sharp corners]-- cycle;
\node[scale=.5] at (2.4,1.5) {$\ldots$};\node[scale=.5] at (8.9,1.5) {$\ldots$};\node[scale=.5] at (16.9,1.5) {$\ldots$};
\filldraw[fill=gray!40, rounded corners=\rcC] (-1,-1) rectangle node[yshift=-25]{$\Sigma'_1$} (5,1) (0,0) circle (.5) (1.5,0) circle (.5) (4,0) circle (.5);\node[scale=.9] at (2.75,0) {$\ldots$};
\filldraw[fill=gray!40, rounded corners=\rcC] (5.5,-1) rectangle node[yshift=-25]{$\Sigma'_2$} (11.5,1) (6.5,0) circle (.5) (8,0) circle (.5) (10.5,0) circle (.5);\node[scale=.9] at (9.25,0) {$\ldots$};\node[yshift=-25] at (12.5,0) {$\ldots$};
\filldraw[fill=gray!40, rounded corners=\rcC] (13.5,-1) rectangle node[yshift=-25]{$\Sigma'_l$} (19.5,1) (14.5,0) circle (.5) (16,0) circle (.5) (18.5,0) circle (.5);\node[scale=.9] at (17.25,0) {$\ldots$};
\draw (.8,1) +(.4,0) [rounded corners=\rcA]-- +(.4,1) -- +(.8,1) -- +(.8,0) +(1.2,0) -- +(1.2,1) -- +(1.4,1) +(1.8,1) -- +(2,1) -- +(2,0) +(2.4,0) -- +(2.4,1) -- (7.3,2) -- (7.3,1) 
+(.4,0) -- +(.4,1) -- +(.8,1) -- +(.8,0) +(1.2,0) -- +(1.2,1) -- +(1.4,1) +(1.8,1) -- +(2,1) -- +(2,0) +(2.4,0) -- +(2.4,1) -- +(4.2,1) (13.5,2) -- (15.3,2) -- (15.3,1) 
+(.4,0) -- +(.4,1) -- +(.8,1) -- +(.8,0) +(1.2,0) -- +(1.2,1) -- +(1.4,1) +(1.8,1) -- +(2,1) -- +(2,0) +(2.4,0) -- +(2.4,1) [rounded corners=\rcB]-- +(2.4,2) --node[below, scale=.8]{$\ID_n$} (.8,3) -- (.8,1);
\node at (-2.3,0) {$\Sigma'\,\,\bigg\{$};
\node at (19.7,1) {$\left.\phantom{|^{|^{|^{|^{|^{|^{|^{|^{|^{|^{|^{|^|}}}}}}}}}}}}\right\}\,\,\Sigma$};
}
\end{equation}
Indeed, since $\Sigma$ has genus zero, the manifold $M$ can only intersect one boundary component of each $\Sigma'_i$.
By renumbering the components of $\partial\Sigma'_i$, we may assume that $M\cap S'_{i1}\not = \emptyset$ and $M\cap S'_{ij} = \emptyset$ for $j\ge 2$.

Recall that we are trying to show that
\begin{equation}\label{eq: dim = = prod N}
\dim\,\hom_{\Rep_{\partial\Sigma}(\cala)}\big(H_{\lambda_1}(S_1)\otimes\ldots\otimes H_{\lambda_m}(S_m),V(\Sigma)\big)\,=\,\,N_0^{\lambda_{1},\ldots,\lambda_{m}}.\,\,\,
\end{equation}
By induction on $n$, we know that for every $i$
\[
\,\,\,\dim\,\hom_{\Rep_{\partial\Sigma'_i}(\cala)}\big(H_{\lambda_{i1}}(S'_{i1})\otimes\ldots\otimes H_{\lambda_{ik_i}}(S'_{ik_i}),V(\Sigma'_i)\big)\,=\,\,N_0^{\lambda_{i1},\ldots,\lambda_{ik_i}}.
\]
Let $b_i$ be the number of connected components of $\partial \ID_n\cap \Sigma'_i$.
Before treating the general case, we prove it for certain small values of $l$ and $b_1,\ldots,b_l$.

If $l=1$ and $b_1=2$ then, up to homeomorphism, $\Sigma=\Sigma'\cup\ID_n$ appears as follows:
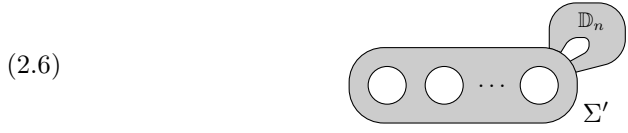
\begin{equation}\label{eq: case l=1, b_1=2}
\tikzmath[scale = .5]{
\def \rcA{2.8}\def \rcB{10}\def \rcC{13}
\filldraw[fill=gray!40, rounded corners=\rcB]
($(4,0)+(75:1.02)$) -- ++(0,1.2) --node[below, scale=.8, xshift=2, yshift=-1]{$\ID_n$} ++(2,0) [rounded corners= 7.5]-- ++(0,-1.5) [sharp corners]-- ($(4,0)+(20:1)$) -- ($(4,0)+(38.5:1)$)
[rounded corners= 2]-- ($(4,0)+(35:1.6)$) -- ++(0,.35) [rounded corners= \rcA] -- ++(-.5,0) [sharp corners]-- ($(4,0)+(57:1.04)$) -- cycle;
\filldraw[fill=gray!40, rounded corners=\rcC] (-1,-1) rectangle (5,1) (0,0) circle (.5) (1.5,0) circle (.5) (4,0) circle (.5);\node[scale=.9] at (2.75,0) {$\ldots$};
\node at (5.5,-.7) {$\Sigma'$};
}
\end{equation}
By Lemma \ref{lem: (0000)} below and 
induction on $n$, $V(\Sigma')$ is isomorphic to the $\partial\Sigma'$-sector associated to this cellular decomposition:
$\tikzmath[scale = .3]{
\filldraw[fill=gray!40, rounded corners=8] (-1,-1) rectangle (5,1) (0,0) circle (.5) (1.5,0) circle (.5) (4,0) circle (.5);\node[scale=.9] at (2.75,0) {$\scriptstyle \ldots$};
\draw (0,0) +(90:.5) -- +(90:1)+(-90:.5) -- +(-90:1) (1.5,0) +(90:.5) -- +(90:1)+(-90:.5) -- +(-90:1) (4,0) +(90:.5) -- +(90:1)+(-90:.5) -- +(-90:1);
}$\,.
It follows that $V(\Sigma)$ is isomorphic to the $\partial\Sigma$-sector associated to the following decomposition of $\Sigma$:
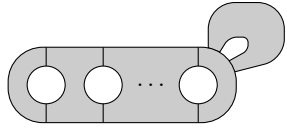
\[
\tikzmath[scale = .5]{
\def \rcA{2.8}\def \rcB{10}\def \rcC{13}
\filldraw[fill=gray!40, rounded corners=\rcB]
($(4,0)+(75:1.02)$) -- ++(0,1.2) -- ++(2,0) [rounded corners= 7.5]-- ++(0,-1.5) [sharp corners]-- ($(4,0)+(20:1)$) -- ($(4,0)+(38.5:1)$)
[rounded corners= 2]-- ($(4,0)+(35:1.6)$) -- ++(0,.35) [rounded corners= \rcA] -- ++(-.5,0) [sharp corners]-- ($(4,0)+(57:1.04)$) -- cycle;
\filldraw[fill=gray!40, rounded corners=\rcC] (-1,-1) rectangle (5,1) (0,0) circle (.5) (1.5,0) circle (.5) (4,0) circle (.5);\node[scale=.9] at (2.75,0) {$\ldots$};
\draw (0,0) +(90:.5) -- +(90:1)+(-90:.5) -- +(-90:1) (1.5,0) +(90:.5) -- +(90:1)+(-90:.5) -- +(-90:1) (4,0) +(90:.5) -- +(90:1)+(-90:.5) -- +(-90:1);
}
\]
That decomposition is of the form covered by Lemma \ref{lem: (0000)}, and so \eqref{eq: dim = = prod N} follows.

We now treat the case $l=2$, $b_1=b_2=1$, where $\Sigma$ appears as follows:
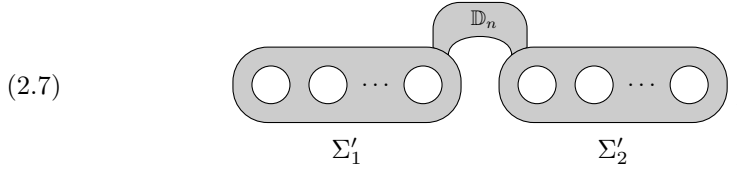
\begin{equation}\label{eq: case l=2, b_1=b_2=1}
\tikzmath[scale = .5]{
\def \rcA{2.8}\def \rcB{10}\def \rcC{13}
\filldraw[fill=gray!40, rounded corners=\rcB] ($(4,0)+(75:1.02)$) -- ($(4,0)+(75:1.02)+(0,1.2)$) --node[below, scale=.8, xshift=2]{$\ID_n$} ($(7,0)+(105:1.02)+(0,1.2)$) [sharp corners]-- ($(7,0)+(105:1.02)$) -- ($(7,0)+(130:1.04)$) to[bend right=90] ($(4,0)+(50:1.04)$) -- cycle;
\filldraw[fill=gray!40, rounded corners=\rcC] (-1,-1) rectangle node[yshift=-25]{$\Sigma'_1$} (5,1) (0,0) circle (.5) (1.5,0) circle (.5) (4,0) circle (.5);\node[scale=.9] at (2.75,0) {$\ldots$};
\filldraw[fill=gray!40, rounded corners=\rcC] (6,-1) rectangle node[yshift=-25]{$\Sigma'_2$} (12,1) (7,0) circle (.5) (8.5,0) circle (.5) (11,0) circle (.5);\node[scale=.9] at (9.75,0) {$\ldots$};
}\;\raisebox{-1.5mm}{.}
\end{equation}
Let $M_1=\partial \ID_n\cap \Sigma'_1$ and $M_2=\partial \ID_n\cap \Sigma'_2$.
Once again, by Lemma \ref{lem: (0000)} and induction on $n$,
we know that $V(\Sigma'_i)$ is isomorphic to the $\partial\Sigma'_i$-sector constructed from the cellular decomposition 
$\tikzmath[scale = .3]{
\filldraw[fill=gray!40, rounded corners=8] (-1,-1) rectangle (5,1) (0,0) circle (.5) (1.5,0) circle (.5) (4,0) circle (.5);\node[scale=.9] at (2.75,0) {$\scriptstyle \ldots$};
\draw (0,0) +(90:.5) -- +(90:1)+(-90:.5) -- +(-90:1) (1.5,0) +(90:.5) -- +(90:1)+(-90:.5) -- +(-90:1) (4,0) +(90:.5) -- +(90:1)+(-90:.5) -- +(-90:1);
}$\,.
It follows that
\[
V(\Sigma)=\big(V(\Sigma'_1)\otimes V(\Sigma'_2)\big)\!\underset{\cala(M)}\boxtimes\! H_0(\partial\ID_n)
\,\cong\,V(\Sigma'_1)\!\underset{\cala(\bar M_1)}\boxtimes\! H_0(\partial\ID_n) \!\underset{\cala(M_2)}\boxtimes\! V(\Sigma'_2)
\]
is isomorphic to the $\partial\Sigma$-sector associated to the following decomposition of $\Sigma$:
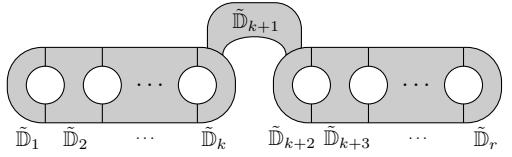
\[
\tikzmath[scale = .5]{
\def \rcA{2.8}\def \rcB{10}\def \rcC{13}
\filldraw[fill=gray!40, rounded corners=\rcB] ($(4,0)+(75:1.02)$) -- ($(4,0)+(75:1.02)+(0,1.2)$) -- ($(7,0)+(105:1.02)+(0,1.2)$) [sharp corners]-- ($(7,0)+(105:1.02)$) -- ($(7,0)+(130:1.04)$) to[bend right=90] ($(4,0)+(50:1.04)$) -- cycle;
\filldraw[fill=gray!40, rounded corners=\rcC] (-1,-1) rectangle (5,1) (0,0) circle (.5) (1.5,0) circle (.5) (4,0) circle (.5);\node[scale=.9] at (2.75,0) {$\ldots$};
\filldraw[fill=gray!40, rounded corners=\rcC] (6,-1) rectangle (12,1) (7,0) circle (.5) (8.5,0) circle (.5) (11,0) circle (.5);\node[scale=.9] at (9.75,0) {$\ldots$};
\draw (0,0) +(90:.5) -- +(90:1)+(-90:.5) -- +(-90:1) (1.5,0) +(90:.5) -- +(90:1)+(-90:.5) -- +(-90:1) (4,0) +(90:.5) -- +(90:1)+(-90:.5) -- +(-90:1) (7,0) +(90:.5) -- +(90:1)+(-90:.5) -- +(-90:1) (8.5,0) +(90:.5) -- +(90:1)+(-90:.5) -- +(-90:1) (11,0) +(90:.5) -- +(90:1)+(-90:.5) -- +(-90:1);
\node[scale=.8] at (-.45,-1.4) {$\tilde\ID_1$};
\node[scale=.8] at (.8,-1.4) {$\tilde\ID_2$};
\node[scale=.6] at (2.625,-1.4) {$\ldots$};
\node[scale=.8] at (4.45,-1.4) {$\tilde\ID_k$};
\node[scale=.8] at (5.55,1.7) {$\tilde\ID_{k+1}$};
\node[scale=.8] at (6.5,-1.4) {$\tilde\ID_{k+2}$};
\node[scale=.8] at (7.95,-1.4) {$\tilde\ID_{k+3}$};
\node[scale=.6] at (9.8,-1.4) {$\ldots$};
\node[scale=.8] at (11.6,-1.4) {$\tilde\ID_r$};
}\,\raisebox{-1.5mm}{.}
\]
As $H_0(\partial\tilde \ID_k)\boxtimes_{\cala(\bar M_1)} H_0(\partial\tilde\ID_{k+1})\boxtimes_{\cala(M_2)}H_0(\partial\tilde \ID_k)\cong
H_0(\partial(\ID_k\cup \ID_{k+1}\cup \ID_{k+2}))$, the sector $V(\Sigma)$ is also isomorphic to the one associated to the following decomposition:
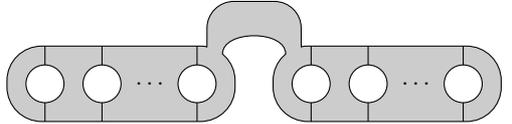
\[
\tikzmath[scale = .5]{
\def \rcA{2.8}\def \rcB{10}\def \rcC{13}
\filldraw[fill=gray!40, rounded corners=\rcC] (-1,-1) rectangle (5,1) (0,0) circle (.5) (1.5,0) circle (.5) (4,0) circle (.5);\node[scale=.9] at (2.75,0) {$\ldots$};
\filldraw[fill=gray!40, rounded corners=\rcC] (6,-1) rectangle (12,1) (7,0) circle (.5) (8.5,0) circle (.5) (11,0) circle (.5);\node[scale=.9] at (9.75,0) {$\ldots$};
\draw (0,0) +(90:.5) -- +(90:1)+(-90:.5) -- +(-90:1) (1.5,0) +(90:.5) -- +(90:1)+(-90:.5) -- +(-90:1) (4,0) +(90:.5) -- +(90:1)+(-90:.5) -- +(-90:1) (7,0) +(90:.5) -- +(90:1)+(-90:.5) -- +(-90:1) (8.5,0) +(90:.5) -- +(90:1)+(-90:.5) -- +(-90:1) (11,0) +(90:.5) -- +(90:1)+(-90:.5) -- +(-90:1);
\fill[fill=gray!40, rounded corners=\rcB] ($(4,0)+(75:1.02)+(0,-.35)$) -- ($(4,0)+(75:1.02)+(0,1.2)$) -- ($(7,0)+(105:1.02)+(0,1.2)$) [sharp corners]-- ($(7,0)+(105:1.02)+(0,-.35)$) -- ($(7,0)+(130.3:1.04)$) to[bend right=90] ($(4,0)+(49.7:1.04)$) -- cycle;
\draw[rounded corners=\rcB] ($(4,0)+(75:1.005)$) -- ($(4,0)+(75:1.02)+(0,1.2)$) -- ($(7,0)+(105:1.02)+(0,1.2)$) [sharp corners]-- ($(7,0)+(105:1.005)$) ($(7,0)+(130.3:1.035)$) to[bend right=90] ($(4,0)+(49.7:1.035)$);
}\,\,\raisebox{-3.5mm}{.}
\]
That decomposition is of the form covered by Lemma \ref{lem: (0000)}, and \eqref{eq: dim = = prod N} follows.

We now treat the general case.
We assume by induction that $V(\Sigma')$ satisfies \eqref{eq: dim = = prod N}.
Consider the refinement $\Sigma=\tilde\ID_1\cup\ldots\cup\tilde\ID_s$ of \eqref{eq: big octopus}
given by $\tilde\ID_i:=\ID_i$ for $i<n$, and where $\ID_n$ is subdivided into cells $\tilde\ID_n$, $\tilde\ID_{n+1},\ldots, \tilde\ID_s$
as follows:
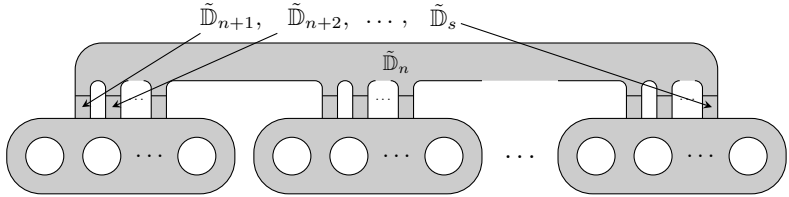
\[
\tikzmath[scale = .5]{
\def \rcA{2.8}\def \rcB{10}\def \rcC{13}
\fill[gray!40] (.8,1)-- +(.4,0) [rounded corners=\rcA]-- +(.4,1) -- +(.8,1) [sharp corners]-- +(.8,0) -- +(1.2,0) [rounded corners=\rcA]-- +(1.2,1) -- +(2,1) [sharp corners]-- +(2,0) -- +(2.4,0) [rounded corners=\rcA]-- +(2.4,1) -- (7.3,2) [sharp corners]-- (7.3,1)  -- +(.4,0) [rounded corners=\rcA]-- +(.4,1) -- +(.8,1) [sharp corners]-- +(.8,0) -- +(1.2,0) [rounded corners=\rcA]-- +(1.2,1) -- +(2,1) [sharp corners]-- +(2,0) -- +(2.4,0) [rounded corners=\rcA]-- +(2.4,1) -- (15.3,2) [sharp corners]-- (15.3,1) -- +(.4,0) [rounded corners=\rcA]-- +(.4,1) -- +(.8,1) [sharp corners]-- +(.8,0) -- +(1.2,0) [rounded corners=\rcA]-- +(1.2,1) -- +(2,1) [sharp corners]-- +(2,0) -- +(2.4,0) [rounded corners=\rcB]-- +(2.4,2) -- (.8,3) [sharp corners]-- cycle;
\node[scale=.5] at (2.4,1.5) {$\ldots$};\node[scale=.5] at (8.9,1.5) {$\ldots$};\node[scale=.5] at (16.9,1.5) {$\ldots$};
\filldraw[fill=gray!40, rounded corners=\rcC] (-1,-1) rectangle (5,1) (0,0) circle (.5) (1.5,0) circle (.5) (4,0) circle (.5);\node[scale=.9] at (2.75,0) {$\ldots$};
\filldraw[fill=gray!40, rounded corners=\rcC] (5.5,-1) rectangle (11.5,1) (6.5,0) circle (.5) (8,0) circle (.5) (10.5,0) circle (.5);\node[scale=.9] at (9.25,0) {$\ldots$};\node at (12.5,0) {$\ldots$};
\filldraw[fill=gray!40, rounded corners=\rcC] (13.5,-1) rectangle (19.5,1) (14.5,0) circle (.5) (16,0) circle (.5) (18.5,0) circle (.5);\node[scale=.9] at (17.25,0) {$\ldots$};
\draw (.8,1) +(.4,0) [rounded corners=\rcA]-- +(.4,1) -- +(.8,1) -- +(.8,0) +(1.2,0) -- +(1.2,1) -- +(1.4,1) +(1.8,1) -- +(2,1) -- +(2,0) +(2.4,0) -- +(2.4,1) -- (7.3,2) -- (7.3,1) 
+(.4,0) -- +(.4,1) -- +(.8,1) -- +(.8,0) +(1.2,0) -- +(1.2,1) -- +(1.4,1) +(1.8,1) -- +(2,1) -- +(2,0) +(2.4,0) -- +(2.4,1) -- +(4.2,1) (13.5,2) -- (15.3,2) -- (15.3,1) 
+(.4,0) -- +(.4,1) -- +(.8,1) -- +(.8,0) +(1.2,0) -- +(1.2,1) -- +(1.4,1) +(1.8,1) -- +(2,1) -- +(2,0) +(2.4,0) -- +(2.4,1) [rounded corners=\rcB]-- +(2.4,2) --node[below, scale=.8]{$\tilde\ID_n$} (.8,3) -- (.8,1);
\draw (0.8,1.6) -- ++(.4,0) ++(.4,0) -- ++(.4,0) ++(.8,0) -- ++(.4,0) (7.3,1.6) -- ++(.4,0) ++(.4,0) -- ++(.4,0) ++(.8,0) -- ++(.4,0) (15.3,1.6) -- ++(.4,0) ++(.4,0) -- ++(.4,0) ++(.8,0) -- ++(.4,0);
\def \xshif {2.8,5.3}
\node[scale=.9, inner sep=1] (a) at ($(\xshif)+(2,-1.6)$) {$\,\,\tilde\ID_{n+1}$,};\draw[-stealth] (a) -- (.98,1.3);
\node[scale=.9, inner sep=1] (b) at ($(\xshif)+(4.3,-1.6)$) {$\,\,\tilde\ID_{n+2},$};\draw[-stealth] (b) -- (1.77,1.3);
\node[scale=.9, inner sep=1] at ($(\xshif)+(6.2,-1.785)$) {$\ldots\,,$};
\node[scale=.9, inner sep=1] (c) at ($(\xshif)+(7.6,-1.6)$) {$\,\,\tilde\ID_s$};\draw[-stealth] (c) -- (17.55,1.3);
}
\]
Recall that by definition, the sector $V(\Sigma)$ is associated to the original cell decomposition $\Sigma=\ID_1\cup\ldots\cup\ID_n$.
Letting $I_i:= \partial\tilde\ID_{n+i}\cap \tilde\ID_n$ and $M_i:=\partial\tilde\ID_{n+i}\cap \Sigma'$, we see that
\[
\begin{split}
V(\Sigma)&=V(\Sigma')\boxtimes_{\cala(M)}H_0(\partial\ID_n)\\
&=V(\Sigma')\boxtimes_{\cala(M_1\cup\ldots\cup M_s)}H_0(\partial\ID_n)\\
&\cong V(\Sigma')\boxtimes_{\cala(M_1\cup\ldots\cup M_s)}\Big(H_0(\partial\tilde\ID_n)\boxtimes_{\cala(I_1\cup\ldots\cup I_s)} \bigotimes_{i\ge 1}H_0(\partial\tilde\ID_{n+i})\Big)\\
&\cong \Big(\Big(V(\Sigma')\otimes H_0(\partial \tilde \ID_n)\Big)\boxtimes_{\cala(M_1\cup I_1)}H_0(\partial\tilde\ID_{n+1})\Big)\ldots 
\boxtimes_{\cala(M_s\cup I_s)}H_0(\partial\tilde\ID_{n+s})
\end{split}
\]
is isomorphic to the sector associated to the refined cell decomposition $\tilde\ID_1\cup\ldots\cup\tilde\ID_s$ of $\Sigma$.

For each $i\le s$, let $\tilde\Sigma_i:=\tilde\ID_1\cup\ldots\cup\tilde\ID_i$.
Recall that by assumption $V(\tilde\Sigma_{n-1})=V(\Sigma')$ satisfies the equation \eqref{eq: dim = = prod N}.
Clearly, $V(\tilde\Sigma_{n})\cong V(\tilde\Sigma_{n-1})\otimes H_0(\partial\tilde\ID_n)$ then also satisfies \eqref{eq: dim = = prod N}.
We now proceed by induction on $i$, and assume that $V(\tilde\Sigma_{i-1})$ satisfies \eqref{eq: dim = = prod N} for some $i\le s$.
By construction, the intersection $\partial \tilde\ID_i\cap \tilde\Sigma_{i-1}=M_i\cup I_i$ is the union of exactly two intervals.
If these two intervals belong to the same connected component of $\tilde\Sigma_{i-1}$, then the situation is as in \eqref{eq: case l=1, b_1=2}; 
if they belong to different connected component of $\tilde\Sigma_{i-1}$, then the situation is as in \eqref{eq: case l=2, b_1=b_2=1}.
In both cases, we can apply our intermediate results and conclude that $V(\tilde\Sigma_i)$ satisfies \eqref{eq: dim = = prod N}.
\end{proof}

The following lemma was used in the previous proof:

\begin{lemma}\label{lem: (0000)}
Let $\Sigma$ be a sphere with $n$ holes.
Equip its boundary components $S_1,S_2,\ldots, S_n$ with the orientation induced from $\Sigma$.
Let $\cala$ be a conformal net with finite index, and let $V(\Sigma)$ be the Hilbert space constructed from the following cellular decomposition of $\Sigma$:
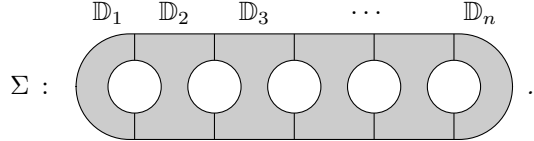
\[
\raisebox{-6pt}{$\Sigma\,:$}\,\,\,\,\,\,\tikzmath[scale=.7]{
\filldraw[fill=gray!40, rounded corners=20] (-.1,-1) rectangle (8.1,1);
\draw (1,0) +(0,1)--+(0,-1)+(-.5,1.4)node{$\ID_1$} (2.5,0) +(0,1)--+(0,-1)+(-.75,1.4)node{$\ID_2$} (4,0) +(0,1)--+(0,-1)+(-.75,1.4)node{$\ID_3$} (5.5,0) +(0,1)--+(0,-1)+(-.15,1.4)node{$\ldots$} (7,0) +(0,1)--+(0,-1)+(.5,1.4)node{$\ID_n$};
\filldraw[fill=white] (1,0) circle (.5) (2.5,0) circle (.5) (4,0) circle (.5) (5.5,0) circle (.5) (7,0) circle (.5);
\useasboundingbox (current bounding box.north west) rectangle ($(current bounding box.south east)-(0,.2)$);}\,\,\,\raisebox{-6pt}{.}
\]
Then the vector space
\[
\hom_{\Rep_{\partial\Sigma}(\cala)}\big(H_{\lambda_1}(S_1)\otimes\ldots\otimes H_{\lambda_n}(S_n),V(\Sigma)\big)
\]
has dimension $N_0^{\lambda_1,\ldots,\lambda_n}$.
\end{lemma}

\begin{proof}
Let us number the boundary components of $\Sigma$ as follows:
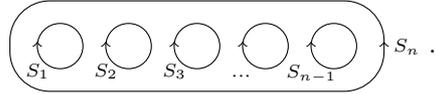
\[
\,\,\tikzmath[scale=.6]{
\draw[rounded corners=15] (-.1,-1) rectangle (8.1,1);
\draw (1,0)(2.5,0)(4,0)(5.5,0)(7,0);
\draw (1,0) circle (.5) +(-135:.7)node[yshift=-1]{$\scriptstyle S_1$} (2.5,0) circle (.5) +(-135:.7)node[yshift=-1]{$\scriptstyle S_2$} (4,0) circle (.5) +(-135:.7)node[yshift=-1]{$\scriptstyle S_3$} (5.5,0) circle (.5) +(-135:.7)node[yshift=-3]{$\scriptstyle \cdots$} (7,0) circle (.5) +(-135:.7)node[yshift=-2]{$\scriptstyle S_{n-1}$};
\draw[->] (8.1,.04)node[right,yshift=-1]{$\scriptstyle S_n$};\draw[->] (.5,.04);\draw[->] (2,.04);\draw[->] (3.5,.04);\draw[->] (5,.04);\draw[->] (6.5,.04);
\useasboundingbox ($(current bounding box.north west)+(0,.1)$) rectangle ($(current bounding box.south east)-(0,.1)$);}.
\]
Let $\Sigma':=\ID_1\cup\ldots\cup\ID_{n-1}$, and let $S'$ be the component of $\partial\Sigma'$ that touches $\ID_n$.
Let also $J_1:=S'\cap S_{n-1}$ and $J_2:=S'\cap S_n$:
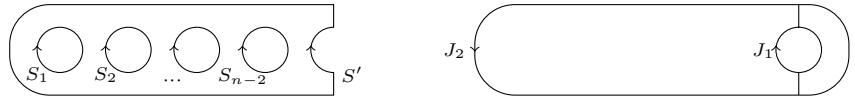
\[
\quad\,\,\tikzmath[scale=.6]{
\draw[rounded corners=15] (7,-1) -- (-.1,-1) -- (-.1,1) [sharp corners]-- (7,1) -- (7,.5) arc (90:270:.5) -- cycle;
\draw (1,0)(2.5,0)(4,0)(5.5,0)(7,0);
\draw (1,0) circle (.5) +(-135:.7)node[yshift=-1]{$\scriptstyle S_1$} (2.5,0) circle (.5) +(-135:.7)node[yshift=-1]{$\scriptstyle S_2$} (4,0) circle (.5) +(-135:.7)node[yshift=-3]{$\scriptstyle \cdots$} (5.5,0) circle (.5) +(-135:.7)node[yshift=-2]{$\scriptstyle S_{n-2}$};
\draw[->] (.5,.04);\draw[->] (2,.04);\draw[->] (3.5,.04);\draw[->] (5,.04);\draw[->] (6.5,.04)node[right,yshift=-10,xshift=8]{$\scriptstyle S'$};
\useasboundingbox ($(current bounding box.north west)+(0,.1)$) rectangle ($(current bounding box.south east)-(0,.1)$);}
\qquad
\,\,\tikzmath[scale=.6]{
\draw[rounded corners=15] (-.1,-1) rectangle (8.1,1);
\draw (1,0)(2.5,0)(4,0)(5.5,0)(7,0);
\draw (7,0) circle (.5) (7,0) +(0,.5) -- +(0,1) +(0,-.5) -- +(0,-1);
\draw[->] (6.5,.04)node[left,yshift=-1, xshift=3]{$\scriptstyle J_1$};
\draw[->] (-.1,-.03) -- (-.1,-.04)node[left,yshift=.4]{$\scriptstyle J_2$};
\useasboundingbox ($(current bounding box.north west)+(0,.1)$) rectangle ($(current bounding box.south east)-(0,.1)$);}\,.
\]
By induction we may assume that
\[
\dim\hom\big(H_{\lambda_1}(S_1)\otimes\ldots\otimes H_{\lambda_{n-2}}(S_{n-2})\otimes H_{\nu}(S'),V(\Sigma')\big)  = N_0^{\lambda_1,\ldots,\lambda_{n-2},\nu}.
\]
Applying~\eqref{eq: sum decomposition of arbitrary sector*} with
respect to the circle $S'$ we obtain a decomposition
\begin{equation*}
  V(\Sigma') \cong \bigoplus_{\nu \in \Delta} M_\nu  \otimes H_\nu(S') 
\end{equation*}
where the $M_\nu$ are $S_1 \sqcup \ldots \sqcup S_{n-2}$-sectors.
Thus
\begin{equation*}
\begin{split}
  \hom( H_{\lambda_1}(S_1) \ox  \dots \ox  H_{\lambda_{n-2}} & (S_{n-2}),  M_\nu) \\
   &  \cong \hom( H_{\lambda_1}(S_1) \ox \dots \ox H_{\lambda_{n-2}}(S_{n-2})
     \ox H_\nu(S' ) , V(\Sigma') ) 
\end{split}
\end{equation*} 
is of dimension $N_0^{\lambda_1,\ldots,\lambda_{n-2},\nu}$.
Let $M:=\ID_{n-1}\cap \ID_n$, so that $V(\Sigma)=V(\Sigma')\boxtimes_{\cala(M)}H_0(\partial\ID_n)$.
We then have
\[
\begin{split}
&\hom\big(H_{\lambda_1}(S_1)\otimes\ldots\otimes H_{\lambda_n}(S_n),V(\Sigma)\big)\\
\cong\,\, &
\hom\big(H_{\lambda_1}(S_1)\otimes\ldots\otimes H_{\lambda_n}(S_n),V(\Sigma')\boxtimes_{\cala(M)}H_0(\partial\ID_n)\big)\\
\cong\,\, &
\hom\big(H_{\lambda_1}(S_1) \otimes\ldots\otimes H_{\lambda_n}(S_n),
           \bigoplus_{\nu \in \Delta} M_\nu  \otimes H_\nu(S') 
                     \boxtimes_{\cala(M)}H_0(\partial\ID_n)\big)\\
\cong\,\, &\bigoplus_{\nu\in\Delta}N_0^{\lambda_1,\ldots,\lambda_{n-2},\nu}\hom\Big(H_{\lambda_{n-1}}(S_{n-1})\otimes H_{\lambda_n}(S_n),H_\nu(S')\boxtimes_{\cala(M)}H_0(\partial\ID_n)\Big).
\end{split}
\]
By Frobenius reciprocity,
\[
\begin{split}
&\hom\big(H_{\lambda_{n-1}}(S_{n-1})\otimes H_{\lambda_n}(S_n),H_\nu(S')\boxtimes_{\cala(M)}H_0(\partial\ID_n)\big)\\
\cong\,\,&\hom\Big(\,\overline{H_0(\partial\ID_n)}\,,\,\big(\,\overline{H_{\lambda_{n-1}}(S_{n-1})}\otimes \overline{H_{\lambda_n}(S_n)}\,\big)\boxtimes_{\cala(J_1)\,\bar\otimes\,\cala(J_2)}H_\nu(S')\Big)\\
\cong\,\,&
\hom\Big(\,\overline{H_0(\partial\ID_n)}\,,\,\overline{H_{\lambda_{n-1}}(S_{n-1})}\boxtimes_{\cala(J_1)}H_\nu(S')\boxtimes_{\cala(\bar J_2)} \overline{H_{\lambda_n}(S_n)}\,\Big)\\
\cong\,\,&
\hom\Big(\,H_0(\bar{\partial\ID}_n)\,,\,H_{\bar\lambda_{n-1}}(\bar S_{n-1})\boxtimes_{\cala(J_1)}H_\nu(S')\boxtimes_{\cala(\bar J_2)} H_{\bar\lambda_n}(\bar S_n)\Big)\\
\end{split}
\]
is of dimension $N_0^{\bar\lambda_{n-1},\nu,\bar\lambda_n}=N_0^{\bar\nu,\lambda_{n-1},\lambda_n}$, where we have used Lemmas \ref{lem: H_0(-S)} and \ref{lem: dual of H_lambda} in order to identify the duals of $H_0(\partial \ID_n)$, $H_{\lambda_{n-1}}(S_{n-1})$, and $H_{\lambda_n}(S_n)$.
It follows that $\hom(H_{\lambda_1}(S_1)\otimes\ldots\otimes H_{\lambda_n}(S_n),V(\Sigma))$ has dimension
\[\textstyle
\sum_{\nu\in\Delta}\,N_0^{\lambda_1,\ldots,\lambda_{n-2},\nu}\,N_0^{\bar\nu,\lambda_{n-1},\lambda_n}\,=\,N_0^{\lambda_1,\ldots,\lambda_{n-1},\lambda_n}.
\qedhere
\]
\end{proof}

Note that while Lemma \ref{lem: computation of V((000))} pins down the isomorphism type of $V(\Sigma)$, it does not say anything
about how well defined the Hilbert space is.
If the surface $\Sigma$ is a disk, then we can get a little bit more:

\begin{corollary}\label{cor: V(disc) independent of cell decomposition}
Let $\cala$ be a conformal net with finite index and let $\ID$ be a disc.
Then the Hilbert space $V(\ID)$ associated to an ordered regular trivalent smooth cell decomposition of $\ID$ is given by the vacuum sector $H_0(\partial \ID)$.
It is independent of the cell decomposition, and well defined up to canonical-up-to-phase unitary isomorphism.
\end{corollary}

\begin{proof}
The isomorphism $V(\ID)\cong H_0(\partial \ID)$ is immediate from Lemma \ref{lem: computation of V((000))}.
Moreover, by Schur's lemma, since $V(\ID)$ is irreducible, it is well defined up to canonical-up-to-phase unitary isomorphism.
\end{proof}

\subsection{Independence of the cell decomposition}\label{sec: Independence of the cell decomposition}

We have seen in \eqref{eq: V(D_i) -- BIS} that it is possible to define a sector $V(\Sigma)\in\Rep_{\partial \Sigma}(\cala)$ for every ordered regular trivalent smooth cell decomposition of $\Sigma$.
For technical reasons, it is convenient to replace the ordering of the set of 2-cells by a slightly different structure:

\begin{definition}
Let $\Sigma$ be a topological surface with smooth boundary.
A \emph{soccer ball decomposition} is a regular trivalent smooth cell decomposition (see Section \ref{sec: The Hilbert space associated to a surface}), equipped
with a partition of the set of $2$-cells into subsets $X_{\mathrm{white}}$ and $X_{\mathrm{black}}$, such that 
every interior vertex is adjacent to two white cells and one back cell, and every boundary vertex is adjacent to two white cells.
\end{definition}

\[
\parbox{5cm}{We illustrate the above\\ definition with an example:
}\qquad
\tikzmath[scale = .8]{\useasboundingbox (1.3,-1.1) rectangle (5.85,1.1);
\filldraw[fill = white, line width=.5] (1.7,.59) ellipse(.1  and .41);\filldraw[fill = white, line width=.5] (1.7,-.58) ellipse(.1  and .42);\clip (1.7,-1.1) rectangle (5.8,1.1);
\draw[gray, line width=.6, rounded corners=8.5, fill=gray!20](.3,.4) -- (1,1) -- (2,1) --  (3,.75) -- (4,1) -- (5,1) -- (5.7,.4) -- (5.7,-.4) -- (5,-1) -- (4,-1) -- (3,-.75) -- (2,-1) -- (1,-1) -- (.3,-.4) -- cycle;
\fill[white](1.19,-.01) .. controls (1.5,.25) and (2,.25) .. (2.31,-.01)(1.19,-.01) .. controls (1.5,-.21) and (2,-.21) .. (2.31,-.01)
(3.69,-.01) .. controls (4,.25) and (4.5,.25) .. (4.81,-.01)(3.69,-.01) .. controls (4,-.21) and (4.5,-.21) .. (4.81,-.01);
\draw[gray, line width=.6] (1,.1) .. controls (1.5,-.25) and (2,-.25) .. (2.5,.1)(1.2,0) .. controls (1.5,.25) and (2,.25) .. (2.3,0)
(3.5,.1) .. controls (4,-.25) and (4.5,-.25) .. (5,.1)(3.7,0) .. controls (4,.25) and (4.5,.25) .. (4.8,0);
\filldraw[fill = white, line width=.5] (1.7,.59) ellipse(.1 and .41);\filldraw[fill = white, line width=.5] (1.7,-.58) ellipse(.1 and .42);
\path[rotate=15](1.8,-1.3)coordinate(a)(3.8,-.5)coordinate(b)(4.3,-1.9)coordinate(c)(4.6,-1.48)coordinate(d);
\path[rotate=-15](1.95,.85)coordinate(e)(2.5,1.3)coordinate(f)(4.7,1.9)coordinate(g)(3.96,.71)coordinate(h);\path[rotate=-8](3.58,-.2)coordinate(i);
\path(2.4,-.35)coordinate(j)(3.2,.2)coordinate(k)(3.05,-.3)coordinate(l)(4.4,.75)coordinate(m)(5.4,-.1)coordinate(n)(3.5,-.875)coordinate(ii);
\fill[fill=gray, line width=.3]  (m) to[bend right=5] (g) -- +(.18,.19) to[bend right=18] ($(m)+(-.02,.25)$) -- cycle ;
\fill[fill=gray, line width=.3] ($(h)+(.035,.2)$) to[bend right =17] ($(d)+(-.11,.18)$) -- (d) to[bend left=10] (c) to[bend left=5] (h) -- cycle;
\fill[gray](ii) -- (i) -- (l) -- (j) to[bend left=10] (a) -- ++(.03,-.17) to[bend right=2] ($(a)+(.03,-.17)+(.6,.135)$) -- ++(.2,.02) -- ++(.2,0) -- ++(.15,-.015) -- ++(.15,-.03) -- cycle;
\draw[line width=.3](a) -- +(-.285,.015)(j) to[bend left=10] (a) -- +(.03,-.17)(l) -- (j) -- +(-.26,.26)(i) -- (l) (h) -- (i) -- (ii)(d) to[bend left=10] (c) -- +(.045,-.275) (c) to[bend left=5] (h) -- +(.035,.2) (b) -- (m) to[bend right=5] (g) to[bend left=10] (n) to[bend left=5] (d) -- +(-.11,.18)(g) -- +(.18,.19)(m) -- +(-.02,.25);
\fill[gray, line width=.1]  (2.4,.9) to[bend right =8] (3.05,.805) to[bend right =3] (3.3,.37) -- (3.05,.05) to[bend right =5] (2.3,.475) -- cycle;
\draw[line width=.3]  (3.05,.805) to[bend right =3] (3.3,.37) -- (3.05,.05) to[bend right =5] (2.3,.475) -- (2.4,.9);
\draw[line width=.3] (2.3,.475) -- (1.785,.43) (3.3,.37) -- (3.66,.44)(3.05,.05) -- (l) (3.75,.7) -- +(-.04,.223) (3.93,.3) -- +(.07,-.145);
\filldraw[fill=gray, line width=.3] (3.66,.44) -- (3.75,.7) -- (4,.585) -- (3.93,.3) -- cycle;
\fill[gray] (5.63,-.33) to[bend right=18] (5.66,.25) -- (5.312,.15) -- (5.28,-.14) -- cycle;
\draw[line width=.3]  (5.66,.25) -- (5.312,.15) -- (5.28,-.14) -- (5.63,-.33);
} 
\medskip \]
Such decompositions exist in great abundance:
starting with any regular cell decomposition, we can get a soccer ball decomposition by replacing each internal vertex by a little black disc.
Here is an illustration of the above procedure:

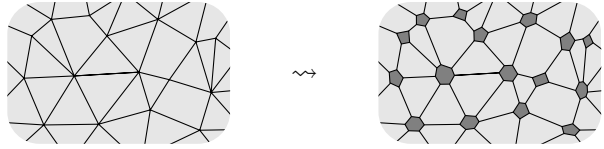
\[
\tikzmath[scale = .1]{\foreach \b/\x/\y in {a/.8/18,b/1.7/13.5,c/-.5/9.5,d/2/5,e/.5/1.5,f/5/20,g/5.4/11.8,h/5.8/7,i/5/0,j/9/19,k/9.8/14,l/9.7/5,m/9/1.5,n/13.2/20,o/13.7/11.3,p/13.5/6.9,q/13/0,r/16.5/18,s/17.5/14,t/16.5/2.5,u/19.5/20,v/23/20,w/25/18,x/22.3/13.8,y/25/10.2,z/24.7/5.4,aa/21.6/2.2,bb/29/19,cc/30.5/15,dd/29/11.5,ee/30.3/8.8,eee/30.2/6.2,ff/28.7/4.1,gg/29.8/.6,hh/25.3/0,A/6.5/17,B/3.9/14.8,C/2.9/9.2,D/5.5/3,E/11.5/17.5,F/13.8/14.8,G/9.5/9.5,H/12.7/3.1,I/20.5/17,J/18/10,K/22/9,L/19.5/5,M/25.5/14,N/28.1/14.5,O/27.5/7.3,P/26/2.3,A'/-2.5/3.5,B'/-2/14,B''/-2/17.3,C'/0/21,D'/7.5/22,D''/10.6/22,E'/16.2/20.5,F'/22/23,G'/27.2/20.7,H'/31.5/19.5,I'/31.2/12.2,J'/33/7.8,K'/31.8/2.4,L'/29/-1.5,M'/18.5/-1.3,M''/16/0,N'/9/-2,O'/-.5/-2.5} {\coordinate (\b) at (\x,\y);}
\clip[rounded corners = 12] (.7,.5) rectangle (30,19.3);
\fill[gray!20](.7,.5) rectangle (30,19.3);
\draw[line join=bevel] (A)--(G)--(J)--(I)--(F)--(J)--(G)--(F)--(E)--(A)--(B)--(C)--(D)--(G)--(H)--(D) (B)--(G)--(C) (J)--(L)--(H)--(J)--(K)--(L)--(P)--(O)--(K)--(M)--(O)--(N)--(M)--(I) 
(D)--(A')--(C)--(B')(B'')--(B)(C')--(A)--(D')(D'')--(E)--(E')--(F)(E')--(I)--(F')(I)--(G')--(M)(H')--(N)--(I')--(O)--(J')(O)--(K')--(P)--(L')(P)--(M')--(L)(M'')--(H)--(N')--(D)--(O');
} 
\qquad\rightsquigarrow\qquad\tikzmath[scale = .1]{\foreach \b/\x/\y in {a/.8/18,b/1.7/13.5,c/-.5/9.5,d/2/5,e/.5/1.5,f/5/20,g/5.4/11.8,h/5.8/7,i/5/0,j/9/19,k/9.8/14,l/9.7/5,m/9/1.5,n/13.2/20,o/13.7/11.3,p/13.5/6.9,q/13/0,r/16.5/18,s/17.5/14,t/16.5/2.5,u/19.5/20,v/23/20,w/25/18,x/22.3/13.8,y/25/10.2,z/24.7/5.4,aa/21.6/2.2,bb/29/19,cc/30.5/15,dd/29/11.5,ee/30.3/8.8,eee/30.2/6.2,ff/28.7/4.1,gg/29.8/.6,hh/25.3/0,A/6.5/17,B/3.9/14.8,C/2.9/9.2,D/5.5/3,E/11.5/17.5,F/13.8/14.8,G/9.5/9.5,H/12.7/3.1,I/20.5/17,J/18/10,K/22/9,L/19.5/5,M/25.5/14,N/28.1/14.5,O/27.5/7.3,P/26/2.3,A'/-2.5/3.5,B'/-2/14,B''/-2/17.3,C'/0/21,D'/7.5/22,D''/10.6/22,E'/16.2/20.5,F'/22/23,G'/27.2/20.7,H'/31.5/19.5,I'/31.2/12.2,J'/33/7.8,K'/31.8/2.4,L'/29/-1.5,M'/18.5/-1.3,M''/16/0,N'/9/-2,O'/-.5/-2.5} {\coordinate (\b) at (\x,\y);}
\clip[rounded corners = 12](.7,.5) rectangle (30,19.3);
\fill[gray!20](.7,.5) rectangle (30,19.3);
\draw[line join=bevel] (A)--(G)--(J)--(I)--(F)--(J)--(G)--(F)--(E)--(A)--(B)--(C)--(D)--(G)--(H)--(D) (B)--(G)--(C) (J)--(L)--(H)--(J)--(K)--(L)--(P)--(O)--(K)--(M)--(O)--(N)--(M)--(I) 
(D)--(A')--(C)--(B')(B'')--(B)(C')--(A)--(D')(D'')--(E)--(E')--(F)(E')--(I)--(F')(I)--(G')--(M)(H')--(N)--(I')--(O)--(J')(O)--(K')--(P)--(L')(P)--(M')--(L)(M'')--(H)--(N')--(D)--(O');
\filldraw[fill=gray]
($(A)!.31!(B)$) -- ($(A)!.12!(G)$) -- ($(A)!.19!(E)$) -- ($(A)!.32!($(f)!.5!(j)$)$) -- ($(A)!.24!($(f)!.45!(a)$)$) -- cycle
($(B)!.3!(A)$) -- ($(B)!.16!(G)$) -- ($(B)!.19!(C)$) -- ($(B)!.3!($(a)!.45!(b)$)$) -- cycle
($(C)!.2!(B)$) -- ($(C)!.16!(G)$) -- ($(C)!.19!(D)$) -- ($(C)!.32!($(c)!.5!(d)$)$) -- ($(C)!.34!($(c)!.5!(b)$)$) -- cycle
($(D)!.18!(C)$) -- ($(D)!.15!(G)$) -- ($(D)!.17!(H)$) -- ($(D)!.36!($(i)!.5!(m)$)$) -- ($(D)!.33!($(e)!.55!(i)$)$) -- ($(D)!.31!($(e)!.5!(d)$)$) -- cycle
($(E)!.2!(A)$) -- ($(E)!.25!(F)$) -- ($(E)!.3!($(n)!.33!(r)$)$) -- ($(E)!.45!($(j)!.5!(n)$)$) -- cycle
($(F)!.24!(E)$) -- ($(F)!.12!(G)$) -- ($(F)!.17!(J)$) -- ($(F)!.2!(I)$) -- ($(F)!.3!($(n)!.67!(r)$)$) -- cycle
($(G)!.2!(A)$) -- ($(G)!.2!(B)$) -- ($(G)!.2!(C)$) -- ($(G)!.17!(D)$) -- ($(G)!.17!(H)$) -- ($(G)!.14!(J)$) -- ($(G)!.19!(F)$) -- cycle
($(H)!.18!(D)$) -- ($(H)!.18!(G)$) -- ($(H)!.15!(J)$) -- ($(H)!.18!(L)$) -- ($(H)!.3!($(t)!.5!(q)$)$) -- ($(H)!.3!($(m)!.5!(q)$)$) -- cycle
($(I)!.18!(F)$) -- ($(I)!.15!(J)$) -- ($(I)!.18!(M)$) -- ($(I)!.3!($(v)!.5!(w)$)$) -- ($(I)!.3!($(v)!.5!(u)$)$) -- ($(I)!.3!($(u)!.5!(r)$)$) -- cycle
($(J)!.18!(F)$) -- ($(J)!.17!(G)$) -- ($(J)!.16!(H)$) -- ($(J)!.24!(L)$) -- ($(J)!.25!(K)$) -- ($(J)!.13!(I)$) -- cycle
($(K)!.265!(J)$) -- ($(K)!.18!(M)$) -- ($(K)!.2!(O)$) -- ($(K)!.23!(L)$) -- cycle
($(L)!.16!(H)$) -- ($(L)!.19!(J)$) -- ($(L)!.21!(K)$) -- ($(L)!.16!(P)$) -- ($(L)!.45!($(aa)!.5!(t)$)$) -- cycle
($(M)!.18!(I)$) -- ($(M)!.2!(K)$) -- ($(M)!.17!(O)$) -- ($(M)!.4!(N)$) -- ($(M)!.3!($(w)!.4!(bb)$)$) -- cycle
($(N)!.2!(M)$) -- ($(N)!.18!(O)$) -- ($(N)!.5!($(cc)!.5!(dd)$)$) -- ($(N)!.2!($(cc)!.5!(bb)$)$) -- cycle
($(O)!.2!(N)$) -- ($(O)!.2!(M)$) -- ($(O)!.16!(K)$) -- ($(O)!.18!(P)$) -- ($(O)!.25!($(ff)!.5!(eee)$)$) -- ($(O)!.25!($(eee)!.5!(ee)$)$) -- ($(O)!.3!($(ee)!.5!(dd)$)$) -- cycle
($(P)!.17!(O)$) -- ($(P)!.18!(L)$) -- ($(P)!.36!($(aa)!.5!(hh)$)$) -- ($(P)!.35!($(hh)!.5!(gg)$)$) -- ($(P)!.36!($(gg)!.5!(ff)$)$) -- cycle;
} 
\]

Given a soccer ball decomposition $X_{\mathrm{white}}=\{\ID_1,\ldots,\ID_n\}$, $X_{\mathrm{black}}=\{B_1,\ldots,B_m\}$ of our surface $\Sigma$, 
we can pick an order on those two sets,
and apply the construction~\eqref{eq: V(D_i) -- BIS} to get
\[
\begin{split}
V(\Sigma)=\,&\bigg(\hspace{-.5mm}\bigg(\hspace{-.5mm}\bigg(\cdots\bigg(\\
&\,H_0(\partial\ID_1)\underset{\cala(\partial \ID_2\cap\ID_1)}\boxtimes H_0(\partial\ID_2)\bigg)\underset{\cala(\partial\ID_3\cap(\ID_1\cup\ID_2))}\boxtimes
\,\cdots\bigg)\underset{\cala(\partial\ID_n\cap\,\ldots)}\boxtimes H_0(\partial\ID_n)\bigg)\\
&\qquad\qquad\,\,\,\,\underset{\cala(\partial B_1)}\boxtimes H_0(\partial B_1)\bigg)\underset{\cala(\partial B_2)}\boxtimes
H_0(\partial B_2)\,\,\cdots\bigg)\underset{\cala(\partial B_m)}\boxtimes H_0(\partial B_m).\phantom{\Bigg|}
\end{split}
\]
The advantage of soccer ball decompositions is that, given choices of vacuum sectors $H_0(\partial \ID_i)$ and $H_0(\partial B_j)$,
the sector $V(\Sigma)\in\Rep_{\partial\Sigma}(\cala)$ is visibly invariant (up to canonical unitary isomorphism) under permutations of the sets $\{\ID_1,\ldots,\ID_n\}$ and $\{B_1,\ldots, B_m\}$---this is seen as follows.
The iterated fusion
\begin{equation}\label{eq: V(white Sigma)}
\,\Big(\!\Big(\raisebox{.3mm}{$\scriptstyle\cdots$}\hspace{-.2mm}\Big(H_0(\partial\ID_1)\boxtimes_{\cala(\partial \ID_2\cap\ID_1)} H_0(\partial\ID_2)\Big)\boxtimes_{\cala(\partial\ID_3\cap\,\ldots)}\cdots\Big)\boxtimes_{\cala(\partial\ID_n\cap\,\ldots)} H_0(\partial\ID_n)\Big)
\end{equation}
can be identified with the graph fusion $\boxtimes_{\{\cala(\ID_i\cap \ID_j)\}} \{H_0(\partial\ID_i)\}$
of the Hilbert spaces $H_0(\partial \ID_i)$ along the algebras $\cala(\ID_i\cap \ID_j)$ (see Appendix~\ref{app:Cyclic fusion} regarding graph fusion).
Here, the manifolds $\ID_i\cap \ID_j$ are equipped with the orientation induced from $\partial\ID_j$ for $j>i$,
the graph that indexes the graph fusion is the adjacency graph $\Gamma$ of the white cells of $X$,
and the orientations on $\Gamma$ are induced by the ordering of its set of vertices.
Renumbering the cells $\ID_1,\ldots,\ID_n$ affects both the orientations of the $\ID_i\cap \ID_j$, and the orientations of the edges of $\Gamma$.
It follows from \eqref{eq: orientations in graph fusion} that those permutations leave \eqref{eq: V(white Sigma)} invariant.
Altogether, $V(\Sigma)$ is canonically isomorphic to
\begin{equation}\label{eq: graph + fill holes}
\Big(\bigboxtimes_{\{\cala(\ID_i\cap \ID_j)\}} \big\{H_0(\partial\ID_i)\big\}\Big)\,\boxtimes_{\cala(\partial B_1\cup\ldots\cup\,\partial B_m)}\big(H_0(\partial B_1)\otimes\ldots\otimes H_0(\partial B_m)\big),
\end{equation}
and is therefore also invariant under the permutations of the cells $B_1,\ldots, B_m$.

By the above discussion, given a topological surface $\Sigma$ with smooth boundary, and a soccer ball decomposition $X=(X_{\mathrm{white}}, X_{\mathrm{black}})$ of that surface, one can associate to it, canonically up to canonical-up-to-phase unitary isomorphism, a $\partial\Sigma$-sector $V(\Sigma)$.
We denote the result 
\begin{equation}\label{eq: V(Sigma;X)}
V(\Sigma;X)=V(\Sigma;X_{\mathrm{white}}, X_{\mathrm{black}})\,\in\,\Rep_{\partial\Sigma}(\cala)
\end{equation}
to emphasize the dependence on the soccer ball decomposition.
The Hilbert space $V(\Sigma;X)$ is only well defined up to canonical-up-to-phase unitary,
because so were the vacuum sectors $H_0(\ID_i)$ and $H_0(B_j)$ that entered its definition.

Our next goal is to show that $V(\Sigma;X)$ does not depend on the soccer ball decomposition $X$.

\begin{definition}
Let $\Sigma$ be a surface with smooth boundary, and let $X$ and $Y$ be soccer ball decompositions.
We write $X\tikz{\useasboundingbox(-.24,-.1)rectangle(.24,.3);\node at (0,.23) {$\scriptstyle \ID$};\node at (0,0) {$\prec$};} Y$
with $\ID\in Y_{\mathrm{white}}$ if $\ID$ is a union of cells of $X$, the decompositions $X$ and $Y$ agree outside of $\ID$, and the inclusion of the $1$-skeleton
of $Y$ into the $1$-skeleton of $X$ is smooth.
If there \vspace{.05cm} exists another soccer ball decomposition $Z$ and a cell $\ID\in Z_{\mathrm{white}}$ such that 
$X\tikz{\useasboundingbox(-.24,-.1)rectangle(.24,.3);\node at (0,.23) {$\scriptstyle \ID$};\node at (0,0) {$\prec$};} Z
\tikz{\useasboundingbox(-.24,-.1)rectangle(.24,.3);\node at (0,.23) {$\scriptstyle \ID$};\node at (0,0) {$\succ$};}  Y$, then we write $X\overset\ID\sim Y$.
For example, we have
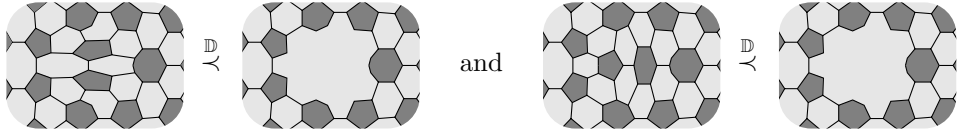
\[\tikzmath[scale = .088]{\pgftransformxscale{.95}\useasboundingbox (.6,-10.5) rectangle (28.5,11);\clip[rounded corners = 12] (.6,-9.5) rectangle (28.5,9.5);
\fill[gray!20] (0,-10) rectangle (30,10);\fill[gray] (5.3,1.1)--(3.5,3.5)--(5.2,5.6)--(7.5,4.5)--(7.5, 1.8)--cycle;\fill[gray] (10,6)--(10,8.3)--(12.5,10)--(15,8)--(14,6)--(12.2,5.2)-- cycle;\fill[gray] (18.2,5.4)-- (17,8)--(19.5,9.8)--(22,8)--(21,5.2)--cycle;\fill[gray] (26,5)--(24.6,7.8)-- (26.3,10)--(28.8,9.8)--(30,7.5)--(28,5)-- cycle;\fill[gray] (20.5,0)--(20.8,1.1)--(22,2.5)--(24.5,2.5)--(25.8,0)--(24.5,-2.5)--(22,-2.5)--(20.8,-1.1)--cycle;\fill[gray] (13.5,.7)--(11,1.8)--(13,3.7)--(16.8,3.8)--(17.2,1.6)--cycle;\draw (0,-1.3)--(0,1.3)(5.3,-1.1)--(5.3,1.1)(25.8,0)--(28,0)(13.5,.7)--(13.5,-.7);\draw (0,1.3)--(1.5,3.5)--(0,6)--(.7,9.1)--(3,10)--(5,8)--(7.5,10)--(10,8.3)--(12.5,10)--(15,8)--(17,8)--(19.5,9.8)--(22,8)--(24.6,7.8)--(26.3,10)--(28.8,9.8)--(30,7.5)--(28,5)--(29.6,2.4)--(28,0)(1.5,3.5)--(3.5,3.5)--(5.2,5.6)--(5,8)(5.2,5.6)--(7.5,4.5)--(10,6)--(10,8.3)(10,6)--(12.2,5.2)--(14,6)--(15,8) (17,8)--(18.2,5.4)--(21,5.2)(3.5,3.5)--(5.3,1.1)--(7.6,1.9)--(7.5,4.5)(24.6,7.8)--(26,5)--(28,5) (26,5)--(24.5,2.5)--(25.8,0)  (22,2.5)--(20.8,1.1)--(20.5,0)(22,8)--(21,5.2)--(22,2.5)--(24.5,2.5);\draw (7.6,1.9)--(11,1.8)--(13.5,.7)--(17.2,1.6)--(20.8,1.1)(11,1.8)--(13,3.7)--(12.2,5.2)(13,3.7)--(16.8,3.8)--(18.2,5.4)(16.8,3.8)--(17.2,1.6);\pgftransformyscale{-1}\fill[gray] (5.3,1.1)--(3.5,3.5)--(5.2,5.6)--(7.5,4.5)--(7.5, 1.8)--cycle;\fill[gray] (10,6)--(10,8.3)--(12.5,10)--(15,8)--(14,6)--(12.2,5.2)-- cycle;\fill[gray] (18.2,5.4)-- (17,8)--(19.5,9.8)--(22,8)--(21,5.2)--cycle;\fill[gray] (26,5)--(24.6,7.8)-- (26.3,10)--(28.8,9.8)--(30,7.5)--(28,5)-- cycle;\fill[gray] (13.5,.7)--(11,1.8)--(13,3.7)--(16.8,3.8)--(17.2,1.6)--cycle;\fill[gray](3,10) -- (5,8) -- (7.5,10) -- cycle (3,-10) -- (5,-8) -- (7.5,-10) -- cycle (28,0) -- (29.6,2.4) -- (29.6,-2.4) -- cycle (0,1.3)--(1.5,3.5)--(0,6)--cycle (0,-1.3)--(1.5,-3.5)--(0,-6)--cycle;\draw (0,1.3)--(1.5,3.5)--(0,6)--(.7,9.1)--(3,10)--(5,8)--(7.5,10)--(10,8.3)--(12.5,10)--(15,8)--(17,8)--(19.5,9.8)--(22,8)--(24.6,7.8)--(26.3,10)--(28.8,9.8)--(30,7.5)--(28,5)--(29.6,2.4)--(28,0)(1.5,3.5)--(3.5,3.5)--(5.2,5.6)--(5,8)(5.2,5.6)--(7.5,4.5)--(10,6)--(10,8.3)(10,6)--(12.2,5.2)--(14,6)--(15,8) (17,8)--(18.2,5.4)--(21,5.2)(3.5,3.5)--(5.3,1.1)--(7.6,1.9)--(7.5,4.5)(24.6,7.8)--(26,5)--(28,5) (26,5)--(24.5,2.5)--(25.8,0)  (22,2.5)--(20.8,1.1)--(20.5,0)(22,8)--(21,5.2)--(22,2.5)--(24.5,2.5);\draw (7.6,1.9)--(11,1.8)--(13.5,.7)--(17.2,1.6)--(20.8,1.1)(11,1.8)--(13,3.7)--(12.2,5.2)(13,3.7)--(16.8,3.8)--(18.2,5.4)(16.8,3.8)--(17.2,1.6);} 
\,\,\tikz{\useasboundingbox(-.24,-.1)rectangle(.24,.3);\node at (0,.23) {$\scriptstyle \ID$};\node at (0,0) {$\prec$};}\,\,\,
\tikzmath[scale = .088]{\pgftransformxscale{.95}\useasboundingbox (.6,-10.5) rectangle (28.5,11);\clip[rounded corners = 12] (.6,-9.5) rectangle (28.5,9.5);\fill[gray!20] (0,-10) rectangle (30,10);\fill[gray] (5.3,1.1)--(3.5,3.5)--(5.2,5.6)--(7.5,4.5)--(7.5, 1.8)--cycle;\fill[gray] (10,6)--(10,8.3)--(12.5,10)--(15,8)--(14,6)--(12.2,5.2)-- cycle;\fill[gray] (18.2,5.4)-- (17,8)--(19.5,9.8)--(22,8)--(21,5.2)--cycle;\fill[gray] (26,5)--(24.6,7.8)-- (26.3,10)--(28.8,9.8)--(30,7.5)--(28,5)-- cycle;\fill[gray] (20.5,0)--(20.8,1.1)--(22,2.5)--(24.5,2.5)--(25.8,0)--(24.5,-2.5)--(22,-2.5)--(20.8,-1.1)--cycle;\fill[gray](3,10) -- (5,8) -- (7.5,10) -- cycle (3,-10) -- (5,-8) -- (7.5,-10) -- cycle (28,0) -- (29.6,2.4) -- (29.6,-2.4) -- cycle (0,1.3)--(1.5,3.5)--(0,6)--cycle (0,-1.3)--(1.5,-3.5)--(0,-6)--cycle;\draw (0,-1.3)--(0,1.3)(5.3,-1.1)--(5.3,1.1)(25.8,0)--(28,0);\draw (0,1.3)--(1.5,3.5)--(0,6)--(.7,9.1)--(3,10)--(5,8)--(7.5,10)--(10,8.3)--(12.5,10)--(15,8)--(17,8)--(19.5,9.8)--(22,8)--(24.6,7.8)--(26.3,10)--(28.8,9.8)--(30,7.5)--(28,5)--(29.6,2.4)--(28,0)(1.5,3.5)--(3.5,3.5)--(5.2,5.6)--(5,8)(5.2,5.6)--(7.5,4.5)--(10,6)--(10,8.3)(10,6)--(12.2,5.2)--(14,6)--(15,8) (17,8)--(18.2,5.4)--(21,5.2)(3.5,3.5)--(5.3,1.1)--(7.6,1.9)--(7.5,4.5)(24.6,7.8)--(26,5)--(28,5) (26,5)--(24.5,2.5)--(25.8,0)  (22,2.5)--(20.8,1.1)--(20.5,0)(22,8)--(21,5.2)--(22,2.5)--(24.5,2.5);\pgftransformyscale{-1}\fill[gray] (5.3,1.1)--(3.5,3.5)--(5.2,5.6)--(7.5,4.5)--(7.5, 1.8)--cycle;\fill[gray] (10,6)--(10,8.3)--(12.5,10)--(15,8)--(14,6)--(12.2,5.2)-- cycle;\fill[gray] (18.2,5.4)-- (17,8)--(19.5,9.8)--(22,8)--(21,5.2)--cycle;\fill[gray] (26,5)--(24.6,7.8)-- (26.3,10)--(28.8,9.8)--(30,7.5)--(28,5)-- cycle;\draw (0,1.3)--(1.5,3.5)--(0,6)--(.7,9.1)--(3,10)--(5,8)--(7.5,10)--(10,8.3)--(12.5,10)--(15,8)--(17,8)--(19.5,9.8)--(22,8)--(24.6,7.8)--(26.3,10)--(28.8,9.8)--(30,7.5)--(28,5)--(29.6,2.4)--(28,0)(1.5,3.5)--(3.5,3.5)--(5.2,5.6)--(5,8)(5.2,5.6)--(7.5,4.5)--(10,6)--(10,8.3)(10,6)--(12.2,5.2)--(14,6)--(15,8) (17,8)--(18.2,5.4)--(21,5.2)(3.5,3.5)--(5.3,1.1)--(7.6,1.9)--(7.5,4.5)(24.6,7.8)--(26,5)--(28,5) (26,5)--(24.5,2.5)--(25.8,0)  (22,2.5)--(20.8,1.1)--(20.5,0)(22,8)--(21,5.2)--(22,2.5)--(24.5,2.5);} 
\,\,\,\,\,\,\,\,\,\text{and}\,\,\,\,\,\,\,\,\,
\tikzmath[scale = .088]{\pgftransformxscale{.95}\useasboundingbox (.6,-10.5) rectangle (28.5,11);\clip[rounded corners = 12] (.6,-9.5) rectangle (28.5,9.5);\fill[gray!20] (0,-10) rectangle (30,10);
\useasboundingbox ($(current bounding box.north west)+(0,1)$) rectangle ($(current bounding box.south east)-(0,1)$);
\fill[gray] (5.3,1.1)--(3.5,3.5)--(5.2,5.6)--(7.5,4.5)--(7.5, 1.8)--cycle;\fill[gray] (10,6)--(10,8.3)--(12.5,10)--(15,8)--(14,6)--(12.2,5.2)-- cycle;\fill[gray] (18.2,5.4)-- (17,8)--(19.5,9.8)--(22,8)--(21,5.2)--cycle;\fill[gray] (26,5)--(24.6,7.8)-- (26.3,10)--(28.8,9.8)--(30,7.5)--(28,5)-- cycle;\fill[gray] (20.5,0)--(20.8,1.1)--(22,2.5)--(24.5,2.5)--(25.8,0)--(24.5,-2.5)--(22,-2.5)--(20.8,-1.1)--cycle;\fill[gray] (9.3,1.2)--(11.8,2.3)--(13.2,0)--(11.8,-2.3)--(9.3,-1.2)--cycle;\fill[gray] (14.9,0)--(15.5,2.8)--(17.5,2.8)--(18.5,0) --(17.5,-2.8)--(15.5,-2.8)--cycle;\draw (0,-1.3)--(0,1.3)(5.3,-1.1)--(5.3,1.1)(25.8,0)--(28,0);\fill[gray](3,10) -- (5,8) -- (7.5,10) -- cycle (3,-10) -- (5,-8) -- (7.5,-10) -- cycle (28,0) -- (29.6,2.4) -- (29.6,-2.4) -- cycle (0,1.3)--(1.5,3.5)--(0,6)--cycle (0,-1.3)--(1.5,-3.5)--(0,-6)--cycle;\draw (0,1.3)--(1.5,3.5)--(0,6)--(.7,9.1)--(3,10)--(5,8)--(7.5,10)--(10,8.3)--(12.5,10)--(15,8)--(17,8)--(19.5,9.8)--(22,8)--(24.6,7.8)--(26.3,10)--(28.8,9.8)--(30,7.5)--(28,5)--(29.6,2.4)--(28,0)(1.5,3.5)--(3.5,3.5)--(5.2,5.6)--(5,8)(5.2,5.6)--(7.5,4.5)--(10,6)--(10,8.3)(10,6)--(12.2,5.2)--(14,6)--(15,8) (17,8)--(18.2,5.4)--(21,5.2)(3.5,3.5)--(5.3,1.1)--(7.6,1.9)--(7.5,4.5)(24.6,7.8)--(26,5)--(28,5) (26,5)--(24.5,2.5)--(25.8,0)  (22,2.5)--(20.8,1.1)--(20.5,0)(22,8)--(21,5.2)--(22,2.5)--(24.5,2.5);\draw (7.6,1.9)--(9.3,1.2)--(9.3,0) (9.3,1.2)--(11.8,2.3)--(12.2,5.2)(11.8,2.3)--(13.2,0)--(14.9,0)--(15.5,2.8)--(17.5,2.8)--(18.5,0)--(20.5,0)(15.5,2.8)--(14,6) (17.5,2.8)--(18.2,5.4);\pgftransformyscale{-1}\fill[gray] (5.3,1.1)--(3.5,3.5)--(5.2,5.6)--(7.5,4.5)--(7.5, 1.8)--cycle;\fill[gray] (10,6)--(10,8.3)--(12.5,10)--(15,8)--(14,6)--(12.2,5.2)-- cycle;\fill[gray] (18.2,5.4)-- (17,8)--(19.5,9.8)--(22,8)--(21,5.2)--cycle;\fill[gray] (26,5)--(24.6,7.8)-- (26.3,10)--(28.8,9.8)--(30,7.5)--(28,5)-- cycle;\draw (0,1.3)--(1.5,3.5)--(0,6)--(.7,9.1)--(3,10)--(5,8)--(7.5,10)--(10,8.3)--(12.5,10)--(15,8)--(17,8)--(19.5,9.8)--(22,8)--(24.6,7.8)--(26.3,10)--(28.8,9.8)--(30,7.5)--(28,5)--(29.6,2.4)--(28,0)(1.5,3.5)--(3.5,3.5)--(5.2,5.6)--(5,8)(5.2,5.6)--(7.5,4.5)--(10,6)--(10,8.3)(10,6)--(12.2,5.2)--(14,6)--(15,8) (17,8)--(18.2,5.4)--(21,5.2)(3.5,3.5)--(5.3,1.1)--(7.6,1.9)--(7.5,4.5)(24.6,7.8)--(26,5)--(28,5) (26,5)--(24.5,2.5)--(25.8,0)  (22,2.5)--(20.8,1.1)--(20.5,0)(22,8)--(21,5.2)--(22,2.5)--(24.5,2.5);\draw (7.6,1.9)--(9.3,1.2)--(9.3,0) (9.3,1.2)--(11.8,2.3)--(12.2,5.2)(11.8,2.3)--(13.2,0)  (14.9,0)--(15.5,2.8)--(17.5,2.8)--(18.5,0)  (20.5,0)(15.5,2.8)--(14,6) (17.5,2.8)--(18.2,5.4);} 
\,\,\tikz{\useasboundingbox(-.24,-.1)rectangle(.24,.3);\node at (0,.23) {$\scriptstyle \ID$};\node at (0,0) {$\prec$};}\,\,\,
\tikzmath[scale = .088]{\pgftransformxscale{.95}\useasboundingbox (.6,-10.5) rectangle (28.5,11);\clip[rounded corners = 12] (.6,-9.5) rectangle (28.5,9.5);\fill[gray!20] (0,-10) rectangle (30,10);\fill[gray] (5.3,1.1)--(3.5,3.5)--(5.2,5.6)--(7.5,4.5)--(7.5, 1.8)--cycle;\fill[gray] (10,6)--(10,8.3)--(12.5,10)--(15,8)--(14,6)--(12.2,5.2)-- cycle;\fill[gray] (18.2,5.4)-- (17,8)--(19.5,9.8)--(22,8)--(21,5.2)--cycle;\fill[gray] (26,5)--(24.6,7.8)-- (26.3,10)--(28.8,9.8)--(30,7.5)--(28,5)-- cycle;\fill[gray] (20.5,0)--(20.8,1.1)--(22,2.5)--(24.5,2.5)--(25.8,0)--(24.5,-2.5)--(22,-2.5)--(20.8,-1.1)--cycle;\fill[gray](3,10) -- (5,8) -- (7.5,10) -- cycle (3,-10) -- (5,-8) -- (7.5,-10) -- cycle (28,0) -- (29.6,2.4) -- (29.6,-2.4) -- cycle (0,1.3)--(1.5,3.5)--(0,6)--cycle (0,-1.3)--(1.5,-3.5)--(0,-6)--cycle;\draw (0,-1.3)--(0,1.3)(5.3,-1.1)--(5.3,1.1)(25.8,0)--(28,0);\draw (0,1.3)--(1.5,3.5)--(0,6)--(.7,9.1)--(3,10)--(5,8)--(7.5,10)--(10,8.3)--(12.5,10)--(15,8)--(17,8)--(19.5,9.8)--(22,8)--(24.6,7.8)--(26.3,10)--(28.8,9.8)--(30,7.5)--(28,5)--(29.6,2.4)--(28,0)(1.5,3.5)--(3.5,3.5)--(5.2,5.6)--(5,8)(5.2,5.6)--(7.5,4.5)--(10,6)--(10,8.3)(10,6)--(12.2,5.2)--(14,6)--(15,8) (17,8)--(18.2,5.4)--(21,5.2)(3.5,3.5)--(5.3,1.1)--(7.6,1.9)--(7.5,4.5)(24.6,7.8)--(26,5)--(28,5) (26,5)--(24.5,2.5)--(25.8,0)  (22,2.5)--(20.8,1.1)--(20.5,0)(22,8)--(21,5.2)--(22,2.5)--(24.5,2.5);\pgftransformyscale{-1}\fill[gray] (5.3,1.1)--(3.5,3.5)--(5.2,5.6)--(7.5,4.5)--(7.5, 1.8)--cycle;\fill[gray] (10,6)--(10,8.3)--(12.5,10)--(15,8)--(14,6)--(12.2,5.2)-- cycle;\fill[gray] (18.2,5.4)-- (17,8)--(19.5,9.8)--(22,8)--(21,5.2)--cycle;\fill[gray] (26,5)--(24.6,7.8)-- (26.3,10)--(28.8,9.8)--(30,7.5)--(28,5)-- cycle;\draw (0,1.3)--(1.5,3.5)--(0,6)--(.7,9.1)--(3,10)--(5,8)--(7.5,10)--(10,8.3)--(12.5,10)--(15,8)--(17,8)--(19.5,9.8)--(22,8)--(24.6,7.8)--(26.3,10)--(28.8,9.8)--(30,7.5)--(28,5)--(29.6,2.4)--(28,0)(1.5,3.5)--(3.5,3.5)--(5.2,5.6)--(5,8)(5.2,5.6)--(7.5,4.5)--(10,6)--(10,8.3)(10,6)--(12.2,5.2)--(14,6)--(15,8) (17,8)--(18.2,5.4)--(21,5.2)(3.5,3.5)--(5.3,1.1)--(7.6,1.9)--(7.5,4.5)(24.6,7.8)--(26,5)--(28,5) (26,5)--(24.5,2.5)--(25.8,0)  (22,2.5)--(20.8,1.1)--(20.5,0)(22,8)--(21,5.2)--(22,2.5)--(24.5,2.5);} 
\]
for the same disk $\ID$, from which it follows that
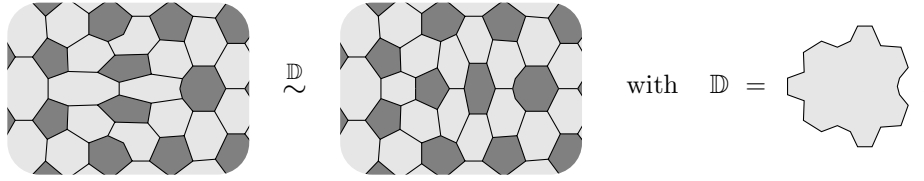
\[
\tikzmath[scale = .12]{\pgftransformxscale{.95}\useasboundingbox (.6,-10.5) rectangle (28.5,11);
\clip[rounded corners = 12] (.6,-9.5) rectangle (28.5,9.5);\fill[gray!20] (0,-10) rectangle (30,10);
\fill[gray] (5.3,1.1)--(3.5,3.5)--(5.2,5.6)--(7.5,4.5)--(7.5, 1.8)--cycle (10,6)--(10,8.3)--(12.5,10)--(15,8)--(14,6)--(12.2,5.2)-- cycle (18.2,5.4)-- (17,8)--(19.5,9.8)--(22,8)--(21,5.2)--cycle (26,5)--(24.6,7.8)-- (26.3,10)--(28.8,9.8)--(30,7.5)--(28,5)-- cycle (20.5,0)--(20.8,1.1)--(22,2.5)--(24.5,2.5)--(25.8,0)--(24.5,-2.5)--(22,-2.5)--(20.8,-1.1)--cycle (13.5,.7)--(11,1.8)--(13,3.7)--(16.8,3.8)--(17.2,1.6)--cycle;\fill[gray](3,10) -- (5,8) -- (7.5,10) -- cycle (3,-10) -- (5,-8) -- (7.5,-10) -- cycle (28,0) -- (29.6,2.4) -- (29.6,-2.4) -- cycle (0,1.3)--(1.5,3.5)--(0,6)--cycle (0,-1.3)--(1.5,-3.5)--(0,-6)--cycle;
\draw (0,-1.3)--(0,1.3)(5.3,-1.1)--(5.3,1.1)(25.8,0)--(28,0)(13.5,.7)--(13.5,-.7);
\draw (0,1.3)--(1.5,3.5)--(0,6)--(.7,9.1)--(3,10)--(5,8)--(7.5,10)--(10,8.3)--(12.5,10)--(15,8)--(17,8)--(19.5,9.8)--(22,8)--(24.6,7.8)--(26.3,10)--(28.8,9.8)--(30,7.5)--(28,5)--(29.6,2.4)--(28,0)(1.5,3.5)--(3.5,3.5)--(5.2,5.6)--(5,8)(5.2,5.6)--(7.5,4.5)--(10,6)--(10,8.3)(10,6)--(12.2,5.2)--(14,6)--(15,8) (17,8)--(18.2,5.4)--(21,5.2)(3.5,3.5)--(5.3,1.1)--(7.6,1.9)--(7.5,4.5)(24.6,7.8)--(26,5)--(28,5) (26,5)--(24.5,2.5)--(25.8,0)  (22,2.5)--(20.8,1.1)--(20.5,0)(22,8)--(21,5.2)--(22,2.5)--(24.5,2.5);\draw (7.6,1.9)--(11,1.8)--(13.5,.7)--(17.2,1.6)--(20.8,1.1)(11,1.8)--(13,3.7)--(12.2,5.2)(13,3.7)--(16.8,3.8)--(18.2,5.4)(16.8,3.8)--(17.2,1.6);
\pgftransformyscale{-1}\fill[gray] (5.3,1.1)--(3.5,3.5)--(5.2,5.6)--(7.5,4.5)--(7.5, 1.8)--cycle;\fill[gray] (10,6)--(10,8.3)--(12.5,10)--(15,8)--(14,6)--(12.2,5.2)-- cycle;\fill[gray] (18.2,5.4)-- (17,8)--(19.5,9.8)--(22,8)--(21,5.2)--cycle;\fill[gray] (26,5)--(24.6,7.8)-- (26.3,10)--(28.8,9.8)--(30,7.5)--(28,5)-- cycle;\fill[gray] (13.5,.7)--(11,1.8)--(13,3.7)--(16.8,3.8)--(17.2,1.6)--cycle;\draw (0,1.3)--(1.5,3.5)--(0,6)--(.7,9.1)--(3,10)--(5,8)--(7.5,10)--(10,8.3)--(12.5,10)--(15,8)--(17,8)--(19.5,9.8)--(22,8)--(24.6,7.8)--(26.3,10)--(28.8,9.8)--(30,7.5)--(28,5)--(29.6,2.4)--(28,0)(1.5,3.5)--(3.5,3.5)--(5.2,5.6)--(5,8)(5.2,5.6)--(7.5,4.5)--(10,6)--(10,8.3)(10,6)--(12.2,5.2)--(14,6)--(15,8) (17,8)--(18.2,5.4)--(21,5.2)(3.5,3.5)--(5.3,1.1)--(7.6,1.9)--(7.5,4.5)(24.6,7.8)--(26,5)--(28,5) (26,5)--(24.5,2.5)--(25.8,0)  (22,2.5)--(20.8,1.1)--(20.5,0)(22,8)--(21,5.2)--(22,2.5)--(24.5,2.5);\draw (7.6,1.9)--(11,1.8)--(13.5,.7)--(17.2,1.6)--(20.8,1.1)(11,1.8)--(13,3.7)--(12.2,5.2)(13,3.7)--(16.8,3.8)--(18.2,5.4)(16.8,3.8)--(17.2,1.6);} 
\tikzmath{\useasboundingbox (-.6,-.7) rectangle (.6,.5); \node[scale=1.2]{$\overset\ID\sim$};} 
\tikzmath[scale = .12]{\pgftransformxscale{.95}\useasboundingbox (.6,-10.5) rectangle (28.5,11);
\clip[rounded corners = 12] (.6,-9.5) rectangle (28.5,9.5);
\fill[gray!20] (0,-10) rectangle (30,10);
\fill[gray]
(5.3,1.1)--(3.5,3.5)--(5.2,5.6)--(7.5,4.5)--(7.6, 1.9)--cycle (10,6)--(10,8.3)--(12.5,10)--(15,8)--(14,6)--(12.2,5.2)-- cycle (18.2,5.4)-- (17,8)--(19.5,9.8)--(22,8)--(21,5.2)--cycle (26,5)--(24.6,7.8)-- (26.3,10)--(28.8,9.8)--(30,7.5)--(28,5)-- cycle (20.5,0)--(20.8,1.1)--(22,2.5)--(24.5,2.5)--(25.8,0)--(24.5,-2.5)--(22,-2.5)--(20.8,-1.1)--cycle (9.3,1.2)--(11.8,2.3)--(13.2,0)--(11.8,-2.3)--(9.3,-1.2)--cycle (14.9,0)--(15.5,2.8)--(17.5,2.8)--(18.5,0) --(17.5,-2.8)--(15.5,-2.8)--cycle
(3,10) -- (5,8) -- (7.5,10) -- cycle (3,-10) -- (5,-8) -- (7.5,-10) -- cycle (28,0) -- (29.6,2.4) -- (29.6,-2.4) -- cycle (0,1.3)--(1.5,3.5)--(0,6)--cycle (0,-1.3)--(1.5,-3.5)--(0,-6)--cycle;
;
\draw (0,-1.3)--(0,1.3)(5.3,-1.1)--(5.3,1.1)(25.8,0)--(28,0);\draw (0,1.3)--(1.5,3.5)--(0,6)--(.7,9.1)--(3,10)--(5,8)--(7.5,10)--(10,8.3)--(12.5,10)--(15,8)--(17,8)--(19.5,9.8)--(22,8)--(24.6,7.8)--(26.3,10)--(28.8,9.8)--(30,7.5)--(28,5)--(29.6,2.4)--(28,0)(1.5,3.5)--(3.5,3.5)--(5.2,5.6)--(5,8)(5.2,5.6)--(7.5,4.5)--(10,6)--(10,8.3)(10,6)--(12.2,5.2)--(14,6)--(15,8) (17,8)--(18.2,5.4)--(21,5.2)(3.5,3.5)--(5.3,1.1)--(7.6,1.9)--(7.5,4.5)(24.6,7.8)--(26,5)--(28,5) (26,5)--(24.5,2.5)--(25.8,0)  (22,2.5)--(20.8,1.1)--(20.5,0)(22,8)--(21,5.2)--(22,2.5)--(24.5,2.5);\draw (7.6,1.9)--(9.3,1.2)--(9.3,0) (9.3,1.2)--(11.8,2.3)--(12.2,5.2)(11.8,2.3)--(13.2,0)--(14.9,0)--(15.5,2.8)--(17.5,2.8)--(18.5,0)--(20.5,0)(15.5,2.8)--(14,6) (17.5,2.8)--(18.2,5.4);
\pgftransformyscale{-1}\fill[gray] (5.3,1.1)--(3.5,3.5)--(5.2,5.6)--(7.5,4.5)--(7.6, 1.9)--cycle;\fill[gray] (10,6)--(10,8.3)--(12.5,10)--(15,8)--(14,6)--(12.2,5.2)-- cycle;\fill[gray] (18.2,5.4)-- (17,8)--(19.5,9.8)--(22,8)--(21,5.2)--cycle;\fill[gray] (26,5)--(24.6,7.8)-- (26.3,10)--(28.8,9.8)--(30,7.5)--(28,5)-- cycle;\draw (0,1.3)--(1.5,3.5)--(0,6)--(.7,9.1)--(3,10)--(5,8)--(7.5,10)--(10,8.3)--(12.5,10)--(15,8)--(17,8)--(19.5,9.8)--(22,8)--(24.6,7.8)--(26.3,10)--(28.8,9.8)--(30,7.5)--(28,5)--(29.6,2.4)--(28,0)(1.5,3.5)--(3.5,3.5)--(5.2,5.6)--(5,8)(5.2,5.6)--(7.5,4.5)--(10,6)--(10,8.3)(10,6)--(12.2,5.2)--(14,6)--(15,8) (17,8)--(18.2,5.4)--(21,5.2)(3.5,3.5)--(5.3,1.1)--(7.6,1.9)--(7.5,4.5)(24.6,7.8)--(26,5)--(28,5) (26,5)--(24.5,2.5)--(25.8,0)  (22,2.5)--(20.8,1.1)--(20.5,0)(22,8)--(21,5.2)--(22,2.5)--(24.5,2.5);\draw (7.6,1.9)--(9.3,1.2)--(9.3,0) (9.3,1.2)--(11.8,2.3)--(12.2,5.2)(11.8,2.3)--(13.2,0)  (14.9,0)--(15.5,2.8)--(17.5,2.8)--(18.5,0)  (20.5,0)(15.5,2.8)--(14,6) (17.5,2.8)--(18.2,5.4);} 
\,\,\,\,\,\,\,\,\,\,\text{with}\,\,\,\,\,\,\, \ID\,\,=\,\,\, \tikzmath[scale = .1]{\pgftransformxscale{.95}\filldraw[fill=gray!20] (5.3,1.1)--(7.6,1.9)--(7.5,4.5)--(10,6)--(12.2,5.2)--(14,6)--(15,8)--(17,8)--(18.2,5.4)--(21,5.2)--(22,2.5)--(20.8,1.1)--(20.5,0)--(20.8,-1.1)--(22,-2.5)--(21,-5.2)--(18.2,-5.4)--(17,-8)--(15,-8)--(14,-6)--(12.2,-5.2)--(10,-6)--(7.5,-4.5)--(7.6,-1.9)--(5.3,-1.1) --cycle;} 
\,\,\,.\vspace{.05cm}
\]

Given a set $S=\{\ID_0,\ldots,\ID_n\}$ of discs with disjoint interiors, we also write $X\overset S \sim Y$, or simply $X \sim Y$, if 
there exist soccer ball decompositions $Z_1,\ldots,Z_n$ such that
\begin{equation}\label{eq: XZZZZZZZY}
X\,\overset{\ID_0}\sim\, Z_1\,\overset{\ID_1}\sim\,\ldots\overset{\ID_{n-1}}\sim Z_n\,\overset{\ID_n}\sim\, Y.
\end{equation}
Finally, we write $X\tikz{\useasboundingbox(-.28,-.1)rectangle(.28,.3);\node at (0,.22) {$\scriptstyle S$};\node at (0,0) {$\prec$};} Y$, or simply $X\prec Y$, 
if $S=\{\ID_0,\ldots,\ID_n\}$ is a collection of white cells of $Y$ and
$X \tikz{\useasboundingbox(-.28,-.1)rectangle(.28,.3);\node at (0,.23) {$\scriptstyle \ID_0$};\node at (0,0) {$\prec$};} Z_1
\tikz{\useasboundingbox(-.28,-.1)rectangle(.28,.3);\node at (0,.23) {$\scriptstyle \ID_1$};\node at (0,0) {$\prec$};} \ldots
\tikz{\useasboundingbox(-.3,-.1)rectangle(.3,.3);\node at (0,.23) {$\scriptstyle \ID_{n-1}$};\node at (0,0) {$\prec$};} Z_n
\tikz{\useasboundingbox(-.28,-.1)rectangle(.28,.3);\node at (0,.23) {$\scriptstyle \ID_n$};\node at (0,0) {$\prec$};} Y$
for some choice $Z_1,\ldots,Z_n$ of soccer ball decompositions.
Note that $\prec$ is a transitive relation.
\end{definition}

If $X\tikz{\useasboundingbox(-.24,-.1)rectangle(.24,.3);\node at (0,.23) {$\scriptstyle \ID$};\node at (0,0) {$\prec$};} Y$, we can construct a unitary isomorphism
\begin{equation}\label{eq: def: little phi}
\phi_{XY}:V(\Sigma;X)\to V(\Sigma;Y),
\end{equation}
well defined up to phase.
Let us number the set $X_{\mathrm{white}}=\{\ID_1,\ldots,\ID_n\}$ of white cells of $X$ in such a way that $\ID_i\in Y_{\mathrm{white}}$ for $i\le p$ and $\ID_i\subset \ID$ for $i>p$.
Similarly, we number the set $X_{\mathrm{black}}=\{B_1,\ldots,B_m\}$ so as to have $B_j\in Y_{\mathrm{black}}$ for $j\le q$ and $B_j\subset \ID$ for $j>q$.
Letting $M:=\partial\ID\cap(\ID_1\cup\ldots\cup\ID_p)$,
the map $\phi_{XY}$ is then defined as the composite
\[
\begin{split}
&\phi_{XY}\,:\,\,\,V(\Sigma;X)=\bigg(\bigboxtimes_{\{\cala(\ID_i\cap \ID_j)\}} \big\{H_0(\partial\ID_i)\big\}_{i,j\le n}\bigg)\boxtimes_{\cala(\bigcup_{k\le m}\partial B_k)}\Big(\textstyle\underset{\scriptscriptstyle k\le m}\bigotimes H_0(\partial B_k)\Big)\\
&\!\stackrel\cong\longrightarrow\,\bigg(\Big(\bigboxtimes_{\{\cala(\ID_i\cap \ID_j)\}} \big\{H_0(\partial\ID_i)\big\}_{i,j\le p}
\Big)\underset{\cala(M)}\boxtimes\Big(\bigboxtimes_{\{\cala(\ID_i\cap \ID_j)\}} \big\{H_0(\partial\ID_i)\big\}_{p<i,j\le n}\Big)\\
&\hspace{2.5cm}\boxtimes_{\cala(\bigcup_{q<k\le m}\partial B_k)}\Big(\textstyle\underset{\scriptscriptstyle q<k\le m}\bigotimes H_0(\partial B_k)\Big)\bigg)
\,\boxtimes_{\cala(\bigcup_{k\le q}\partial B_k)}\Big(\textstyle\underset{\scriptscriptstyle k\le q}\bigotimes H_0(\partial B_k)\Big)\\
&\!\stackrel\cong\longrightarrow\,\bigg(\Big(\bigboxtimes_{\{\cala(\ID_i\cap \ID_j)\}} \big\{H_0(\partial\ID_i)\big\}_{i,j\le p}\Big)
\underset{\cala(M)}\boxtimes H_0(\partial\ID)\bigg)
\,\boxtimes_{\cala(\bigcup_{k\le q}\partial B_k)}\Big(\textstyle\underset{\scriptscriptstyle k\le q}\bigotimes H_0(\partial B_k)\Big)\\
&\hspace{10.6cm}=V(\Sigma;Y),
\end{split}
\]
where the first isomorphism is an instance of \eqref{eq: ass for graph fusion}, and the second one follows from Corollary \ref{cor: V(disc) independent of cell decomposition}.
If $X$, $Y$, $Z$ are soccer ball decompositions such that 
$X\tikz{\useasboundingbox(-.24,-.1)rectangle(.24,.3);\node at (0,.23) {$\scriptstyle \ID$};\node at (0,0) {$\prec$};} Y
\tikz{\useasboundingbox(-.24,-.1)rectangle(.24,.3);\node at (0.01,.23) {$\scriptstyle \ID'$};\node at (0,0) {$\prec$};} Z$ and
$X\tikz{\useasboundingbox(-.24,-.1)rectangle(.24,.3);\node at (0.01,.23) {$\scriptstyle \ID'$};\node at (0,0) {$\prec$};} Z$,
then it follows from Corollary \ref{cor: V(disc) independent of cell decomposition} that
\[
\phi_{YZ}\circ\phi_{XY}=\phi_{XZ}
\]
up to phase.
If two soccer ball decompositions are related by $X\prec Y$ then, by composing suitable instances of \eqref{eq: def: little phi}, we obtain an induced unitary isomorphism
$\phi_{XY}:V(\Sigma;X)\to V(\Sigma;Y)$, well defined up to phase.
As before, whenever $X\prec Y\prec Z$, those isomorphisms satisfy $\phi_{YZ}\circ\phi_{XY}=\phi_{XZ}$ up to phase.

If $X$ and $Y$ are soccer ball decompositions related by $X\overset\ID\sim Y$, then we also define
\begin{equation}\label{eq: PHI_D}
\Phi_\ID\,:\,V(\Sigma;X)\to V(\Sigma;Y)
\end{equation}
to be the composite $\phi_{YZ}^{-1}\circ \phi_{XZ}$,
where $Z$ is such that $X\tikz{\useasboundingbox(-.24,-.1)rectangle(.24,.3);\node at (0,.23) {$\scriptstyle \ID$};\node at (0,0) {$\prec$};} Z
\tikz{\useasboundingbox(-.24,-.1)rectangle(.24,.3);\node at (0,.23) {$\scriptstyle \ID$};\node at (0,0) {$\succ$};}  Y$.
The map $\Phi_\ID$ is well defined up to phase but, a priori, it depends on the choice of disk $\ID$.

\begin{maintheorem}\label{Mthm: conformal blocks}
Given a conformal net $\cala$ with finite index, there exists a construction
\begin{equation}\label{eq: main thm eq1}
\qquad\Sigma\,\,\mapsto\, V(\Sigma)\in\Rep_{\partial\Sigma}(\cala),
\end{equation}
that assigns, to every compact oriented topological surface $\Sigma$ with smooth boundary, a sector $V(\Sigma)=V(\Sigma;\cala)$,
well defined up to canonical-up-to-phase unitary isomorphism.
Moreover, if the surface $\Sigma$ is a disc, there exists a unitary isomorphism, canonical up to phase, between $V(\Sigma)$ and the vacuum sector $H_0(\partial\Sigma)$.

Given surfaces $\Sigma_1$, $\Sigma_2$ as above, and a homeomorphism $f:\Sigma_1\to \Sigma_2$ that is 
either orientation preserving or orientation reversing (which is automatic if $\Sigma_1$ is connected), and smooth on the boundary,
there is an induced map
\[
V(f):V(\Sigma_1)\to V(\Sigma_2),
\]
well defined up to phase.
The map $V(f)=V(f,\cala)$ is complex linear if $f$ is orientation preserving and complex antilinear if $f$ is orientation reversing.
The above construction satisfies $V(f\circ g)=V(f)\circ V(g)$ up to phase for composable $f$ and $g$.
Moreover, if $f_1,f_2:\Sigma_1\to \Sigma_2$ are isotopic relative to $\partial \Sigma_1$,
then $V(f_1)=V(f_2)$, up to phase.

Given an interval $I\subset \partial \Sigma_1$ and an element $a\in\cala(I)$,
the map $V(f)$ satisfies $V(f)a=\cala(f|_I)(a)V(f)$ if $f$ is orientation preserving and $V(f)a=\cala(f|_I)(a^*)V(f)$ if $f$ is orientation reversing.

Finally, the map $f\mapsto V(f)$ is continuous for the natural topology on the set of homeomorphisms $f:\Sigma_1\to \Sigma_2$ that are smooth on the boundary modulo isotopy rel boundary\,\footnote{The topology here is the quotient topology on
$\mathrm{Pullback}\,(\mathrm{Homeo}(\Sigma_1,\Sigma_2)\to\mathrm{Homeo}(\partial\Sigma_1,\partial\Sigma_2)\leftarrow\Diff(\partial\Sigma_1,\partial\Sigma_2))/\sim$
where $\mathrm{Homeo}$ is equipped with the $\mathcal C^0$ topology, and $\Diff$ is equipped with the $\mathcal C^\infty$ topology}, and the strong topology on $\PU_\pm(V(\Sigma))$.
\end{maintheorem}

Part of the content of the above theorem can be understood as follows.
Let $\mathsf{2MAN}$ be the following groupoid:
\begin{equation}\label{eq: def of 2MAN}
\mathsf{2MAN}:=\begin{cases}
\,\,\,\text{Objects:}\,\,\,\,\,\,\parbox{9cm}{Compact oriented topological surfaces equipped with a smooth structure on their boundary.}
\\\\
\text{Morphisms:}\,\,\,\parbox{9cm}{Isotopy classes of homeomorphisms that are either orientation preserving or orientation reversing, and that are smooth on the boundary, where the isotopies are taken relative to the boundary.}
\end{cases}
\end{equation}
and let $\mathsf{pHILB}$ be the groupoid whose objects are complex Hilbert space, and whose morphisms are the projective unitary and the projective antiunitary maps.
Then there is a symmetric monoidal continuous functor\footnote{The comments in Footnote \ref{footnote: BZ/2} also apply to the functor $V$.}
\begin{equation}\label{eq: 2MAN --> pHILB}
V:\mathsf{2MAN}\to \mathsf{pHILB}.
\end{equation}
It is interesting to note that in the groupoid $\mathsf{2MAN}$, the automorphism group of the unit disk is $\Diff(S^1)$, and so $V$ recovers 
the projective action of $\Diff(S^1)$ on the vacuum sector of $\cala$.
At another extreme, the automorphism group of a closed surface $\Sigma$ is its mapping class group.
We will see in Section \ref{sec: Conformal blocks} that in the case of a closed surface the Hilbert space $V(\Sigma)$ is finite dimensional.
We conjecture that it agrees with the spaces of conformal blocks, as defined in conformal field theory.
By applying the functor $V$, one therefore recovers the projective action of the mapping class group on the spaces of conformal blocks.
For a general surface, the automorphism group $G(\Sigma):=\aut_\mathsf{2MAN}(\Sigma)$ fits into a short exact sequence
\begin{equation}\label{def: G(Sigma)}
1\,\to\, \Gamma(\Sigma)\,\to\, G(\Sigma)\,\to\, D\,\to\, 1
\end{equation}
of topological groups, where $D$ is the finite index subgroup of diffeomorphisms of $\partial \Sigma$ that admit an extension to all of $\Sigma$, and $\Gamma(\Sigma)$ is the mapping class group of $\Sigma$ relative to its boundary, equipped with the discrete topology.
This group $G(\Sigma)$ was already considered in \cite[Thm.~2.11]{Posthuma(The-Heisenberg-group-and-conformal-field-theory)};
we call it the \emph{extended mapping class group} of $\Sigma$.
As a corollary of the above theorem,
the extended mapping class group $G(\Sigma)$ admits a continuous projective (anti)unitary action on the Hilbert space $V(\Sigma)$.

\begin{proof}[Proof of Theorem \ref{Mthm: conformal blocks}]
We apply the strategy outlined in Section \ref{sec: ``up to non-canonical isomorphism''}, and first construct the definition complex $\IX_\Sigma$ associated to $V(\Sigma)$.
The vertices of $\IX_\Sigma$ are the soccer ball decompositions of $\Sigma$.
The edges between two vertices $X$ and $Y$ are the discs $\ID\subset \Sigma$ such that $X\stackrel \ID\sim Y$.
Finally, there is a 2-cell of $\IX_\Sigma$ attached to the loop
\begin{equation}\label{eq: 2-cell in X_S}
\tikzmath[scale=.8]{\useasboundingbox (-1,-.9) rectangle (8.5,.5);
\node (a) at (0,0) {$X_1$};\node (b) at (2,0) {$X_2$};\node (c) at (4,0) {$X_3$};\node (d) at (7,0) {$X_n$};\node[yshift=-.2] at ($(c)!.5!(d)$) {$\scriptstyle \cdots$};
\draw (a) --node[above]{$\scriptstyle \ID_1$} (b);\draw (b) --node[above]{$\scriptstyle \ID_2$} (c);\draw (c) -- +(1.2,0)(d) -- +(-1.2,0);
\draw[rounded corners=6.8] (d) --node[above]{$\scriptstyle \ID_n$} (8.5,0) -- (8.5,-.6) -- (-1,-.6) -- (-1,0) -- (a);}
\end{equation}
if there exists a soccer ball decomposition $Z$ such that $X_i\prec Z$ for all $i$,
and such that every $\ID_i$ is contained in some white cell of $Z$.

We have seen in \eqref{eq: V(Sigma;X)} how to construct a sector $V(\Sigma;X)\in\Rep_{\partial\Sigma}(\cala)$ for every vertex $X$ of $\IX_\Sigma$,
and we have seen in \eqref{eq: PHI_D} how to construct a unitary isomorphism $\Phi_\ID:V(\Sigma;X)\to V(\Sigma;Y)$ for every edge $\ID$ of $\IX_\Sigma$.
Given a 2-cell \eqref{eq: 2-cell in X_S}, let $Y_i$ and $Z$ be soccer ball decompositions such that 
$X_i\tikz{\useasboundingbox(-.28,-.1)rectangle(.28,.3);\node at (0,.23) {$\scriptstyle \ID_i$};\node at (0,0) {$\prec$};} Y_i
\tikz{\useasboundingbox(-.28,-.1)rectangle(.28,.3);\node at (0,.23) {$\scriptstyle \ID_i$};\node at (0,0) {$\succ$};}  X_{i+1}$
and $Y_i\prec Z$.
The above relations induce diagrams
\[
\tikzmath[scale=1.2]{\node (a) at (-1.5,1) {$V(\Sigma;X_i)$};\node (b) at (1.5,1) {$V(\Sigma;Y_i)$};\node (c) at (0,0) {$V(\Sigma;Z)$};
\draw[->] (a) --node[above]{$\scriptstyle\phi_{X_i,Y_i}$} (b);\draw[->] (a) --node[below, pos=.3, xshift=-6]{$\scriptstyle\phi_{X_i,Z}$} (c);\draw[->] (b) --node[below, pos=.3, xshift=8]{$\scriptstyle\phi_{Y_i,Z}$} (c);
}\quad\text{and}\quad
\tikzmath[scale=1.2]{\node (a) at (-1.5,1) {$V(\Sigma;Y_i)$};\node (b) at (1.5,1) {$V(\Sigma;X_{i+1})$};\node (c) at (0,0) {$V(\Sigma;Z)$};
\draw[<-] (a) --node[above]{$\scriptstyle\phi_{Y_i,X_{i+1}}$} (b);\draw[->] (a) --node[below, pos=.3, xshift=-6]{$\scriptstyle\phi_{Y_i,Z}$} (c);\draw[->] (b) --node[below, pos=.3, xshift=12]{$\scriptstyle\phi_{X_{i+1},Z}$} (c);
}
\smallskip\]
that commute up to phase.
When assembled together, they show that the composite
\[
V(\Sigma;X_1)\xrightarrow{\,\,\,\Phi_{\ID_1}\,\,\,} V(\Sigma;X_2)\xrightarrow{\,\,\,\Phi_{\ID_2}\,\,\,} \,\,\,\cdots\,\,\,
\xrightarrow{\!\Phi_{\ID_{n-1}}\!\!}\, V(\Sigma;X_n)\xrightarrow{\,\,\Phi_{\ID_n}\,\,}V(\Sigma;X_1)
\]
is a scalar multiple of the identity, as desired.

To finish showing that $V(\Sigma)$ is well defined up to canonical-up-to-phase unitary isomorphism, we need to prove that $\IX_\Sigma$ is simply connected:
this is the content of Lemma~\ref{lem: definition complex simply connected}.
Finally, if $\Sigma$ is a disc, there is a unitary isomorphism canonical up to phase $V(\Sigma)\cong H_0(\partial\Sigma)$ by Corollary \ref{cor: V(disc) independent of cell decomposition}.

Given a homeomorphism $f:\Sigma_1\to \Sigma_2$ that is smooth on the boundary,
a soccer ball decomposition $X$ of $\Sigma_1$ induces a corresponding soccer ball decomposition $f_*X$ of $\Sigma_2$.
If $\ID$ is a cell of $X$, then the identity map from $H_0(\partial\ID)\in\Rep_{\partial\ID}(\cala)$ to $(f^{-1})^*(H_0(\partial\ID))\in\Rep_{f(\partial\ID)}(\cala)$
is complex linear if $f$ is orientation preserving, and complex antilinear otherwise.
The sector $(f^{-1})^*(H_0(\partial\ID))$ is isomorphic to $H_0(f(\partial\ID))$, and so we get a map
$H_0(f):H_0(\partial\ID)\to H_0(f(\partial\ID))$, canonical up to phase.
By construction, these maps satisfy
\begin{equation}\label{eq: H_0(f)a=}
\begin{cases}
H_0(f)a=
\cala(f|_I)(a)H_0(f)&\text{ if $f$ is orientation preserving }\\
H_0(f)a=
\cala(f|_I)(a^*)H_0(f)&\text{ if $f$ is orientation reversing }
\end{cases}
\end{equation}
for every interval $I\subset \partial\ID$ and $a\in \cala(I)$.

Applying the maps $H_0(f)$ to every occurrence of a vacuum sector in \eqref{eq: graph + fill holes}, we get an induced map $V(f;X):V(\Sigma_1;X)\to V(\Sigma_2;f_*X)$
which satisfies the same equivariance properties \eqref{eq: H_0(f)a=} as the maps $H_0(f)$ used in its construction, but now for $I\subset \partial \Sigma$.
For every edge $X\overset\ID\sim Y$ of the definition complex $\IX_{\Sigma_1}$,
there is a corresponding edge $f_*X\!\overset{f(\ID)}\sim\! f_*Y$ of $\IX_{\Sigma_2}$.
The diagrams
\[
\tikzmath{ \matrix [matrix of math nodes,column sep=1.6cm,row sep=.7cm]
{ |(a)| V(\Sigma_1;X) \pgfmatrixnextcell |(c)| V(\Sigma_2;f_*X)\\
|(b)|  V(\Sigma_1;Y) \pgfmatrixnextcell |(d)| V(\Sigma_2;f_*Y)\\ }; 
\draw[->] (a) -- node [left]	{$\scriptstyle \Phi_\ID$} (b); \draw[->] (c) -- node [right]	{$\scriptstyle \Phi_{f(\ID)}$} (d);
\draw[->] (a) --node[above]{$\scriptstyle V(f;X)$} (c); \draw[->] (b) --node[above]{$\scriptstyle V(f;Y)$} (d); }
\]
are commutative, and so the maps $V(f;X)$ induce a map $V(f):V(\Sigma_1)\to V(\Sigma_2)$, well defined up to phase.
The relation $V(f\circ g)=V(f)\circ V(g)$ is immediate from the definition.

Next, we need to show that if $f$ is isotopic to the identity, then $V(f)=\mathrm{Id}_{V(\Sigma)}$.
Write $f$ as a composition of homeomorphisms $f_1\circ\ldots\circ f_n$ such that for each $i$ there is a disk $D_i\subset \Sigma$ that contains the support of $f_i$.
Since $V(f)=V(f_1)\circ\ldots\circ V(f_n)$,
it is enough to show that each $V(f_i)$ is the identity.
Pick a soccer ball decomposition $X_i$ whose 1-skeleton does not intersect $D_i$.
The definition of $V(\Sigma;X_i)$ only uses the 1-skeleton of $X_i$, and so, clearly, $V(f_i;X_i):V(\Sigma;X_i)\to V(\Sigma;X_i)$ is the identity.
If follows that $V(f_i):V(\Sigma)\to V(\Sigma)$ is the identity, and therefore so is $V(f)$.

Finally, we show that the assignment $f\mapsto V(f)$ is continuous, which is
equivalent to showing that for every surface $\Sigma$, the homomorphism
\[
G(\Sigma)\to\PU_\pm(V(\Sigma))
\]
is continuous, where $G(\Sigma)$ is the group defined in \eqref{def: G(Sigma)}.
Let $S_1,\ldots,S_n$ be the boundary components of $\Sigma$, and let $\widetilde\Diff_+(S_i)$ denote the universal cover of $\Diff_+(S_i)$.
Then $\prod \widetilde\Diff_+(S_i)$ admits a natural homomorphism to $G(\Sigma)$, which is a homeomorphism on a neighborhood of the identity.
It is therefore enough to check the continuity of the map $\prod \widetilde\Diff_+(S_i)\to \PU_\pm(V(\Sigma))$.
Equivalently, we have to show that each map $\widetilde\Diff_+(S_i)\to \PU_\pm(V(\Sigma))$ is continuous.
For every choice of intervals $K\subset I\subset S_i$, let us write $\Diff_{0,K}(I)$ for the group of diffeomorphisms of $I$ that fix the complement of $K$ pointwise.
The topology on $\widetilde\Diff_+(S_i)$ is the finest one for which all the inclusions $\Diff_{0,K}(I)\hookrightarrow\Diff(S_i)$ are continuous.
It is therefore sufficient to show that the homomorphisms $\Diff_{0,K}(I)\to \PU(V(\Sigma))$ are continuous.
To finish the argument, we note that the above homomorphisms can be factored as
\[
\Diff_{0,K}(I)\to \PU(\cala(I))\to\PU(V(\Sigma)),
\]
where the first map is continuous by \cite[Lemma 2.11]{BDH(nets)}, and the second one is continuous because $V(\Sigma)$ is a $\partial\Sigma$-sector.
\end{proof}

\begin{lemma}\label{lem: definition complex simply connected}
The cell complex $\IX_\Sigma$ is connected and simply connected.
\end{lemma}

\begin{proof}
We will show that given a finite graph $\Gamma$ in the 1-skeleton of $\IX_\Sigma$, the inclusion map $\iota:\Gamma\hookrightarrow \IX_\Sigma$ extends to the cone on $\Gamma$:
\begin{equation}\label{eq: extension problem for definition complex}
\tikzmath{ \matrix [matrix of math nodes,column sep=1.2cm,row sep=5mm, inner sep=5]
{ |(a)| \Gamma \pgfmatrixnextcell |(b)| \IX_\Sigma\\  |(c)| \mathit{Cone}(\Gamma)\\ }; 
\draw[->] (a) --node[above, pos=.48]{$\scriptstyle\iota$} (b);\draw[->] (a) -- (c);\draw[->,dashed] (c) --node[below,xshift=3]{$\scriptstyle\tilde\iota$} (b);}\!.
\end{equation}
The case when $\Gamma$ consists of just two points shows that $\IX_\Sigma$ is connected, and the case when $\Gamma$ is a loop shows that $\IX_\Sigma$ is simply connected.

In order to construct the extension $\tilde\iota$, we need to: (1) pick a vertex $Y$ of $\IX_\Sigma$,
(2) provide a path $\gamma_X$ from $X$ to $Y$ for every vertex $X\in\Gamma$, and
(3) for every edge $X_1\stackrel \ID\sim X_2$ of $\Gamma$, construct a map $\delta_\ID:D^2\to\IX_\Sigma$ that bounds the triangle
\begin{equation}\label{eq: triangle to bound}
\tikzmath{
\node (a) at (0,0) {$X_1$};\node (b) at (0,-1) {$X_2$};\node (c) at (1.8,-.5) {$Y.$};
\draw (a) --node[left]{$\scriptstyle \ID$} (b);
\draw (a) to[bend right=-5]node[above]{$\scriptstyle \gamma_{X_1}$} (c);
\draw (b) to[bend right=5]node[below, yshift=-1]{$\scriptstyle \gamma_{X_2}$} (c);
}
\end{equation}

(1)
We first construct a soccer ball decomposition $Y$ as follows.
Start by picking a cell decomposition $Y_1$ (no soccer ball structure) 
that is transverse to all the decompositions $X_i \in \Gamma$.\footnote{See http://mathoverflow.net/questions/176227/topological-transversality for a proof that this always exists.}
It is taken fine enough so that it contains vertices in each face of each $X_i$, and so that it has at least one edge intersecting every edge of every $X_i$.
Add extra edges and vertices if necessary so that the intersection of the 1-skeleton of $Y_1$ with every 2-cell of every $X_i$ remains connected.\footnote{{\it Caveat:} Even though we believe that this is possible, we don't actually know how to construct $Y_1$ with this last property. What we can achieve instead by adding edges to $Y_1$ is that for each $2$-cell of each $X_i$, the intersection of the $2$-cell with the $1$-skeleton of $Y_1$ has one `big' connected component, that touches every edge of the cell.  In order not to overcomplicate our argument, we will proceed assuming the intersection of the 1-skeleton of $Y_1$ with every 2-cell of every $X_i$ is actually connected.  At the end of the proof, in footnote~\ref{cavaetresolved}, we address the issue by indicating the necessary modification to the argument in case those intersections are not connected.
}
Now double every edge and call the result $Y_2$. The cell decomposition $Y_2$ has the same vertices as $Y_1$, but it has a new bigon face in the place of every edge of $Y_1$.
Finally, replace the vertices by little black disks, and equip the result with a smooth structure, so as to get a soccer ball decomposition.
The black disks are chosen small enough, so as to not intersect the edges of the~$X_i$.

We give an illustration of the above process.
The soccer ball decomposition denoted $X$ below represents one of the vertices of $\Gamma$, and the cell decomposition $Y_1$ is the beginning point of our construction:
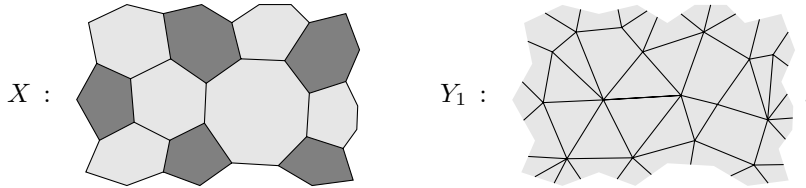
\[
X\,:\,\,\,\,
\tikzmath[scale = .12]{\foreach \b/\x/\y in {a/.8/18,b/1.7/13.5,c/-.5/9.5,d/2/5,e/.5/1.5,f/5/20,g/5.4/11.8,h/5.8/7,i/5/0,j/9/19,k/9.8/14,l/9.7/5,m/9/1.5,n/13.2/20,o/13.7/11.3,p/13.5/6.9,q/13/0,r/16.5/18,s/17.5/14,t/16.5/2.5,u/19.5/20,v/23/20,w/25/18,x/22.3/13.8,y/25/10.2,z/24.7/5.4,aa/21.6/2.2,bb/29/19,cc/30.5/15,dd/29/11.5,ee/30.3/8.8,eee/30.2/6.2,ff/28.7/4.1,gg/29.8/.6,hh/25.3/0,A/6.5/17,B/3.9/14.8,C/2.9/9.2,D/5.5/3,E/11.5/17.5,F/13.8/14.8,G/9.5/9.5,H/12.7/3.1,I/20.5/17,J/18/10,K/22/9,L/19.5/5,M/25.5/14,N/28.1/14.5,O/27.5/7.3,P/26/2.3,} {\coordinate (\b) at (\x,\y);}
\fill[gray!20](j)--(f)--(a)--(b)--(c)--(d)--(e)--(i)--(m)--(q)--(t)--(aa)--(hh)--(gg)--(ff)--(eee)--(ee)--(dd)--(cc)--(bb)--(w)--(v)--(u)--(r)--(n)--cycle;
\fill[gray] (b)--(c)--(d)--(h)--(g)--cycle (j)--(k)--(o)--(s)--(r)--(n)--cycle (w)--(x)--(y)--(dd)--(cc)--(bb)--cycle (l)--(m)--(q)--(t)--(p)--cycle (ff)--(gg)--(hh)--(aa)--(z)--cycle;
\draw (j)--(f)--(a)--(b)--(c)--(d)--(e)--(i)--(m)--(l)--(h)--(g)--(k)--(j)--(n)--(r)--(s)--(o)--(p)--(t)--(aa)--(z)--(y)--(x)--(w)--(bb)--(cc)--(dd)--(ee)--(eee)--(ff)--(gg) --(hh)--(aa) (b)--(g) (d)--(h) (l)--(p) (m)--(q)--(t) (k)--(o) (r)--(u)--(v)--(w) (s)--(x) (y)--(dd) (z)--(ff); 
}
\qquad\quad
Y_1\,:\,\,\,\,
\tikzmath[scale = .12]{\foreach \b/\x/\y in {a/.8/18,b/1.7/13.5,c/-.5/9.5,d/2/5,e/.5/1.5,f/5/20,g/5.4/11.8,h/5.8/7,i/5/0,j/9/19,k/9.8/14,l/9.7/5,m/9/1.5,n/13.2/20,o/13.7/11.3,p/13.5/6.9,q/13/0,r/16.5/18,s/17.5/14,t/16.5/2.5,u/19.5/20,v/23/20,w/25/18,x/22.3/13.8,y/25/10.2,z/24.7/5.4,aa/21.6/2.2,bb/29/19,cc/30.5/15,dd/29/11.5,ee/30.3/8.8,eee/30.2/6.2,ff/28.7/4.1,gg/29.8/.6,hh/25.3/0,A/6.5/17,B/3.9/14.8,C/2.9/9.2,D/5.5/3,E/11.5/17.5,F/13.8/14.8,G/9.5/9.5,H/12.7/3.1,I/20.5/17,J/18/10,K/22/9,L/19.5/5,M/25.5/14,N/28.1/14.5,O/27.5/7.3,P/26/2.3,} {\coordinate (\b) at (\x,\y);}
\fill[gray!20](j)--(f)--(a)--(b)--(c)--(d)--(e)--(i)--(m)--(q)--(t)--(aa)--(hh)--(gg)--(ff)--(eee)--(ee)--(dd)--(cc)--(bb)--(w)--(v)--(u)--(r)--(n)--cycle;
\draw[line join=bevel] (A)--(G)--(J)--(I)--(F)--(J)--(G)--(F)--(E)--(A)--(B)--(C)--(D)--(G)--(H)--(D) (B)--(G)--(C) (J)--(L)--(H)--(J)--(K)--(L)--(P)--(O)--(K)--(M)--(O)--(N)--(M)--(I) ($(a)!.45!(b)$)--(B) ($(c)!.5!(b)$)--(C)--($(c)!.5!(d)$) ($(f)!.45!(a)$)--(A)--($(f)!.5!(j)$) ($(d)!.5!(e)$)--(D)--($(e)!.33!(i)$) (D)--($(e)!.8!(i)$) (D)--($(m)!.5!(i)$) ($(m)!.5!(q)$)--(H)--($(q)!.5!(t)$) ($(t)!.5!(aa)$)--(L) ($(aa)!.5!(hh)$)--(P)--($(hh)!.5!(gg)$) (P)--($(ff)!.5!(gg)$) ($(ff)!.5!(eee)$)--(O)--($(ee)!.5!(dd)$) (O)--($(ee)!.5!(eee)$) ($(r)!.5!(u)$)--(I)--($(u)!.5!(v)$) (I)--($(w)!.5!(v)$) ($(dd)!.5!(cc)$)--(N)--($(cc)!.5!(bb)$) (M)--($(w)!.4!(bb)$) ($(j)!.5!(n)$)--(E)--($(n)!.33!(r)$) (F)--($(n)!.67!(r)$);
}\,\,. 
\]
It is convenient to draw $Y_1$ on top of $X$ in order to check that they satisfy the desired properties:
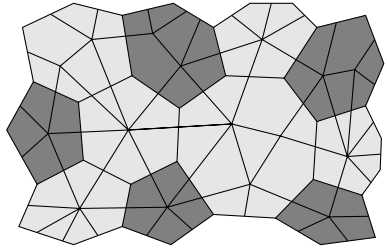
\[
\tikzmath[scale = .16]{\foreach \b/\x/\y in {a/.8/18,b/1.7/13.5,c/-.5/9.5,d/2/5,e/.5/1.5,f/5/20,g/5.4/11.8,h/5.8/7,i/5/0,j/9/19,k/9.8/14,l/9.7/5,m/9/1.5,n/13.2/20,o/13.7/11.3,p/13.5/6.9,q/13/0,r/16.5/18,s/17.5/14,t/16.5/2.5,u/19.5/20,v/23/20,w/25/18,x/22.3/13.8,y/25/10.2,z/24.7/5.4,aa/21.6/2.2,bb/29/19,cc/30.5/15,dd/29/11.5,ee/30.3/8.8,eee/30.2/6.2,ff/28.7/4.1,gg/29.8/.6,hh/25.3/0,A/6.5/17,B/3.9/14.8,C/2.9/9.2,D/5.5/3,E/11.5/17.5,F/13.8/14.8,G/9.5/9.5,H/12.7/3.1,I/20.5/17,J/18/10,K/22/9,L/19.5/5,M/25.5/14,N/28.1/14.5,O/27.5/7.3,P/26/2.3,} {\coordinate (\b) at (\x,\y);}
\fill[gray!20](j)--(f)--(a)--(b)--(c)--(d)--(e)--(i)--(m)--(q)--(t)--(aa)--(hh)--(gg)--(ff)--(eee)--(ee)--(dd)--(cc)--(bb)--(w)--(v)--(u)--(r)--(n)--cycle;
\fill[gray] (b)--(c)--(d)--(h)--(g)--cycle (j)--(k)--(o)--(s)--(r)--(n)--cycle (w)--(x)--(y)--(dd)--(cc)--(bb)--cycle (l)--(m)--(q)--(t)--(p)--cycle (ff)--(gg)--(hh)--(aa)--(z)--cycle;
\draw (j)--(f)--(a)--(b)--(c)--(d)--(e)--(i)--(m)--(l)--(h)--(g)--(k)--(j)--(n)--(r)--(s)--(o)--(p)--(t)--(aa)--(z)--(y)--(x)--(w)--(bb)--(cc)--(dd)--(ee)--(eee)--(ff)--(gg) --(hh)--(aa) (b)--(g) (d)--(h) (l)--(p) (m)--(q)--(t) (k)--(o) (r)--(u)--(v)--(w) (s)--(x) (y)--(dd) (z)--(ff); 
\draw[line join=bevel] (A)--(G)--(J)--(I)--(F)--(J)--(G)--(F)--(E)--(A)--(B)--(C)--(D)--(G)--(H)--(D) (B)--(G)--(C) (J)--(L)--(H)--(J)--(K)--(L)--(P)--(O)--(K)--(M)--(O)--(N)--(M)--(I) ($(a)!.45!(b)$)--(B) ($(c)!.5!(b)$)--(C)--($(c)!.5!(d)$) ($(f)!.45!(a)$)--(A)--($(f)!.5!(j)$) ($(d)!.5!(e)$)--(D)--($(e)!.33!(i)$) (D)--($(e)!.8!(i)$) (D)--($(m)!.5!(i)$) ($(m)!.5!(q)$)--(H)--($(q)!.5!(t)$) ($(t)!.5!(aa)$)--(L) ($(aa)!.5!(hh)$)--(P)--($(hh)!.5!(gg)$) (P)--($(ff)!.5!(gg)$) ($(ff)!.5!(eee)$)--(O)--($(ee)!.5!(dd)$) (O)--($(ee)!.5!(eee)$) ($(r)!.5!(u)$)--(I)--($(u)!.5!(v)$) (I)--($(w)!.5!(v)$) ($(dd)!.5!(cc)$)--(N)--($(cc)!.5!(bb)$) (M)--($(w)!.4!(bb)$) ($(j)!.5!(n)$)--(E)--($(n)!.33!(r)$) (F)--($(n)!.67!(r)$);
}
\]
The cell decomposition $Y_1$ is transverse to $X$, it intersects every edge of $X$, and its 1-skeleton remains connected when intersected with every 2-cell of $X$.
It is therefore a valid beginning point for our construction (as far as $X$ is concerned).
The next steps are then as follows:
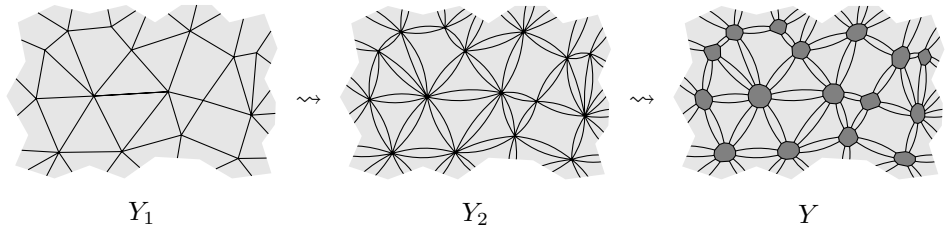
\[
\tikzmath[scale = .115]{\foreach \b/\x/\y in {a/.8/18,b/1.7/13.5,c/-.5/9.5,d/2/5,e/.5/1.5,f/5/20,g/5.4/11.8,h/5.8/7,i/5/0,j/9/19,k/9.8/14,l/9.7/5,m/9/1.5,n/13.2/20,o/13.7/11.3,p/13.5/6.9,q/13/0,r/16.5/18,s/17.5/14,t/16.5/2.5,u/19.5/20,v/23/20,w/25/18,x/22.3/13.8,y/25/10.2,z/24.7/5.4,aa/21.6/2.2,bb/29/19,cc/30.5/15,dd/29/11.5,ee/30.3/8.8,eee/30.2/6.2,ff/28.7/4.1,gg/29.8/.6,hh/25.3/0,A/6.5/17,B/3.9/14.8,C/2.9/9.2,D/5.5/3,E/11.5/17.5,F/13.8/14.8,G/9.5/9.5,H/12.7/3.1,I/20.5/17,J/18/10,K/22/9,L/19.5/5,M/25.5/14,N/28.1/14.5,O/27.5/7.3,P/26/2.3,} {\coordinate (\b) at (\x,\y);}
\fill[gray!20](j)--(f)--(a)--(b)--(c)--(d)--(e)--(i)--(m)--(q)--(t)--(aa)--(hh)--(gg)--(ff)--(eee)--(ee)--(dd)--(cc)--(bb)--(w)--(v)--(u)--(r)--(n)--cycle;
\draw[line join=bevel] (A)--(G)--(J)--(I)--(F)--(J)--(G)--(F)--(E)--(A)--(B)--(C)--(D)--(G)--(H)--(D) (B)--(G)--(C) (J)--(L)--(H)--(J)--(K)--(L)--(P)--(O)--(K)--(M)--(O)--(N)--(M)--(I) ($(a)!.45!(b)$)--(B) ($(c)!.5!(b)$)--(C)--($(c)!.5!(d)$) ($(f)!.45!(a)$)--(A)--($(f)!.5!(j)$) ($(d)!.5!(e)$)--(D)--($(e)!.33!(i)$) (D)--($(e)!.8!(i)$) (D)--($(m)!.5!(i)$) ($(m)!.5!(q)$)--(H)--($(q)!.5!(t)$) ($(t)!.5!(aa)$)--(L) ($(aa)!.5!(hh)$)--(P)--($(hh)!.5!(gg)$) (P)--($(ff)!.5!(gg)$) ($(ff)!.5!(eee)$)--(O)--($(ee)!.5!(dd)$) (O)--($(ee)!.5!(eee)$) ($(r)!.5!(u)$)--(I)--($(u)!.5!(v)$) (I)--($(w)!.5!(v)$) ($(dd)!.5!(cc)$)--(N)--($(cc)!.5!(bb)$) (M)--($(w)!.4!(bb)$) ($(j)!.5!(n)$)--(E)--($(n)!.33!(r)$) (F)--($(n)!.67!(r)$);
\node at (15,-4) {$Y_1$};} 
\,\,\,\,\raisebox{.26cm}{$\rightsquigarrow$}\,\,\,\,\tikzmath[scale = .115]{\foreach \b/\x/\y in {a/.8/18,b/1.7/13.5,c/-.5/9.5,d/2/5,e/.5/1.5,f/5/20,g/5.4/11.8,h/5.8/7,i/5/0,j/9/19,k/9.8/14,l/9.7/5,m/9/1.5,n/13.2/20,o/13.7/11.3,p/13.5/6.9,q/13/0,r/16.5/18,s/17.5/14,t/16.5/2.5,u/19.5/20,v/23/20,w/25/18,x/22.3/13.8,y/25/10.2,z/24.7/5.4,aa/21.6/2.2,bb/29/19,cc/30.5/15,dd/29/11.5,ee/30.3/8.8,eee/30.2/6.2,ff/28.7/4.1,gg/29.8/.6,hh/25.3/0,A/6.5/17,B/3.9/14.8,C/2.9/9.2,D/5.5/3,E/11.5/17.5,F/13.8/14.8,G/9.5/9.5,H/12.7/3.1,I/20.5/17,J/18/10,K/22/9,L/19.5/5,M/25.5/14,N/28.1/14.5,O/27.5/7.3,P/26/2.3,A'/-2.5/3.5,B'/-2/13.5,B''/-2/16.5,C'/0/21,D'/7.5/22,D''/10.5/22,E'/16.5/21,F'/22/23,G'/28.5/22,H'/31.5/19.5,I'/31.5/12,J'/33/8,K'/31.5/2.5,L'/29/-2,M'/19/-.5,M''/16/0,N'/9/-1.5,O'/-.5/-2.5} {\coordinate (\b) at (\x,\y);}
\clip(j)--(f)--(a)--(b)--(c)--(d)--(e)--(i)--(m)--(q)--(t)--(aa)--(hh)--(gg)--(ff)--(eee)--(ee)--(dd)--(cc)--(bb)--(w)--(v)--(u)--(r)--(n)--cycle (12,-2.1) rectangle (18,-6); \fill[gray!20] (-2,-2) rectangle (32,22); 
\draw[line join=bevel](A') to[bend right=15] (D)(A') to[bend right=-15] (D)(A') to[bend right=15] (C)(A') to[bend right=-15] (C)(B') to[bend right=15] (C)(B') to[bend right=-15] (C)(B'') to[bend right=15] (B)(B'') to[bend right=-15] (B)(C') to[bend right=15] (A)(C') to[bend right=-15] (A)(D') to[bend right=15] (A)(D') to[bend right=-15] (A)(D'') to[bend right=18] (E)(D'') to[bend right=-18] (E)(E') to[bend right=15] (E)(E') to[bend right=-15] (E)(E') to[bend right=15] (F)(E') to[bend right=-15] (F)(E') to[bend right=15] (I)(E') to[bend right=-15] (I)(F') to[bend right=15] (I)(F') to[bend right=-15] (I)(G') to[bend right=12] (I)(G') to[bend right=-12] (I)(G') to[bend right=13] (M)(G') to[bend right=-13] (M)(H') to[bend right=15] (N)(H') to[bend right=-15] (N)(I') to[bend right=17] (N)(I') to[bend right=-17] (N)(I') to[bend right=13] (O)(I') to[bend right=-13] (O)(J') to[bend right=15] (O)(J') to[bend right=-15] (O)(K') to[bend right=15] (O)(K') to[bend right=-15] (O)(K') to[bend right=15] (P)(K') to[bend right=-15] (P)(L') to[bend right=15] (P)(L') to[bend right=-15] (P)(M') to[bend right=15] (P)(M') to[bend right=-15] (P)(M') to[bend right=15] (L)(M') to[bend right=-15] (L)(M'') to[bend right=15] (H)(M'') to[bend right=-15] (H)(N') to[bend right=15] (H)(N') to[bend right=-15] (H)(N') to[bend right=17] (D)(N') to[bend right=-15] (D)(O') to[bend right=15] (D)(O') to[bend right=-15] (D);\draw[line join=bevel](A) to[bend left=20] (G)(A) to[bend left=-4] (G)(J) to[bend left=15] (I)(J) to[bend left=-15] (I)(I) to[bend left=15] (F)(I) to[bend left=-15] (F)(F) to[bend left=15] (J)(F) to[bend left=-15] (J)(J) to[bend left=15] (G)(J) to[bend left=-13] (G)(G) to[bend left=15] (F)(G) to[bend left=-15] (F)(F) to[bend left=21] (E)(F) to[bend left=-19] (E)(E) to[bend left=15] (A)(E) to[bend left=-15] (A)(A) to[bend left=17] (B)(A) to[bend left=-17] (B)(B) to[bend left=15] (C)(B) to[bend left=-15] (C)(C) to[bend left=15] (D)(C) to[bend left=-15] (D)(D) to[bend left=15] (G)(D) to[bend left=-15] (G)(G) to[bend left=15] (H)(G) to[bend left=-15] (H)(H) to[bend left=15] (D)(H) to[bend left=-15] (D)(B) to[bend left=11] (G)(B) to[bend left=-13] (G)(G) to[bend left=15] (C)(G) to[bend left=-15] (C)(J) to[bend left=15] (L)(J) to[bend left=-15] (L)(L) to[bend left=18] (H)(L) to[bend left=-10] (H)(H) to[bend left=12] (J)(H) to[bend left=-12] (J)(J) to[bend left=20] (K)(J) to[bend left=-14] (K)(K) to[bend left=19] (L)(K) to[bend left=-15] (L)(L) to[bend left=15] (P)(L) to[bend left=-15] (P)(P) to[bend left=15] (O)(P) to[bend left=-15] (O)(O) to[bend left=15] (K)(O) to[bend left=-15] (K)(K) to[bend left=18] (M)(K) to[bend left=-12] (M)(M) to[bend left=6] (O)(M) to[bend left=-19] (O)(O) to[bend left=7] (N)(O) to[bend left=-16] (N)(N) to[bend left=21] (M)(N) to[bend left=-21] (M)(M) to[bend left=15] (I)(M) to[bend left=-15] (I);
\node at (15,-4) {$Y_2$};} 
\,\,\,\,\raisebox{.26cm}{$\rightsquigarrow$}\,\,\,\,\tikzmath[scale = .115]{\foreach \b/\x/\y in {a/.8/18,b/1.7/13.5,c/-.5/9.5,d/2/5,e/.5/1.5,f/5/20,g/5.4/11.8,h/5.8/7,i/5/0,j/9/19,k/9.8/14,l/9.7/5,m/9/1.5,n/13.2/20,o/13.7/11.3,p/13.5/6.9,q/13/0,r/16.5/18,s/17.5/14,t/16.5/2.5,u/19.5/20,v/23/20,w/25/18,x/22.3/13.8,y/25/10.2,z/24.7/5.4,aa/21.6/2.2,bb/29/19,cc/30.5/15,dd/29/11.5,ee/30.3/8.8,eee/30.2/6.2,ff/28.7/4.1,gg/29.8/.6,hh/25.3/0,A/6.5/17,B/3.9/14.8,C/2.9/9.2,D/5.5/3,E/11.5/17.5,F/13.8/14.8,G/9.5/9.5,H/12.7/3.1,I/20.5/17,J/18/10,K/22/9,L/19.5/5,M/25.5/14,N/28.1/14.5,O/27.5/7.3,P/26/2.3,A'/-2.5/3.5,B'/-2/13.5,B''/-2/16.5,C'/0/21,D'/7.5/22,D''/10.5/22,E'/16.5/21,F'/22/23,G'/28.5/22,H'/31.5/19.5,I'/31.5/12,J'/33/8,K'/31.5/2.5,L'/29/-2,M'/19/-.5,M''/16/0,N'/9/-1.5,O'/-.5/-2.5} {\coordinate (\b) at (\x,\y);}
\clip(j)--(f)--(a)--(b)--(c)--(d)--(e)--(i)--(m)--(q)--(t)--(aa)--(hh)--(gg)--(ff)--(eee)--(ee)--(dd)--(cc)--(bb)--(w)--(v)--(u)--(r)--(n)--cycle (12,-2.1) rectangle (18,-6); \fill[gray!20] (-2,-2) rectangle (32,22); 
\draw[line join=bevel](A') to[bend right=15] (D)(A') to[bend right=-15] (D)(A') to[bend right=15] (C)(A') to[bend right=-15] (C)(B') to[bend right=15] (C)(B') to[bend right=-15] (C)(B'') to[bend right=15] (B)(B'') to[bend right=-15] (B)(C') to[bend right=15] (A)(C') to[bend right=-15] (A)(D') to[bend right=15] (A)(D') to[bend right=-15] (A)(D'') to[bend right=18] (E)(D'') to[bend right=-18] (E)(E') to[bend right=15] (E)(E') to[bend right=-15] (E)(E') to[bend right=15] (F)(E') to[bend right=-15] (F)(E') to[bend right=15] (I)(E') to[bend right=-15] (I)(F') to[bend right=15] (I)(F') to[bend right=-15] (I)(G') to[bend right=12] (I)(G') to[bend right=-12] (I)(G') to[bend right=13] (M)(G') to[bend right=-13] (M)(H') to[bend right=15] (N)(H') to[bend right=-15] (N)(I') to[bend right=17] (N)(I') to[bend right=-17] (N)(I') to[bend right=13] (O)(I') to[bend right=-13] (O)(J') to[bend right=15] (O)(J') to[bend right=-15] (O)(K') to[bend right=15] (O)(K') to[bend right=-15] (O)(K') to[bend right=15] (P)(K') to[bend right=-15] (P)(L') to[bend right=15] (P)(L') to[bend right=-15] (P)(M') to[bend right=15] (P)(M') to[bend right=-15] (P)(M') to[bend right=15] (L)(M') to[bend right=-15] (L)(M'') to[bend right=15] (H)(M'') to[bend right=-15] (H)(N') to[bend right=15] (H)(N') to[bend right=-15] (H)(N') to[bend right=17] (D)(N') to[bend right=-15] (D)(O') to[bend right=15] (D)(O') to[bend right=-15] (D);\draw[line join=bevel](A) to[bend left=20] (G)(A) to[bend left=-4] (G)(J) to[bend left=15] (I)(J) to[bend left=-15] (I)(I) to[bend left=15] (F)(I) to[bend left=-15] (F)(F) to[bend left=15] (J)(F) to[bend left=-15] (J)(J) to[bend left=15] (G)(J) to[bend left=-13] (G)(G) to[bend left=15] (F)(G) to[bend left=-15] (F)(F) to[bend left=21] (E)(F) to[bend left=-19] (E)(E) to[bend left=15] (A)(E) to[bend left=-15] (A)(A) to[bend left=17] (B)(A) to[bend left=-17] (B)(B) to[bend left=15] (C)(B) to[bend left=-15] (C)(C) to[bend left=15] (D)(C) to[bend left=-15] (D)(D) to[bend left=15] (G)(D) to[bend left=-15] (G)(G) to[bend left=15] (H)(G) to[bend left=-15] (H)(H) to[bend left=15] (D)(H) to[bend left=-15] (D)(B) to[bend left=11] (G)(B) to[bend left=-13] (G)(G) to[bend left=15] (C)(G) to[bend left=-15] (C)(J) to[bend left=15] (L)(J) to[bend left=-15] (L)(L) to[bend left=18] (H)(L) to[bend left=-10] (H)(H) to[bend left=12] (J)(H) to[bend left=-12] (J)(J) to[bend left=20] (K)(J) to[bend left=-14] (K)(K) to[bend left=19] (L)(K) to[bend left=-15] (L)(L) to[bend left=15] (P)(L) to[bend left=-15] (P)(P) to[bend left=15] (O)(P) to[bend left=-15] (O)(O) to[bend left=15] (K)(O) to[bend left=-15] (K)(K) to[bend left=18] (M)(K) to[bend left=-12] (M)(M) to[bend left=6] (O)(M) to[bend left=-19] (O)(O) to[bend left=7] (N)(O) to[bend left=-16] (N)(N) to[bend left=21] (M)(N) to[bend left=-21] (M)(M) to[bend left=15] (I)(M) to[bend left=-15] (I);
\filldraw[fill=gray]
($($(A)!.3!(B)$)+(-45:-.25)$)--($($(A)!.3!(B)$)+(-45:.25)$)--($($(A)!.125!(G)$)+(22:-.08)$)--($($(A)!.125!(G)$)+(22:.35)$)--($($(A)!.19!(E)$)+(100:-.22)$)--($($(A)!.19!(E)$)+(100:.22)$)--($($(A)!.32!($(f)!.5!(j)$)$)+(-10:.22)$)--($($(A)!.32!($(f)!.5!(j)$)$)+(-10:-.2)$)--($($(A)!.24!($(f)!.45!(a)$)$)+(50:.22)$)--($($(A)!.24!($(f)!.45!(a)$)$)+(50:-.25)$)--cycle ($($(B)!.3!(A)$)+(-40:-.24)$)--($($(B)!.3!(A)$)+(-40:.25)$)--($($(B)!.151!(G)$)+(45:.22)$)--($($(B)!.151!(G)$)+(45:-.25)$)--($($(B)!.19!(C)$)+(-16:.25)$)--($($(B)!.19!(C)$)+(-16:-.25)$)--($($(B)!.3!($(a)!.45!(b)$)$)+(70:-.34)$)--($($(B)!.3!($(a)!.45!(b)$)$)+(70:.09)$)--cycle ($($(C)!.2!(B)$)+(-20:-.27)$)--($($(C)!.2!(B)$)+(-20:.25)$)--($($(C)!.16!(G)$)+(95:.25)$)--($($(C)!.16!(G)$)+(95:-.25)$)--($($(C)!.18!(D)$)+(20:.29)$)--($($(C)!.18!(D)$)+(20:-.29)$)--($($(C)!.32!($(c)!.51!(d)$)$)+(-50:.28)$)--($($(C)!.32!($(c)!.52!(d)$)$)+(-50:-.19)$)--($($(C)!.34!($(c)!.5!(b)$)$)+(60:-.33)$)--($($(C)!.34!($(c)!.5!(b)$)$)+(60:.22)$)--cycle ($($(D)!.18!(C)$)+(20:-.27)$)--($($(D)!.18!(C)$)+(20:.28)$)--($($(D)!.155!(G)$)+(-30:-.27)$)--($($(D)!.15!(G)$)+(-30:.28)$)--($($(D)!.155!(H)$)+(90:.26)$)--($($(D)!.155!(H)$)+(90:-.26)$)--($($(D)!.36!($(i)!.5!(m)$)$)+(30:.34)$)--($($(D)!.36!($(i)!.5!(m)$)$)+(30:-.18)$)--($($(D)!.34!($(e)!.5!(i)$)$)+(-35:.35)$)--($($(D)!.36!($(e)!.5!(i)$)$)+(-35:-.22)$)--($($(D)!.31!($(e)!.5!(d)$)$)+(83:-.28)$)--($($(D)!.31!($(e)!.5!(d)$)$)+(83:.28)$)--cycle ($($(E)!.19!(A)$)+(-80:-.23)$)--($($(E)!.19!(A)$)+(-80:.23)$)--($($(E)!.25!(F)$)+(35:-.26)$)--($($(E)!.25!(F)$)+(35:.26)$)--($($(E)!.3!($(n)!.33!(r)$)$)+(-50:.22)$)--($($(E)!.3!($(n)!.33!(r)$)$)+(-50:-.25)$)--($($(E)!.45!($(j)!.5!(n)$)$)+(15:.24)$)--($($(E)!.45!($(j)!.5!(n)$)$)+(15:-.26)$)--cycle ($($(F)!.24!(E)$)+(50:.24)$)--($($(F)!.24!(E)$)+(50:-.28)$)--($($(F)!.12!(G)$)+(-40:-.20)$)--($($(F)!.12!(G)$)+(-40:.20)$)--($($(F)!.17!(J)$)+(30:-.25)$)--($($(F)!.17!(J)$)+(30:.28)$)--($($(F)!.19!(I)$)+(-70:.30)$)--($($(F)!.19!(I)$)+(-70:-.32)$)--($($(F)!.28!($(n)!.67!(r)$)$)+(-10:.32)$)--($($(F)!.28!($(n)!.67!(r)$)$)+(-10:-.26)$)--cycle ($($(G)!.175!(A)$)+(20:.45)$)--($($(G)!.19!(A)$)+(20:-.08)$)--($($(G)!.2!(B)$)+(45:.26)$)--($($(G)!.195!(B)$)+(45:-.28)$)--($($(G)!.215!(C)$)+(95:.30)$)--($($(G)!.213!(C)$)+(95:-.32)$)--($($(G)!.17!(D)$)+(-25:-.3)$)--($($(G)!.17!(D)$)+(-25:.29)$)--($($(G)!.17!(H)$)+(30:-.28)$)--($($(G)!.17!(H)$)+(30:.28)$)--($($(G)!.14!(J)$)+(90:-.29)$)--($($(G)!.14!(J)$)+(90:.26)$)--($($(G)!.19!(F)$)+(-45:.30)$)--($($(G)!.19!(F)$)+(-45:-.30)$)--cycle ($($(H)!.175!(D)$)+(95:-.30)$)--($($(H)!.175!(D)$)+(95:.29)$)--($($(H)!.18!(G)$)+(25:-.3)$)--($($(H)!.18!(G)$)+(25:.28)$)--($($(H)!.15!(J)$)+(-35:-.26)$)--($($(H)!.15!(J)$)+(-35:.25)$)--($($(H)!.18!(L)$)+(-80:-.22)$)--($($(H)!.17!(L)$)+(-80:.32)$)--($($(H)!.32!($(t)!.5!(q)$)$)+(30:.23)$)--($($(H)!.32!($(t)!.5!(q)$)$)+(30:-.23)$)--($($(H)!.31!($(m)!.5!(q)$)$)+(-15:.22)$)--($($(H)!.31!($(m)!.5!(q)$)$)+(-15:-.3)$)--cycle ($($(I)!.175!(F)$)+(-75:-.29)$)--($($(I)!.18!(F)$)+(-75:.29)$)--($($(I)!.15!(J)$)+(-10:-.26)$)--($($(I)!.15!(J)$)+(-10:.26)$)--($($(I)!.18!(M)$)+(55:-.25)$)--($($(I)!.18!(M)$)+(55:.27)$)--($($(I)!.3!($(v)!.5!(w)$)$)+(-55:.22)$)--($($(I)!.29!($(v)!.5!(w)$)$)+(-55:-.3)$)--($($(I)!.31!($(v)!.5!(u)$)$)+(0:.22)$)--($($(I)!.31!($(v)!.5!(u)$)$)+(0:-.22)$)--($($(I)!.3!($(u)!.5!(r)$)$)+(35:.32)$)--($($(I)!.3!($(u)!.5!(r)$)$)+(35:-.12)$)--cycle ($($(J)!.18!(F)$)+(30:.29)$)--($($(J)!.18!(F)$)+(30:-.29)$)--($($(J)!.165!(G)$)+(90:.26)$)--($($(J)!.165!(G)$)+(90:-.3)$)--($($(J)!.16!(H)$)+(-35:-.26)$)--($($(J)!.155!(H)$)+(-35:.27)$)--($($(J)!.24!(L)$)+(20:-.25)$)--($($(J)!.24!(L)$)+(20:.28)$)--($($(J)!.25!(K)$)+(85:-.22)$)--($($(J)!.25!(K)$)+(85:.3)$)--($($(J)!.13!(I)$)+(-20:.25)$)--($($(J)!.13!(I)$)+(-20:-.23)$)--cycle ($($(K)!.265!(J)$)+(70:-.22)$)--($($(K)!.265!(J)$)+(70:.3)$)--($($(K)!.18!(M)$)+(-25:-.3)$)--($($(K)!.18!(M)$)+(-25:.22)$)--($($(K)!.2!(O)$)+(70:.28)$)--($($(K)!.195!(O)$)+(70:-.25)$)--($($(K)!.23!(L)$)+(-20:.3)$)--($($(K)!.225!(L)$)+(-20:-.27)$)--cycle ($($(L)!.158!(H)$)+(-80:.32)$)--($($(L)!.16!(H)$)+(-80:-.17)$)--($($(L)!.19!(J)$)+(20:-.24)$)--($($(L)!.19!(J)$)+(20:.24)$)--($($(L)!.21!(K)$)+(-35:-.21)$)--($($(L)!.21!(K)$)+(-35:.29)$)--($($(L)!.16!(P)$)+(70:.26)$)--($($(L)!.16!(P)$)+(70:-.28)$)--($($(L)!.45!($(aa)!.5!(t)$)$)+(-5:.35)$)--($($(L)!.45!($(aa)!.5!(t)$)$)+(-5:-.18)$)--cycle ($($(M)!.18!(I)$)+(65:.26)$)--($($(M)!.18!(I)$)+(65:-.26)$)--($($(M)!.195!(K)$)+(-35:-.34)$)--($($(M)!.2!(K)$)+(-35:.24)$)--($($(M)!.17!(O)$)+(20:-.34)$)--($($(M)!.17!(O)$)+(20:.12)$)--($($(M)!.38!(N)$)+(95:-.26)$)--($($(M)!.41!(N)$)+(95:.26)$)--($($(M)!.295!($(w)!.4!(bb)$)$)+(-15:.44)$)--($($(M)!.3!($(w)!.4!(bb)$)$)+(-15:-.14)$)--cycle ($($(N)!.22!(M)$)+(80:.17)$)--($($(N)!.22!(M)$)+(80:-.22)$)--($($(N)!.18!(O)$)+(-10:-.16)$)--($($(N)!.18!(O)$)+(-10:.34)$)--($($(N)!.5!($(cc)!.5!(dd)$)$)+(70:-.25)$)--($($(N)!.5!($(cc)!.5!(dd)$)$)+(70:.26)$)--($($(N)!.22!($(cc)!.5!(bb)$)$)+(-25:.18)$)--($($(N)!.2!($(cc)!.5!(bb)$)$)+(-25:-.14)$)--cycle ($($(O)!.19!(N)$)+(-10:.34)$)--($($(O)!.2!(N)$)+(-10:-.16)$)--($($(O)!.2!(M)$)+(30:.16)$)--($($(O)!.18!(M)$)+(30:-.38)$)--($($(O)!.16!(K)$)+(85:.22)$)--($($(O)!.155!(K)$)+(85:-.22)$)--($($(O)!.16!(P)$)+(-13:-.22)$)--($($(O)!.16!(P)$)+(-13:.22)$)--($($(O)!.25!($(ff)!.5!(eee)$)$)+(45:-.22)$)--($($(O)!.25!($(ff)!.5!(eee)$)$)+(45:.16)$)--($($(O)!.27!($(eee)!.5!(ee)$)$)+(-90:.14)$)--($($(O)!.275!($(eee)!.5!(ee)$)$)+(-90:-.18)$)--($($(O)!.3!($(ee)!.5!(dd)$)$)+(-70:.28)$)--($($(O)!.31!($(ee)!.5!(dd)$)$)+(-70:-.18)$)--cycle ($($(P)!.17!(O)$)+(-13:.22)$)--($($(P)!.17!(O)$)+(-13:-.22)$)--($($(P)!.18!(L)$)+(65:.29)$)--($($(P)!.18!(L)$)+(65:-.3)$)--($($(P)!.36!($(aa)!.5!(hh)$)$)+(-45:-.34)$)--($($(P)!.36!($(aa)!.5!(hh)$)$)+(-45:.18)$)--($($(P)!.35!($(hh)!.5!(gg)$)$)+(20:-.25)$)--($($(P)!.36!($(hh)!.5!(gg)$)$)+(20:.2)$)--($($(P)!.36!($(gg)!.5!(ff)$)$)+(90:-.24)$)--($($(P)!.36!($(gg)!.5!(ff)$)$)+(90:.3)$)--cycle;
\node at (15,-4) {$Y$};} 
\]
By superimposing the soccer ball decompositions $X$ and $Y$ we obtain:\medskip
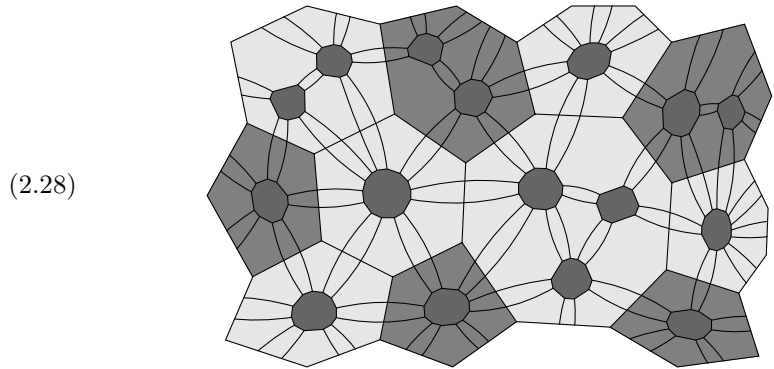
\begin{equation}\label{eq: picture of a transverse refinement}
\tikzmath[scale = .24]{\foreach \b/\x/\y in {a/.8/18,b/1.7/13.5,c/-.5/9.5,d/2/5,e/.5/1.5,f/5/20,g/5.4/11.8,h/5.8/7,i/5/0,j/9/19,k/9.8/14,l/9.7/5,m/9/1.5,n/13.2/20,o/13.7/11.3,p/13.5/6.9,q/13/0,r/16.5/18,s/17.5/14,t/16.5/2.5,u/19.5/20,v/23/20,w/25/18,x/22.3/13.8,y/25/10.2,z/24.7/5.4,aa/21.6/2.2,bb/29/19,cc/30.5/15,dd/29/11.5,ee/30.3/8.8,eee/30.2/6.2,ff/28.7/4.1,gg/29.8/.6,hh/25.3/0,A/6.5/17,B/3.9/14.8,C/2.9/9.2,D/5.5/3,E/11.5/17.5,F/13.8/14.8,G/9.5/9.5,H/12.7/3.1,I/20.5/17,J/18/10,K/22/9,L/19.5/5,M/25.5/14,N/28.1/14.5,O/27.5/7.3,P/26/2.3,A'/-2.5/3.5,B'/-2/13.5,B''/-2/16.5,C'/0/21,D'/7.5/22,D''/10.5/22,E'/16.5/21,F'/22/23,G'/28.5/22,H'/31.5/19.5,I'/31.5/12,J'/33/8,K'/31.5/2.5,L'/29/-2,M'/19/-.5,M''/16/0,N'/9/-1.5,O'/-.5/-2.5} {\coordinate (\b) at (\x,\y);}
\clip(j)--(f)--(a)--(b)--(c)--(d)--(e)--(i)--(m)--(q)--(t)--(aa)--(hh)--(gg)--(ff)--(eee)--(ee)--(dd)--(cc)--(bb)--(w)--(v)--(u)--(r)--(n)--cycle; \fill[gray!20] (-2,-2) rectangle (32,22); 
\fill[gray] (b)--(c)--(d)--(h)--(g)--cycle (j)--(k)--(o)--(s)--(r)--(n)--cycle (w)--(x)--(y)--(dd)--(cc)--(bb)--cycle (l)--(m)--(q)--(t)--(p)--cycle (ff)--(gg)--(hh)--(aa)--(z)--cycle;
\draw[line join=bevel](A') to[bend right=15] (D)(A') to[bend right=-15] (D)(A') to[bend right=15] (C)(A') to[bend right=-15] (C)(B') to[bend right=15] (C)(B') to[bend right=-15] (C)(B'') to[bend right=15] (B)(B'') to[bend right=-15] (B)(C') to[bend right=15] (A)(C') to[bend right=-15] (A)(D') to[bend right=15] (A)(D') to[bend right=-15] (A)(D'') to[bend right=18] (E)(D'') to[bend right=-18] (E)(E') to[bend right=15] (E)(E') to[bend right=-15] (E)(E') to[bend right=15] (F)(E') to[bend right=-15] (F)(E') to[bend right=15] (I)(E') to[bend right=-15] (I)(F') to[bend right=15] (I)(F') to[bend right=-15] (I)(G') to[bend right=12] (I)(G') to[bend right=-12] (I)(G') to[bend right=13] (M)(G') to[bend right=-13] (M)(H') to[bend right=15] (N)(H') to[bend right=-15] (N)(I') to[bend right=17] (N)(I') to[bend right=-17] (N)(I') to[bend right=13] (O)(I') to[bend right=-13] (O)(J') to[bend right=15] (O)(J') to[bend right=-15] (O)(K') to[bend right=15] (O)(K') to[bend right=-15] (O)(K') to[bend right=15] (P)(K') to[bend right=-15] (P)(L') to[bend right=15] (P)(L') to[bend right=-15] (P)(M') to[bend right=15] (P)(M') to[bend right=-15] (P)(M') to[bend right=15] (L)(M') to[bend right=-15] (L)(M'') to[bend right=15] (H)(M'') to[bend right=-15] (H)(N') to[bend right=15] (H)(N') to[bend right=-15] (H)(N') to[bend right=17] (D)(N') to[bend right=-15] (D)(O') to[bend right=15] (D)(O') to[bend right=-15] (D);\draw[line join=bevel](A) to[bend left=20] (G)(A) to[bend left=-4] (G)(J) to[bend left=15] (I)(J) to[bend left=-15] (I)(I) to[bend left=15] (F)(I) to[bend left=-15] (F)(F) to[bend left=15] (J)(F) to[bend left=-15] (J)(J) to[bend left=15] (G)(J) to[bend left=-13] (G)(G) to[bend left=15] (F)(G) to[bend left=-15] (F)(F) to[bend left=21] (E)(F) to[bend left=-19] (E)(E) to[bend left=15] (A)(E) to[bend left=-15] (A)(A) to[bend left=17] (B)(A) to[bend left=-17] (B)(B) to[bend left=15] (C)(B) to[bend left=-15] (C)(C) to[bend left=15] (D)(C) to[bend left=-15] (D)(D) to[bend left=15] (G)(D) to[bend left=-15] (G)(G) to[bend left=15] (H)(G) to[bend left=-15] (H)(H) to[bend left=15] (D)(H) to[bend left=-15] (D)(B) to[bend left=11] (G)(B) to[bend left=-13] (G)(G) to[bend left=15] (C)(G) to[bend left=-15] (C)(J) to[bend left=15] (L)(J) to[bend left=-15] (L)(L) to[bend left=18] (H)(L) to[bend left=-10] (H)(H) to[bend left=12] (J)(H) to[bend left=-12] (J)(J) to[bend left=20] (K)(J) to[bend left=-14] (K)(K) to[bend left=19] (L)(K) to[bend left=-15] (L)(L) to[bend left=15] (P)(L) to[bend left=-15] (P)(P) to[bend left=15] (O)(P) to[bend left=-15] (O)(O) to[bend left=15] (K)(O) to[bend left=-15] (K)(K) to[bend left=18] (M)(K) to[bend left=-12] (M)(M) to[bend left=6] (O)(M) to[bend left=-19] (O)(O) to[bend left=7] (N)(O) to[bend left=-16] (N)(N) to[bend left=21] (M)(N) to[bend left=-21] (M)(M) to[bend left=15] (I)(M) to[bend left=-15] (I);
\filldraw[fill=black!60]
($($(A)!.3!(B)$)+(-45:-.25)$)--($($(A)!.3!(B)$)+(-45:.25)$)--($($(A)!.125!(G)$)+(22:-.08)$)--($($(A)!.125!(G)$)+(22:.35)$)--($($(A)!.19!(E)$)+(100:-.22)$)--($($(A)!.19!(E)$)+(100:.22)$)--($($(A)!.32!($(f)!.5!(j)$)$)+(-10:.22)$)--($($(A)!.32!($(f)!.5!(j)$)$)+(-10:-.2)$)--($($(A)!.24!($(f)!.45!(a)$)$)+(50:.22)$)--($($(A)!.24!($(f)!.45!(a)$)$)+(50:-.25)$)--cycle ($($(B)!.3!(A)$)+(-40:-.24)$)--($($(B)!.3!(A)$)+(-40:.25)$)--($($(B)!.151!(G)$)+(45:.22)$)--($($(B)!.151!(G)$)+(45:-.25)$)--($($(B)!.19!(C)$)+(-16:.25)$)--($($(B)!.19!(C)$)+(-16:-.25)$)--($($(B)!.3!($(a)!.45!(b)$)$)+(70:-.34)$)--($($(B)!.3!($(a)!.45!(b)$)$)+(70:.09)$)--cycle ($($(C)!.2!(B)$)+(-20:-.27)$)--($($(C)!.2!(B)$)+(-20:.25)$)--($($(C)!.16!(G)$)+(95:.25)$)--($($(C)!.16!(G)$)+(95:-.25)$)--($($(C)!.18!(D)$)+(20:.29)$)--($($(C)!.18!(D)$)+(20:-.29)$)--($($(C)!.32!($(c)!.51!(d)$)$)+(-50:.28)$)--($($(C)!.32!($(c)!.52!(d)$)$)+(-50:-.19)$)--($($(C)!.34!($(c)!.5!(b)$)$)+(60:-.33)$)--($($(C)!.34!($(c)!.5!(b)$)$)+(60:.22)$)--cycle ($($(D)!.18!(C)$)+(20:-.27)$)--($($(D)!.18!(C)$)+(20:.28)$)--($($(D)!.155!(G)$)+(-30:-.27)$)--($($(D)!.15!(G)$)+(-30:.28)$)--($($(D)!.155!(H)$)+(90:.26)$)--($($(D)!.155!(H)$)+(90:-.26)$)--($($(D)!.36!($(i)!.5!(m)$)$)+(30:.34)$)--($($(D)!.36!($(i)!.5!(m)$)$)+(30:-.18)$)--($($(D)!.34!($(e)!.5!(i)$)$)+(-35:.35)$)--($($(D)!.36!($(e)!.5!(i)$)$)+(-35:-.22)$)--($($(D)!.31!($(e)!.5!(d)$)$)+(83:-.28)$)--($($(D)!.31!($(e)!.5!(d)$)$)+(83:.28)$)--cycle ($($(E)!.19!(A)$)+(-80:-.23)$)--($($(E)!.19!(A)$)+(-80:.23)$)--($($(E)!.25!(F)$)+(35:-.26)$)--($($(E)!.25!(F)$)+(35:.26)$)--($($(E)!.3!($(n)!.33!(r)$)$)+(-50:.22)$)--($($(E)!.3!($(n)!.33!(r)$)$)+(-50:-.25)$)--($($(E)!.45!($(j)!.5!(n)$)$)+(15:.24)$)--($($(E)!.45!($(j)!.5!(n)$)$)+(15:-.26)$)--cycle ($($(F)!.24!(E)$)+(50:.24)$)--($($(F)!.24!(E)$)+(50:-.28)$)--($($(F)!.12!(G)$)+(-40:-.20)$)--($($(F)!.12!(G)$)+(-40:.20)$)--($($(F)!.17!(J)$)+(30:-.25)$)--($($(F)!.17!(J)$)+(30:.28)$)--($($(F)!.19!(I)$)+(-70:.30)$)--($($(F)!.19!(I)$)+(-70:-.32)$)--($($(F)!.28!($(n)!.67!(r)$)$)+(-10:.32)$)--($($(F)!.28!($(n)!.67!(r)$)$)+(-10:-.26)$)--cycle ($($(G)!.175!(A)$)+(20:.45)$)--($($(G)!.19!(A)$)+(20:-.08)$)--($($(G)!.2!(B)$)+(45:.26)$)--($($(G)!.195!(B)$)+(45:-.28)$)--($($(G)!.215!(C)$)+(95:.30)$)--($($(G)!.213!(C)$)+(95:-.32)$)--($($(G)!.17!(D)$)+(-25:-.3)$)--($($(G)!.17!(D)$)+(-25:.29)$)--($($(G)!.17!(H)$)+(30:-.28)$)--($($(G)!.17!(H)$)+(30:.28)$)--($($(G)!.14!(J)$)+(90:-.29)$)--($($(G)!.14!(J)$)+(90:.26)$)--($($(G)!.19!(F)$)+(-45:.30)$)--($($(G)!.19!(F)$)+(-45:-.30)$)--cycle ($($(H)!.175!(D)$)+(95:-.30)$)--($($(H)!.175!(D)$)+(95:.29)$)--($($(H)!.18!(G)$)+(25:-.3)$)--($($(H)!.18!(G)$)+(25:.28)$)--($($(H)!.15!(J)$)+(-35:-.26)$)--($($(H)!.15!(J)$)+(-35:.25)$)--($($(H)!.18!(L)$)+(-80:-.22)$)--($($(H)!.17!(L)$)+(-80:.32)$)--($($(H)!.32!($(t)!.5!(q)$)$)+(30:.23)$)--($($(H)!.32!($(t)!.5!(q)$)$)+(30:-.23)$)--($($(H)!.31!($(m)!.5!(q)$)$)+(-15:.22)$)--($($(H)!.31!($(m)!.5!(q)$)$)+(-15:-.3)$)--cycle ($($(I)!.175!(F)$)+(-75:-.29)$)--($($(I)!.18!(F)$)+(-75:.29)$)--($($(I)!.15!(J)$)+(-10:-.26)$)--($($(I)!.15!(J)$)+(-10:.26)$)--($($(I)!.18!(M)$)+(55:-.25)$)--($($(I)!.18!(M)$)+(55:.27)$)--($($(I)!.3!($(v)!.5!(w)$)$)+(-55:.22)$)--($($(I)!.29!($(v)!.5!(w)$)$)+(-55:-.3)$)--($($(I)!.31!($(v)!.5!(u)$)$)+(0:.22)$)--($($(I)!.31!($(v)!.5!(u)$)$)+(0:-.22)$)--($($(I)!.3!($(u)!.5!(r)$)$)+(35:.32)$)--($($(I)!.3!($(u)!.5!(r)$)$)+(35:-.12)$)--cycle ($($(J)!.18!(F)$)+(30:.29)$)--($($(J)!.18!(F)$)+(30:-.29)$)--($($(J)!.165!(G)$)+(90:.26)$)--($($(J)!.165!(G)$)+(90:-.3)$)--($($(J)!.16!(H)$)+(-35:-.26)$)--($($(J)!.155!(H)$)+(-35:.27)$)--($($(J)!.24!(L)$)+(20:-.25)$)--($($(J)!.24!(L)$)+(20:.28)$)--($($(J)!.25!(K)$)+(85:-.22)$)--($($(J)!.25!(K)$)+(85:.3)$)--($($(J)!.13!(I)$)+(-20:.25)$)--($($(J)!.13!(I)$)+(-20:-.23)$)--cycle ($($(K)!.265!(J)$)+(70:-.22)$)--($($(K)!.265!(J)$)+(70:.3)$)--($($(K)!.18!(M)$)+(-25:-.3)$)--($($(K)!.18!(M)$)+(-25:.22)$)--($($(K)!.2!(O)$)+(70:.28)$)--($($(K)!.195!(O)$)+(70:-.25)$)--($($(K)!.23!(L)$)+(-20:.3)$)--($($(K)!.225!(L)$)+(-20:-.27)$)--cycle ($($(L)!.158!(H)$)+(-80:.32)$)--($($(L)!.16!(H)$)+(-80:-.17)$)--($($(L)!.19!(J)$)+(20:-.24)$)--($($(L)!.19!(J)$)+(20:.24)$)--($($(L)!.21!(K)$)+(-35:-.21)$)--($($(L)!.21!(K)$)+(-35:.29)$)--($($(L)!.16!(P)$)+(70:.26)$)--($($(L)!.16!(P)$)+(70:-.28)$)--($($(L)!.45!($(aa)!.5!(t)$)$)+(-5:.35)$)--($($(L)!.45!($(aa)!.5!(t)$)$)+(-5:-.18)$)--cycle ($($(M)!.18!(I)$)+(65:.26)$)--($($(M)!.18!(I)$)+(65:-.26)$)--($($(M)!.195!(K)$)+(-35:-.34)$)--($($(M)!.2!(K)$)+(-35:.24)$)--($($(M)!.17!(O)$)+(20:-.34)$)--($($(M)!.17!(O)$)+(20:.12)$)--($($(M)!.38!(N)$)+(95:-.26)$)--($($(M)!.41!(N)$)+(95:.26)$)--($($(M)!.295!($(w)!.4!(bb)$)$)+(-15:.44)$)--($($(M)!.3!($(w)!.4!(bb)$)$)+(-15:-.14)$)--cycle ($($(N)!.22!(M)$)+(80:.17)$)--($($(N)!.22!(M)$)+(80:-.22)$)--($($(N)!.18!(O)$)+(-10:-.16)$)--($($(N)!.18!(O)$)+(-10:.34)$)--($($(N)!.5!($(cc)!.5!(dd)$)$)+(70:-.25)$)--($($(N)!.5!($(cc)!.5!(dd)$)$)+(70:.26)$)--($($(N)!.22!($(cc)!.5!(bb)$)$)+(-25:.18)$)--($($(N)!.2!($(cc)!.5!(bb)$)$)+(-25:-.14)$)--cycle ($($(O)!.19!(N)$)+(-10:.34)$)--($($(O)!.2!(N)$)+(-10:-.16)$)--($($(O)!.2!(M)$)+(30:.16)$)--($($(O)!.18!(M)$)+(30:-.38)$)--($($(O)!.16!(K)$)+(85:.22)$)--($($(O)!.155!(K)$)+(85:-.22)$)--($($(O)!.16!(P)$)+(-13:-.22)$)--($($(O)!.16!(P)$)+(-13:.22)$)--($($(O)!.25!($(ff)!.5!(eee)$)$)+(45:-.22)$)--($($(O)!.25!($(ff)!.5!(eee)$)$)+(45:.16)$)--($($(O)!.27!($(eee)!.5!(ee)$)$)+(-90:.14)$)--($($(O)!.275!($(eee)!.5!(ee)$)$)+(-90:-.18)$)--($($(O)!.3!($(ee)!.5!(dd)$)$)+(-70:.28)$)--($($(O)!.31!($(ee)!.5!(dd)$)$)+(-70:-.18)$)--cycle ($($(P)!.17!(O)$)+(-13:.22)$)--($($(P)!.17!(O)$)+(-13:-.22)$)--($($(P)!.18!(L)$)+(65:.29)$)--($($(P)!.18!(L)$)+(65:-.3)$)--($($(P)!.36!($(aa)!.5!(hh)$)$)+(-45:-.34)$)--($($(P)!.36!($(aa)!.5!(hh)$)$)+(-45:.18)$)--($($(P)!.35!($(hh)!.5!(gg)$)$)+(20:-.25)$)--($($(P)!.36!($(hh)!.5!(gg)$)$)+(20:.2)$)--($($(P)!.36!($(gg)!.5!(ff)$)$)+(90:-.24)$)--($($(P)!.36!($(gg)!.5!(ff)$)$)+(90:.3)$)--cycle;
\draw[thick](j)--(f)--(a)--(b)--(c)--(d)--(e)--(i)--(m)--(q)--(t)--(aa)--(hh)--(gg)--(ff)--(eee)--(ee)--(dd)--(cc)--(bb)--(w)--(v)--(u)--(r)--(n)--cycle;
\draw[line cap = round] (b) -- (g) -- (h) -- (d) (h) -- (l) -- (m) (l) -- (p) -- (t) (g) -- (k) -- (o) -- (p) (j) -- (k) (r) --(s) -- (o) (w) -- (x) -- (s) (dd) -- (y) -- (x) (y) -- (z) -- (aa) (z) -- (ff);
} 
\medskip
\end{equation}

(2)
For every $X\in \Gamma$, we now describe two related soccer ball decompositions $X'$ and $X''$ that satisfy
\[
X\sim X'\sim X'' \sim Y.
\]
Each one of the three $\sim$'s above corresponds, as in \eqref{eq: XZZZZZZZY}, to a path in $\IX_\Sigma$ and we call their concatenation $\gamma_X$.

Let $A$ be the set of edges of $X$ that separate two white cells,
and let $B$ be the set of edges of $Y$ that intersect edges in $A$.
The first soccer ball decomposition $X'$ is defined as follows.
Its 1-skeleton is the union of the 1-skeleton of $X$ and of the edges $e\cap \ID$, where $\ID$ is a white cell of $X$, and $e$ is an edge of $Y$, $e\not\in B$.
The soccer ball structure on $X'$ is inherited from that of $Y$, and the smooth structure around the new trivalent vertices (at the intersection of the $1$-skeletons of $X$ and of $Y$) is chosen arbitrarily.
The second soccer ball decomposition $X''$ is then obtained from $X'$ by removing all the edges in $A$, and adding all those in $B$.

In the example when $X$ and $Y$ are as in \eqref{eq: picture of a transverse refinement}, then the intermediate soccer ball decompositions $X'$ and $X''$ appear as follows:
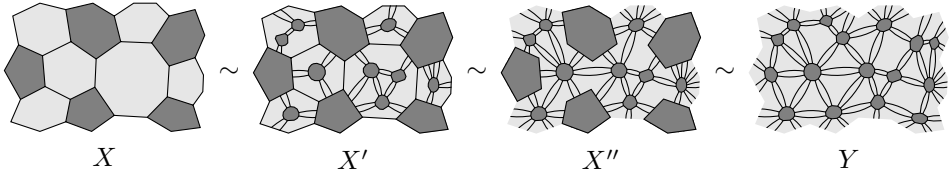
\[
\tikzmath[scale = .085]{\foreach \b/\x/\y in {a/.8/18,b/1.7/13.5,c/-.5/9.5,d/2/5,e/.5/1.5,f/5/20,g/5.4/11.8,h/5.8/7,i/5/0,j/9/19,k/9.8/14,l/9.7/5,m/9/1.5,n/13.2/20,o/13.7/11.3,p/13.5/6.9,q/13/0,r/16.5/18,s/17.5/14,t/16.5/2.5,u/19.5/20,v/23/20,w/25/18,x/22.3/13.8,y/25/10.2,z/24.7/5.4,aa/21.6/2.2,bb/29/19,cc/30.5/15,dd/29/11.5,ee/30.3/8.8,eee/30.2/6.2,ff/28.7/4.1,gg/29.8/.6,hh/25.3/0,A/6.5/17,B/3.9/14.8,C/2.9/9.2,D/5.5/3,E/11.5/17.5,F/13.8/14.8,G/9.5/9.5,H/12.7/3.1,I/20.5/17,J/18/10,K/22/9,L/19.5/5,M/25.5/14,N/28.1/14.5,O/27.5/7.3,P/26/2.3,} {\coordinate (\b) at (\x,\y);}
\fill[gray!20](j)--(f)--(a)--(b)--(c)--(d)--(e)--(i)--(m)--(q)--(t)--(aa)--(hh)--(gg)--(ff)--(eee)--(ee)--(dd)--(cc)--(bb)--(w)--(v)--(u)--(r)--(n)--cycle;
\fill[gray] (b)--(c)--(d)--(h)--(g)--cycle (j)--(k)--(o)--(s)--(r)--(n)--cycle (w)--(x)--(y)--(dd)--(cc)--(bb)--cycle (l)--(m)--(q)--(t)--(p)--cycle (ff)--(gg)--(hh)--(aa)--(z)--cycle;
\draw (j)--(f)--(a)--(b)--(c)--(d)--(e)--(i)--(m)--(l)--(h)--(g)--(k)--(j)--(n)--(r)--(s)--(o)--(p)--(t)--(aa)--(z)--(y)--(x)--(w)--(bb)--(cc)--(dd)--(ee)--(eee)--(ff)--(gg) --(hh)--(aa) (b)--(g) (d)--(h) (l)--(p) (m)--(q)--(t) (k)--(o) (r)--(u)--(v)--(w) (s)--(x) (y)--(dd) (z)--(ff); 
\node at (15,-4) {$X$};} 
\,\,\,\raisebox{.24cm}{$\sim$}\,\,\,\tikzmath[scale = .085]{\foreach \b/\x/\y in {a/.8/18,b/1.7/13.5,c/-.5/9.5,d/2/5,e/.5/1.5,f/5/20,g/5.4/11.8,h/5.8/7,i/5/0,j/9/19,k/9.8/14,l/9.7/5,m/9/1.5,n/13.2/20,o/13.7/11.3,p/13.5/6.9,q/13/0,r/16.5/18,s/17.5/14,t/16.5/2.5,u/19.5/20,v/23/20,w/25/18,x/22.3/13.8,y/25/10.2,z/24.7/5.4,aa/21.6/2.2,bb/29/19,cc/30.5/15,dd/29/11.5,ee/30.3/8.8,eee/30.2/6.2,ff/28.7/4.1,gg/29.8/.6,hh/25.3/0,A/6.5/17,B/3.9/14.8,C/2.9/9.2,D/5.5/3,E/11.5/17.5,F/13.8/14.8,G/9.5/9.5,H/12.7/3.1,I/20.5/17,J/18/10,K/22/9,L/19.5/5,M/25.5/14,N/28.1/14.5,O/27.5/7.3,P/26/2.3,A'/-2.5/3.5,B'/-2/13.5,B''/-2/16.5,C'/0/21,D'/7.5/22,D''/10.5/22,E'/16.5/21,F'/22/23,G'/28.5/22,H'/31.5/19.5,I'/31.5/12,J'/33/8,K'/31.5/2.5,L'/29/-2,M'/19/-.5,M''/16/0,N'/9/-1.5,O'/-.5/-2.5} {\coordinate (\b) at (\x,\y);}
\clip(j)--(f)--(a)--(b)--(c)--(d)--(e)--(i)--(m)--(q)--(t)--(aa)--(hh)--(gg)--(ff)--(eee)--(ee)--(dd)--(cc)--(bb)--(w)--(v)--(u)--(r)--(n)--cycle (12,-2.1) rectangle (18,-6); \fill[gray!20] (-2,-2) rectangle (32,22); 
\draw[line join=bevel](C') to[bend right=15] (A)(C') to[bend right=-15] (A)(D'') to[bend right=18] (E)(D'') to[bend right=-18] (E)(F') to[bend right=15] (I)(F') to[bend right=-15] (I)(J') to[bend right=15] (O)(J') to[bend right=-15] (O)(O') to[bend right=15] (D)(O') to[bend right=-15] (D)(I) to[bend left=15] (F)(I) to[bend left=-15] (F)(F) to[bend left=15] (J)(F) to[bend left=-15] (J)(G) to[bend left=15] (F)(G) to[bend left=-15] (F)(E) to[bend left=15] (A)(E) to[bend left=-15] (A)(A) to[bend left=17] (B)(A) to[bend left=-17] (B)(B) to[bend left=15] (C)(B) to[bend left=-15] (C)(C) to[bend left=15] (D)(C) to[bend left=-15] (D)(G) to[bend left=15] (H)(G) to[bend left=-15] (H)(H) to[bend left=15] (D)(H) to[bend left=-15] (D)(G) to[bend left=15] (C)(G) to[bend left=-15] (C)(J) to[bend left=15] (L)(J) to[bend left=-15] (L)(L) to[bend left=18] (H)(L) to[bend left=-10] (H)(H) to[bend left=12] (J)(H) to[bend left=-12] (J)(J) to[bend left=20] (K)(J) to[bend left=-14] (K)(K) to[bend left=19] (L)(K) to[bend left=-15] (L)(L) to[bend left=15] (P)(L) to[bend left=-15] (P)(P) to[bend left=15] (O)(P) to[bend left=-15] (O)(K) to[bend left=18] (M)(K) to[bend left=-12] (M)(M) to[bend left=6] (O)(M) to[bend left=-19] (O)(O) to[bend left=7] (N)(O) to[bend left=-16] (N)(M) to[bend left=15] (I)(M) to[bend left=-15] (I);
\filldraw[fill=gray]
($($(A)!.3!(B)$)+(-45:-.25)$)--($($(A)!.3!(B)$)+(-45:.25)$)--($($(A)!.125!(G)$)+(22:-.08)$)--($($(A)!.125!(G)$)+(22:.35)$)--($($(A)!.19!(E)$)+(100:-.22)$)--($($(A)!.19!(E)$)+(100:.22)$)--($($(A)!.32!($(f)!.5!(j)$)$)+(-10:.22)$)--($($(A)!.32!($(f)!.5!(j)$)$)+(-10:-.2)$)--($($(A)!.24!($(f)!.45!(a)$)$)+(50:.22)$)--($($(A)!.24!($(f)!.45!(a)$)$)+(50:-.25)$)--cycle ($($(B)!.3!(A)$)+(-40:-.24)$)--($($(B)!.3!(A)$)+(-40:.25)$)--($($(B)!.151!(G)$)+(45:.22)$)--($($(B)!.151!(G)$)+(45:-.25)$)--($($(B)!.19!(C)$)+(-16:.25)$)--($($(B)!.19!(C)$)+(-16:-.25)$)--($($(B)!.3!($(a)!.45!(b)$)$)+(70:-.34)$)--($($(B)!.3!($(a)!.45!(b)$)$)+(70:.09)$)--cycle ($($(C)!.2!(B)$)+(-20:-.27)$)--($($(C)!.2!(B)$)+(-20:.25)$)--($($(C)!.16!(G)$)+(95:.25)$)--($($(C)!.16!(G)$)+(95:-.25)$)--($($(C)!.18!(D)$)+(20:.29)$)--($($(C)!.18!(D)$)+(20:-.29)$)--($($(C)!.32!($(c)!.51!(d)$)$)+(-50:.28)$)--($($(C)!.32!($(c)!.52!(d)$)$)+(-50:-.19)$)--($($(C)!.34!($(c)!.5!(b)$)$)+(60:-.33)$)--($($(C)!.34!($(c)!.5!(b)$)$)+(60:.22)$)--cycle ($($(D)!.18!(C)$)+(20:-.27)$)--($($(D)!.18!(C)$)+(20:.28)$)--($($(D)!.155!(G)$)+(-30:-.27)$)--($($(D)!.15!(G)$)+(-30:.28)$)--($($(D)!.155!(H)$)+(90:.26)$)--($($(D)!.155!(H)$)+(90:-.26)$)--($($(D)!.36!($(i)!.5!(m)$)$)+(30:.34)$)--($($(D)!.36!($(i)!.5!(m)$)$)+(30:-.18)$)--($($(D)!.34!($(e)!.5!(i)$)$)+(-35:.35)$)--($($(D)!.36!($(e)!.5!(i)$)$)+(-35:-.22)$)--($($(D)!.31!($(e)!.5!(d)$)$)+(83:-.28)$)--($($(D)!.31!($(e)!.5!(d)$)$)+(83:.28)$)--cycle ($($(E)!.19!(A)$)+(-80:-.23)$)--($($(E)!.19!(A)$)+(-80:.23)$)--($($(E)!.25!(F)$)+(35:-.26)$)--($($(E)!.25!(F)$)+(35:.26)$)--($($(E)!.3!($(n)!.33!(r)$)$)+(-50:.22)$)--($($(E)!.3!($(n)!.33!(r)$)$)+(-50:-.25)$)--($($(E)!.45!($(j)!.5!(n)$)$)+(15:.24)$)--($($(E)!.45!($(j)!.5!(n)$)$)+(15:-.26)$)--cycle ($($(F)!.24!(E)$)+(50:.24)$)--($($(F)!.24!(E)$)+(50:-.28)$)--($($(F)!.12!(G)$)+(-40:-.20)$)--($($(F)!.12!(G)$)+(-40:.20)$)--($($(F)!.17!(J)$)+(30:-.25)$)--($($(F)!.17!(J)$)+(30:.28)$)--($($(F)!.19!(I)$)+(-70:.30)$)--($($(F)!.19!(I)$)+(-70:-.32)$)--($($(F)!.28!($(n)!.67!(r)$)$)+(-10:.32)$)--($($(F)!.28!($(n)!.67!(r)$)$)+(-10:-.26)$)--cycle ($($(G)!.175!(A)$)+(20:.45)$)--($($(G)!.19!(A)$)+(20:-.08)$)--($($(G)!.2!(B)$)+(45:.26)$)--($($(G)!.195!(B)$)+(45:-.28)$)--($($(G)!.215!(C)$)+(95:.30)$)--($($(G)!.213!(C)$)+(95:-.32)$)--($($(G)!.17!(D)$)+(-25:-.3)$)--($($(G)!.17!(D)$)+(-25:.29)$)--($($(G)!.17!(H)$)+(30:-.28)$)--($($(G)!.17!(H)$)+(30:.28)$)--($($(G)!.14!(J)$)+(90:-.29)$)--($($(G)!.14!(J)$)+(90:.26)$)--($($(G)!.19!(F)$)+(-45:.30)$)--($($(G)!.19!(F)$)+(-45:-.30)$)--cycle ($($(H)!.175!(D)$)+(95:-.30)$)--($($(H)!.175!(D)$)+(95:.29)$)--($($(H)!.18!(G)$)+(25:-.3)$)--($($(H)!.18!(G)$)+(25:.28)$)--($($(H)!.15!(J)$)+(-35:-.26)$)--($($(H)!.15!(J)$)+(-35:.25)$)--($($(H)!.18!(L)$)+(-80:-.22)$)--($($(H)!.17!(L)$)+(-80:.32)$)--($($(H)!.32!($(t)!.5!(q)$)$)+(30:.23)$)--($($(H)!.32!($(t)!.5!(q)$)$)+(30:-.23)$)--($($(H)!.31!($(m)!.5!(q)$)$)+(-15:.22)$)--($($(H)!.31!($(m)!.5!(q)$)$)+(-15:-.3)$)--cycle ($($(I)!.175!(F)$)+(-75:-.29)$)--($($(I)!.18!(F)$)+(-75:.29)$)--($($(I)!.15!(J)$)+(-10:-.26)$)--($($(I)!.15!(J)$)+(-10:.26)$)--($($(I)!.18!(M)$)+(55:-.25)$)--($($(I)!.18!(M)$)+(55:.27)$)--($($(I)!.3!($(v)!.5!(w)$)$)+(-55:.22)$)--($($(I)!.29!($(v)!.5!(w)$)$)+(-55:-.3)$)--($($(I)!.31!($(v)!.5!(u)$)$)+(0:.22)$)--($($(I)!.31!($(v)!.5!(u)$)$)+(0:-.22)$)--($($(I)!.3!($(u)!.5!(r)$)$)+(35:.32)$)--($($(I)!.3!($(u)!.5!(r)$)$)+(35:-.12)$)--cycle ($($(J)!.18!(F)$)+(30:.29)$)--($($(J)!.18!(F)$)+(30:-.29)$)--($($(J)!.165!(G)$)+(90:.26)$)--($($(J)!.165!(G)$)+(90:-.3)$)--($($(J)!.16!(H)$)+(-35:-.26)$)--($($(J)!.155!(H)$)+(-35:.27)$)--($($(J)!.24!(L)$)+(20:-.25)$)--($($(J)!.24!(L)$)+(20:.28)$)--($($(J)!.25!(K)$)+(85:-.22)$)--($($(J)!.25!(K)$)+(85:.3)$)--($($(J)!.13!(I)$)+(-20:.25)$)--($($(J)!.13!(I)$)+(-20:-.23)$)--cycle ($($(K)!.265!(J)$)+(70:-.22)$)--($($(K)!.265!(J)$)+(70:.3)$)--($($(K)!.18!(M)$)+(-25:-.3)$)--($($(K)!.18!(M)$)+(-25:.22)$)--($($(K)!.2!(O)$)+(70:.28)$)--($($(K)!.195!(O)$)+(70:-.25)$)--($($(K)!.23!(L)$)+(-20:.3)$)--($($(K)!.225!(L)$)+(-20:-.27)$)--cycle ($($(L)!.158!(H)$)+(-80:.32)$)--($($(L)!.16!(H)$)+(-80:-.17)$)--($($(L)!.19!(J)$)+(20:-.24)$)--($($(L)!.19!(J)$)+(20:.24)$)--($($(L)!.21!(K)$)+(-35:-.21)$)--($($(L)!.21!(K)$)+(-35:.29)$)--($($(L)!.16!(P)$)+(70:.26)$)--($($(L)!.16!(P)$)+(70:-.28)$)--($($(L)!.45!($(aa)!.5!(t)$)$)+(-5:.35)$)--($($(L)!.45!($(aa)!.5!(t)$)$)+(-5:-.18)$)--cycle ($($(M)!.18!(I)$)+(65:.26)$)--($($(M)!.18!(I)$)+(65:-.26)$)--($($(M)!.195!(K)$)+(-35:-.34)$)--($($(M)!.2!(K)$)+(-35:.24)$)--($($(M)!.17!(O)$)+(20:-.34)$)--($($(M)!.17!(O)$)+(20:.12)$)--($($(M)!.38!(N)$)+(95:-.26)$)--($($(M)!.41!(N)$)+(95:.26)$)--($($(M)!.295!($(w)!.4!(bb)$)$)+(-15:.44)$)--($($(M)!.3!($(w)!.4!(bb)$)$)+(-15:-.14)$)--cycle ($($(N)!.22!(M)$)+(80:.17)$)--($($(N)!.22!(M)$)+(80:-.22)$)--($($(N)!.18!(O)$)+(-10:-.16)$)--($($(N)!.18!(O)$)+(-10:.34)$)--($($(N)!.5!($(cc)!.5!(dd)$)$)+(70:-.25)$)--($($(N)!.5!($(cc)!.5!(dd)$)$)+(70:.26)$)--($($(N)!.22!($(cc)!.5!(bb)$)$)+(-25:.18)$)--($($(N)!.2!($(cc)!.5!(bb)$)$)+(-25:-.14)$)--cycle ($($(O)!.19!(N)$)+(-10:.34)$)--($($(O)!.2!(N)$)+(-10:-.16)$)--($($(O)!.2!(M)$)+(30:.16)$)--($($(O)!.18!(M)$)+(30:-.38)$)--($($(O)!.16!(K)$)+(85:.22)$)--($($(O)!.155!(K)$)+(85:-.22)$)--($($(O)!.16!(P)$)+(-13:-.22)$)--($($(O)!.16!(P)$)+(-13:.22)$)--($($(O)!.25!($(ff)!.5!(eee)$)$)+(45:-.22)$)--($($(O)!.25!($(ff)!.5!(eee)$)$)+(45:.16)$)--($($(O)!.27!($(eee)!.5!(ee)$)$)+(-90:.14)$)--($($(O)!.275!($(eee)!.5!(ee)$)$)+(-90:-.18)$)--($($(O)!.3!($(ee)!.5!(dd)$)$)+(-70:.28)$)--($($(O)!.31!($(ee)!.5!(dd)$)$)+(-70:-.18)$)--cycle ($($(P)!.17!(O)$)+(-13:.22)$)--($($(P)!.17!(O)$)+(-13:-.22)$)--($($(P)!.18!(L)$)+(65:.29)$)--($($(P)!.18!(L)$)+(65:-.3)$)--($($(P)!.36!($(aa)!.5!(hh)$)$)+(-45:-.34)$)--($($(P)!.36!($(aa)!.5!(hh)$)$)+(-45:.18)$)--($($(P)!.35!($(hh)!.5!(gg)$)$)+(20:-.25)$)--($($(P)!.36!($(hh)!.5!(gg)$)$)+(20:.2)$)--($($(P)!.36!($(gg)!.5!(ff)$)$)+(90:-.24)$)--($($(P)!.36!($(gg)!.5!(ff)$)$)+(90:.3)$)--cycle;
\fill[gray] (b)--(c)--(d)--(h)--(g)--cycle (j)--(k)--(o)--(s)--(r)--(n)--cycle (w)--(x)--(y)--(dd)--(cc)--(bb)--cycle (l)--(m)--(q)--(t)--(p)--cycle (ff)--(gg)--(hh)--(aa)--(z)--cycle;
\draw[thick](j)--(f)--(a)--(b)--(c)--(d)--(e)--(i)--(m)--(q)--(t)--(aa)--(hh)--(gg)--(ff)--(eee)--(ee)--(dd)--(cc)--(bb)--(w)--(v)--(u)--(r)--(n)--cycle;
\draw[line cap = round] (b) -- (g) -- (h) -- (d) (h) -- (l) -- (m) (l) -- (p) -- (t) (g) -- (k) -- (o) -- (p) (j) -- (k) (r) --(s) -- (o) (w) -- (x) -- (s) (dd) -- (y) -- (x) (y) -- (z) -- (aa) (z) -- (ff);
\node at (15,-4) {$X'$};} 
\,\,\,\raisebox{.24cm}{$\sim$}\,\,\,\tikzmath[scale = .085]{\foreach \b/\x/\y in {a/.8/18,b/1.7/13.5,c/-.5/9.5,d/2/5,e/.5/1.5,f/5/20,g/5.4/11.8,h/5.8/7,i/5/0,j/9/19,k/9.8/14,l/9.7/5,m/9/1.5,n/13.2/20,o/13.7/11.3,p/13.5/6.9,q/13/0,r/16.5/18,s/17.5/14,t/16.5/2.5,u/19.5/20,v/23/20,w/25/18,x/22.3/13.8,y/25/10.2,z/24.7/5.4,aa/21.6/2.2,bb/29/19,cc/30.5/15,dd/29/11.5,ee/30.3/8.8,eee/30.2/6.2,ff/28.7/4.1,gg/29.8/.6,hh/25.3/0,A/6.5/17,B/3.9/14.8,C/2.9/9.2,D/5.5/3,E/11.5/17.5,F/13.8/14.8,G/9.5/9.5,H/12.7/3.1,I/20.5/17,J/18/10,K/22/9,L/19.5/5,M/25.5/14,N/28.1/14.5,O/27.5/7.3,P/26/2.3,A'/-2.5/3.5,B'/-2/13.5,B''/-2/16.5,C'/0/21,D'/7.5/22,D''/10.5/22,E'/16.5/21,F'/22/23,G'/28.5/22,H'/31.5/19.5,I'/31.5/12,J'/33/8,K'/31.5/2.5,L'/29/-2,M'/19/-.5,M''/16/0,N'/9/-1.5,O'/-.5/-2.5} {\coordinate (\b) at (\x,\y);}
\clip(j)--(f)--(a)--(b)--(c)--(d)--(e)--(i)--(m)--(q)--(t)--(aa)--(hh)--(gg)--(ff)--(eee)--(ee)--(dd)--(cc)--(bb)--(w)--(v)--(u)--(r)--(n)--cycle (12,-2.1) rectangle (18,-6); \fill[gray!20] (-2,-2) rectangle (32,22); 
\draw[line join=bevel](A') to[bend right=15] (D)(A') to[bend right=-15] (D)(A') to[bend right=15] (C)(A') to[bend right=-15] (C)(B') to[bend right=15] (C)(B') to[bend right=-15] (C)(B'') to[bend right=15] (B)(B'') to[bend right=-15] (B)(C') to[bend right=15] (A)(C') to[bend right=-15] (A)(D') to[bend right=15] (A)(D') to[bend right=-15] (A)(D'') to[bend right=18] (E)(D'') to[bend right=-18] (E)(E') to[bend right=15] (E)(E') to[bend right=-15] (E)(E') to[bend right=15] (F)(E') to[bend right=-15] (F)(E') to[bend right=15] (I)(E') to[bend right=-15] (I)(F') to[bend right=15] (I)(F') to[bend right=-15] (I)(G') to[bend right=12] (I)(G') to[bend right=-12] (I)(G') to[bend right=13] (M)(G') to[bend right=-13] (M)(H') to[bend right=15] (N)(H') to[bend right=-15] (N)(I') to[bend right=17] (N)(I') to[bend right=-17] (N)(I') to[bend right=13] (O)(I') to[bend right=-13] (O)(J') to[bend right=15] (O)(J') to[bend right=-15] (O)(K') to[bend right=15] (O)(K') to[bend right=-15] (O)(K') to[bend right=15] (P)(K') to[bend right=-15] (P)(L') to[bend right=15] (P)(L') to[bend right=-15] (P)(M') to[bend right=15] (P)(M') to[bend right=-15] (P)(M') to[bend right=15] (L)(M') to[bend right=-15] (L)(M'') to[bend right=15] (H)(M'') to[bend right=-15] (H)(N') to[bend right=15] (H)(N') to[bend right=-15] (H)(N') to[bend right=17] (D)(N') to[bend right=-15] (D)(O') to[bend right=15] (D)(O') to[bend right=-15] (D);\draw[line join=bevel](A) to[bend left=20] (G)(A) to[bend left=-4] (G)(J) to[bend left=15] (I)(J) to[bend left=-15] (I)(I) to[bend left=15] (F)(I) to[bend left=-15] (F)(F) to[bend left=15] (J)(F) to[bend left=-15] (J)(J) to[bend left=15] (G)(J) to[bend left=-13] (G)(G) to[bend left=15] (F)(G) to[bend left=-15] (F)(F) to[bend left=21] (E)(F) to[bend left=-19] (E)(E) to[bend left=15] (A)(E) to[bend left=-15] (A)(A) to[bend left=17] (B)(A) to[bend left=-17] (B)(B) to[bend left=15] (C)(B) to[bend left=-15] (C)(C) to[bend left=15] (D)(C) to[bend left=-15] (D)(D) to[bend left=15] (G)(D) to[bend left=-15] (G)(G) to[bend left=15] (H)(G) to[bend left=-15] (H)(H) to[bend left=15] (D)(H) to[bend left=-15] (D)(B) to[bend left=11] (G)(B) to[bend left=-13] (G)(G) to[bend left=15] (C)(G) to[bend left=-15] (C)(J) to[bend left=15] (L)(J) to[bend left=-15] (L)(L) to[bend left=18] (H)(L) to[bend left=-10] (H)(H) to[bend left=12] (J)(H) to[bend left=-12] (J)(J) to[bend left=20] (K)(J) to[bend left=-14] (K)(K) to[bend left=19] (L)(K) to[bend left=-15] (L)(L) to[bend left=15] (P)(L) to[bend left=-15] (P)(P) to[bend left=15] (O)(P) to[bend left=-15] (O)(O) to[bend left=15] (K)(O) to[bend left=-15] (K)(K) to[bend left=18] (M)(K) to[bend left=-12] (M)(M) to[bend left=6] (O)(M) to[bend left=-19] (O)(O) to[bend left=7] (N)(O) to[bend left=-16] (N)(N) to[bend left=21] (M)(N) to[bend left=-21] (M)(M) to[bend left=15] (I)(M) to[bend left=-15] (I);
\filldraw[fill=gray]
($($(A)!.3!(B)$)+(-45:-.25)$)--($($(A)!.3!(B)$)+(-45:.25)$)--($($(A)!.125!(G)$)+(22:-.08)$)--($($(A)!.125!(G)$)+(22:.35)$)--($($(A)!.19!(E)$)+(100:-.22)$)--($($(A)!.19!(E)$)+(100:.22)$)--($($(A)!.32!($(f)!.5!(j)$)$)+(-10:.22)$)--($($(A)!.32!($(f)!.5!(j)$)$)+(-10:-.2)$)--($($(A)!.24!($(f)!.45!(a)$)$)+(50:.22)$)--($($(A)!.24!($(f)!.45!(a)$)$)+(50:-.25)$)--cycle ($($(B)!.3!(A)$)+(-40:-.24)$)--($($(B)!.3!(A)$)+(-40:.25)$)--($($(B)!.151!(G)$)+(45:.22)$)--($($(B)!.151!(G)$)+(45:-.25)$)--($($(B)!.19!(C)$)+(-16:.25)$)--($($(B)!.19!(C)$)+(-16:-.25)$)--($($(B)!.3!($(a)!.45!(b)$)$)+(70:-.34)$)--($($(B)!.3!($(a)!.45!(b)$)$)+(70:.09)$)--cycle ($($(C)!.2!(B)$)+(-20:-.27)$)--($($(C)!.2!(B)$)+(-20:.25)$)--($($(C)!.16!(G)$)+(95:.25)$)--($($(C)!.16!(G)$)+(95:-.25)$)--($($(C)!.18!(D)$)+(20:.29)$)--($($(C)!.18!(D)$)+(20:-.29)$)--($($(C)!.32!($(c)!.51!(d)$)$)+(-50:.28)$)--($($(C)!.32!($(c)!.52!(d)$)$)+(-50:-.19)$)--($($(C)!.34!($(c)!.5!(b)$)$)+(60:-.33)$)--($($(C)!.34!($(c)!.5!(b)$)$)+(60:.22)$)--cycle ($($(D)!.18!(C)$)+(20:-.27)$)--($($(D)!.18!(C)$)+(20:.28)$)--($($(D)!.155!(G)$)+(-30:-.27)$)--($($(D)!.15!(G)$)+(-30:.28)$)--($($(D)!.155!(H)$)+(90:.26)$)--($($(D)!.155!(H)$)+(90:-.26)$)--($($(D)!.36!($(i)!.5!(m)$)$)+(30:.34)$)--($($(D)!.36!($(i)!.5!(m)$)$)+(30:-.18)$)--($($(D)!.34!($(e)!.5!(i)$)$)+(-35:.35)$)--($($(D)!.36!($(e)!.5!(i)$)$)+(-35:-.22)$)--($($(D)!.31!($(e)!.5!(d)$)$)+(83:-.28)$)--($($(D)!.31!($(e)!.5!(d)$)$)+(83:.28)$)--cycle ($($(E)!.19!(A)$)+(-80:-.23)$)--($($(E)!.19!(A)$)+(-80:.23)$)--($($(E)!.25!(F)$)+(35:-.26)$)--($($(E)!.25!(F)$)+(35:.26)$)--($($(E)!.3!($(n)!.33!(r)$)$)+(-50:.22)$)--($($(E)!.3!($(n)!.33!(r)$)$)+(-50:-.25)$)--($($(E)!.45!($(j)!.5!(n)$)$)+(15:.24)$)--($($(E)!.45!($(j)!.5!(n)$)$)+(15:-.26)$)--cycle ($($(F)!.24!(E)$)+(50:.24)$)--($($(F)!.24!(E)$)+(50:-.28)$)--($($(F)!.12!(G)$)+(-40:-.20)$)--($($(F)!.12!(G)$)+(-40:.20)$)--($($(F)!.17!(J)$)+(30:-.25)$)--($($(F)!.17!(J)$)+(30:.28)$)--($($(F)!.19!(I)$)+(-70:.30)$)--($($(F)!.19!(I)$)+(-70:-.32)$)--($($(F)!.28!($(n)!.67!(r)$)$)+(-10:.32)$)--($($(F)!.28!($(n)!.67!(r)$)$)+(-10:-.26)$)--cycle ($($(G)!.175!(A)$)+(20:.45)$)--($($(G)!.19!(A)$)+(20:-.08)$)--($($(G)!.2!(B)$)+(45:.26)$)--($($(G)!.195!(B)$)+(45:-.28)$)--($($(G)!.215!(C)$)+(95:.30)$)--($($(G)!.213!(C)$)+(95:-.32)$)--($($(G)!.17!(D)$)+(-25:-.3)$)--($($(G)!.17!(D)$)+(-25:.29)$)--($($(G)!.17!(H)$)+(30:-.28)$)--($($(G)!.17!(H)$)+(30:.28)$)--($($(G)!.14!(J)$)+(90:-.29)$)--($($(G)!.14!(J)$)+(90:.26)$)--($($(G)!.19!(F)$)+(-45:.30)$)--($($(G)!.19!(F)$)+(-45:-.30)$)--cycle ($($(H)!.175!(D)$)+(95:-.30)$)--($($(H)!.175!(D)$)+(95:.29)$)--($($(H)!.18!(G)$)+(25:-.3)$)--($($(H)!.18!(G)$)+(25:.28)$)--($($(H)!.15!(J)$)+(-35:-.26)$)--($($(H)!.15!(J)$)+(-35:.25)$)--($($(H)!.18!(L)$)+(-80:-.22)$)--($($(H)!.17!(L)$)+(-80:.32)$)--($($(H)!.32!($(t)!.5!(q)$)$)+(30:.23)$)--($($(H)!.32!($(t)!.5!(q)$)$)+(30:-.23)$)--($($(H)!.31!($(m)!.5!(q)$)$)+(-15:.22)$)--($($(H)!.31!($(m)!.5!(q)$)$)+(-15:-.3)$)--cycle ($($(I)!.175!(F)$)+(-75:-.29)$)--($($(I)!.18!(F)$)+(-75:.29)$)--($($(I)!.15!(J)$)+(-10:-.26)$)--($($(I)!.15!(J)$)+(-10:.26)$)--($($(I)!.18!(M)$)+(55:-.25)$)--($($(I)!.18!(M)$)+(55:.27)$)--($($(I)!.3!($(v)!.5!(w)$)$)+(-55:.22)$)--($($(I)!.29!($(v)!.5!(w)$)$)+(-55:-.3)$)--($($(I)!.31!($(v)!.5!(u)$)$)+(0:.22)$)--($($(I)!.31!($(v)!.5!(u)$)$)+(0:-.22)$)--($($(I)!.3!($(u)!.5!(r)$)$)+(35:.32)$)--($($(I)!.3!($(u)!.5!(r)$)$)+(35:-.12)$)--cycle ($($(J)!.18!(F)$)+(30:.29)$)--($($(J)!.18!(F)$)+(30:-.29)$)--($($(J)!.165!(G)$)+(90:.26)$)--($($(J)!.165!(G)$)+(90:-.3)$)--($($(J)!.16!(H)$)+(-35:-.26)$)--($($(J)!.155!(H)$)+(-35:.27)$)--($($(J)!.24!(L)$)+(20:-.25)$)--($($(J)!.24!(L)$)+(20:.28)$)--($($(J)!.25!(K)$)+(85:-.22)$)--($($(J)!.25!(K)$)+(85:.3)$)--($($(J)!.13!(I)$)+(-20:.25)$)--($($(J)!.13!(I)$)+(-20:-.23)$)--cycle ($($(K)!.265!(J)$)+(70:-.22)$)--($($(K)!.265!(J)$)+(70:.3)$)--($($(K)!.18!(M)$)+(-25:-.3)$)--($($(K)!.18!(M)$)+(-25:.22)$)--($($(K)!.2!(O)$)+(70:.28)$)--($($(K)!.195!(O)$)+(70:-.25)$)--($($(K)!.23!(L)$)+(-20:.3)$)--($($(K)!.225!(L)$)+(-20:-.27)$)--cycle ($($(L)!.158!(H)$)+(-80:.32)$)--($($(L)!.16!(H)$)+(-80:-.17)$)--($($(L)!.19!(J)$)+(20:-.24)$)--($($(L)!.19!(J)$)+(20:.24)$)--($($(L)!.21!(K)$)+(-35:-.21)$)--($($(L)!.21!(K)$)+(-35:.29)$)--($($(L)!.16!(P)$)+(70:.26)$)--($($(L)!.16!(P)$)+(70:-.28)$)--($($(L)!.45!($(aa)!.5!(t)$)$)+(-5:.35)$)--($($(L)!.45!($(aa)!.5!(t)$)$)+(-5:-.18)$)--cycle ($($(M)!.18!(I)$)+(65:.26)$)--($($(M)!.18!(I)$)+(65:-.26)$)--($($(M)!.195!(K)$)+(-35:-.34)$)--($($(M)!.2!(K)$)+(-35:.24)$)--($($(M)!.17!(O)$)+(20:-.34)$)--($($(M)!.17!(O)$)+(20:.12)$)--($($(M)!.38!(N)$)+(95:-.26)$)--($($(M)!.41!(N)$)+(95:.26)$)--($($(M)!.295!($(w)!.4!(bb)$)$)+(-15:.44)$)--($($(M)!.3!($(w)!.4!(bb)$)$)+(-15:-.14)$)--cycle ($($(N)!.22!(M)$)+(80:.17)$)--($($(N)!.22!(M)$)+(80:-.22)$)--($($(N)!.18!(O)$)+(-10:-.16)$)--($($(N)!.18!(O)$)+(-10:.34)$)--($($(N)!.5!($(cc)!.5!(dd)$)$)+(70:-.25)$)--($($(N)!.5!($(cc)!.5!(dd)$)$)+(70:.26)$)--($($(N)!.22!($(cc)!.5!(bb)$)$)+(-25:.18)$)--($($(N)!.2!($(cc)!.5!(bb)$)$)+(-25:-.14)$)--cycle ($($(O)!.19!(N)$)+(-10:.34)$)--($($(O)!.2!(N)$)+(-10:-.16)$)--($($(O)!.2!(M)$)+(30:.16)$)--($($(O)!.18!(M)$)+(30:-.38)$)--($($(O)!.16!(K)$)+(85:.22)$)--($($(O)!.155!(K)$)+(85:-.22)$)--($($(O)!.16!(P)$)+(-13:-.22)$)--($($(O)!.16!(P)$)+(-13:.22)$)--($($(O)!.25!($(ff)!.5!(eee)$)$)+(45:-.22)$)--($($(O)!.25!($(ff)!.5!(eee)$)$)+(45:.16)$)--($($(O)!.27!($(eee)!.5!(ee)$)$)+(-90:.14)$)--($($(O)!.275!($(eee)!.5!(ee)$)$)+(-90:-.18)$)--($($(O)!.3!($(ee)!.5!(dd)$)$)+(-70:.28)$)--($($(O)!.31!($(ee)!.5!(dd)$)$)+(-70:-.18)$)--cycle ($($(P)!.17!(O)$)+(-13:.22)$)--($($(P)!.17!(O)$)+(-13:-.22)$)--($($(P)!.18!(L)$)+(65:.29)$)--($($(P)!.18!(L)$)+(65:-.3)$)--($($(P)!.36!($(aa)!.5!(hh)$)$)+(-45:-.34)$)--($($(P)!.36!($(aa)!.5!(hh)$)$)+(-45:.18)$)--($($(P)!.35!($(hh)!.5!(gg)$)$)+(20:-.25)$)--($($(P)!.36!($(hh)!.5!(gg)$)$)+(20:.2)$)--($($(P)!.36!($(gg)!.5!(ff)$)$)+(90:-.24)$)--($($(P)!.36!($(gg)!.5!(ff)$)$)+(90:.3)$)--cycle;
\fill[gray] (b)--(c)--(d)--(h)--(g)--cycle (j)--(k)--(o)--(s)--(r)--(n)--cycle (w)--(x)--(y)--(dd)--(cc)--(bb)--cycle (l)--(m)--(q)--(t)--(p)--cycle (ff)--(gg)--(hh)--(aa)--(z)--cycle;
\draw[thick](b)--(c)--(d)(m)--(q)--(t)(aa)--(hh)--(gg)--(ff)(dd)--(cc)--(bb)--(w)(r)--(n)--(j);
\draw (b) -- (g) -- (h) -- (d) (m) -- (l) -- (p) -- (t) (j) -- (k) -- (o) -- (s) -- ($(r)+(-.1,0)$) (w)+(.1,0) -- (x) -- (y) -- ($(dd)+(0,.1)$) (ff)+(0,-.1) -- (z) -- (aa);
\node at (15,-4) {$X''$};} 
\,\,\,\raisebox{.24cm}{$\sim$}\,\,\,
\tikzmath[scale = .085]{\foreach \b/\x/\y in {a/.8/18,b/1.7/13.5,c/-.5/9.5,d/2/5,e/.5/1.5,f/5/20,g/5.4/11.8,h/5.8/7,i/5/0,j/9/19,k/9.8/14,l/9.7/5,m/9/1.5,n/13.2/20,o/13.7/11.3,p/13.5/6.9,q/13/0,r/16.5/18,s/17.5/14,t/16.5/2.5,u/19.5/20,v/23/20,w/25/18,x/22.3/13.8,y/25/10.2,z/24.7/5.4,aa/21.6/2.2,bb/29/19,cc/30.5/15,dd/29/11.5,ee/30.3/8.8,eee/30.2/6.2,ff/28.7/4.1,gg/29.8/.6,hh/25.3/0,A/6.5/17,B/3.9/14.8,C/2.9/9.2,D/5.5/3,E/11.5/17.5,F/13.8/14.8,G/9.5/9.5,H/12.7/3.1,I/20.5/17,J/18/10,K/22/9,L/19.5/5,M/25.5/14,N/28.1/14.5,O/27.5/7.3,P/26/2.3,A'/-2.5/3.5,B'/-2/13.5,B''/-2/16.5,C'/0/21,D'/7.5/22,D''/10.5/22,E'/16.5/21,F'/22/23,G'/28.5/22,H'/31.5/19.5,I'/31.5/12,J'/33/8,K'/31.5/2.5,L'/29/-2,M'/19/-.5,M''/16/0,N'/9/-1.5,O'/-.5/-2.5} {\coordinate (\b) at (\x,\y);}
\clip(j)--(f)--(a)--(b)--(c)--(d)--(e)--(i)--(m)--(q)--(t)--(aa)--(hh)--(gg)--(ff)--(eee)--(ee)--(dd)--(cc)--(bb)--(w)--(v)--(u)--(r)--(n)--cycle (12,-2.1) rectangle (18,-6);
\fill[gray!20] (-2,-2) rectangle (32,22);
\draw[line join=bevel](A') to[bend right=15] (D)(A') to[bend right=-15] (D)(A') to[bend right=15] (C)(A') to[bend right=-15] (C)(B') to[bend right=15] (C)(B') to[bend right=-15] (C)(B'') to[bend right=15] (B)(B'') to[bend right=-15] (B)(C') to[bend right=15] (A)(C') to[bend right=-15] (A)(D') to[bend right=15] (A)(D') to[bend right=-15] (A)(D'') to[bend right=18] (E)(D'') to[bend right=-18] (E)(E') to[bend right=15] (E)(E') to[bend right=-15] (E)(E') to[bend right=15] (F)(E') to[bend right=-15] (F)(E') to[bend right=15] (I)(E') to[bend right=-15] (I)(F') to[bend right=15] (I)(F') to[bend right=-15] (I)(G') to[bend right=12] (I)(G') to[bend right=-12] (I)(G') to[bend right=13] (M)(G') to[bend right=-13] (M)(H') to[bend right=15] (N)(H') to[bend right=-15] (N)(I') to[bend right=17] (N)(I') to[bend right=-17] (N)(I') to[bend right=13] (O)(I') to[bend right=-13] (O)(J') to[bend right=15] (O)(J') to[bend right=-15] (O)(K') to[bend right=15] (O)(K') to[bend right=-15] (O)(K') to[bend right=15] (P)(K') to[bend right=-15] (P)(L') to[bend right=15] (P)(L') to[bend right=-15] (P)(M') to[bend right=15] (P)(M') to[bend right=-15] (P)(M') to[bend right=15] (L)(M') to[bend right=-15] (L)(M'') to[bend right=15] (H)(M'') to[bend right=-15] (H)(N') to[bend right=15] (H)(N') to[bend right=-15] (H)(N') to[bend right=17] (D)(N') to[bend right=-15] (D)(O') to[bend right=15] (D)(O') to[bend right=-15] (D);\draw[line join=bevel](A) to[bend left=20] (G)(A) to[bend left=-4] (G)(J) to[bend left=15] (I)(J) to[bend left=-15] (I)(I) to[bend left=15] (F)(I) to[bend left=-15] (F)(F) to[bend left=15] (J)(F) to[bend left=-15] (J)(J) to[bend left=15] (G)(J) to[bend left=-13] (G)(G) to[bend left=15] (F)(G) to[bend left=-15] (F)(F) to[bend left=21] (E)(F) to[bend left=-19] (E)(E) to[bend left=15] (A)(E) to[bend left=-15] (A)(A) to[bend left=17] (B)(A) to[bend left=-17] (B)(B) to[bend left=15] (C)(B) to[bend left=-15] (C)(C) to[bend left=15] (D)(C) to[bend left=-15] (D)(D) to[bend left=15] (G)(D) to[bend left=-15] (G)(G) to[bend left=15] (H)(G) to[bend left=-15] (H)(H) to[bend left=15] (D)(H) to[bend left=-15] (D)(B) to[bend left=11] (G)(B) to[bend left=-13] (G)(G) to[bend left=15] (C)(G) to[bend left=-15] (C)(J) to[bend left=15] (L)(J) to[bend left=-15] (L)(L) to[bend left=18] (H)(L) to[bend left=-10] (H)(H) to[bend left=12] (J)(H) to[bend left=-12] (J)(J) to[bend left=20] (K)(J) to[bend left=-14] (K)(K) to[bend left=19] (L)(K) to[bend left=-15] (L)(L) to[bend left=15] (P)(L) to[bend left=-15] (P)(P) to[bend left=15] (O)(P) to[bend left=-15] (O)(O) to[bend left=15] (K)(O) to[bend left=-15] (K)(K) to[bend left=18] (M)(K) to[bend left=-12] (M)(M) to[bend left=6] (O)(M) to[bend left=-19] (O)(O) to[bend left=7] (N)(O) to[bend left=-16] (N)(N) to[bend left=21] (M)(N) to[bend left=-21] (M)(M) to[bend left=15] (I)(M) to[bend left=-15] (I);
\filldraw[fill=gray]
($($(A)!.3!(B)$)+(-45:-.25)$)--($($(A)!.3!(B)$)+(-45:.25)$)--($($(A)!.125!(G)$)+(22:-.08)$)--($($(A)!.125!(G)$)+(22:.35)$)--($($(A)!.19!(E)$)+(100:-.22)$)--($($(A)!.19!(E)$)+(100:.22)$)--($($(A)!.32!($(f)!.5!(j)$)$)+(-10:.22)$)--($($(A)!.32!($(f)!.5!(j)$)$)+(-10:-.2)$)--($($(A)!.24!($(f)!.45!(a)$)$)+(50:.22)$)--($($(A)!.24!($(f)!.45!(a)$)$)+(50:-.25)$)--cycle ($($(B)!.3!(A)$)+(-40:-.24)$)--($($(B)!.3!(A)$)+(-40:.25)$)--($($(B)!.151!(G)$)+(45:.22)$)--($($(B)!.151!(G)$)+(45:-.25)$)--($($(B)!.19!(C)$)+(-16:.25)$)--($($(B)!.19!(C)$)+(-16:-.25)$)--($($(B)!.3!($(a)!.45!(b)$)$)+(70:-.34)$)--($($(B)!.3!($(a)!.45!(b)$)$)+(70:.09)$)--cycle ($($(C)!.2!(B)$)+(-20:-.27)$)--($($(C)!.2!(B)$)+(-20:.25)$)--($($(C)!.16!(G)$)+(95:.25)$)--($($(C)!.16!(G)$)+(95:-.25)$)--($($(C)!.18!(D)$)+(20:.29)$)--($($(C)!.18!(D)$)+(20:-.29)$)--($($(C)!.32!($(c)!.51!(d)$)$)+(-50:.28)$)--($($(C)!.32!($(c)!.52!(d)$)$)+(-50:-.19)$)--($($(C)!.34!($(c)!.5!(b)$)$)+(60:-.33)$)--($($(C)!.34!($(c)!.5!(b)$)$)+(60:.22)$)--cycle ($($(D)!.18!(C)$)+(20:-.27)$)--($($(D)!.18!(C)$)+(20:.28)$)--($($(D)!.155!(G)$)+(-30:-.27)$)--($($(D)!.15!(G)$)+(-30:.28)$)--($($(D)!.155!(H)$)+(90:.26)$)--($($(D)!.155!(H)$)+(90:-.26)$)--($($(D)!.36!($(i)!.5!(m)$)$)+(30:.34)$)--($($(D)!.36!($(i)!.5!(m)$)$)+(30:-.18)$)--($($(D)!.34!($(e)!.5!(i)$)$)+(-35:.35)$)--($($(D)!.36!($(e)!.5!(i)$)$)+(-35:-.22)$)--($($(D)!.31!($(e)!.5!(d)$)$)+(83:-.28)$)--($($(D)!.31!($(e)!.5!(d)$)$)+(83:.28)$)--cycle ($($(E)!.19!(A)$)+(-80:-.23)$)--($($(E)!.19!(A)$)+(-80:.23)$)--($($(E)!.25!(F)$)+(35:-.26)$)--($($(E)!.25!(F)$)+(35:.26)$)--($($(E)!.3!($(n)!.33!(r)$)$)+(-50:.22)$)--($($(E)!.3!($(n)!.33!(r)$)$)+(-50:-.25)$)--($($(E)!.45!($(j)!.5!(n)$)$)+(15:.24)$)--($($(E)!.45!($(j)!.5!(n)$)$)+(15:-.26)$)--cycle ($($(F)!.24!(E)$)+(50:.24)$)--($($(F)!.24!(E)$)+(50:-.28)$)--($($(F)!.12!(G)$)+(-40:-.20)$)--($($(F)!.12!(G)$)+(-40:.20)$)--($($(F)!.17!(J)$)+(30:-.25)$)--($($(F)!.17!(J)$)+(30:.28)$)--($($(F)!.19!(I)$)+(-70:.30)$)--($($(F)!.19!(I)$)+(-70:-.32)$)--($($(F)!.28!($(n)!.67!(r)$)$)+(-10:.32)$)--($($(F)!.28!($(n)!.67!(r)$)$)+(-10:-.26)$)--cycle ($($(G)!.175!(A)$)+(20:.45)$)--($($(G)!.19!(A)$)+(20:-.08)$)--($($(G)!.2!(B)$)+(45:.26)$)--($($(G)!.195!(B)$)+(45:-.28)$)--($($(G)!.215!(C)$)+(95:.30)$)--($($(G)!.213!(C)$)+(95:-.32)$)--($($(G)!.17!(D)$)+(-25:-.3)$)--($($(G)!.17!(D)$)+(-25:.29)$)--($($(G)!.17!(H)$)+(30:-.28)$)--($($(G)!.17!(H)$)+(30:.28)$)--($($(G)!.14!(J)$)+(90:-.29)$)--($($(G)!.14!(J)$)+(90:.26)$)--($($(G)!.19!(F)$)+(-45:.30)$)--($($(G)!.19!(F)$)+(-45:-.30)$)--cycle ($($(H)!.175!(D)$)+(95:-.30)$)--($($(H)!.175!(D)$)+(95:.29)$)--($($(H)!.18!(G)$)+(25:-.3)$)--($($(H)!.18!(G)$)+(25:.28)$)--($($(H)!.15!(J)$)+(-35:-.26)$)--($($(H)!.15!(J)$)+(-35:.25)$)--($($(H)!.18!(L)$)+(-80:-.22)$)--($($(H)!.17!(L)$)+(-80:.32)$)--($($(H)!.32!($(t)!.5!(q)$)$)+(30:.23)$)--($($(H)!.32!($(t)!.5!(q)$)$)+(30:-.23)$)--($($(H)!.31!($(m)!.5!(q)$)$)+(-15:.22)$)--($($(H)!.31!($(m)!.5!(q)$)$)+(-15:-.3)$)--cycle ($($(I)!.175!(F)$)+(-75:-.29)$)--($($(I)!.18!(F)$)+(-75:.29)$)--($($(I)!.15!(J)$)+(-10:-.26)$)--($($(I)!.15!(J)$)+(-10:.26)$)--($($(I)!.18!(M)$)+(55:-.25)$)--($($(I)!.18!(M)$)+(55:.27)$)--($($(I)!.3!($(v)!.5!(w)$)$)+(-55:.22)$)--($($(I)!.29!($(v)!.5!(w)$)$)+(-55:-.3)$)--($($(I)!.31!($(v)!.5!(u)$)$)+(0:.22)$)--($($(I)!.31!($(v)!.5!(u)$)$)+(0:-.22)$)--($($(I)!.3!($(u)!.5!(r)$)$)+(35:.32)$)--($($(I)!.3!($(u)!.5!(r)$)$)+(35:-.12)$)--cycle ($($(J)!.18!(F)$)+(30:.29)$)--($($(J)!.18!(F)$)+(30:-.29)$)--($($(J)!.165!(G)$)+(90:.26)$)--($($(J)!.165!(G)$)+(90:-.3)$)--($($(J)!.16!(H)$)+(-35:-.26)$)--($($(J)!.155!(H)$)+(-35:.27)$)--($($(J)!.24!(L)$)+(20:-.25)$)--($($(J)!.24!(L)$)+(20:.28)$)--($($(J)!.25!(K)$)+(85:-.22)$)--($($(J)!.25!(K)$)+(85:.3)$)--($($(J)!.13!(I)$)+(-20:.25)$)--($($(J)!.13!(I)$)+(-20:-.23)$)--cycle ($($(K)!.265!(J)$)+(70:-.22)$)--($($(K)!.265!(J)$)+(70:.3)$)--($($(K)!.18!(M)$)+(-25:-.3)$)--($($(K)!.18!(M)$)+(-25:.22)$)--($($(K)!.2!(O)$)+(70:.28)$)--($($(K)!.195!(O)$)+(70:-.25)$)--($($(K)!.23!(L)$)+(-20:.3)$)--($($(K)!.225!(L)$)+(-20:-.27)$)--cycle ($($(L)!.158!(H)$)+(-80:.32)$)--($($(L)!.16!(H)$)+(-80:-.17)$)--($($(L)!.19!(J)$)+(20:-.24)$)--($($(L)!.19!(J)$)+(20:.24)$)--($($(L)!.21!(K)$)+(-35:-.21)$)--($($(L)!.21!(K)$)+(-35:.29)$)--($($(L)!.16!(P)$)+(70:.26)$)--($($(L)!.16!(P)$)+(70:-.28)$)--($($(L)!.45!($(aa)!.5!(t)$)$)+(-5:.35)$)--($($(L)!.45!($(aa)!.5!(t)$)$)+(-5:-.18)$)--cycle ($($(M)!.18!(I)$)+(65:.26)$)--($($(M)!.18!(I)$)+(65:-.26)$)--($($(M)!.195!(K)$)+(-35:-.34)$)--($($(M)!.2!(K)$)+(-35:.24)$)--($($(M)!.17!(O)$)+(20:-.34)$)--($($(M)!.17!(O)$)+(20:.12)$)--($($(M)!.38!(N)$)+(95:-.26)$)--($($(M)!.41!(N)$)+(95:.26)$)--($($(M)!.295!($(w)!.4!(bb)$)$)+(-15:.44)$)--($($(M)!.3!($(w)!.4!(bb)$)$)+(-15:-.14)$)--cycle ($($(N)!.22!(M)$)+(80:.17)$)--($($(N)!.22!(M)$)+(80:-.22)$)--($($(N)!.18!(O)$)+(-10:-.16)$)--($($(N)!.18!(O)$)+(-10:.34)$)--($($(N)!.5!($(cc)!.5!(dd)$)$)+(70:-.25)$)--($($(N)!.5!($(cc)!.5!(dd)$)$)+(70:.26)$)--($($(N)!.22!($(cc)!.5!(bb)$)$)+(-25:.18)$)--($($(N)!.2!($(cc)!.5!(bb)$)$)+(-25:-.14)$)--cycle ($($(O)!.19!(N)$)+(-10:.34)$)--($($(O)!.2!(N)$)+(-10:-.16)$)--($($(O)!.2!(M)$)+(30:.16)$)--($($(O)!.18!(M)$)+(30:-.38)$)--($($(O)!.16!(K)$)+(85:.22)$)--($($(O)!.155!(K)$)+(85:-.22)$)--($($(O)!.16!(P)$)+(-13:-.22)$)--($($(O)!.16!(P)$)+(-13:.22)$)--($($(O)!.25!($(ff)!.5!(eee)$)$)+(45:-.22)$)--($($(O)!.25!($(ff)!.5!(eee)$)$)+(45:.16)$)--($($(O)!.27!($(eee)!.5!(ee)$)$)+(-90:.14)$)--($($(O)!.275!($(eee)!.5!(ee)$)$)+(-90:-.18)$)--($($(O)!.3!($(ee)!.5!(dd)$)$)+(-70:.28)$)--($($(O)!.31!($(ee)!.5!(dd)$)$)+(-70:-.18)$)--cycle ($($(P)!.17!(O)$)+(-13:.22)$)--($($(P)!.17!(O)$)+(-13:-.22)$)--($($(P)!.18!(L)$)+(65:.29)$)--($($(P)!.18!(L)$)+(65:-.3)$)--($($(P)!.36!($(aa)!.5!(hh)$)$)+(-45:-.34)$)--($($(P)!.36!($(aa)!.5!(hh)$)$)+(-45:.18)$)--($($(P)!.35!($(hh)!.5!(gg)$)$)+(20:-.25)$)--($($(P)!.36!($(hh)!.5!(gg)$)$)+(20:.2)$)--($($(P)!.36!($(gg)!.5!(ff)$)$)+(90:-.24)$)--($($(P)!.36!($(gg)!.5!(ff)$)$)+(90:.3)$)--cycle;
\node at (15,-4) {$Y$};
} 
\]

(3) Finally, given an edge $X_1\stackrel \ID\sim X_2$ of $\Gamma$, we construct the map $\delta_\ID:D^2\to\IX_\Sigma$ that bounds the triangle \eqref{eq: triangle to bound}.
That triangle can be decomposed as follows
\[
\quad\tikzmath{
\node (a) at (0,0) {$X_1$};\node (a') at (2,0) {$X'_1$};\node (a'') at (4,0) {$X''_1$};
\node (b) at (0,-1) {$X_2$};\node (b') at (2,-1) {$X'_2$};\node (b'') at (4,-1) {$X''_2$};
\node (c) at (6.1,-.5) {$Y$};
\draw (a) --node[left]{$\scriptstyle \ID$} (b);
\draw (a') --node[left]{$\scriptstyle \ID$} (b');
\draw (a'') --node[left]{$\scriptstyle \tilde\ID$} (b'');
\draw (a) -- (a') -- (a'') to[bend left = 10] (c);
\draw (b) -- (b') -- (b'') to[bend right = 10] (c);
\node[draw, circle, scale=.8, inner sep=2] at (0.95,-.5){1};
\node[draw, circle, scale=.8, inner sep=2] at (2.95,-.5){2};
\node[draw, circle, scale=.8, inner sep=2] at (4.8,-.5){3};
}
\]
and we claim that each one of the above cycles
\hspace{.1mm}$\tikzmath{\node[scale=.8, draw, circle, inner sep=1.5]{1};}$\hspace{.2mm},
\hspace{.1mm}$\tikzmath{\node[scale=.8, draw, circle, inner sep=1.5]{2};}$\hspace{.2mm},
\hspace{.1mm}$\tikzmath{\node[scale=.8, draw, circle, inner sep=1.5]{3};}$
bounds a 2-cell of $\IX_\Sigma$.
The union of those three 2-cells provides the desired filler $\delta_\ID$.

By the definition of $X_1\stackrel \ID\sim X_2$, there is a soccer ball decomposition $Z$ such that
$X_1\tikz{\useasboundingbox(-.24,-.1)rectangle(.24,.3);\node at (0,.23) {$\scriptstyle \ID$};\node at (0,0) {$\prec$};} Z\tikz{\useasboundingbox(-.24,-.1)rectangle(.24,.3);\node at (0,.23) {$\scriptstyle \ID$};\node at (0,0) {$\succ$};}  X_2$.
The first loop \hspace{.1mm}$\tikzmath{\node[scale=.8, draw, circle, inner sep=1.5]{1};}$\hspace{.1mm} bounds a 2-cell because 
\[
X_1\prec Z,\qquad X_2\prec Z,\qquad X'_1\prec Z,\qquad \text{and}\qquad X'_2\prec Z.
\]
Let $\tilde\ID\subset\Sigma$ be the union of all the cells of $X_1''$ (equivalently $X_2''$) whose interior intersects $\ID$.
It follows from the construction of $Y$ that $\tilde\ID$ is a disc.
Let $Z'$ be the soccer ball decomposition obtained from $X'_1$ (equivalently $X_2'$)
by removing all the edges $e\in A$, and all the edges and vertices
that are in the interior of $\ID$.
Note that $\tilde\ID$ is one of the cells of $Z'$.
We then have
\[
X'_1\prec Z',\qquad X'_2\prec Z',\qquad X''_1\prec Z',\qquad \text{and}\qquad X''_2\prec Z',
\]
and so the second loop \hspace{.1mm}$\tikzmath{\node[scale=.8, draw, circle, inner sep=1.5]{2};}$\hspace{.1mm} also bounds a 2-cell.
To finish the argument, we need to construct a soccer ball decomposition $Z''$ such that
\[
X''_1\prec Z'',\qquad X''_2\prec Z'',\qquad \text{and}\qquad Y\prec Z''.
\]
This will then show that the third loop \hspace{.1mm}$\tikzmath{\node[scale=.8, draw, circle, inner sep=1.5]{3};}$\hspace{.1mm} bounds a 2-cell.

Let $B_1,\ldots,B_n$ be the set of black cells of $X_1''$ (equivalently $X_2''$) that are not contained in $\ID$.
For each one of them, let $\tilde B_i$ be the union of all the cells of $Y$ that intersect $B_i$.
It follows from the construction of $Y$  that $\tilde B_i$ is a disc.
Let $U\subset \ID$ be the union of all the cells of $Y$ that intersect some black cell $B\subset \ID$ of $X_1$ or $X_2$ and let $U_1,\ldots,U_k$ be the connected components of $U$.
Finally, let $D_j$ be the smallest disk that is contained in $\ID$ and that contains $U_j$.
It follows from the construction of $Y$ that the set $S:=\{\tilde B_1,\ldots,\tilde B_n, D_1,\ldots, D_k\}$ consists of pairwise disjoint discs.
Moreover, it defines a unique soccer ball decomposition $Z''$ such that 
$X_1''\tikz{\useasboundingbox(-.28,-.1)rectangle(.28,.3);\node at (0,.22) {$\scriptstyle S$};\node at (0,0) {$\prec$};} Z''$,
$X_2''\tikz{\useasboundingbox(-.28,-.1)rectangle(.28,.3);\node at (0,.22) {$\scriptstyle S$};\node at (0,0) {$\prec$};} Z''$,
and
$Y\tikz{\useasboundingbox(-.28,-.1)rectangle(.28,.3);\node at (0,.22) {$\scriptstyle S$};\node at (0,0) {$\prec$};} Z''$.
This finishes the proof that $\IX_\Sigma$ is simply connected.\footnote{\label{cavaetresolved}{\it Caveat resolution:} If one follows the definition of $X'$ presented above, using a cell decomposition $Y_1$ whose 1-skeleton has a disconnected intersection with some 2-cell of some $X_i$, then one may encounter the following problem.
Let $\ID$ be a white cell of $X$, and let $e$ be an edge that separates it from some other white cell of $X$.
If there  exists a small component of the intersection of the 1-skeleton of $Y_1$ with $\ID$ that only touches $e$, then the $1$-skeleton of $X'$ will be disconnected and $X'$ won't be a soccer ball decomposition.  However, we may simply redefine $X'$ by removing those small components; this yields a new soccer ball decomposition which one uses in place of $X'$.
The rest of the argument is unaffected by this modification.}
\end{proof}

\subsection{Factorization along circles and along intervals}

From now on, all surfaces will be assumed to be compact oriented topological surfaces with smooth boundary.

Let $\Sigma_1$ and $\Sigma_2$ be two surfaces, pick an orientation-reversing diffeomorphism $\varphi:M_1\to M$ between submanifolds $M_1\subset \partial \Sigma_1$ and $M\subset \partial \Sigma_2$, 
and equip the trivalent graph $\Gamma:=\partial \Sigma_1\cup_M \partial \Sigma_2$ with a smooth structure (Definition \ref{def: smooth structure on trivalent graph}), compatibly with the existing smooth structures on $\partial \Sigma_1$ and $\partial \Sigma_2$.
The smooth structure on $\Gamma$ restricts to a smooth structure on the boundary of $\Sigma_1\cup_M \Sigma_2$,
and so we can define $V(\Sigma_1\cup_M \Sigma_2)$ by Theorem \ref{Mthm: conformal blocks}.

\begin{maintheorem}\label{maintheorem gluing}
Let $\Sigma_1$, $\Sigma_2$, $M_1$, $M$, $\varphi$ be as above.
Then there is a unitary isomorphism
\begin{equation}\label{eq: main thm eq2}
g:\,V(\Sigma_1\cup_M \Sigma_2)\to V(\Sigma_1)\boxtimes_{\cala(M)} V(\Sigma_2),
\end{equation}
well defined up to phase.
\end{maintheorem}

Before embarking on the proof of the theorem, we will need the following variation of the notion of soccer ball decomposition:

\begin{definition}
Let $\Sigma$ be surface, and let $S\subset \partial \Sigma$ 
be a submanifold of its boundary diffeomorphic to a disjoint union of circles.
A regular trivalent smooth cell decomposition $X$ of $\Sigma$ 
(see Section \ref{sec: The Hilbert space associated to a surface}), equipped with a partition of the set 
of $2$-cells into subsets $X_{\mathrm{white}}$ and $X_{\mathrm{black}}$ is called an 
\emph{$S$-open soccer ball decomposition of $\Sigma$}
if every interior vertex is adjacent to two white cells and one back cell, 
every boundary vertex not in $S$ is adjacent to two white cells, and 
every boundary vertex  in $S$ is adjacent to a white and a black cell.
\end{definition}

We shall sometimes abbreviate the above terminology by ``open soccer ball decomposition'' when the manifold $S$ is obvious from the context.

\[
\parbox{5cm}{An example of an open\\ soccer ball decomposition:
}\qquad
\tikzmath[scale = .8]{\useasboundingbox (1.3,-1.1) rectangle (5.85,1.1);
\draw[gray, line width=.6, rounded corners=8.5, fill=gray!20](.3,.4) -- (1,1) -- (2,1) --  (3,.75) -- (4,1) -- (5,1) -- (5.7,.4) -- (5.7,-.4) -- (5,-1) -- (4,-1) -- (3,-.75) -- (2,-1) -- (1,-1) -- (.3,-.4) -- cycle;
\path[rotate=15](1.8,-1.3)coordinate(a)+(-.012,-.175)coordinate(aa)(4.3,-1.9)coordinate(c)(4.6,-1.48)coordinate(d);\path[rotate=-15](4.7,1.9)coordinate(g)(3.96,.71)coordinate(h);\path[rotate=-8](3.58,-.2)coordinate(i);\path(2.4,-.35)coordinate(j)+(-.26,.26)coordinate(jj)(3.05,-.3)coordinate(l)(4.4,.75)coordinate(m)+(-.02,.25)coordinate(mm)(a)+(-.285,.015)coordinate(b)(c)+(.045,-.275)coordinate(cc)(g)+(.18,.19)coordinate(gg)(d)+(-.11,.18)coordinate(dd)(h)+(.035,.2)coordinate(hh)(3.5,-.875)coordinate(ii);
\path(3.05,.805)coordinate(o) (3.3,.37)coordinate(p)(3.05,.05)coordinate(q)(2.3,.475)coordinate(r)(2.4,.9)coordinate(s)(5.66,.25)coordinate(e)(5.312,.15)coordinate(f)(5.28,-.14)coordinate(k)(5.63,-.33)coordinate(t)(3.66,.44)coordinate(u)(3.75,.7)coordinate(v)+(-.04,.223)coordinate(vv)(4,.585)coordinate(w)(3.93,.3)coordinate(x)+(.07,-.145)coordinate(xx)(1.785,.43)coordinate(rr);
\fill[gray] (hh) -- (h) -- (i) -- (l) -- (q) -- (p) -- (u) -- (x) -- (xx);\fill[gray] (rr) -- (r) -- (s) to[bend right=4] ($(s)+(-.7,.105)$);\fill[gray] (b) -- (a) to[bend right=10] (j) -- (jj) to[bend right=-10] (1.7,-.16);\fill[gray] (v) -- (w) -- (m) -- (mm) to[bend right=4] (vv);\fill[gray] (g) to[bend left=6] (f) -- (e) to[bend right=18] (gg);\fill[gray] (c) to[bend right=10] (d) to[bend right=3] (k) -- (t) to[bend left=30, looseness=.85] (cc);
\fill[white](1.19,-.01) .. controls (1.5,.25) and (2,.25) .. (2.31,-.01)(1.19,-.01) .. controls (1.5,-.21) and (2,-.21) .. (2.31,-.01)
(3.69,-.01) .. controls (4,.25) and (4.5,.25) .. (4.81,-.01)(3.69,-.01) .. controls (4,-.21) and (4.5,-.21) .. (4.81,-.01);
\draw[gray, line width=.6] (1,.1) .. controls (1.5,-.25) and (2,-.25) .. (2.5,.1)(1.2,0) .. controls (1.5,.25) and (2,.25) .. (2.3,0)
(3.5,.1) .. controls (4,-.25) and (4.5,-.25) .. (5,.1)(3.7,0) .. controls (4,.25) and (4.5,.25) .. (4.8,0);
\draw[line width=.3]  (e) -- (f) -- (5.28,-.14) -- (5.63,-.33);\draw[line width=.3] (u) -- (v) -- (w) -- (x) -- cycle;\draw[line width=.3]  (o) to[bend right =3] (p) -- (q) to[bend right =5] (r) -- (s);\draw[line width=.3] (r) -- (rr) (p) -- (u)(q) -- (l) (v) -- (vv) (x) -- (xx);\draw[line width=.3](a) -- (b)(j) to[bend left=10] (a) -- (aa)(l) -- (j) -- (jj)(i) -- (l) (h) -- (i) -- (ii)(d) to[bend left=10] (c) -- (cc) (c) to[bend left=5] (h) -- (hh) (w) -- (m) to[bend right=5] (g) to[bend left=6] (f)(k) to[bend left=3] (d) -- (dd)(g) -- (gg)(m) -- (mm);
\fill[white](1.7,-1.05) rectangle (.25,1.05);\filldraw[fill = white, line width=.5] (1.7,.59) ellipse(.1 and .41);\filldraw[fill = white, line width=.5] (1.7,-.58) ellipse(.1 and .42);} 
\medskip
\]
Note that open soccer ball decompositions are not soccer ball decompositions, 
unless $S=\emptyset$.
Indeed, in a soccer ball decomposition, 
the cells along the boundary are all white, 
whereas is an open soccer ball decomposition, the cells along the submanifold $S$ 
alternate between black and white.

\begin{remark}
Let $\Sigma$ be a surface, and let $S\subset \partial \Sigma$ be a disjoint union of circles.
Let $\Delta=D_1\cup\ldots\cup D_k$ be a union of discs,
and let $\Sigma^+:=\Sigma\cup_\varphi \Delta$ for some diffeomorphism $\varphi:S \to \partial\Delta$.
Then a regular trivalent smooth cell decomposition $X$ of $\Sigma$, 
with a partition of the set of $2$-cells into $X_{\mathrm{white}}$ and $X_{\mathrm{black}}$ is an $S$-open soccer ball decomposition 
if and only if $(X_{\mathrm{white}}\cup \{D_1,\ldots,D_k\},X_{\mathrm{black}})$ 
is a soccer ball decomposition of $\Sigma^+$.
\end{remark}

Let $\Sigma$ and $\Sigma'$ be surfaces,
let $S$ be a closed $1$-manifold equipped with orientation preserving embeddings $\bar S\to\partial \Sigma$ and $S\to \partial \Sigma'$, where $\bar S$ denotes $S$ with the opposite orientation,
and let $\Sigma^+=\Sigma\cup_S\Sigma'$ be the result of gluing $\Sigma$ and $\Sigma'$ along $S$.
Given a soccer ball decomposition $X$ of $\Sigma$ and an $S$-open soccer ball decomposition $Y$ of $\Sigma'$,
we say that $X$ and $Y$ are \emph{compatible} if $X\cup Y$ is a soccer ball decomposition of $\Sigma^+$.
If we write the boundary of $\Sigma'$ as $\partial \Sigma'=S\sqcup T$, and we denote the elements of our open soccer ball decomposition by 
$Y_{\mathrm{white}}=\{\ID_1,\ldots,\ID_n\}$ and $Y_{\mathrm{black}}=\{B_1,\ldots,B_m\}$, then we may consider the functor
\[
\begin{split}
\Rep_{\bar S}(\cala)\to&\,\Rep_T(\cala)\\
H\,\,\,\mapsto&\,\,\,\, H\triangleright Y
\end{split}
\]
given by
\[
\begin{split}
H\triangleright Y\,\,:=\Big(H\,\boxtimes_{\cala(S\cap (\ID_1\cup\ldots\cup\ID_n))} \Big(\bigboxtimes_{\{\cala(\ID_i\cap \ID_j)\}} \big\{H_0(\partial\ID_i)\big\}\Big)\!\Big)&\\
\boxtimes_{\cala(\partial B_1\cup\ldots\cup\,\partial B_m)}\big(H_0(\partial B_1)\otimes\ldots\otimes&\, H_0(\partial B_m)\big).
\end{split}
\]
Note that if $X$ is a soccer ball decomposition that is compatible with $Y$, then we have a canonical unitary isomorphism
\begin{equation}\label{eq:V(S;XuY)=Xtri V(S_2;Y)}
V(\Sigma^+;X\cup Y) \,\cong\, V(\Sigma;X) \triangleright Y.
\end{equation}

Similarly, for $\Sigma^+=\Sigma\cup_S\Sigma'$ as above,
if $X$ is an $S$-open soccer ball decomposition of $\Sigma$ and $Y$ a compatible soccer ball decomposition of $\Sigma'$,
then
\[
V(\Sigma^+;X\cup Y) \,\cong\, X\triangleleft V(\Sigma';Y)
\]
for an analogously defined functor $X\triangleleft -$.

\begin{lemma}
Let $\Sigma^+=\Sigma\cup_S\Sigma'$ be as above.
Then for any $S$-open soccer ball decomposition $Y$ of $\Sigma'$,
the maps \eqref{eq:V(S;XuY)=Xtri V(S_2;Y)} induce unitary isomorphisms
\begin{equation}\label{eq:V(S;XuY)=Xtri V(S_2;Y)++}
\lambda_Y:\,V(\Sigma^+) \,\to\, V(\Sigma)\hspace{.3mm}\triangleright Y,
\end{equation}
canonical up to phase.
\end{lemma}

\begin{proof}
Recall the definition complex $\IX_{\Sigma}$ from the proof of Theorem~\ref{Mthm: conformal blocks}, and
let $\tilde\IX_{\Sigma}\subset \IX_{\Sigma}$ be the subcomplex 
whose vertices are the soccer ball decompositions of $\Sigma$ that are compatible with $Y$.
For such a soccer ball decomposition $X\in \tilde\IX_{\Sigma}$, we may consider the composite
\[
\lambda_{X,Y}:\,V(\Sigma^+)\to V(\Sigma^+;X\cup Y) \to V(\Sigma;X)\triangleright Y \to V(\Sigma)\triangleright Y,
\]
where the middle map is the isomorphism \eqref{eq:V(S;XuY)=Xtri V(S_2;Y)}, and the two outer ones are constructed in the proof of Theorem \ref{Mthm: conformal blocks}.

We now show that $\lambda_{X,Y}$ is independent of $X$, up to phase.
Recall the maps $\Phi_\ID$ from \eqref{eq: PHI_D}.
If ${X_1}\in \tilde\IX_{\Sigma}$ is another compatible soccer ball decomposition related by ${X_1}\stackrel \ID\sim X$, then it is clear by construction that the diagram
\[
\tikzmath{\node (a) at (-1.75,0) {$V(\Sigma^+)$};\node (b) at (1,.9) {$V(\Sigma^+;X\cup Y)$};\node (b') at (1,-.9) {$V(\Sigma^+;{X_1}\cup Y)$};\node (c) at (4,.9) {$V(\Sigma;X)\triangleright Y$};
\node (c') at (4,-.9) {$V(\Sigma;{X_1})\triangleright Y$};\node (d) at (7,0) {$V(\Sigma)\triangleright Y$};
\draw[->](a) -- (b);\draw[->](a) -- (b');\draw[->](b) --node[scale=1.1, right]{$\scriptstyle \Phi_\ID$} (b');
\draw[->](b) -- (c);\draw[->](b') -- (c');\draw[->](c) --node[scale=1.1, right]{$\scriptstyle \Phi_\ID\hspace{.3mm}\triangleright Y$} (c');\draw[->](c) -- (d);\draw[->](c') -- (d);}
\]
commutes up to phase, and so $\lambda_{X,Y}=\lambda_{{X_1},Y}$ up to phase.
The result follows because $\tilde\IX_{\Sigma}$ is connected.
\end{proof}

The isomorphisms \eqref{eq:V(S;XuY)=Xtri V(S_2;Y)++} satisfy the following version of associativity.
Given surfaces $\Sigma$, $\Sigma'$, $\Sigma''$ and closed 1-manifolds $S$ and $T$ along with orientation preserving embeddings 
$\bar S\to \partial \Sigma$, $S\sqcup \bar T\to \Sigma'$, and $T\to \partial \Sigma''$,
then for every open soccer ball decompositions $X$ of $\Sigma$ and $Y$ of $\Sigma''$, the diagram
\begin{equation}\label{eq: ass of lambda}
\tikzmath{
\matrix [matrix of math nodes,column sep=1.6cm,row sep=1cm]
{
|(a)| V\big(\Sigma\cup_S\Sigma'\cup_T\Sigma''\big) \pgfmatrixnextcell |(b)| X\triangleleft V(\Sigma'\cup_T\Sigma'') \\
|(c)| V(\Sigma\cup_S\Sigma')\triangleright Y \pgfmatrixnextcell |(d)| X\triangleleft V(\Sigma')\triangleright Y \\}; 
\draw[->] (a) -- node [above] {$\scriptstyle \lambda_X$} (b);
\draw[->] (a) -- node [left] {$\scriptstyle \lambda_Y$} (c);
\draw[->] (a) -- node [above, xshift=5, pos=.53] {$\scriptstyle \lambda_{X\sqcup Y}$} (d);
\draw[->] (b) -- node [right] {$\scriptstyle X\triangleleft \lambda_Y$} (d);
\draw[->] (c) -- node [above] {$\scriptstyle \lambda_X\triangleright Y$} (d);
}
\end{equation}
commutes up to phase.

\begin{proof}[Proof of Theorem \ref{maintheorem gluing}]
We first deal with two special cases:
the first when $M$ consists solely of intervals, 
and the second when $M$ consists solely of circles.

If $M$ is a union of intervals, consider the following subset $\Gamma\subset \IX_\Sigma$ of the $1$-skeleton of the definition complex of $\Sigma:=\Sigma_1\cup_M\Sigma_2$.
The vertices of $\Gamma$ are the soccer ball decompositions $X$ with the property that each connected component of $M\subset \Sigma$ is a single edge of $X$,
and an edge $X\overset\ID\sim Y$ is in $\Gamma$ if either $\ID\subset \Sigma_1$ or $\ID\subset \Sigma_2$.
Note that $\Gamma$ is a connected graph.
Note also that if a soccer ball decomposition $X$ belongs to $\Gamma$, 
then all the $2$-cells of $X$, that are adjacent to $M$, are white.
Given $X\in \Gamma$, let us write $X|_{\Sigma_1}$ and $X|_{\Sigma_2}$ for the restrictions of the soccer ball decomposition $X$ to $\Sigma_1$ and to $\Sigma_2$.
For every $X\in \Gamma$, there is then an obvious isomorphism
\[
g_X:V(\Sigma;X)\to V(\Sigma_1;X|_{\Sigma_1})\boxtimes_{\cala(M)} V(\Sigma_2;X|_{\Sigma_2}),
\]
well defined up to phase.
Moreover, for every edge $X\overset\ID\sim Y$ of $\Gamma$, the diagram
\begin{equation}\label{eq: V(S)=V(S_1)[x]V(S_2)+}
\tikzmath{ \matrix [matrix of math nodes,column sep=1.6cm,row sep=.7cm]
{ |(a)| V(\Sigma;X) \pgfmatrixnextcell |(c)| V(\Sigma_1;X|_{\Sigma_1})\boxtimes_{\cala(M)} V(\Sigma_2;X|_{\Sigma_2})\\
|(b)|  V(\Sigma;Y) \pgfmatrixnextcell |(d)| V(\Sigma_1;Y|_{\Sigma_1})\boxtimes_{\cala(M)} V(\Sigma_2;Y|_{\Sigma_2})\\ }; 
\draw[->] (a) --node[left]{$\scriptstyle \Phi_\ID$} (b); \draw[->] 
($(c.south)+(-.3,0)$) --node[right, xshift=4] {$\scriptstyle \Phi_\ID\boxtimes 1\,$ {\footnotesize or} $\,\scriptstyle 1\boxtimes\Phi_\ID$} 
($(d.north)+(-.3,0)$);
\draw[->] (a) --node[above]{$\scriptstyle g_X$} (c); \draw[->] (b) --node[above]{$\scriptstyle g_Y$} (d); }
\end{equation}
commutes up to phase, where $\Phi_\ID$ are as in \eqref{eq: PHI_D}.
Since $\Gamma$ is connected, it follows from \eqref{eq: V(S)=V(S_1)[x]V(S_2)+} that the maps $g_X$ for $X\in \Gamma$ descend to a unitary isomorphism
\[
g:V(\Sigma)\to V(\Sigma_1)\boxtimes_{\cala(M)} V(\Sigma_2),
\] 
well defined up to phase.

If $M$ is a union of circles, let us define the auxiliary surfaces
\[
\Sigma_1^+:=\Sigma_1 \cup_M (M\times [0,1]),\,\,\,
\Sigma_2^+:=(M\times [0,1]) \cup_M \Sigma_2,\,\,\,
\Sigma^+:=\Sigma_1 \cup_M (M\times [0,1]) \cup_M \Sigma_2.
\]
Recall the definition of $\mathsf{2MAN}$ from \eqref{eq: def of 2MAN},
and note that there are canonical isomorphisms
$\Sigma_1^+\cong \Sigma_1$, $\Sigma_2^+\cong \Sigma_2$, $\Sigma^+\cong \Sigma$ in that groupoid
(even though there are no canonical homeomorphisms).
By Theorem \ref{Mthm: conformal blocks}, we therefore have unitary isomorphisms
$V(\Sigma_1^+)\cong V(\Sigma_1)$, $V(\Sigma_2^+)\cong V(\Sigma_2)$, $V(\Sigma^+)\cong V(\Sigma)$,
well defined up to phase.
Therefore, instead of \eqref{eq: main thm eq2}, we may as well construct an isomorphism
\begin{equation}\label{eq: the map g}
g\,:\,\,V(\Sigma^+)\to V(\Sigma_1^+)\boxtimes_{\cala(M)} V(\Sigma_2^+).
\end{equation}

Given open soccer ball decompositions $X$ of $\Sigma_1$ and $Y$ of $\Sigma_2$, 
we let \eqref{eq: the map g} be the composite
\[
\begin{split}V(\Sigma^+)\,&\to\,X\triangleleft V(M\times [0,1])\triangleright Y\\&\to\,X\triangleleft L^2\cala(M)\triangleright Y\\&\to\,X\triangleleft L^2\cala(M)\boxtimes_{\cala(M)}L^2\cala(M)
\triangleright Y\\&\to\,X\triangleleft V(M\times [0,1])\boxtimes_{\cala(M)}V(M\times [0,1])\triangleright Y\\&\to\,V(\Sigma_1^+)\boxtimes_{\cala(M)} V(\Sigma_2^+)\end{split}
\]
where the first map is $\lambda_{X\sqcup Y}$ from \eqref{eq: ass of lambda}, the second and fourth maps are provided by \eqref{eq:HSigmaL^2}, established in Theorem~\ref{thm: H_Sigma == L^2 cala(S)},
and the last map is the inverse of $\lambda_X\boxtimes \lambda_Y$.
We need to show that the above map does not depend on $X$ and on $Y$.
The following diagram is easily seen to commutate up to phase:
\begin{equation}\label{eq: rightmost composition}
\tikzmath[scale=.95]{
\node[scale=.95] (a) at (0,5) {$V(\Sigma^+)$};
\node[scale=.95] (b) at (0,4) {$X\triangleleft V(M\times [0,1])\triangleright Y$};
\node[scale=.95] (c) at (-.45,3) {$X\triangleleft L^2\cala(M)\triangleright Y$};
\node[scale=.95] (d) at (-1.6,2) {$X\triangleleft L^2\cala(M)\boxtimes_{\cala(M)}L^2\cala(M)\triangleright Y$};
\node[scale=.95] (e) at (-1.5,1) {$X\triangleleft V(M\times [0,1])\boxtimes_{\cala(M)}V(M\times [0,1])\triangleright Y$};
\node[scale=.95] (f) at (-.5,0) {$V(\Sigma_1^+)\boxtimes_{\cala(M)} V(\Sigma_2^+)$};
\node[scale=.95] (A) at (4.3,4.1) {$X\triangleleft V(\Sigma_2^+)$};
\node[scale=.95] (B) at (4.3,2.7) {$X\triangleleft L^2\cala(M)\boxtimes_{\cala(M)}V(M\times [0,1])\triangleright Y$};
\node[scale=.95] (C') at (5.9,1.45) {$X\triangleleft L^2\cala(M)\boxtimes_{\cala(M)} V(\Sigma_2^+)$};
\node[scale=.95] (D) at (4.5,.3) {$X\triangleleft V(M\times [0,1])\boxtimes_{\cala(M)} V(\Sigma_2^+)$};
\draw[->] (a) -- (b.north-|a);
\draw[->] (b.south-|a) -- (c.north-|a);
\draw[->] (c.south-|a) -- (d.north-|a);
\draw[->] (d.south-|a) -- (e.north-|a);
\draw[->] (e.south-|a) -- (f.north-|a);
\draw[->] (A) -- (b);\draw[->] (a) -- (A);
\draw[->] (B.south)+(-1.7,0) -- ($(d.east)+(0,.1)$);
\draw[->] (B.south)+(-.5,0) -- ($(e.north east)+(-.3,.1)$);
\draw[->] (B.south)+(.5,0) -- ($(C'.north)+(0,0)$);
\draw[->] (C'.south)+(.3,0) -- ($(D.north)+(.02,0)$);
\draw[->] (e.south east)+(0,.25) -- ($(D.north)+(-1.5,0)$);
\draw[->] (b.south east) -- (B);
\draw[->] (D.west)+(0,-.17) -- (f.east);
\draw[->] (A.east)+(0,-.2) to[out=-15,in=38,looseness=1.2] ($(C'.north)+(1.8,0)$);
}
\end{equation}
The vertical map from $V(\Sigma^+)$ to $V(\Sigma_1^+)\boxtimes_{\cala(M)} V(\Sigma_2^+)$ is our definition of the map $g$ \eqref{eq: the map g}, and
it is equal to the rightmost composition, up to phase.
The map $g$ therefore does not depend on $Y$, up to phase.
Since $X$ and $Y$ enter symmetrically in the definition of $g$, the map is also independent of $X$.

We now treat the general case.
Write $M=M_1\sqcup M_2$ as a disjoint union, where $M_1$ consists only of intervals, and $M_2$ consists only of circles.
Since $\cala(M_2)$ is a direct sum of type $I$ factors, 
$L^2\cala(M_2)$ is an $\cala(M_2)\,\bar\otimes\,\cala(M_2)^\op$-module and we have canonical isomorphisms
\[
\begin{split}
\,\,V(\Sigma_1)\boxtimes_{\cala(M)} V(\Sigma_2)
\,&\cong\,
V(\Sigma_1)\boxtimes_{\cala(M_1)\,\bar\otimes\,\cala(M_2)} V(\Sigma_2)\phantom{._{\big)}}\\
&\cong
\Big(V(\Sigma_1)\boxtimes_{\cala(M_1)} V(\Sigma_2)\Big)
\boxtimes_{\cala(M_2)\,\bar\otimes\,\cala(M_2)^\op} L^2\cala(M_2)\\
&\cong
\Big(V(\Sigma_1)\boxtimes_{\cala(M_1)} V(\Sigma_2)\Big)
\boxtimes_{\cala(M_2)\,\bar\otimes\,\cala(M_2)^\op} V(M_2\times [0,1]),
\intertext{where the last equality follows from \eqref{eq:HSigmaL^2}, that is Theorem \ref{thm: H_Sigma == L^2 cala(S)}.\goodbreak
By the two special cases that we treated earlier, we have further isomorphisms}
&\cong
V\big(\Sigma_1\cup_{M_1} \Sigma_2\big)
\boxtimes_{\cala(M_2)\,\bar\otimes\,\cala(M_2)^\op} V(M_2\times [0,1])\\
&\cong
V\big((\Sigma_1\cup_{M_1} \Sigma_2)\cup_{M_2\sqcup \bar M_2} (M_2\times [0,1])\big).
\end{split}
\]
There is a canonical isomorphism in $\mathsf{2MAN}$ between
$
\Sigma^+:=(\Sigma_1\cup_{M_1} \Sigma_2)\cup_{M_2\sqcup \bar M_2} (M_2\times [0,1])
$
and $\Sigma:=\Sigma_1\cup_M \Sigma_2$, which by Theorem \ref{Mthm: conformal blocks} translates to a unitary isomorphism $V(\Sigma^+)\cong V(\Sigma)$, canonical up to phase.
Composing all the above maps, we obtain our desired unitary isomorphism $g:V(\Sigma)\to V(\Sigma_1)\boxtimes_{\cala(M)} V(\Sigma_2)$, canonically up to phase.
\end{proof}

The gluing isomorphisms \eqref{eq: main thm eq2} satisfy the following version of associativity.
Let $\Sigma_1$, $\Sigma_2$, $\Sigma_3$ be surfaces, $M_1\subset \partial \Sigma_1$, $M\subset \partial \Sigma_2$, $N_2\subset \partial \Sigma_2$, $N\subset \partial \Sigma_3$ submanifolds with $M$ and $N_2$ disjoint, and let $M_1\to M$ and $N_2\to N$ be orientation reversing diffeomorphisms.
Pick a smooth structure on the trivalent graph $\partial \Sigma_1\cup_M\partial \Sigma_2\cup_N\partial \Sigma_3$, compatibly with those of $\partial \Sigma_i$.
Then the following diagram commutes up to phase:
\begin{equation}\label{eq: 4 g's}
\hspace{-1cm}\tikzmath{
\matrix [matrix of math nodes,column sep=1.2cm,row sep=.8cm]
{ 
|(a)| V(\Sigma_1\cup_M\Sigma_2\cup_N\Sigma_3) \pgfmatrixnextcell |(b)|  V(\Sigma_1)\boxtimes_{\cala(M)} V(\Sigma_2\cup_N\Sigma_3)\\
|(c)| V(\Sigma_1\cup_M\Sigma_2)\boxtimes_{\cala(N)}V(\Sigma_3) \pgfmatrixnextcell |(d)| V(\Sigma_1)\boxtimes_{\cala(M)} V(\Sigma_2)\boxtimes_{\cala(N)}V(\Sigma_3).\\ 
}; 
\draw[->] (a) -- node [above]	{$\scriptstyle g$} (b);
\draw[->] (c) -- node [above]	{$\scriptstyle g\boxtimes 1$} (d);
\draw[->] (a) -- node [left]		{$\scriptstyle g$} (c);
\draw[->] (b) -- node [right]	{$\scriptstyle 1\boxtimes g$} (d);
}\hspace{-1cm}
\end{equation}
Indeed, our definition of the gluing isomorphism is local, and two gluings that happen far away do not interfere with each other.

\begin{remark}
As a consequence of Theorem \ref{maintheorem gluing}, 
the definition of $V(\Sigma)$ presented in Section~\ref{sec: The Hilbert space associated to a surface} agrees with the more canonical one
provided by Theorem \ref{Mthm: conformal blocks} (which is based on \eqref{eq: graph + fill holes} and \eqref{eq: PHI_D}).
\end{remark}

\subsection{Conformal blocks}\label{sec: Conformal blocks}
Recall that a ``pair-of-pants'' is just another name for a sphere with three holes, and that a pair-of-pants
decomposition of a surface $\Sigma$ consists in a collection of circles in its interior that decompose it into pairs-of-pants $P_1,\ldots,P_k$
(we allow the situation where two legs of the same pair-of-pants get glued to each other).

With Theorem \ref{maintheorem gluing} in hand, it is now easy to compute the value of $V(\Sigma)$ for any surface $\Sigma$.
Recall that we always assume our conformal net $\cala$ has finite index, and the set $\Delta$ of isomorphism classes of irreducible $\cala$-sectors is finite.
Recall also that for $\lambda,\mu,\nu\in \Delta$, the multiplicity of $H_0$ inside $H_\lambda\boxtimes H_\mu\boxtimes H_\nu$ is denoted $N_0^{\lambda\mu\nu}$.

\begin{proposition}\label{prop: Conf Blocks 1}
Let $\Sigma$ be a surface, with boundary components $S_1,\ldots,S_n$.
Let $S_{n+1},\ldots,S_m$ be oriented circles in the interior of $\Sigma$ that form a pairs-of-pants decomposition $P_1,\ldots,P_k$.
Then the multiplicity of $H_{\lambda_1}(S_1)\otimes \ldots \otimes H_{\lambda_n}(S_n)$ in $V(\Sigma)$ is finite, and is given by
\begin{equation}\label{eq: state sum + labelings}
\raisebox{.15cm}{$\displaystyle\sum_{\substack{\text{labelings $\lambda_{n+1},\ldots,\lambda_m$}\\\text{of $S_{n+1},\ldots,S_m$}\\\text{by elements of $\Delta$}}}
\,\,\prod_{\substack{\text{pairs-of-pants $P_i{}_{\phantom{(}}$}\\\text{in the decomposition}\\\text{of $\Sigma$}}}
\,\,N_0^{\mu_1^i\mu_2^i\mu_3^i},
$}\end{equation}
where $\mu_1^i,\mu_2^i,\mu_3^i\in\Delta$ are computed in the following way:
letting $S_{j_1}$, $S_{j_2}$, $S_{j_3}$ be the three boundary component of $P_i$, then we set
$\mu_a^i:=\lambda_{j_a}$ if the orientation $S_{j_a}$ is the one induced by $P_i$,
and $\mu_a^i:=\bar \lambda_{j_a}$ otherwise.
\end{proposition}

\begin{proof}
For a manifold $M$, a Hilbert space $H$, and an isomorphism class of sector $\lambda\in \Delta$, let us write
\[
\begin{split}
\quad M^\epsilon&:=\begin{cases} M\text{ for }\epsilon=+1\\ \overline M\text{ for }\epsilon=-1\end{cases}\quad\text{where the bar denotes orientation reversal,}\\
H^\epsilon&:=\begin{cases} H\text{ for }\epsilon=+1\\ \overline H\text{ for }\epsilon=-1\end{cases}\,\quad\text{where the bar denotes complex conjugate,}\\
\lambda^\epsilon&:=\begin{cases} \lambda\text{ for }\epsilon=+1\\ \bar \lambda\text{ for }\epsilon=-1\end{cases}\hspace{.172cm}\quad\text{where the bar denotes dual sector.}
\end{split}
\]
Let $\Sigma^+$ be the surface that is obtained by gluing $A:=P_1\sqcup\ldots\sqcup P_k$ to $B:=(S_{n+1}\sqcup\ldots\sqcup S_m)\times [0,1]$ along $\partial B$ in the obvious manner.
There is a canonical isomorphism in $\mathsf{2MAN}$ between $\Sigma$ and $\Sigma^+$, and so by Theorem \ref{Mthm: conformal blocks}
the spaces $V(\Sigma)$ and $V(\Sigma^+)$ are isomorphic as $\partial\Sigma$-sectors.

For each pair-of-pants $P_i$ in the decomposition, let $j_{i1}, j_{i2}, j_{i3}\in\{1,\ldots,m\}$ and $\epsilon_{i1}, \epsilon_{i2}, \epsilon_{i3}\in\{\pm1\}$
be such that $\partial P_i=S_{j_{i1}}^{\epsilon_{i1}}\sqcup S_{j_{i2}}^{\epsilon_{i2}}\sqcup S_{j_{i3}}^{\epsilon_{i3}}$ as oriented manifolds.
Recall that by Lemma \ref{lem: computation of V((000))} the Hilbert space associated to a pair-of-pants with boundary components $S_1$, $S_2$, $S_3$ is given by
\[
V\bigg(\tikzmath[scale=.5]{\filldraw[fill=gray!30] (0,0) circle (1); \filldraw[fill=white] (-.4,0) circle (.25)(.4,0) circle (.25);}\bigg) = \bigoplus_{\lambda,\mu,\nu\in\Delta} N_0^{\lambda\mu\nu} H_\lambda(S_1)\otimes H_\mu(S_2)\otimes H_\nu(S_3).
\]
We therefore have
\[
\begin{split}
V(A)\,&\cong\,\bigotimes_{i=1}^k V(P_i)\cong\, \bigotimes_{i=1}^k \bigg(\bigoplus_{\lambda,\mu,\nu\in\Delta} N_0^{\lambda\mu\nu} H_\lambda(S_{j_{i1}}^{\epsilon_{i1}})\otimes H_\mu(S_{j_{i2}}^{\epsilon_{i2}})\otimes H_\nu(S_{j_{i3}}^{\epsilon_{i3}})\bigg)\\
&\cong\,
\bigotimes_{i=1}^k \bigg(\bigoplus_{\lambda,\mu,\nu\in\Delta}\! N_0^{\lambda\mu\nu} \big(H_{\lambda^{\epsilon_{i1}}}(S_{j_{i1}})\big)^{\epsilon_{i1}}\!\otimes \big(H_{\mu^{\epsilon_{i2}}}(S_{j_{i2}})\big)^{\epsilon_{i2}}\!\otimes \big(H_{\nu^{\epsilon_{i3}}}(S_{j_{i3}})\big)^{\epsilon_{i3}}\!\bigg)\\
&\cong\,
\bigotimes_{i=1}^k \bigg(\bigoplus_{\lambda,\mu,\nu\in\Delta}\! N_0^{\lambda^{\epsilon_{i1}}\mu^{\epsilon_{i2}}\nu^{\epsilon_{i3}}} \big(H_{\lambda}(S_{j_{i1}})\big)^{\epsilon_{i1}}\!\otimes \big(H_{\mu}(S_{j_{i2}})\big)^{\epsilon_{i2}}\!\otimes \big(H_{\nu}(S_{j_{i3}})\big)^{\epsilon_{i3}}\!\bigg),
\end{split}
\]
where the middle isomorphism follows from \cite[\lemdualofHlambda]{BDH(nets)}.
The multiplicity of
\[
\begin{split}
\Big(H_{\lambda_1}(S_1)\otimes\ldots\otimes H_{\lambda_n}(S_n)\Big)\otimes \Big(H_{\lambda_{n+1}}(&S_{n+1})\otimes\ldots\otimes H_{\lambda_m}(S_m)\Big)\\&\otimes \Big(\overline{H_{\lambda'_{n+1}}(S_{n+1})}\otimes\ldots\otimes \overline{H_{\lambda'_m}(S_m)}\Big)
\end{split}
\]
inside $V(A)$ is thus given by
\[
m_{\lambda_1,\ldots,\lambda_n,\lambda_{n+1},\ldots,\lambda_m,\lambda'_{n+1},\ldots,\lambda'_m}:=\prod_{i=1}^k\,N_0^{\nu_1^i\nu_2^i\nu_3^i}
\]
where $\nu_a^i$ is defined to be $\lambda_{j_{ia}}$ if the orientation of $S_{j_{ia}}$ is the one induced by $P_i$,
and $\bar \lambda'_{j_{ia}}$ otherwise.

We have $V(B)\cong L^2\cala (S_{n+1})\otimes\ldots\otimes L^2\cala (S_m)$ by \eqref{eq:HSigmaL^2}, that is Theorem \ref{thm: H_Sigma == L^2 cala(S)}, and so
the operation $-\boxtimes_{\cala(\partial B)} V(B)$ amounts to contracting the indices $\lambda_i$ with $\lambda'_i$ for $i\in\{n+1,\ldots,m\}$.
The multiplicity of $H_{\lambda_1}(S_1)\otimes \ldots \otimes H_{\lambda_n}(S_n)$ inside $V(\Sigma)\cong V(\Sigma^+)\cong V(A)\boxtimes_{\cala(\partial B)} V(B)$ is therefore given by
\[
\raisebox{.1cm}{$\displaystyle\sum_{\,\,\,\,(\lambda_{n+1},\ldots,\lambda_m)\in\Delta^{m-n}}\,
m_{\lambda_1,\ldots,\lambda_n,\lambda_{n+1},\ldots,\lambda_m,\lambda_{n+1},\ldots,\lambda_m}$}\,,
\]
which is precisely \eqref{eq: state sum + labelings}.
\end{proof}

Note that for $\partial\Sigma=\emptyset$, the above result implies that $V(\Sigma)$ is a finite dimensional:

\begin{corollary}
Let $\Sigma$ be a closed surface and
let $S_1,\ldots,S_m$ be oriented circles in $\Sigma$ that form a pair-of-pants decomposition $P_1,\ldots,P_k$.
Then $V(\Sigma)$ is finite dimensional, and its dimension is given by
\begin{equation}\label{eq: state sum + labelings-dim}
\raisebox{.15cm}{$\displaystyle\sum_{\substack{\text{labelings $\lambda_{1},\ldots,\lambda_m$}\\\text{of $S_1,\ldots,S_m$}\\\text{by elements of $\Delta$}}}
\,\,\prod_{\substack{\text{pairs-of-pants $P_i{}_{\phantom{(}}$}\\\text{in the decomposition}\\\text{of $\Sigma$}}}
\,\,N_0^{\mu_1^i\mu_2^i\mu_3^i}
$}\end{equation}
where $\mu_1^i,\mu_2^i,\mu_3^i$ are as defined in Proposition \ref{prop: Conf Blocks 1}.\hfill $\square$
\end{corollary}

Given the formula \eqref{eq: state sum + labelings-dim}, it now seems reasonable 
that for a closed surface $\Sigma$, our vector space $V(\Sigma)$ agrees with
the notion of ``space of conformal blocks'' in conformal field 
theory~\cite{Friedan-Shenker(The-analytic-geometry-of-two-dimensional-conformal-field-theory)}, \cite{Moore-Seiberg(classical+quantum-conf-field-theory)}.
Moreover, for a surface $\Sigma$ with boundary components $S_1,\ldots,S_n$, the vector space
\begin{equation}\label{eq: conf bl with points}
V(\Sigma;\lambda_1,\ldots,\lambda_n):=\hom \big(H_{\lambda_1}(S_1)\otimes \ldots \otimes H_{\lambda_n}(S_n),V(\Sigma)\big)
\end{equation}
corresponds to the notion of ``space of conformal blocks with field insertions''.
Note that those statements, relating our space of conformal blocks to the CFT notion of space of conformal blocks, are really only conjectures because, for example,
there is no theorem that identifies the representation theory of a conformal net
with that of the corresponding chiral conformal field theory.

For every choice of labels $\lambda_1,\ldots,\lambda_n\in\Delta$, there is a projective action of $\Gamma(\Sigma):=\ker(G(\Sigma)\to\Diff(\partial\Sigma))$ on $V(\Sigma;\lambda_1,\ldots,\lambda_n)$.
Note however that the projective action of $G(\Sigma)$ on $V(\Sigma)$ contains strictly more information than the collection of all the projective actions of $\Gamma(\Sigma)$ on the spaces $V(\Sigma;\lambda_1,\ldots,\lambda_n)$.

\begin{remark}\label{rmk: modular functor}
If one only considers surfaces whose boundary components
are \hyphenation{para-me-trized}parametrized by $S^1$, and if one restricts attention to those gluings that are along full boundary components and that are compatible with the parametrizations,
then Theorems \ref{Mthm: conformal blocks} and \ref{maintheorem gluing}
provide a projective version of a modular functor (as in \cite[Def.~5.7.10]{Bakalov-Kirillov(Lect-tens-cat+mod-func)}, but without any control over the central cocycle).
\end{remark}

\section{Modularity}
\label{sec:modularity}

In this section, we will use the constructions from the previous sections to analyze the structure of the category of sectors of a conformal net with finite index.

\subsection{Fusion of sectors}
Let us define the \emph{standard pair-of-pants} to be the surface
\[
P:=(S^1\times[0,1])\sqcup(S^1\times[0,1])/\sim\,,
\]
where the equivalence relation identifies the point $(z,1)$ in the first copy of $S^1\times[0,1]$
with the point $(-\bar z,1)$ in the second copy of $S^1\times[0,1]$, for every $z\in S^1_\vdash =\{z\in S^1| \Re\mathrm{e}(z)\le 0\}$.
As mentioned at \eqref{eq: def H-sigma} and \eqref{eq:HSigmaL^2}, and proven in Theorems \ref{thm: any cover of the circle} and \ref{thm: H_Sigma == L^2 cala(S)},
the Hilbert space $V(S^1\times[0,1]) \cong L^2\cala(S^1)$ associated to $S^1\times[0,1]$ is well defined up to canonical unitary isomorphism, as opposed to being merely 
well defined up to canonical-up-to-phase unitary isomorphism.
We may therefore declare
\begin{equation}\label{eq: def: H_P}
H_P\,:=\,L^2\cala(S^1)\boxtimes_{\cala(S^1_{\vdash})}L^2\cala(S^1)
\end{equation}
to be the Hilbert space associated to $P$.
This is a lift of $V(P)$ from a Hilbert space well defined up to canonical-up-to-phase isomorphism to an honest Hilbert space.

Let us call $S_1$, $S_2$, $S_3$ the three boundary components of $P$,
and orient them as indicated in the following picture:
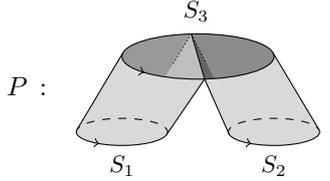
\[
P\,:\,\,\,\,\tikzmath{
\filldraw[fill=gray!30] (-85:1 and .3) coordinate (a) -- ($(-1,-1)+(-10:.6 and .18)$) arc (-10:-190:.6 and .18) -- (170:1 and .3);
\filldraw[fill=gray!30] (a) -- ($(1,-1)+(-170:.6 and .18)$) arc (-170:10:.6 and .18) -- (10:1 and .3);
\draw[dashed] ($(-1,-1)+(-10:.6 and .18)$) arc (-10:170:.6 and .18)($(1,-1)+(-170:.6 and .18)$) arc (190:10:.6 and .18);
\fill[black!25] (95:1 and .3) arc (95:360-85:1 and .3) -- cycle;
\fill[black!40] (95:1 and .3) arc (95:244.3:1 and .3) -- cycle;
\fill[black!60] (95:1 and .3) arc (95:-85:1 and .3) -- cycle;
\fill[black!45] (95:1 and .3) arc (95:-77.7:1 and .3) -- cycle;
\draw (0,0) circle (1 and .3);
\draw (a) -- (95:1 and .3);
\draw[densely dotted] (95:1 and .3) -- +(-.244,-.4);
\draw[densely dotted] (95:1 and .3) -- +(.202,-.4);
\draw[<-] (-1,-1) +(240:.6 and .18) -- +(239.5:.6 and .18);
\draw[<-] (1,-1) +(240:.6 and .18) -- +(239.5:.6 and .18);
\draw[<-] (225:1 and .3) -- (224.5:1 and .3);
\node[scale=.9] at (-1,-1.45) {$S_1$};
\node[scale=.9] at (1,-1.45) {$S_2$};
\node[scale=.9] at (-.025,.6) {$S_3$};
}\,,
\]
so that $\partial P = S_1\sqcup S_2\sqcup \bar S_3$.
Below, we will use the picture 
$\tikzmath[scale=.25]{\fill[gray!60] (.9,1.2) circle (.6 and .18);\filldraw[fill=gray!30](-.6,0) arc (-180:0:.6 and .18) arc (180:0:.3 and .18) arc (-180:0:.6 and .18)[rounded corners=1.7]-- ++(0,.3) -- ++(-.9,.6) [sharp corners]-- ++(0,.3) arc (0:-180:.6 and .18)[rounded corners=1.7]-- ++(0,-.3) -- ++(-.9,-.6) [sharp corners]-- cycle;
\draw[dashed] (0,0) +(0:.6 and .18) arc (0:180:.6 and .18);\draw[dashed] (1.8,0) +(0:.6 and .18) arc (0:180:.6 and .18);\draw (.9,1.2) +(0:.6 and .18) arc (0:180:.6 and .18);}$
for that same manifold $P$.

Recall from \cite[\secconformalembeddings]{BDH(nets)} 
that there is a monoidal structure 
\[
H,K\,\,\mapsto\,\,\, H\,\boxtimes^\mathsf{h}K\,:=\,H\,\boxtimes_{\cala(S^1_\vdash)}K
\]
on the category $\Rep(\cala)$, called horizontal fusion.
The $S_1$-$S_2$-$\bar S_3$-sector $H_P$ is constructed in such a way that there are canonical unitary isomorphisms
\[
\begin{split}
H\,\boxtimes^\mathsf{h}K\,\,\,&=\,\, H\,\boxtimes_{\cala(S^1_\vdash)}K\\
&\cong\,\, \big(H\boxtimes_{\cala(S_1)}L^2\cala(S_1)\big)\,\boxtimes_{\cala(S^1_\vdash)} \big(K\boxtimes_{\cala(S_2)}L^2\cala(S_2)\big)\\
&\cong\,\, (H\otimes K)\boxtimes_{\cala(S_1\cup S_2)}\big(L^2\cala(S_1)\boxtimes_{\cala(S^1_{\vdash})}L^2\cala(S_2)\big)\\
&=\,\, (H\otimes K)\boxtimes_{\cala(S_1\cup S_2)}H_P.
\end{split}
\]
In other words, $H_P$ represents the operation of horizontal fusion.

We will be interested in various ways of fusing copies of $H_P$ to each other.
For $a,b\in\{1,2,3\}$, we shall write $H_P\,\mbox{${}_a\boxtimes_b$}\, H_P$ for the fusion $H_P\boxtimes_{\cala(S^1)}H_P$,
where the algebra $\cala(S^1)$ acts on first copy of $H_P$ via its isomorphism to $\cala(S_a)$ (or its opposite),
and on the second $H_P$ via its isomorphism to $\cala(S_b)$ (or its opposite).
For example, the Hilbert spaces $H_P\,\mbox{${}_3\boxtimes_1$}\, H_P$ and $H_P\,\mbox{${}_2\boxtimes_3$}\, H_P$ correspond to the following two surfaces:
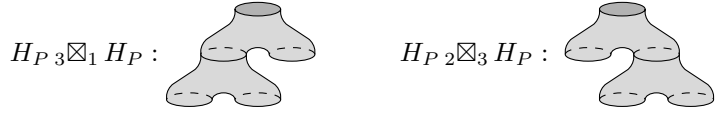
\[
H_P\,\mbox{${}_3\boxtimes_1$}\, H_P
:
\tikzmath[scale=.5]{
\fill[gray!60] (.9,1.2) circle (.6 and .18);
\filldraw[fill=gray!30]
(-1.5,-1.2) arc (-180:0:.6 and .18) arc (180:0:.3 and .18) arc (-180:0:.6 and .18) 
[rounded corners=3]-- ++(0,.3) -- ++(-.9,.6) [sharp corners]-- ++(0,.3) 
arc (180:0:.3 and .18) arc (-180:0:.6 and .18)
[rounded corners=3]-- ++(0,.3) -- ++(-.9,.6) [sharp corners]-- ++(0,.3) 
arc (0:-180:.6 and .18)
[rounded corners=3]-- ++(0,-.3) -- ++(-.9,-.6) [sharp corners]-- ++(0,-.3)
[rounded corners=3]-- ++(0,-.3) -- ++(-.9,-.6) [sharp corners]-- cycle;
\draw[dashed] (-.9,-1.2) +(0:.6 and .18) arc (0:180:.6 and .18);
\draw[dashed] (.9,-1.2) +(0:.6 and .18) arc (0:180:.6 and .18);
\draw[dashed] (0,0) +(0:.6 and .18) arc (0:180:.6 and .18);
\draw (0,0) +(0:.6 and .18) arc (0:-180:.6 and .18);
\draw[dashed] (1.8,0) +(0:.6 and .18) arc (0:180:.6 and .18);
\draw (.9,1.2) +(0:.6 and .18) arc (0:180:.6 and .18);
}\,\qquad\quad
H_P\,\mbox{${}_2\boxtimes_3$}\, H_P
:\,\,
\tikzmath[scale=.5]{
\pgftransformxscale{-1}
\fill[gray!60] (.9,1.2) circle (.6 and .18);
\filldraw[fill=gray!30]
(-1.5,-1.2) arc (-180:0:.6 and .18) arc (180:0:.3 and .18) arc (-180:0:.6 and .18) 
[rounded corners=3]-- ++(0,.3) -- ++(-.9,.6) [sharp corners]-- ++(0,.3) 
arc (180:0:.3 and .18) arc (-180:0:.6 and .18)
[rounded corners=3]-- ++(0,.3) -- ++(-.9,.6) [sharp corners]-- ++(0,.3) 
arc (0:-180:.6 and .18)
[rounded corners=3]-- ++(0,-.3) -- ++(-.9,-.6) [sharp corners]-- ++(0,-.3)
[rounded corners=3]-- ++(0,-.3) -- ++(-.9,-.6) [sharp corners]-- cycle;
\draw[dashed] (-.9,-1.2) +(0:.6 and .18) arc (0:180:.6 and .18);
\draw[dashed] (.9,-1.2) +(0:.6 and .18) arc (0:180:.6 and .18);
\draw[dashed] (0,0) +(0:.6 and .18) arc (0:180:.6 and .18);
\draw (0,0) +(0:.6 and .18) arc (0:-180:.6 and .18);
\draw[dashed] (1.8,0) +(0:.6 and .18) arc (0:180:.6 and .18);
\draw (.9,1.2) +(0:.6 and .18) arc (0:180:.6 and .18);
}\,
\medskip\]
Note that the associator for $\boxtimes^\mathsf{h}$ corresponds to a particular unitary isomorphism
\[
\alpha:\,H_P\,\mbox{${}_3\boxtimes_1$}\, H_P\,\to\, H_P\,\mbox{${}_2\boxtimes_3$}\, H_P,
\]
namely the one given by
\[
\begin{split}H_P\,\mbox{${}_3\boxtimes_1$}\, H_P\,&=\,\Big(L^2\cala(S^1)\boxtimes_{\cala(S^1_{\vdash})}L^2\cala(S^1)\Big)
\,\mbox{${}_3\boxtimes_1$}\, \Big(L^2\cala(S^1)\boxtimes_{\cala(S^1_{\vdash})}L^2\cala(S^1)\Big)\\
&=\,\Big(L^2\cala(S^1)\boxtimes_{\cala(S^1_{\vdash})}L^2\cala(S^1)\Big)\boxtimes_{\cala(S^1)}L^2\cala(S^1)\boxtimes_{\cala(S^1_{\vdash})}L^2\cala(S^1)\\
&\to\, L^2\cala(S^1)\boxtimes_{\cala(S^1_{\vdash})}L^2\cala(S^1)\boxtimes_{\cala(S^1_{\vdash})}L^2\cala(S^1)\\
&\to\, L^2\cala(S^1)\boxtimes_{\cala(S^1_{\vdash})}L^2\cala(S^1)\boxtimes_{\cala(S^1)^\op}\Big(L^2\cala(S^1)\boxtimes_{\cala(S^1_{\vdash})}L^2\cala(S^1)\Big)\\
&=\,\Big(L^2\cala(S^1)\boxtimes_{\cala(S^1_{\vdash})}L^2\cala(S^1)\Big)\,\mbox{${}_2\boxtimes_3$}\, \Big(L^2\cala(S^1)\boxtimes_{\cala(S^1_{\vdash})}L^2\cala(S^1)\Big)\\
&=\,\,H_P\,\mbox{${}_2\boxtimes_3$}\, H_P.
\end{split}
\]

This definition immediately raises a question.
Since $H_P$ is a representative of the ``up-to-phase equivalence class'' $V(P)$,
how does the above map $\alpha$ compare to the geometrically defined associator
$\alpha_{geo}:V(P)\,\mbox{${}_3\boxtimes_1$}\, V(P)\,\to\, V(P)\,\mbox{${}_2\boxtimes_3$}\, V(P)$
\[
\alpha_{geo}:\,\,
\tikzmath{
\node at (0,0)
{$V\big(\tikzmath[scale=.25]{\fill[gray!60] (.9,1.2) circle (.6 and .18);\filldraw[fill=gray!30](-.6,0) arc (-180:0:.6 and .18) arc (180:0:.3 and .18) arc (-180:0:.6 and .18)[rounded corners=1.7]-- ++(0,.3) -- ++(-.9,.6) [sharp corners]-- ++(0,.3) arc (0:-180:.6 and .18)[rounded corners=1.7]-- ++(0,-.3) -- ++(-.9,-.6) [sharp corners]-- cycle;
\draw[dashed] (0,0) +(0:.6 and .18) arc (0:180:.6 and .18);\draw[dashed] (1.8,0) +(0:.6 and .18) arc (0:180:.6 and .18);\draw (.9,1.2) +(0:.6 and .18) arc (0:180:.6 and .18);}
\big)$};
\node at (-.225,-.7)
{$V\big(\tikzmath[scale=.25]{\fill[gray!60] (.9,1.2) circle (.6 and .18);\filldraw[fill=gray!30](-.6,0) arc (-180:0:.6 and .18) arc (180:0:.3 and .18) arc (-180:0:.6 and .18)[rounded corners=1.7]-- ++(0,.3) -- ++(-.9,.6) [sharp corners]-- ++(0,.3) arc (0:-180:.6 and .18)[rounded corners=1.7]-- ++(0,-.3) -- ++(-.9,-.6) [sharp corners]-- cycle;
\draw[dashed] (0,0) +(0:.6 and .18) arc (0:180:.6 and .18);\draw[dashed] (1.8,0) +(0:.6 and .18) arc (0:180:.6 and .18);\draw (.9,1.2) +(0:.6 and .18) arc (0:180:.6 and .18);}
\big)$};
\node at (-.272,-.35) {$\scriptscriptstyle\boxtimes\hspace{.5mm} \cala(\tikzmath[scale=.25]{\draw circle (.6 and .18);})\hspace{.5mm}$};
}
\!\xrightarrow{g^{-1}}\,
V\Big(\,
\tikzmath[scale=.25]{
\fill[gray!60] (.9,1.2) circle (.6 and .18);
\filldraw[fill=gray!30]
(-1.5,-1.2) arc (-180:0:.6 and .18) arc (180:0:.3 and .18) arc (-180:0:.6 and .18) 
[rounded corners=1.5]-- ++(0,.3) -- ++(-.9,.6) [sharp corners]-- ++(0,.3) 
arc (180:0:.3 and .18) arc (-180:0:.6 and .18)
[rounded corners=1.5]-- ++(0,.3) -- ++(-.9,.6) [sharp corners]-- ++(0,.3) 
arc (0:-180:.6 and .18)
[rounded corners=1.5]-- ++(0,-.3) -- ++(-.9,-.6) [sharp corners]-- ++(0,-.3)
[rounded corners=1.5]-- ++(0,-.3) -- ++(-.9,-.6) [sharp corners]-- cycle;
\draw[dashed] (-.9,-1.2) +(0:.6 and .18) arc (0:180:.6 and .18);
\draw[dashed] (.9,-1.2) +(0:.6 and .18) arc (0:180:.6 and .18);
\draw[dashed] (1.8,0) +(0:.6 and .18) arc (0:180:.6 and .18);
\draw (.9,1.2) +(0:.6 and .18) arc (0:180:.6 and .18);
}
\Big)
\,\,\xrightarrow{V(\underline{\alpha})}\,\,
V\Big(
\tikzmath[scale=.25]{
\pgftransformxscale{-1}
\fill[gray!60] (.9,1.2) circle (.6 and .18);
\filldraw[fill=gray!30]
(-1.5,-1.2) arc (-180:0:.6 and .18) arc (180:0:.3 and .18) arc (-180:0:.6 and .18) 
[rounded corners=1.5]-- ++(0,.3) -- ++(-.9,.6) [sharp corners]-- ++(0,.3) 
arc (180:0:.3 and .18) arc (-180:0:.6 and .18)
[rounded corners=1.5]-- ++(0,.3) -- ++(-.9,.6) [sharp corners]-- ++(0,.3) 
arc (0:-180:.6 and .18)
[rounded corners=1.5]-- ++(0,-.3) -- ++(-.9,-.6) [sharp corners]-- ++(0,-.3)
[rounded corners=1.5]-- ++(0,-.3) -- ++(-.9,-.6) [sharp corners]-- cycle;
\draw[dashed] (-.9,-1.2) +(0:.6 and .18) arc (0:180:.6 and .18);
\draw[dashed] (.9,-1.2) +(0:.6 and .18) arc (0:180:.6 and .18);
\draw[dashed] (1.8,0) +(0:.6 and .18) arc (0:180:.6 and .18);
\draw (.9,1.2) +(0:.6 and .18) arc (0:180:.6 and .18);
}
\,\Big)
\,\xrightarrow{\,\,g\,\,}\!
\tikzmath{
\node at (0,0)
{$V\big(\tikzmath[scale=.25]{\fill[gray!60] (.9,1.2) circle (.6 and .18);\filldraw[fill=gray!30](-.6,0) arc (-180:0:.6 and .18) arc (180:0:.3 and .18) arc (-180:0:.6 and .18)[rounded corners=1.7]-- ++(0,.3) -- ++(-.9,.6) [sharp corners]-- ++(0,.3) arc (0:-180:.6 and .18)[rounded corners=1.7]-- ++(0,-.3) -- ++(-.9,-.6) [sharp corners]-- cycle;
\draw[dashed] (0,0) +(0:.6 and .18) arc (0:180:.6 and .18);\draw[dashed] (1.8,0) +(0:.6 and .18) arc (0:180:.6 and .18);\draw (.9,1.2) +(0:.6 and .18) arc (0:180:.6 and .18);}
\big)$};
\node at (.225,-.7)
{$V\big(\tikzmath[scale=.25]{\fill[gray!60] (.9,1.2) circle (.6 and .18);\filldraw[fill=gray!30](-.6,0) arc (-180:0:.6 and .18) arc (180:0:.3 and .18) arc (-180:0:.6 and .18)[rounded corners=1.7]-- ++(0,.3) -- ++(-.9,.6) [sharp corners]-- ++(0,.3) arc (0:-180:.6 and .18)[rounded corners=1.7]-- ++(0,-.3) -- ++(-.9,-.6) [sharp corners]-- cycle;
\draw[dashed] (0,0) +(0:.6 and .18) arc (0:180:.6 and .18);\draw[dashed] (1.8,0) +(0:.6 and .18) arc (0:180:.6 and .18);\draw (.9,1.2) +(0:.6 and .18) arc (0:180:.6 and .18);}
\big)$};
\node at (.17,-.35) {$\scriptscriptstyle\boxtimes\hspace{1mm} \cala(\tikzmath[scale=.25]{\draw circle (.6 and .18);})\hspace{1mm}$};
}
\]
obtained by composing two gluing isomorphisms \eqref{eq: main thm eq2} and the image $V(\underline\alpha)$ of the obvious diffeomorphism
$
\underline\alpha:\tikzmath[scale=.20]{
\fill[gray!60] (.9,1.2) circle (.6 and .18);
\filldraw[fill=gray!30]
(-1.5,-1.2) arc (-180:0:.6 and .18) arc (180:0:.3 and .18) arc (-180:0:.6 and .18) 
[rounded corners=1.2]-- ++(0,.3) -- ++(-.9,.6) [sharp corners]-- ++(0,.3) 
arc (180:0:.3 and .18) arc (-180:0:.6 and .18)
[rounded corners=1.2]-- ++(0,.3) -- ++(-.9,.6) [sharp corners]-- ++(0,.3) 
arc (0:-180:.6 and .18)
[rounded corners=1.2]-- ++(0,-.3) -- ++(-.9,-.6) [sharp corners]-- ++(0,-.3)
[rounded corners=1.2]-- ++(0,-.3) -- ++(-.9,-.6) [sharp corners]-- cycle;
\draw[dashed] (-.9,-1.2) +(0:.6 and .18) arc (0:180:.6 and .18);
\draw[dashed] (.9,-1.2) +(0:.6 and .18) arc (0:180:.6 and .18);
\draw[dashed] (1.8,0) +(0:.6 and .18) arc (0:180:.6 and .18);
\draw (.9,1.2) +(0:.6 and .18) arc (0:180:.6 and .18);
}
\stackrel{\simeq}\longrightarrow
\tikzmath[scale=.20]{
\pgftransformxscale{-1}
\fill[gray!60] (.9,1.2) circle (.6 and .18);
\filldraw[fill=gray!30]
(-1.5,-1.2) arc (-180:0:.6 and .18) arc (180:0:.3 and .18) arc (-180:0:.6 and .18) 
[rounded corners=1.2]-- ++(0,.3) -- ++(-.9,.6) [sharp corners]-- ++(0,.3) 
arc (180:0:.3 and .18) arc (-180:0:.6 and .18)
[rounded corners=1.2]-- ++(0,.3) -- ++(-.9,.6) [sharp corners]-- ++(0,.3) 
arc (0:-180:.6 and .18)
[rounded corners=1.2]-- ++(0,-.3) -- ++(-.9,-.6) [sharp corners]-- ++(0,-.3)
[rounded corners=1.2]-- ++(0,-.3) -- ++(-.9,-.6) [sharp corners]-- cycle;
\draw[dashed] (-.9,-1.2) +(0:.6 and .18) arc (0:180:.6 and .18);
\draw[dashed] (.9,-1.2) +(0:.6 and .18) arc (0:180:.6 and .18);
\draw[dashed] (1.8,0) +(0:.6 and .18) arc (0:180:.6 and .18);
\draw (.9,1.2) +(0:.6 and .18) arc (0:180:.6 and .18);
}$\, under the functor $V$?
Since $\alpha_{geo}$ is only well defined up to phase, the best we can hope for is that the equation $\alpha_{geo}=\alpha$ holds up to phase:

\begin{proposition}\label{prop: al = al}
The two maps $\alpha,\alpha_{geo}:\,H_P\,\mbox{${}_3\boxtimes_1$}\, H_P\,\to\, H_P\,\mbox{${}_2\boxtimes_3$}\, H_P$ are equal up to a phase.
\end{proposition}

Before embarking on the proof, we introduce the following graphical notation (and, implicitly, obvious variations thereof):
\[
\begin{split}
\tikzmath[scale=.4]{\fill[gray!50] (0,1.2) circle (.6 and .18);
\filldraw[fill=gray!20](-.6,0) arc (-180:0:.6 and .18) -- +(0,1.2) arc (0:-180:.6 and .18) -- cycle;
\draw[dashed, ultra thin] (0,0) +(0:.6 and .18) arc (0:180:.6 and .18);
\draw (0,1.2) +(0:.6 and .18) arc (0:180:.6 and .18);
\node[scale=.7] at (.02,.53) {$L^2$};}\,
&\,:=L^2\cala(S^1)
\phantom{|_{|_{|_{|_{|_{|}}}}}}\\  
\,\,\qquad\tikzmath[scale=.4]{\fill[gray!50] (0,1.2) circle (.6 and .18);
\filldraw[fill=gray!20](-.6,0) arc (-180:0:.6 and .18) -- +(0,1.2) arc (0:-180:.6 and .18) -- cycle;
\draw[dashed, ultra thin] (0,0) +(0:.6 and .18) arc (0:180:.6 and .18);
\draw (0,1.2) +(0:.6 and .18) arc (0:180:.6 and .18);
\node at (.02,.5) {$\scriptstyle V$};}\,
&\,:=V\big(\,\tikzmath[scale=.25]{\fill[gray!50] (0,1.2) circle (.6 and .18);
\filldraw[fill=gray!20](-.6,0) arc (-180:0:.6 and .18) -- +(0,1.2) arc (0:-180:.6 and .18) -- cycle;
\draw[dashed, ultra thin] (0,0) +(0:.6 and .18) arc (0:180:.6 and .18);
\draw (0,1.2) +(0:.6 and .18) arc (0:180:.6 and .18);}\,
\big)=V(S^1\times [0,1])
\phantom{|_{|_{|_{|_{|_{|}}}}}}
\end{split}
\]
\[
\bigskip
\begin{split}
\tikzmath[scale=.4]{\fill[gray!50] (.9,1.2) circle (.6 and .18);\filldraw[fill=gray!20](-.6,0) arc (-180:0:.6 and .18) arc (180:0:.3 and .18) arc (-180:0:.6 and .18)[rounded corners=2.4]-- ++(0,.3) -- ++(-.9,.6) [sharp corners]-- ++(0,.3) arc (0:-180:.6 and .18)[rounded corners=2.4]-- ++(0,-.3) -- ++(-.9,-.6) [sharp corners]-- cycle;
\draw[dashed, ultra thin] (0,0) +(0:.6 and .18) arc (0:180:.6 and .18);\draw[dashed, ultra thin] (1.8,0) +(0:.6 and .18) arc (0:180:.6 and .18);\draw (.9,1.2) +(0:.6 and .18) arc (0:180:.6 and .18);
\draw(.9,.18) -- (.9,1.02);
\node[scale=.7] at (.45,.53) {$L^2$};
\node[scale=.7] at (1.38,.53) {$L^2$};}
&\,:=\,\tikzmath[scale=.35]{\fill[gray!50] (0,1.2) circle (.6 and .18);
\filldraw[fill=gray!20](-.6,0) arc (-180:0:.6 and .18) -- +(0,1.2) arc (0:-180:.6 and .18) -- cycle;
\draw[dashed, ultra thin] (0,0) +(0:.6 and .18) arc (0:180:.6 and .18);
\draw (0,1.2) +(0:.6 and .18) arc (0:180:.6 and .18);
\node[scale=.7] at (.02,.54) {$L^2$};}\,\,
\boxtimes_{\cala(S^1_{\vdash})}
\tikzmath[scale=.35]{\fill[gray!50] (0,1.2) circle (.6 and .18);
\filldraw[fill=gray!20](-.6,0) arc (-180:0:.6 and .18) -- +(0,1.2) arc (0:-180:.6 and .18) -- cycle;
\draw[dashed, ultra thin] (0,0) +(0:.6 and .18) arc (0:180:.6 and .18);
\draw (0,1.2) +(0:.6 and .18) arc (0:180:.6 and .18);
\node[scale=.7] at (.02,.54) {$L^2$};}\,
=H_P
\phantom{|_{|_{|_{|_{|_{|_{|_{|}}}}}}}}\\  
\tikzmath[scale=.4]{\fill[gray!50] (.9,1.2) circle (.6 and .18);\filldraw[fill=gray!20](-.6,0) arc (-180:0:.6 and .18) arc (180:0:.3 and .18) arc (-180:0:.6 and .18)[rounded corners=2.4]-- ++(0,.3) -- ++(-.9,.6) [sharp corners]-- ++(0,.3) arc (0:-180:.6 and .18)[rounded corners=2.4]-- ++(0,-.3) -- ++(-.9,-.6) [sharp corners]-- cycle;
\draw[dashed, ultra thin] (0,0) +(0:.6 and .18) arc (0:180:.6 and .18);\draw[dashed, ultra thin] (1.8,0) +(0:.6 and .18) arc (0:180:.6 and .18);\draw (.9,1.2) +(0:.6 and .18) arc (0:180:.6 and .18);
\draw(.9,.18) -- (.9,1.02);
\node at (.5,.53) {$\scriptstyle V$};
\node at (1.3,.53) {$\scriptstyle V$};}
&\,:=\tikzmath[scale=.35]{\fill[gray!50] (0,1.2) circle (.6 and .18);
\filldraw[fill=gray!20](-.6,0) arc (-180:0:.6 and .18) -- +(0,1.2) arc (0:-180:.6 and .18) -- cycle;
\draw[dashed, ultra thin] (0,0) +(0:.6 and .18) arc (0:180:.6 and .18);
\draw (0,1.2) +(0:.6 and .18) arc (0:180:.6 and .18);
\node[scale=.9] at (0,.53) {$\scriptstyle V$};}\,\,
\boxtimes_{\cala(S^1_{\vdash})}
\tikzmath[scale=.35]{\fill[gray!50] (0,1.2) circle (.6 and .18);
\filldraw[fill=gray!20](-.6,0) arc (-180:0:.6 and .18) -- +(0,1.2) arc (0:-180:.6 and .18) -- cycle;
\draw[dashed, ultra thin] (0,0) +(0:.6 and .18) arc (0:180:.6 and .18);
\draw (0,1.2) +(0:.6 and .18) arc (0:180:.6 and .18);
\node[scale=.9] at (0,.53) {$\scriptstyle V$};
}
\phantom{|_{|_{|_{|_{|_{|_{|_{|}}}}}}}}\\  
\tikzmath[scale=.4]{\fill[gray!50] (.9,1.2) circle (.6 and .18);\filldraw[fill=gray!20](-.6,0) arc (-180:0:.6 and .18) arc (180:0:.3 and .18) arc (-180:0:.6 and .18)[rounded corners=2.4]-- ++(0,.3) -- ++(-.9,.6) [sharp corners]-- ++(0,.3) arc (0:-180:.6 and .18)[rounded corners=2.4]-- ++(0,-.3) -- ++(-.9,-.6) [sharp corners]-- cycle;
\draw[dashed, ultra thin] (0,0) +(0:.6 and .18) arc (0:180:.6 and .18);\draw[dashed, ultra thin] (1.8,0) +(0:.6 and .18) arc (0:180:.6 and .18);\draw (.9,1.2) +(0:.6 and .18) arc (0:180:.6 and .18);
\node at (.9,.57) {$\scriptstyle V$};}
&\,:=V\big(\,\tikzmath[scale=.25]{\fill[gray!50] (0,1.2) circle (.6 and .18);
\filldraw[fill=gray!20](-.6,0) arc (-180:0:.6 and .18) -- +(0,1.2) arc (0:-180:.6 and .18) -- cycle;
\draw[dashed, ultra thin] (0,0) +(0:.6 and .18) arc (0:180:.6 and .18);
\draw (0,1.2) +(0:.6 and .18) arc (0:180:.6 and .18);}\,\cup_{S^1_{\vdash}} \,\tikzmath[scale=.25]{\fill[gray!50] (0,1.2) circle (.6 and .18);
\filldraw[fill=gray!20](-.6,0) arc (-180:0:.6 and .18) -- +(0,1.2) arc (0:-180:.6 and .18) -- cycle;
\draw[dashed, ultra thin] (0,0) +(0:.6 and .18) arc (0:180:.6 and .18);
\draw (0,1.2) +(0:.6 and .18) arc (0:180:.6 and .18);}\,\big)
=V(P)
\\  
\tikzmath[scale=.4]{
\fill[gray!50] (.9,1.2) circle (.6 and .18);
\filldraw[fill=gray!20]
(-1.5,-1.2) arc (-180:0:.6 and .18) arc (180:0:.3 and .18) arc (-180:0:.6 and .18) 
[rounded corners=2.4]-- ++(0,.3) -- ++(-.9,.6) [sharp corners]-- ++(0,.3) 
arc (180:0:.3 and .18) arc (-180:0:.6 and .18)
[rounded corners=2.4]-- ++(0,.3) -- ++(-.9,.6) [sharp corners]-- ++(0,.3) 
arc (0:-180:.6 and .18)
[rounded corners=2.4]-- ++(0,-.3) -- ++(-.9,-.6) [sharp corners]-- ++(0,-.3)
[rounded corners=2.4]-- ++(0,-.3) -- ++(-.9,-.6) [sharp corners]-- cycle;
\draw[dashed, ultra thin] (-.9,-1.2) +(0:.6 and .18) arc (0:180:.6 and .18);
\draw[dashed, ultra thin] (.9,-1.2) +(0:.6 and .18) arc (0:180:.6 and .18);
\draw[dashed, ultra thin] (0,0) +(0:.6 and .18) arc (0:180:.6 and .18);
\draw (0,0) +(0:.6 and .18) arc (0:-180:.6 and .18);
\draw[dashed, ultra thin] (1.8,0) +(0:.6 and .18) arc (0:180:.6 and .18);
\draw (.9,1.2) +(0:.6 and .18) arc (0:180:.6 and .18);
\draw(.9,.18) -- (.9,1.02)(0,.18-1.2) -- (0,1.02-1.2);
\node[scale=.7] at (.45,.53) {$L^2$};
\node[scale=.7] at (1.38,.53) {$L^2$};
\node[scale=.7] at (.45-.9,.53-1.2) {$L^2$};
\node[scale=.7] at (1.38-.9,.53-1.2) {$L^2$};}
\,&\,:=
\tikzmath[scale=1.1]{
\node[scale=1.1] at (0,0)
{$\tikzmath[scale=.25]{\fill[gray!50] (.9,1.2) circle (.6 and .18);\filldraw[fill=gray!20](-.6,0) arc (-180:0:.6 and .18) arc (180:0:.3 and .18) arc (-180:0:.6 and .18)[rounded corners=1.7]-- ++(0,.3) -- ++(-.9,.6) [sharp corners]-- ++(0,.3) arc (0:-180:.6 and .18)[rounded corners=1.7]-- ++(0,-.3) -- ++(-.9,-.6) [sharp corners]-- cycle;
\draw[dashed, ultra thin] (0,0) +(0:.6 and .18) arc (0:180:.6 and .18);\draw[dashed, ultra thin] (1.8,0) +(0:.6 and .18) arc (0:180:.6 and .18);\draw (.9,1.2) +(0:.6 and .18) arc (0:180:.6 and .18);
\draw(.9,.18) -- (.9,1.02);
\node[scale=.6] at (.445,.53) {$\scriptstyle L^2$};
\node[scale=.6] at (1.35,.53) {$\scriptstyle L^2$};
}$};
\node[scale=1.1] at (-.225,-.68)
{$\tikzmath[scale=.25]{\fill[gray!50] (.9,1.2) circle (.6 and .18);\filldraw[fill=gray!20](-.6,0) arc (-180:0:.6 and .18) arc (180:0:.3 and .18) arc (-180:0:.6 and .18)[rounded corners=1.7]-- ++(0,.3) -- ++(-.9,.6) [sharp corners]-- ++(0,.3) arc (0:-180:.6 and .18)[rounded corners=1.7]-- ++(0,-.3) -- ++(-.9,-.6) [sharp corners]-- cycle;
\draw[dashed, ultra thin] (0,0) +(0:.6 and .18) arc (0:180:.6 and .18);\draw[dashed, ultra thin] (1.8,0) +(0:.6 and .18) arc (0:180:.6 and .18);\draw (.9,1.2) +(0:.6 and .18) arc (0:180:.6 and .18);
\draw(.9,.18) -- (.9,1.02);
\node[scale=.6] at (.445,.53) {$\scriptstyle L^2$};
\node[scale=.6] at (1.35,.53) {$\scriptstyle L^2$};
}$};
\node[scale=1.1] at (-.42,-.34) {$\scriptscriptstyle\boxtimes\cala(\tikzmath[scale=.25]{\draw circle (.6 and .18);})$};
}
\\  
\tikzmath[scale=.4]{
\fill[gray!50] (.9,1.2) circle (.6 and .18);
\filldraw[fill=gray!20]
(-1.5,-1.2) arc (-180:0:.6 and .18) arc (180:0:.3 and .18) arc (-180:0:.6 and .18) 
[rounded corners=2.4]-- ++(0,.3) -- ++(-.9,.6) [sharp corners]-- ++(0,.3) 
arc (180:0:.3 and .18) arc (-180:0:.6 and .18)
[rounded corners=2.4]-- ++(0,.3) -- ++(-.9,.6) [sharp corners]-- ++(0,.3) 
arc (0:-180:.6 and .18)
[rounded corners=2.4]-- ++(0,-.3) -- ++(-.9,-.6) [sharp corners]-- ++(0,-.3)
[rounded corners=2.4]-- ++(0,-.3) -- ++(-.9,-.6) [sharp corners]-- cycle;
\draw[dashed, ultra thin] (-.9,-1.2) +(0:.6 and .18) arc (0:180:.6 and .18);
\draw[dashed, ultra thin] (.9,-1.2) +(0:.6 and .18) arc (0:180:.6 and .18);
\draw[dashed, ultra thin] (0,0) +(0:.6 and .18) arc (0:180:.6 and .18);
\draw (0,0) +(0:.6 and .18) arc (0:-180:.6 and .18);
\draw[dashed, ultra thin] (1.8,0) +(0:.6 and .18) arc (0:180:.6 and .18);
\draw (.9,1.2) +(0:.6 and .18) arc (0:180:.6 and .18);
\draw(.9,.18) -- (.9,1.02)(0,.18-1.2) -- (0,1.02-1.2);
\node at (.5,.53) {$\scriptstyle V$};
\node at (1.3,.53) {$\scriptstyle V$};
\node at (.5-.9,.53-1.2) {$\scriptstyle V$};
\node at (1.3-.9,.53-1.2) {$\scriptstyle V$};}
\,&\,:=
\tikzmath[scale=1.1]{
\node[scale=1.1] at (0,0)
{$\tikzmath[scale=.25]{\fill[gray!50] (.9,1.2) circle (.6 and .18);\filldraw[fill=gray!20](-.6,0) arc (-180:0:.6 and .18) arc (180:0:.3 and .18) arc (-180:0:.6 and .18)[rounded corners=1.7]-- ++(0,.3) -- ++(-.9,.6) [sharp corners]-- ++(0,.3) arc (0:-180:.6 and .18)[rounded corners=1.7]-- ++(0,-.3) -- ++(-.9,-.6) [sharp corners]-- cycle;
\draw[dashed, ultra thin] (0,0) +(0:.6 and .18) arc (0:180:.6 and .18);\draw[dashed, ultra thin] (1.8,0) +(0:.6 and .18) arc (0:180:.6 and .18);\draw (.9,1.2) +(0:.6 and .18) arc (0:180:.6 and .18);
\draw(.9,.18) -- (.9,1.02);
\node[scale=.9] at (.5,.53) {$\scriptscriptstyle V$};
\node[scale=.9] at (1.3,.53) {$\scriptscriptstyle V$};
}
$};
\node[scale=1.1] at (-.225,-.68)
{$\tikzmath[scale=.25]{\fill[gray!50] (.9,1.2) circle (.6 and .18);\filldraw[fill=gray!20](-.6,0) arc (-180:0:.6 and .18) arc (180:0:.3 and .18) arc (-180:0:.6 and .18)[rounded corners=1.7]-- ++(0,.3) -- ++(-.9,.6) [sharp corners]-- ++(0,.3) arc (0:-180:.6 and .18)[rounded corners=1.7]-- ++(0,-.3) -- ++(-.9,-.6) [sharp corners]-- cycle;
\draw[dashed, ultra thin] (0,0) +(0:.6 and .18) arc (0:180:.6 and .18);\draw[dashed, ultra thin] (1.8,0) +(0:.6 and .18) arc (0:180:.6 and .18);\draw (.9,1.2) +(0:.6 and .18) arc (0:180:.6 and .18);
\draw(.9,.18) -- (.9,1.02);
\node[scale=.9] at (.5,.53) {$\scriptscriptstyle V$};
\node[scale=.9] at (1.3,.53) {$\scriptscriptstyle V$};
}
$};
\node[scale=1.1] at (-.42,-.34) {$\scriptscriptstyle\boxtimes\cala(\tikzmath[scale=.25]{\draw circle (.6 and .18);})$};
}
\\  
\tikzmath[scale=.4]{
\fill[gray!50] (.9,1.2) circle (.6 and .18);
\filldraw[fill=gray!20]
(-1.5,-1.2) arc (-180:0:.6 and .18) arc (180:0:.3 and .18) arc (-180:0:.6 and .18) 
[rounded corners=2.4]-- ++(0,.3) -- ++(-.9,.6) [sharp corners]-- ++(0,.3) 
arc (180:0:.3 and .18) arc (-180:0:.6 and .18)
[rounded corners=2.4]-- ++(0,.3) -- ++(-.9,.6) [sharp corners]-- ++(0,.3) 
arc (0:-180:.6 and .18)
[rounded corners=2.4]-- ++(0,-.3) -- ++(-.9,-.6) [sharp corners]-- ++(0,-.3)
[rounded corners=2.4]-- ++(0,-.3) -- ++(-.9,-.6) [sharp corners]-- cycle;
\draw[dashed, ultra thin] (-.9,-1.2) +(0:.6 and .18) arc (0:180:.6 and .18);
\draw[dashed, ultra thin] (.9,-1.2) +(0:.6 and .18) arc (0:180:.6 and .18);
\draw[dashed, ultra thin] (0,0) +(0:.6 and .18) arc (0:180:.6 and .18);
\draw (0,0) +(0:.6 and .18) arc (0:-180:.6 and .18);
\draw[dashed, ultra thin] (1.8,0) +(0:.6 and .18) arc (0:180:.6 and .18);
\draw (.9,1.2) +(0:.6 and .18) arc (0:180:.6 and .18);
\node at (.9,.57) {$\scriptstyle V$};
\node at (0,.57-1.2) {$\scriptstyle V$};}
\,&\,:=
\tikzmath[scale=1.1]{
\node[scale=1.1] at (0,0)
{$\tikzmath[scale=.25]{\fill[gray!50] (.9,1.2) circle (.6 and .18);\filldraw[fill=gray!20](-.6,0) arc (-180:0:.6 and .18) arc (180:0:.3 and .18) arc (-180:0:.6 and .18)[rounded corners=1.7]-- ++(0,.3) -- ++(-.9,.6) [sharp corners]-- ++(0,.3) arc (0:-180:.6 and .18)[rounded corners=1.7]-- ++(0,-.3) -- ++(-.9,-.6) [sharp corners]-- cycle;
\draw[dashed, ultra thin] (0,0) +(0:.6 and .18) arc (0:180:.6 and .18);\draw[dashed, ultra thin] (1.8,0) +(0:.6 and .18) arc (0:180:.6 and .18);\draw (.9,1.2) +(0:.6 and .18) arc (0:180:.6 and .18);
\node[scale=.9] at (.9,.58) {$\scriptscriptstyle V$};
}$};
\node[scale=1.1] at (-.225,-.68)
{$\tikzmath[scale=.25]{\fill[gray!50] (.9,1.2) circle (.6 and .18);\filldraw[fill=gray!20](-.6,0) arc (-180:0:.6 and .18) arc (180:0:.3 and .18) arc (-180:0:.6 and .18)[rounded corners=1.7]-- ++(0,.3) -- ++(-.9,.6) [sharp corners]-- ++(0,.3) arc (0:-180:.6 and .18)[rounded corners=1.7]-- ++(0,-.3) -- ++(-.9,-.6) [sharp corners]-- cycle;
\draw[dashed, ultra thin] (0,0) +(0:.6 and .18) arc (0:180:.6 and .18);\draw[dashed, ultra thin] (1.8,0) +(0:.6 and .18) arc (0:180:.6 and .18);\draw (.9,1.2) +(0:.6 and .18) arc (0:180:.6 and .18);
\node[scale=.9] at (.9,.58) {$\scriptscriptstyle V$};
}$};
\node[scale=1.1] at (-.42,-.34) {$\scriptscriptstyle\boxtimes\cala(\tikzmath[scale=.25]{\draw circle (.6 and .18);})$};}
\\  
\tikzmath[scale=.4]{
\fill[gray!50] (.9,1.2) circle (.6 and .18);
\filldraw[fill=gray!20]
(-1.5,-1.2) arc (-180:0:.6 and .18) arc (180:0:.3 and .18) arc (-180:0:.6 and .18) 
[rounded corners=2.4]-- ++(0,.3) -- ++(-.9,.6) [sharp corners]-- ++(0,.3) 
arc (180:0:.3 and .18) arc (-180:0:.6 and .18)
[rounded corners=2.4]-- ++(0,.3) -- ++(-.9,.6) [sharp corners]-- ++(0,.3) 
arc (0:-180:.6 and .18)
[rounded corners=2.4]-- ++(0,-.3) -- ++(-.9,-.6) [sharp corners]-- ++(0,-.3)
[rounded corners=2.4]-- ++(0,-.3) -- ++(-.9,-.6) [sharp corners]-- cycle;
\draw[dashed, ultra thin] (-.9,-1.2) +(0:.6 and .18) arc (0:180:.6 and .18);
\draw[dashed, ultra thin] (.9,-1.2) +(0:.6 and .18) arc (0:180:.6 and .18);
\draw[dashed, ultra thin] (1.8,0) +(0:.6 and .18) arc (0:180:.6 and .18);
\draw (.9,1.2) +(0:.6 and .18) arc (0:180:.6 and .18);
\node at (.03,0) {$\scriptstyle V$};
}\,
&\,:=
V\Big(\,\tikzmath[scale=.25]{
\fill[gray!50] (.9,1.2) circle (.6 and .18);
\filldraw[fill=gray!20]
(-1.5,-1.2) arc (-180:0:.6 and .18) arc (180:0:.3 and .18) arc (-180:0:.6 and .18) 
[rounded corners=1.5]-- ++(0,.3) -- ++(-.9,.6) [sharp corners]-- ++(0,.3) 
arc (180:0:.3 and .18) arc (-180:0:.6 and .18)
[rounded corners=1.5]-- ++(0,.3) -- ++(-.9,.6) [sharp corners]-- ++(0,.3) 
arc (0:-180:.6 and .18)
[rounded corners=1.5]-- ++(0,-.3) -- ++(-.9,-.6) [sharp corners]-- ++(0,-.3)
[rounded corners=1.5]-- ++(0,-.3) -- ++(-.9,-.6) [sharp corners]-- cycle;
\draw[dashed, ultra thin] (-.9,-1.2) +(0:.6 and .18) arc (0:180:.6 and .18);
\draw[dashed, ultra thin] (.9,-1.2) +(0:.6 and .18) arc (0:180:.6 and .18);
\draw[dashed, ultra thin] (1.8,0) +(0:.6 and .18) arc (0:180:.6 and .18);
\draw (.9,1.2) +(0:.6 and .18) arc (0:180:.6 and .18);
}\Big)
=V\big(P\,\mbox{${}_3\cup_1$}\,P\big)
\\  
\tikzmath[scale=.4]{\fill[gray!50] (1.8,1.2) circle (.6 and .18);
\filldraw[fill=gray!20, line join=bevel](-.6,0) arc (-180:0:.6 and .18) arc (180:0:.3 and .18) arc (-180:0:.6 and .18) arc (180:0:.3 and .18) arc (-180:0:.6 and .18)[rounded corners=4.8]-- ++(-.1,.4) -- ++(-1.6,.4) [sharp corners]-- ++(-.1,.4) arc (0:-180:.6 and .18)[rounded corners=4.8]-- ++(-.1,-.4) -- ++(-1.6,-.4)  [sharp corners]-- cycle;
\draw[dashed, ultra thin] (0,0) +(0:.6 and .18) arc (0:180:.6 and .18);
\draw[dashed, ultra thin] (1.8,0) +(0:.6 and .18) arc (0:180:.6 and .18);
\draw[dashed, ultra thin] (3.6,0) +(0:.6 and .18) arc (0:180:.6 and .18);
\draw (1.8,1.2) +(0:.6 and .18) arc (0:180:.6 and .18);
\draw[line join=bevel] (.9,.18) -- (1.8,1.02) -- (2.7,.18);
\node[scale=.7] at (.4,.3) {$L^2$};
\node[scale=.7] at (1.8,.3) {$L^2$};
\node[scale=.7] at (3.35,.3) {$L^2$};}
&\,:=
\,\tikzmath[scale=.35]{\fill[gray!50] (0,1.2) circle (.6 and .18);
\filldraw[fill=gray!20](-.6,0) arc (-180:0:.6 and .18) -- +(0,1.2) arc (0:-180:.6 and .18) -- cycle;
\draw[dashed, ultra thin] (0,0) +(0:.6 and .18) arc (0:180:.6 and .18);
\draw (0,1.2) +(0:.6 and .18) arc (0:180:.6 and .18);
\node[scale=.7] at (.02,.54) {$L^2$};}\,\,
\boxtimes_{\cala(S^1_{\vdash})}
\tikzmath[scale=.35]{\fill[gray!50] (0,1.2) circle (.6 and .18);
\filldraw[fill=gray!20](-.6,0) arc (-180:0:.6 and .18) -- +(0,1.2) arc (0:-180:.6 and .18) -- cycle;
\draw[dashed, ultra thin] (0,0) +(0:.6 and .18) arc (0:180:.6 and .18);
\draw (0,1.2) +(0:.6 and .18) arc (0:180:.6 and .18);
\node[scale=.7] at (.02,.54) {$L^2$};}\,\,
\boxtimes_{\cala(S^1_{\vdash})}
\tikzmath[scale=.35]{\fill[gray!50] (0,1.2) circle (.6 and .18);
\filldraw[fill=gray!20](-.6,0) arc (-180:0:.6 and .18) -- +(0,1.2) arc (0:-180:.6 and .18) -- cycle;
\draw[dashed, ultra thin] (0,0) +(0:.6 and .18) arc (0:180:.6 and .18);
\draw (0,1.2) +(0:.6 and .18) arc (0:180:.6 and .18);
\node[scale=.7] at (.02,.54) {$L^2$};}
\phantom{|_{|_{|_{|_{|_{|_{|}}}}}}^{|_{|_{|_{|_{|}}}}}}\\  
\tikzmath[scale=.4]{\fill[gray!50] (1.8,1.2) circle (.6 and .18);
\filldraw[fill=gray!20, line join=bevel](-.6,0) arc (-180:0:.6 and .18) arc (180:0:.3 and .18) arc (-180:0:.6 and .18) arc (180:0:.3 and .18) arc (-180:0:.6 and .18)[rounded corners=4.8]-- ++(-.1,.4) -- ++(-1.6,.4) [sharp corners]-- ++(-.1,.4) arc (0:-180:.6 and .18)[rounded corners=4.8]-- ++(-.1,-.4) -- ++(-1.6,-.4)  [sharp corners]-- cycle;
\draw[dashed, ultra thin] (0,0) +(0:.6 and .18) arc (0:180:.6 and .18);
\draw[dashed, ultra thin] (1.8,0) +(0:.6 and .18) arc (0:180:.6 and .18);
\draw[dashed, ultra thin] (3.6,0) +(0:.6 and .18) arc (0:180:.6 and .18);
\draw (1.8,1.2) +(0:.6 and .18) arc (0:180:.6 and .18);
\draw[line join=bevel] (.9,.18) -- (1.8,1.02) -- (2.7,.18);
\node at (.4,.28) {$\scriptstyle V$};
\node at (1.8,.28) {$\scriptstyle V$};
\node at (3.35,.28) {$\scriptstyle V$};}
&\,:=
\tikzmath[scale=.35]{\fill[gray!50] (0,1.2) circle (.6 and .18);
\filldraw[fill=gray!20](-.6,0) arc (-180:0:.6 and .18) -- +(0,1.2) arc (0:-180:.6 and .18) -- cycle;
\draw[dashed, ultra thin] (0,0) +(0:.6 and .18) arc (0:180:.6 and .18);
\draw (0,1.2) +(0:.6 and .18) arc (0:180:.6 and .18);
\node[scale=.9] at (0,.53) {$\scriptstyle V$};
}\,\,
\boxtimes_{\cala(S^1_{\vdash})}
\tikzmath[scale=.35]{\fill[gray!50] (0,1.2) circle (.6 and .18);
\filldraw[fill=gray!20](-.6,0) arc (-180:0:.6 and .18) -- +(0,1.2) arc (0:-180:.6 and .18) -- cycle;
\draw[dashed, ultra thin] (0,0) +(0:.6 and .18) arc (0:180:.6 and .18);
\draw (0,1.2) +(0:.6 and .18) arc (0:180:.6 and .18);
\node[scale=.9] at (0,.53) {$\scriptstyle V$};
}\,\,
\boxtimes_{\cala(S^1_{\vdash})}
\tikzmath[scale=.35]{\fill[gray!50] (0,1.2) circle (.6 and .18);
\filldraw[fill=gray!20](-.6,0) arc (-180:0:.6 and .18) -- +(0,1.2) arc (0:-180:.6 and .18) -- cycle;
\draw[dashed, ultra thin] (0,0) +(0:.6 and .18) arc (0:180:.6 and .18);
\draw (0,1.2) +(0:.6 and .18) arc (0:180:.6 and .18);
\node[scale=.9] at (0,.53) {$\scriptstyle V$};
}
\phantom{|_{|_{|_{|_{|_{|_{|}}}}}}}\\  
\tikzmath[scale=.4]{\fill[gray!50] (1.8,1.2) circle (.6 and .18);
\filldraw[fill=gray!20, line join=bevel](-.6,0) arc (-180:0:.6 and .18) arc (180:0:.3 and .18) arc (-180:0:.6 and .18) arc (180:0:.3 and .18) arc (-180:0:.6 and .18)[rounded corners=4.8]-- ++(-.1,.4) -- ++(-1.6,.4) [sharp corners]-- ++(-.1,.4) arc (0:-180:.6 and .18)[rounded corners=4.8]-- ++(-.1,-.4) -- ++(-1.6,-.4)  [sharp corners]-- cycle;
\draw[dashed, ultra thin] (0,0) +(0:.6 and .18) arc (0:180:.6 and .18);
\draw[dashed, ultra thin] (1.8,0) +(0:.6 and .18) arc (0:180:.6 and .18);
\draw[dashed, ultra thin] (3.6,0) +(0:.6 and .18) arc (0:180:.6 and .18);
\draw (1.8,1.2) +(0:.6 and .18) arc (0:180:.6 and .18);
\node at (1.8,.57) {$\scriptstyle V$};}
&\,:=V\big(\,\tikzmath[scale=.25]{\fill[gray!50] (0,1.2) circle (.6 and .18);
\filldraw[fill=gray!20](-.6,0) arc (-180:0:.6 and .18) -- +(0,1.2) arc (0:-180:.6 and .18) -- cycle;
\draw[dashed, ultra thin] (0,0) +(0:.6 and .18) arc (0:180:.6 and .18);
\draw (0,1.2) +(0:.6 and .18) arc (0:180:.6 and .18);}\,\cup_{S^1_{\vdash}} \,\tikzmath[scale=.25]{\fill[gray!50] (0,1.2) circle (.6 and .18);
\filldraw[fill=gray!20](-.6,0) arc (-180:0:.6 and .18) -- +(0,1.2) arc (0:-180:.6 and .18) -- cycle;
\draw[dashed, ultra thin] (0,0) +(0:.6 and .18) arc (0:180:.6 and .18);
\draw (0,1.2) +(0:.6 and .18) arc (0:180:.6 and .18);}\,\cup_{S^1_{\vdash}} \,\tikzmath[scale=.25]{\fill[gray!50] (0,1.2) circle (.6 and .18);
\filldraw[fill=gray!20](-.6,0) arc (-180:0:.6 and .18) -- +(0,1.2) arc (0:-180:.6 and .18) -- cycle;
\draw[dashed, ultra thin] (0,0) +(0:.6 and .18) arc (0:180:.6 and .18);
\draw (0,1.2) +(0:.6 and .18) arc (0:180:.6 and .18);}\,\big)\,
\end{split}
\]

\begin{proof}[Proof of Proposition \ref{prop: al = al}]
Keeping in mind that the identification between $H_P$ and $V(P)$ is given by
\[
\tikzmath[scale=.4]{\fill[gray!50] (.9,1.2) circle (.6 and .18);\filldraw[fill=gray!20](-.6,0) arc (-180:0:.6 and .18) arc (180:0:.3 and .18) arc (-180:0:.6 and .18)[rounded corners=2.4]-- ++(0,.3) -- ++(-.9,.6) [sharp corners]-- ++(0,.3) arc (0:-180:.6 and .18)[rounded corners=2.4]-- ++(0,-.3) -- ++(-.9,-.6) [sharp corners]-- cycle;
\draw[dashed, ultra thin] (0,0) +(0:.6 and .18) arc (0:180:.6 and .18);\draw[dashed, ultra thin] (1.8,0) +(0:.6 and .18) arc (0:180:.6 and .18);\draw (.9,1.2) +(0:.6 and .18) arc (0:180:.6 and .18);
\draw(.9,.18) -- (.9,1.02);
\node[scale=.7] at (.45,.53) {$L^2$};
\node[scale=.7] at (1.38,.53) {$L^2$};}
\xrightarrow{(\text{\footnotesize \ref{H_Sigm=L^2hatA(S)+}})}
\tikzmath[scale=.4]{\fill[gray!50] (.9,1.2) circle (.6 and .18);\filldraw[fill=gray!20](-.6,0) arc (-180:0:.6 and .18) arc (180:0:.3 and .18) arc (-180:0:.6 and .18)[rounded corners=2.4]-- ++(0,.3) -- ++(-.9,.6) [sharp corners]-- ++(0,.3) arc (0:-180:.6 and .18)[rounded corners=2.4]-- ++(0,-.3) -- ++(-.9,-.6) [sharp corners]-- cycle;
\draw[dashed, ultra thin] (0,0) +(0:.6 and .18) arc (0:180:.6 and .18);\draw[dashed, ultra thin] (1.8,0) +(0:.6 and .18) arc (0:180:.6 and .18);\draw (.9,1.2) +(0:.6 and .18) arc (0:180:.6 and .18);
\draw(.9,.18) -- (.9,1.02);
\node at (.5,.53) {$\scriptstyle V$};
\node at (1.3,.53) {$\scriptstyle V$};}
\xrightarrow{\,g^{-1}}
\tikzmath[scale=.4]{\fill[gray!50] (.9,1.2) circle (.6 and .18);\filldraw[fill=gray!20](-.6,0) arc (-180:0:.6 and .18) arc (180:0:.3 and .18) arc (-180:0:.6 and .18)[rounded corners=2.4]-- ++(0,.3) -- ++(-.9,.6) [sharp corners]-- ++(0,.3) arc (0:-180:.6 and .18)[rounded corners=2.4]-- ++(0,-.3) -- ++(-.9,-.6) [sharp corners]-- cycle;
\draw[dashed, ultra thin] (0,0) +(0:.6 and .18) arc (0:180:.6 and .18);\draw[dashed, ultra thin] (1.8,0) +(0:.6 and .18) arc (0:180:.6 and .18);\draw (.9,1.2) +(0:.6 and .18) arc (0:180:.6 and .18);
\node at (.9,.57) {$\scriptstyle V$};}\,\,,
\]
our task is to prove the commutativity of this diagram:
\[
\tikzmath
{
\node (A) at (0,4) {$\tikzmath[scale=.4]{
\fill[gray!50] (.9,1.2) circle (.6 and .18);
\filldraw[fill=gray!20]
(-1.5,-1.2) arc (-180:0:.6 and .18) arc (180:0:.3 and .18) arc (-180:0:.6 and .18) 
[rounded corners=2.4]-- ++(0,.3) -- ++(-.9,.6) [sharp corners]-- ++(0,.3) 
arc (180:0:.3 and .18) arc (-180:0:.6 and .18)
[rounded corners=2.4]-- ++(0,.3) -- ++(-.9,.6) [sharp corners]-- ++(0,.3) 
arc (0:-180:.6 and .18)
[rounded corners=2.4]-- ++(0,-.3) -- ++(-.9,-.6) [sharp corners]-- ++(0,-.3)
[rounded corners=2.4]-- ++(0,-.3) -- ++(-.9,-.6) [sharp corners]-- cycle;
\draw[dashed, ultra thin] (-.9,-1.2) +(0:.6 and .18) arc (0:180:.6 and .18);
\draw[dashed, ultra thin] (.9,-1.2) +(0:.6 and .18) arc (0:180:.6 and .18);
\draw[dashed, ultra thin] (0,0) +(0:.6 and .18) arc (0:180:.6 and .18);
\draw (0,0) +(0:.6 and .18) arc (0:-180:.6 and .18);
\draw[dashed, ultra thin] (1.8,0) +(0:.6 and .18) arc (0:180:.6 and .18);
\draw (.9,1.2) +(0:.6 and .18) arc (0:180:.6 and .18);
\draw(.9,.18) -- (.9,1.02)(0,.18-1.2) -- (0,1.02-1.2);
\node[scale=.7] at (.45,.53) {$L^2$};
\node[scale=.7] at (1.38,.53) {$L^2$};
\node[scale=.7] at (.45-.9,.53-1.2) {$L^2$};
\node[scale=.7] at (1.38-.9,.53-1.2) {$L^2$};}
$};
\node (B) at (5,4) {$\tikzmath[scale=.4]{\fill[gray!50] (1.8,1.2) circle (.6 and .18);
\filldraw[fill=gray!20, line join=bevel](-.6,0) arc (-180:0:.6 and .18) arc (180:0:.3 and .18) arc (-180:0:.6 and .18) arc (180:0:.3 and .18) arc (-180:0:.6 and .18)[rounded corners=4.8]-- ++(-.1,.4) -- ++(-1.6,.4) [sharp corners]-- ++(-.1,.4) arc (0:-180:.6 and .18)[rounded corners=4.8]-- ++(-.1,-.4) -- ++(-1.6,-.4)  [sharp corners]-- cycle;
\draw[dashed, ultra thin] (0,0) +(0:.6 and .18) arc (0:180:.6 and .18);
\draw[dashed, ultra thin] (1.8,0) +(0:.6 and .18) arc (0:180:.6 and .18);
\draw[dashed, ultra thin] (3.6,0) +(0:.6 and .18) arc (0:180:.6 and .18);
\draw (1.8,1.2) +(0:.6 and .18) arc (0:180:.6 and .18);
\draw[line join=bevel] (.9,.18) -- (1.8,1.02) -- (2.7,.18);
\node[scale=.7] at (.4,.3) {$L^2$};
\node[scale=.7] at (1.8,.3) {$L^2$};
\node[scale=.7] at (3.35,.3) {$L^2$};}
$};
\node (C) at (10,4) {$\tikzmath[scale=.4]{\pgftransformxscale{-1}
\fill[gray!50] (.9,1.2) circle (.6 and .18);
\filldraw[fill=gray!20]
(-1.5,-1.2) arc (-180:0:.6 and .18) arc (180:0:.3 and .18) arc (-180:0:.6 and .18) 
[rounded corners=2.4]-- ++(0,.3) -- ++(-.9,.6) [sharp corners]-- ++(0,.3) 
arc (180:0:.3 and .18) arc (-180:0:.6 and .18)
[rounded corners=2.4]-- ++(0,.3) -- ++(-.9,.6) [sharp corners]-- ++(0,.3) 
arc (0:-180:.6 and .18)
[rounded corners=2.4]-- ++(0,-.3) -- ++(-.9,-.6) [sharp corners]-- ++(0,-.3)
[rounded corners=2.4]-- ++(0,-.3) -- ++(-.9,-.6) [sharp corners]-- cycle;
\draw[dashed, ultra thin] (-.9,-1.2) +(0:.6 and .18) arc (0:180:.6 and .18);
\draw[dashed, ultra thin] (.9,-1.2) +(0:.6 and .18) arc (0:180:.6 and .18);
\draw[dashed, ultra thin] (0,0) +(0:.6 and .18) arc (0:180:.6 and .18);
\draw (0,0) +(0:.6 and .18) arc (0:-180:.6 and .18);
\draw[dashed, ultra thin] (1.8,0) +(0:.6 and .18) arc (0:180:.6 and .18);
\draw (.9,1.2) +(0:.6 and .18) arc (0:180:.6 and .18);
\draw(.9,.18) -- (.9,1.02)(0,.18-1.2) -- (0,1.02-1.2);
\node[scale=.7] at (.42,.53) {$L^2$};
\node[scale=.7] at (1.35,.53) {$L^2$};
\node[scale=.7] at (.42-.9,.53-1.2) {$L^2$};
\node[scale=.7] at (1.35-.9,.53-1.2) {$L^2$};}
$};
\node (D) at (0,2) {$\tikzmath[scale=.4]{
\fill[gray!50] (.9,1.2) circle (.6 and .18);
\filldraw[fill=gray!20]
(-1.5,-1.2) arc (-180:0:.6 and .18) arc (180:0:.3 and .18) arc (-180:0:.6 and .18) 
[rounded corners=2.4]-- ++(0,.3) -- ++(-.9,.6) [sharp corners]-- ++(0,.3) 
arc (180:0:.3 and .18) arc (-180:0:.6 and .18)
[rounded corners=2.4]-- ++(0,.3) -- ++(-.9,.6) [sharp corners]-- ++(0,.3) 
arc (0:-180:.6 and .18)
[rounded corners=2.4]-- ++(0,-.3) -- ++(-.9,-.6) [sharp corners]-- ++(0,-.3)
[rounded corners=2.4]-- ++(0,-.3) -- ++(-.9,-.6) [sharp corners]-- cycle;
\draw[dashed, ultra thin] (-.9,-1.2) +(0:.6 and .18) arc (0:180:.6 and .18);
\draw[dashed, ultra thin] (.9,-1.2) +(0:.6 and .18) arc (0:180:.6 and .18);
\draw[dashed, ultra thin] (0,0) +(0:.6 and .18) arc (0:180:.6 and .18);
\draw (0,0) +(0:.6 and .18) arc (0:-180:.6 and .18);
\draw[dashed, ultra thin] (1.8,0) +(0:.6 and .18) arc (0:180:.6 and .18);
\draw (.9,1.2) +(0:.6 and .18) arc (0:180:.6 and .18);
\draw(.9,.18) -- (.9,1.02)(0,.18-1.2) -- (0,1.02-1.2);
\node at (.5,.53) {$\scriptstyle V$};
\node at (1.3,.53) {$\scriptstyle V$};
\node at (.5-.9,.53-1.2) {$\scriptstyle V$};
\node at (1.3-.9,.53-1.2) {$\scriptstyle V$};}
$};
\node (E) at (10,2) {$\tikzmath[scale=.4]{\pgftransformxscale{-1}
\fill[gray!50] (.9,1.2) circle (.6 and .18);
\filldraw[fill=gray!20]
(-1.5,-1.2) arc (-180:0:.6 and .18) arc (180:0:.3 and .18) arc (-180:0:.6 and .18) 
[rounded corners=2.4]-- ++(0,.3) -- ++(-.9,.6) [sharp corners]-- ++(0,.3) 
arc (180:0:.3 and .18) arc (-180:0:.6 and .18)
[rounded corners=2.4]-- ++(0,.3) -- ++(-.9,.6) [sharp corners]-- ++(0,.3) 
arc (0:-180:.6 and .18)
[rounded corners=2.4]-- ++(0,-.3) -- ++(-.9,-.6) [sharp corners]-- ++(0,-.3)
[rounded corners=2.4]-- ++(0,-.3) -- ++(-.9,-.6) [sharp corners]-- cycle;
\draw[dashed, ultra thin] (-.9,-1.2) +(0:.6 and .18) arc (0:180:.6 and .18);
\draw[dashed, ultra thin] (.9,-1.2) +(0:.6 and .18) arc (0:180:.6 and .18);
\draw[dashed, ultra thin] (0,0) +(0:.6 and .18) arc (0:180:.6 and .18);
\draw (0,0) +(0:.6 and .18) arc (0:-180:.6 and .18);
\draw[dashed, ultra thin] (1.8,0) +(0:.6 and .18) arc (0:180:.6 and .18);
\draw (.9,1.2) +(0:.6 and .18) arc (0:180:.6 and .18);
\draw(.9,.18) -- (.9,1.02)(0,.18-1.2) -- (0,1.02-1.2);
\node at (.5,.53) {$\scriptstyle V$};
\node at (1.3,.53) {$\scriptstyle V$};
\node at (.5-.9,.53-1.2) {$\scriptstyle V$};
\node at (1.3-.9,.53-1.2) {$\scriptstyle V$};}
$};
\node (F) at (0,0) {$\tikzmath[scale=.4]{
\fill[gray!50] (.9,1.2) circle (.6 and .18);
\filldraw[fill=gray!20]
(-1.5,-1.2) arc (-180:0:.6 and .18) arc (180:0:.3 and .18) arc (-180:0:.6 and .18) 
[rounded corners=2.4]-- ++(0,.3) -- ++(-.9,.6) [sharp corners]-- ++(0,.3) 
arc (180:0:.3 and .18) arc (-180:0:.6 and .18)
[rounded corners=2.4]-- ++(0,.3) -- ++(-.9,.6) [sharp corners]-- ++(0,.3) 
arc (0:-180:.6 and .18)
[rounded corners=2.4]-- ++(0,-.3) -- ++(-.9,-.6) [sharp corners]-- ++(0,-.3)
[rounded corners=2.4]-- ++(0,-.3) -- ++(-.9,-.6) [sharp corners]-- cycle;
\draw[dashed, ultra thin] (-.9,-1.2) +(0:.6 and .18) arc (0:180:.6 and .18);
\draw[dashed, ultra thin] (.9,-1.2) +(0:.6 and .18) arc (0:180:.6 and .18);
\draw[dashed, ultra thin] (0,0) +(0:.6 and .18) arc (0:180:.6 and .18);
\draw (0,0) +(0:.6 and .18) arc (0:-180:.6 and .18);
\draw[dashed, ultra thin] (1.8,0) +(0:.6 and .18) arc (0:180:.6 and .18);
\draw (.9,1.2) +(0:.6 and .18) arc (0:180:.6 and .18);
\node at (.9,.57) {$\scriptstyle V$};
\node at (0,.57-1.2) {$\scriptstyle V$};}
$};
\node (G) at (3.3333,0) {$\tikzmath[scale=.4]{
\fill[gray!50] (.9,1.2) circle (.6 and .18);
\filldraw[fill=gray!20]
(-1.5,-1.2) arc (-180:0:.6 and .18) arc (180:0:.3 and .18) arc (-180:0:.6 and .18) 
[rounded corners=2.4]-- ++(0,.3) -- ++(-.9,.6) [sharp corners]-- ++(0,.3) 
arc (180:0:.3 and .18) arc (-180:0:.6 and .18)
[rounded corners=2.4]-- ++(0,.3) -- ++(-.9,.6) [sharp corners]-- ++(0,.3) 
arc (0:-180:.6 and .18)
[rounded corners=2.4]-- ++(0,-.3) -- ++(-.9,-.6) [sharp corners]-- ++(0,-.3)
[rounded corners=2.4]-- ++(0,-.3) -- ++(-.9,-.6) [sharp corners]-- cycle;
\draw[dashed, ultra thin] (-.9,-1.2) +(0:.6 and .18) arc (0:180:.6 and .18);
\draw[dashed, ultra thin] (.9,-1.2) +(0:.6 and .18) arc (0:180:.6 and .18);
\draw[dashed, ultra thin] (1.8,0) +(0:.6 and .18) arc (0:180:.6 and .18);
\draw (.9,1.2) +(0:.6 and .18) arc (0:180:.6 and .18);
\node at (.03,0) {$\scriptstyle V$};
}$};
\node (H) at (6.6666,0) {$\tikzmath[scale=.4]{\pgftransformxscale{-1}
\fill[gray!50] (.9,1.2) circle (.6 and .18);
\filldraw[fill=gray!20]
(-1.5,-1.2) arc (-180:0:.6 and .18) arc (180:0:.3 and .18) arc (-180:0:.6 and .18) 
[rounded corners=2.4]-- ++(0,.3) -- ++(-.9,.6) [sharp corners]-- ++(0,.3) 
arc (180:0:.3 and .18) arc (-180:0:.6 and .18)
[rounded corners=2.4]-- ++(0,.3) -- ++(-.9,.6) [sharp corners]-- ++(0,.3) 
arc (0:-180:.6 and .18)
[rounded corners=2.4]-- ++(0,-.3) -- ++(-.9,-.6) [sharp corners]-- ++(0,-.3)
[rounded corners=2.4]-- ++(0,-.3) -- ++(-.9,-.6) [sharp corners]-- cycle;
\draw[dashed, ultra thin] (-.9,-1.2) +(0:.6 and .18) arc (0:180:.6 and .18);
\draw[dashed, ultra thin] (.9,-1.2) +(0:.6 and .18) arc (0:180:.6 and .18);
\draw[dashed, ultra thin] (1.8,0) +(0:.6 and .18) arc (0:180:.6 and .18);
\draw (.9,1.2) +(0:.6 and .18) arc (0:180:.6 and .18);
\node at (.03,0) {$\scriptstyle V$};
}$};
\node (I) at (10,0) {$\tikzmath[scale=.4]{\pgftransformxscale{-1}
\fill[gray!50] (.9,1.2) circle (.6 and .18);
\filldraw[fill=gray!20]
(-1.5,-1.2) arc (-180:0:.6 and .18) arc (180:0:.3 and .18) arc (-180:0:.6 and .18) 
[rounded corners=2.4]-- ++(0,.3) -- ++(-.9,.6) [sharp corners]-- ++(0,.3) 
arc (180:0:.3 and .18) arc (-180:0:.6 and .18)
[rounded corners=2.4]-- ++(0,.3) -- ++(-.9,.6) [sharp corners]-- ++(0,.3) 
arc (0:-180:.6 and .18)
[rounded corners=2.4]-- ++(0,-.3) -- ++(-.9,-.6) [sharp corners]-- ++(0,-.3)
[rounded corners=2.4]-- ++(0,-.3) -- ++(-.9,-.6) [sharp corners]-- cycle;
\draw[dashed, ultra thin] (-.9,-1.2) +(0:.6 and .18) arc (0:180:.6 and .18);
\draw[dashed, ultra thin] (.9,-1.2) +(0:.6 and .18) arc (0:180:.6 and .18);
\draw[dashed, ultra thin] (0,0) +(0:.6 and .18) arc (0:180:.6 and .18);
\draw (0,0) +(0:.6 and .18) arc (0:-180:.6 and .18);
\draw[dashed, ultra thin] (1.8,0) +(0:.6 and .18) arc (0:180:.6 and .18);
\draw (.9,1.2) +(0:.6 and .18) arc (0:180:.6 and .18);
\node at (.9,.57) {$\scriptstyle V$};
\node at (0,.57-1.2) {$\scriptstyle V$};}
$};
\draw[->] (A) --node[above, scale=.9]{\parbox{2.6cm}{\small \centerline{\it cancel} $-\,\boxtimes_{\cala(S^1)}L^2\cala(S^1)$}} (B);
\draw[->] (B) --node[above, scale=.9]{\parbox{2.6cm}{\small \centerline{\it un-cancel} $-\,\boxtimes_{\cala(S^1)}L^2\cala(S^1)$}} (C);
\draw[->] (A) --node[left]{$\scriptstyle (\text{\footnotesize \ref{H_Sigm=L^2hatA(S)+}})$} (D);
\draw[->] (C) --node[right]{$\scriptstyle (\text{\footnotesize \ref{H_Sigm=L^2hatA(S)+}})$} (E);
\draw[->] (D) --node[left]{$\scriptstyle g^{-1}$} (F);
\draw[->] (E) --node[right]{$\scriptstyle g^{-1}$} (I);
\draw[->] (F) --node[above]{$\scriptstyle g^{-1}$} (G);
\draw[->] (G) --node[above]{$\scriptstyle V(\underline\alpha)$} (H);
\draw[->] (H) --node[above]{$\scriptstyle g$} (I);
}
\medskip
\]
Note that the top composite is $\alpha$ and that the bottom composite is $\alpha_{geo}$.\goodbreak

We can split this square into two mirror-image sub-diagrams by inserting
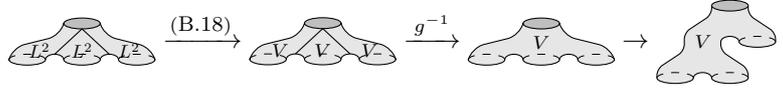
\[
\tikzmath[scale=.4]{\fill[gray!50] (1.8,1.2) circle (.6 and .18);
\filldraw[fill=gray!20, line join=bevel](-.6,0) arc (-180:0:.6 and .18) arc (180:0:.3 and .18) arc (-180:0:.6 and .18) arc (180:0:.3 and .18) arc (-180:0:.6 and .18)[rounded corners=4.8]-- ++(-.1,.4) -- ++(-1.6,.4) [sharp corners]-- ++(-.1,.4) arc (0:-180:.6 and .18)[rounded corners=4.8]-- ++(-.1,-.4) -- ++(-1.6,-.4)  [sharp corners]-- cycle;
\draw[dashed, ultra thin] (0,0) +(0:.6 and .18) arc (0:180:.6 and .18);
\draw[dashed, ultra thin] (1.8,0) +(0:.6 and .18) arc (0:180:.6 and .18);
\draw[dashed, ultra thin] (3.6,0) +(0:.6 and .18) arc (0:180:.6 and .18);
\draw (1.8,1.2) +(0:.6 and .18) arc (0:180:.6 and .18);
\draw[line join=bevel] (.9,.18) -- (1.8,1.02) -- (2.7,.18);
\node[scale=.7] at (.4,.3) {$L^2$};
\node[scale=.7] at (1.8,.3) {$L^2$};
\node[scale=.7] at (3.35,.3) {$L^2$};}
\xrightarrow{(\text{\footnotesize \ref{H_Sigm=L^2hatA(S)+}})}
\tikzmath[scale=.4]{\fill[gray!50] (1.8,1.2) circle (.6 and .18);
\filldraw[fill=gray!20, line join=bevel](-.6,0) arc (-180:0:.6 and .18) arc (180:0:.3 and .18) arc (-180:0:.6 and .18) arc (180:0:.3 and .18) arc (-180:0:.6 and .18)[rounded corners=4.8]-- ++(-.1,.4) -- ++(-1.6,.4) [sharp corners]-- ++(-.1,.4) arc (0:-180:.6 and .18)[rounded corners=4.8]-- ++(-.1,-.4) -- ++(-1.6,-.4)  [sharp corners]-- cycle;
\draw[dashed, ultra thin] (0,0) +(0:.6 and .18) arc (0:180:.6 and .18);
\draw[dashed, ultra thin] (1.8,0) +(0:.6 and .18) arc (0:180:.6 and .18);
\draw[dashed, ultra thin] (3.6,0) +(0:.6 and .18) arc (0:180:.6 and .18);
\draw (1.8,1.2) +(0:.6 and .18) arc (0:180:.6 and .18);
\draw[line join=bevel] (.9,.18) -- (1.8,1.02) -- (2.7,.18);
\node at (.4,.28) {$\scriptstyle V$};
\node at (1.8,.28) {$\scriptstyle V$};
\node at (3.35,.28) {$\scriptstyle V$};}
\xrightarrow{\,g^{-1}}
\tikzmath[scale=.4]{\fill[gray!50] (1.8,1.2) circle (.6 and .18);
\filldraw[fill=gray!20, line join=bevel](-.6,0) arc (-180:0:.6 and .18) arc (180:0:.3 and .18) arc (-180:0:.6 and .18) arc (180:0:.3 and .18) arc (-180:0:.6 and .18)[rounded corners=4.8]-- ++(-.1,.4) -- ++(-1.6,.4) [sharp corners]-- ++(-.1,.4) arc (0:-180:.6 and .18)[rounded corners=4.8]-- ++(-.1,-.4) -- ++(-1.6,-.4)  [sharp corners]-- cycle;
\draw[dashed, ultra thin] (0,0) +(0:.6 and .18) arc (0:180:.6 and .18);
\draw[dashed, ultra thin] (1.8,0) +(0:.6 and .18) arc (0:180:.6 and .18);
\draw[dashed, ultra thin] (3.6,0) +(0:.6 and .18) arc (0:180:.6 and .18);
\draw (1.8,1.2) +(0:.6 and .18) arc (0:180:.6 and .18);
\node at (1.8,.57) {$\scriptstyle V$};}
\to
\tikzmath[scale=.4]{
\fill[gray!50] (.9,1.2) circle (.6 and .18);
\filldraw[fill=gray!20]
(-1.5,-1.2) arc (-180:0:.6 and .18) arc (180:0:.3 and .18) arc (-180:0:.6 and .18) 
[rounded corners=2.4]-- ++(0,.3) -- ++(-.9,.6) [sharp corners]-- ++(0,.3) 
arc (180:0:.3 and .18) arc (-180:0:.6 and .18)
[rounded corners=2.4]-- ++(0,.3) -- ++(-.9,.6) [sharp corners]-- ++(0,.3) 
arc (0:-180:.6 and .18)
[rounded corners=2.4]-- ++(0,-.3) -- ++(-.9,-.6) [sharp corners]-- ++(0,-.3)
[rounded corners=2.4]-- ++(0,-.3) -- ++(-.9,-.6) [sharp corners]-- cycle;
\draw[dashed, ultra thin] (-.9,-1.2) +(0:.6 and .18) arc (0:180:.6 and .18);
\draw[dashed, ultra thin] (.9,-1.2) +(0:.6 and .18) arc (0:180:.6 and .18);
\draw[dashed, ultra thin] (1.8,0) +(0:.6 and .18) arc (0:180:.6 and .18);
\draw (.9,1.2) +(0:.6 and .18) arc (0:180:.6 and .18);
\node at (.03,0) {$\scriptstyle V$};}
\medskip
\]
down the middle (where the last map is induced by the obvious diffeomorphism and there is an analogous map to the target of $V(\underline\alpha)$).
It is therefore enough to prove the commutativity of the following $7$-gon:
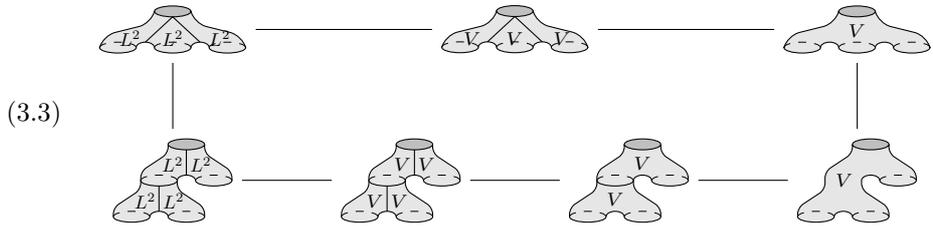
\begin{equation}\label{7gon of 3-pants}
\tikzmath
{\pgftransformxscale{1.5}
\node (A) at (0,0) {$\tikzmath[scale=.4]{
\fill[gray!50] (.9,1.2) circle (.6 and .18);
\filldraw[fill=gray!20]
(-1.5,-1.2) arc (-180:0:.6 and .18) arc (180:0:.3 and .18) arc (-180:0:.6 and .18) 
[rounded corners=2.4]-- ++(0,.3) -- ++(-.9,.6) [sharp corners]-- ++(0,.3) 
arc (180:0:.3 and .18) arc (-180:0:.6 and .18)
[rounded corners=2.4]-- ++(0,.3) -- ++(-.9,.6) [sharp corners]-- ++(0,.3) 
arc (0:-180:.6 and .18)
[rounded corners=2.4]-- ++(0,-.3) -- ++(-.9,-.6) [sharp corners]-- ++(0,-.3)
[rounded corners=2.4]-- ++(0,-.3) -- ++(-.9,-.6) [sharp corners]-- cycle;
\draw[dashed, ultra thin] (-.9,-1.2) +(0:.6 and .18) arc (0:180:.6 and .18);
\draw[dashed, ultra thin] (.9,-1.2) +(0:.6 and .18) arc (0:180:.6 and .18);
\draw[dashed, ultra thin] (0,0) +(0:.6 and .18) arc (0:180:.6 and .18);
\draw (0,0) +(0:.6 and .18) arc (0:-180:.6 and .18);
\draw[dashed, ultra thin] (1.8,0) +(0:.6 and .18) arc (0:180:.6 and .18);
\draw (.9,1.2) +(0:.6 and .18) arc (0:180:.6 and .18);
\draw(.9,.18) -- (.9,1.02)(0,.18-1.2) -- (0,1.02-1.2);
\node[scale=.7] at (.45,.53) {$L^2$};
\node[scale=.7] at (1.38,.53) {$L^2$};
\node[scale=.7] at (.45-.9,.53-1.2) {$L^2$};
\node[scale=.7] at (1.38-.9,.53-1.2) {$L^2$};}
$};
\node (B) at (0,2) {$\tikzmath[scale=.4]{\fill[gray!50] (1.8,1.2) circle (.6 and .18);
\filldraw[fill=gray!20, line join=bevel](-.6,0) arc (-180:0:.6 and .18) arc (180:0:.3 and .18) arc (-180:0:.6 and .18) arc (180:0:.3 and .18) arc (-180:0:.6 and .18)[rounded corners=4.8]-- ++(-.1,.4) -- ++(-1.6,.4) [sharp corners]-- ++(-.1,.4) arc (0:-180:.6 and .18)[rounded corners=4.8]-- ++(-.1,-.4) -- ++(-1.6,-.4)  [sharp corners]-- cycle;
\draw[dashed, ultra thin] (0,0) +(0:.6 and .18) arc (0:180:.6 and .18);
\draw[dashed, ultra thin] (1.8,0) +(0:.6 and .18) arc (0:180:.6 and .18);
\draw[dashed, ultra thin] (3.6,0) +(0:.6 and .18) arc (0:180:.6 and .18);
\draw (1.8,1.2) +(0:.6 and .18) arc (0:180:.6 and .18);
\draw[line join=bevel] (.9,.18) -- (1.8,1.02) -- (2.7,.18);
\node[scale=.7] at (.4,.3) {$L^2$};
\node[scale=.7] at (1.8,.3) {$L^2$};
\node[scale=.7] at (3.35,.3) {$L^2$};}
$};
\node (D) at (2,0) {$\tikzmath[scale=.4]{
\fill[gray!50] (.9,1.2) circle (.6 and .18);
\filldraw[fill=gray!20]
(-1.5,-1.2) arc (-180:0:.6 and .18) arc (180:0:.3 and .18) arc (-180:0:.6 and .18) 
[rounded corners=2.4]-- ++(0,.3) -- ++(-.9,.6) [sharp corners]-- ++(0,.3) 
arc (180:0:.3 and .18) arc (-180:0:.6 and .18)
[rounded corners=2.4]-- ++(0,.3) -- ++(-.9,.6) [sharp corners]-- ++(0,.3) 
arc (0:-180:.6 and .18)
[rounded corners=2.4]-- ++(0,-.3) -- ++(-.9,-.6) [sharp corners]-- ++(0,-.3)
[rounded corners=2.4]-- ++(0,-.3) -- ++(-.9,-.6) [sharp corners]-- cycle;
\draw[dashed, ultra thin] (-.9,-1.2) +(0:.6 and .18) arc (0:180:.6 and .18);
\draw[dashed, ultra thin] (.9,-1.2) +(0:.6 and .18) arc (0:180:.6 and .18);
\draw[dashed, ultra thin] (0,0) +(0:.6 and .18) arc (0:180:.6 and .18);
\draw (0,0) +(0:.6 and .18) arc (0:-180:.6 and .18);
\draw[dashed, ultra thin] (1.8,0) +(0:.6 and .18) arc (0:180:.6 and .18);
\draw (.9,1.2) +(0:.6 and .18) arc (0:180:.6 and .18);
\draw(.9,.18) -- (.9,1.02)(0,.18-1.2) -- (0,1.02-1.2);
\node at (.5,.53) {$\scriptstyle V$};
\node at (1.3,.53) {$\scriptstyle V$};
\node at (.5-.9,.53-1.2) {$\scriptstyle V$};
\node at (1.3-.9,.53-1.2) {$\scriptstyle V$};}
$};
\node (F) at (4,0) {$\tikzmath[scale=.4]{
\fill[gray!50] (.9,1.2) circle (.6 and .18);
\filldraw[fill=gray!20]
(-1.5,-1.2) arc (-180:0:.6 and .18) arc (180:0:.3 and .18) arc (-180:0:.6 and .18) 
[rounded corners=2.4]-- ++(0,.3) -- ++(-.9,.6) [sharp corners]-- ++(0,.3) 
arc (180:0:.3 and .18) arc (-180:0:.6 and .18)
[rounded corners=2.4]-- ++(0,.3) -- ++(-.9,.6) [sharp corners]-- ++(0,.3) 
arc (0:-180:.6 and .18)
[rounded corners=2.4]-- ++(0,-.3) -- ++(-.9,-.6) [sharp corners]-- ++(0,-.3)
[rounded corners=2.4]-- ++(0,-.3) -- ++(-.9,-.6) [sharp corners]-- cycle;
\draw[dashed, ultra thin] (-.9,-1.2) +(0:.6 and .18) arc (0:180:.6 and .18);
\draw[dashed, ultra thin] (.9,-1.2) +(0:.6 and .18) arc (0:180:.6 and .18);
\draw[dashed, ultra thin] (0,0) +(0:.6 and .18) arc (0:180:.6 and .18);
\draw (0,0) +(0:.6 and .18) arc (0:-180:.6 and .18);
\draw[dashed, ultra thin] (1.8,0) +(0:.6 and .18) arc (0:180:.6 and .18);
\draw (.9,1.2) +(0:.6 and .18) arc (0:180:.6 and .18);
\node at (.9,.57) {$\scriptstyle V$};
\node at (0,.57-1.2) {$\scriptstyle V$};}
$};
\node (G) at (6,0) {$\tikzmath[scale=.4]{
\fill[gray!50] (.9,1.2) circle (.6 and .18);
\filldraw[fill=gray!20]
(-1.5,-1.2) arc (-180:0:.6 and .18) arc (180:0:.3 and .18) arc (-180:0:.6 and .18) 
[rounded corners=2.4]-- ++(0,.3) -- ++(-.9,.6) [sharp corners]-- ++(0,.3) 
arc (180:0:.3 and .18) arc (-180:0:.6 and .18)
[rounded corners=2.4]-- ++(0,.3) -- ++(-.9,.6) [sharp corners]-- ++(0,.3) 
arc (0:-180:.6 and .18)
[rounded corners=2.4]-- ++(0,-.3) -- ++(-.9,-.6) [sharp corners]-- ++(0,-.3)
[rounded corners=2.4]-- ++(0,-.3) -- ++(-.9,-.6) [sharp corners]-- cycle;
\draw[dashed, ultra thin] (-.9,-1.2) +(0:.6 and .18) arc (0:180:.6 and .18);
\draw[dashed, ultra thin] (.9,-1.2) +(0:.6 and .18) arc (0:180:.6 and .18);
\draw[dashed, ultra thin] (1.8,0) +(0:.6 and .18) arc (0:180:.6 and .18);
\draw (.9,1.2) +(0:.6 and .18) arc (0:180:.6 and .18);
\node at (.03,0) {$\scriptstyle V$};
}$};
\node (H) at (6,2) {$\tikzmath[scale=.4]{\fill[gray!50] (1.8,1.2) circle (.6 and .18);
\filldraw[fill=gray!20, line join=bevel](-.6,0) arc (-180:0:.6 and .18) arc (180:0:.3 and .18) arc (-180:0:.6 and .18) arc (180:0:.3 and .18) arc (-180:0:.6 and .18)[rounded corners=4.8]-- ++(-.1,.4) -- ++(-1.6,.4) [sharp corners]-- ++(-.1,.4) arc (0:-180:.6 and .18)[rounded corners=4.8]-- ++(-.1,-.4) -- ++(-1.6,-.4)  [sharp corners]-- cycle;
\draw[dashed, ultra thin] (0,0) +(0:.6 and .18) arc (0:180:.6 and .18);
\draw[dashed, ultra thin] (1.8,0) +(0:.6 and .18) arc (0:180:.6 and .18);
\draw[dashed, ultra thin] (3.6,0) +(0:.6 and .18) arc (0:180:.6 and .18);
\draw (1.8,1.2) +(0:.6 and .18) arc (0:180:.6 and .18);
\node at (1.8,.57) {$\scriptstyle V$};}
$};
\node (X) at (3,2) {$\tikzmath[scale=.4]{\fill[gray!50] (1.8,1.2) circle (.6 and .18);
\filldraw[fill=gray!20, line join=bevel](-.6,0) arc (-180:0:.6 and .18) arc (180:0:.3 and .18) arc (-180:0:.6 and .18) arc (180:0:.3 and .18) arc (-180:0:.6 and .18)[rounded corners=4.8]-- ++(-.1,.4) -- ++(-1.6,.4) [sharp corners]-- ++(-.1,.4) arc (0:-180:.6 and .18)[rounded corners=4.8]-- ++(-.1,-.4) -- ++(-1.6,-.4)  [sharp corners]-- cycle;
\draw[dashed, ultra thin] (0,0) +(0:.6 and .18) arc (0:180:.6 and .18);
\draw[dashed, ultra thin] (1.8,0) +(0:.6 and .18) arc (0:180:.6 and .18);
\draw[dashed, ultra thin] (3.6,0) +(0:.6 and .18) arc (0:180:.6 and .18);
\draw (1.8,1.2) +(0:.6 and .18) arc (0:180:.6 and .18);
\draw[line join=bevel] (.9,.18) -- (1.8,1.02) -- (2.7,.18);
\node at (.4,.28) {$\scriptstyle V$};
\node at (1.8,.28) {$\scriptstyle V$};
\node at (3.35,.28) {$\scriptstyle V$};}
$};
\draw (A) -- (B);
\draw (A) -- (D);
\draw (D) -- (F);
\draw (F) -- (G);
\draw (B) -- (X);
\draw (G) -- (H);
\draw (X) -- (H);
}
\medskip
\end{equation}

Let $X$ be an $(S^1\times \{0\})$-open soccer ball decomposition of $S^1\times[0,1]$.
We will prove that \eqref{7gon of 3-pants} is commutative after applying the invertible functor $X\triangleleft -$.
Namely, we will prove the commutativity of this diagram:
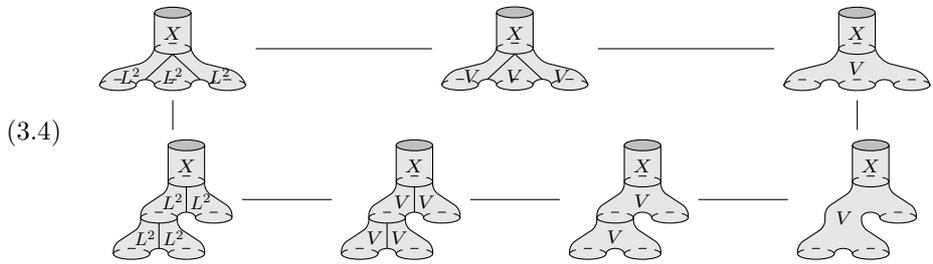
\begin{equation}\label{7gon of 3-pants'}
\tikzmath
{\pgftransformxscale{1.5}
\node (A) at (0,0) {$\tikzmath[scale=.4]{
\fill[gray!50] (.9,2.4) circle (.6 and .18);
\filldraw[fill=gray!20]
(-1.5,-1.2) arc (-180:0:.6 and .18) arc (180:0:.3 and .18) arc (-180:0:.6 and .18) 
[rounded corners=2.4]-- ++(0,.3) -- ++(-.9,.6) [sharp corners]-- ++(0,.3) 
arc (180:0:.3 and .18) arc (-180:0:.6 and .18)
[rounded corners=2.4]-- ++(0,.3) -- ++(-.9,.6) [sharp corners]-- ++(0,.3) 
-- ++(0,1.2) arc (0:-180:.6 and .18) -- ++(0,-1.2)
[rounded corners=2.4]-- ++(0,-.3) -- ++(-.9,-.6) [sharp corners]-- ++(0,-.3)
[rounded corners=2.4]-- ++(0,-.3) -- ++(-.9,-.6) [sharp corners]-- cycle;
\draw[dashed, ultra thin] (-.9,-1.2) +(0:.6 and .18) arc (0:180:.6 and .18);
\draw[dashed, ultra thin] (.9,-1.2) +(0:.6 and .18) arc (0:180:.6 and .18);
\draw[dashed, ultra thin] (0,0) +(0:.6 and .18) arc (0:180:.6 and .18);
\draw (0,0) +(0:.6 and .18) arc (0:-180:.6 and .18);
\draw[dashed, ultra thin] (1.8,0) +(0:.6 and .18) arc (0:180:.6 and .18);
\draw[dashed, ultra thin] (.9,1.2) +(0:.6 and .18) arc (0:180:.6 and .18);
\draw (.9,1.2) +(0:.6 and .18) arc (0:-180:.6 and .18);
\node at (.9,1.2+.5) {$\scriptstyle X$};
\draw (.9,2.4) +(0:.6 and .18) arc (0:180:.6 and .18);
\draw(.9,.18) -- (.9,1.02)(0,.18-1.2) -- (0,1.02-1.2);
\node[scale=.7] at (.45,.53) {$L^2$};
\node[scale=.7] at (1.38,.53) {$L^2$};
\node[scale=.7] at (.45-.9,.53-1.2) {$L^2$};
\node[scale=.7] at (1.38-.9,.53-1.2) {$L^2$};}
$};
\node (B) at (0,2) {$\tikzmath[scale=.4]{\fill[gray!50] (1.8,2.4) circle (.6 and .18);
\filldraw[fill=gray!20, line join=bevel](-.6,0) arc (-180:0:.6 and .18) arc (180:0:.3 and .18) arc (-180:0:.6 and .18) arc (180:0:.3 and .18) arc (-180:0:.6 and .18)[rounded corners=4.8]-- ++(-.1,.4) -- ++(-1.6,.4) [sharp corners]-- ++(-.1,.4) -- ++(0,1.2) arc (0:-180:.6 and .18) -- ++(0,-1.2)[rounded corners=4.8]-- ++(-.1,-.4) -- ++(-1.6,-.4)  [sharp corners]-- cycle;
\draw[dashed, ultra thin] (0,0) +(0:.6 and .18) arc (0:180:.6 and .18);
\draw[dashed, ultra thin] (1.8,0) +(0:.6 and .18) arc (0:180:.6 and .18);
\draw[dashed, ultra thin] (3.6,0) +(0:.6 and .18) arc (0:180:.6 and .18);
\draw[dashed, ultra thin] (1.8,1.2) +(0:.6 and .18) arc (0:180:.6 and .18);
\draw (1.8,2.4) +(0:.6 and .18) arc (0:180:.6 and .18);
\draw (1.8,1.2) +(0:.6 and .18) arc (0:-180:.6 and .18);
\node at (1.8,1.2+.5) {$\scriptstyle X$};
\draw[line join=bevel] (.9,.18) -- (1.8,1.02) -- (2.7,.18);
\node[scale=.7] at (.4,.3) {$L^2$};
\node[scale=.7] at (1.8,.3) {$L^2$};
\node[scale=.7] at (3.35,.3) {$L^2$};}
$};
\node (D) at (2,0) {$\tikzmath[scale=.4]{
\fill[gray!50] (.9,2.4) circle (.6 and .18);
\filldraw[fill=gray!20]
(-1.5,-1.2) arc (-180:0:.6 and .18) arc (180:0:.3 and .18) arc (-180:0:.6 and .18) 
[rounded corners=2.4]-- ++(0,.3) -- ++(-.9,.6) [sharp corners]-- ++(0,.3) 
arc (180:0:.3 and .18) arc (-180:0:.6 and .18)
[rounded corners=2.4]-- ++(0,.3) -- ++(-.9,.6) [sharp corners]-- ++(0,.3) 
-- ++(0,1.2) arc (0:-180:.6 and .18) -- ++(0,-1.2)[rounded corners=2.4]-- ++(0,-.3) -- ++(-.9,-.6) [sharp corners]-- ++(0,-.3)
[rounded corners=2.4]-- ++(0,-.3) -- ++(-.9,-.6) [sharp corners]-- cycle;
\draw[dashed, ultra thin] (-.9,-1.2) +(0:.6 and .18) arc (0:180:.6 and .18);
\draw[dashed, ultra thin] (.9,-1.2) +(0:.6 and .18) arc (0:180:.6 and .18);
\draw[dashed, ultra thin] (0,0) +(0:.6 and .18) arc (0:180:.6 and .18);
\draw (0,0) +(0:.6 and .18) arc (0:-180:.6 and .18);
\draw[dashed, ultra thin] (1.8,0) +(0:.6 and .18) arc (0:180:.6 and .18);
\draw[dashed, ultra thin] (.9,1.2) +(0:.6 and .18) arc (0:180:.6 and .18);
\draw (.9,2.4) +(0:.6 and .18) arc (0:180:.6 and .18);
\draw (.9,1.2) +(0:.6 and .18) arc (0:-180:.6 and .18);
\node at (.9,1.2+.5) {$\scriptstyle X$};
\draw(.9,.18) -- (.9,1.02)(0,.18-1.2) -- (0,1.02-1.2);
\node at (.5,.53) {$\scriptstyle V$};
\node at (1.3,.53) {$\scriptstyle V$};
\node at (.5-.9,.53-1.2) {$\scriptstyle V$};
\node at (1.3-.9,.53-1.2) {$\scriptstyle V$};}
$};
\node (F) at (4,0) {$\tikzmath[scale=.4]{
\fill[gray!50] (.9,2.4) circle (.6 and .18);
\filldraw[fill=gray!20]
(-1.5,-1.2) arc (-180:0:.6 and .18) arc (180:0:.3 and .18) arc (-180:0:.6 and .18) 
[rounded corners=2.4]-- ++(0,.3) -- ++(-.9,.6) [sharp corners]-- ++(0,.3) 
arc (180:0:.3 and .18) arc (-180:0:.6 and .18)
[rounded corners=2.4]-- ++(0,.3) -- ++(-.9,.6) [sharp corners]-- ++(0,.3) 
-- ++(0,1.2) arc (0:-180:.6 and .18) -- ++(0,-1.2)[rounded corners=2.4]-- ++(0,-.3) -- ++(-.9,-.6) [sharp corners]-- ++(0,-.3)
[rounded corners=2.4]-- ++(0,-.3) -- ++(-.9,-.6) [sharp corners]-- cycle;
\draw[dashed, ultra thin] (-.9,-1.2) +(0:.6 and .18) arc (0:180:.6 and .18);
\draw[dashed, ultra thin] (.9,-1.2) +(0:.6 and .18) arc (0:180:.6 and .18);
\draw[dashed, ultra thin] (0,0) +(0:.6 and .18) arc (0:180:.6 and .18);
\draw (0,0) +(0:.6 and .18) arc (0:-180:.6 and .18);
\draw[dashed, ultra thin] (1.8,0) +(0:.6 and .18) arc (0:180:.6 and .18);
\draw[dashed, ultra thin] (.9,1.2) +(0:.6 and .18) arc (0:180:.6 and .18);
\draw (.9,2.4) +(0:.6 and .18) arc (0:180:.6 and .18);
\draw (.9,1.2) +(0:.6 and .18) arc (0:-180:.6 and .18);
\node at (.9,1.2+.5) {$\scriptstyle X$};
\node at (.9,.57) {$\scriptstyle V$};
\node at (0,.57-1.2) {$\scriptstyle V$};}
$};
\node (G) at (6,0) {$\tikzmath[scale=.4]{
\fill[gray!50] (.9,2.4) circle (.6 and .18);
\filldraw[fill=gray!20]
(-1.5,-1.2) arc (-180:0:.6 and .18) arc (180:0:.3 and .18) arc (-180:0:.6 and .18) 
[rounded corners=2.4]-- ++(0,.3) -- ++(-.9,.6) [sharp corners]-- ++(0,.3) 
arc (180:0:.3 and .18) arc (-180:0:.6 and .18)
[rounded corners=2.4]-- ++(0,.3) -- ++(-.9,.6) [sharp corners]-- ++(0,.3) 
-- ++(0,1.2) arc (0:-180:.6 and .18) -- ++(0,-1.2)[rounded corners=2.4]-- ++(0,-.3) -- ++(-.9,-.6) [sharp corners]-- ++(0,-.3)
[rounded corners=2.4]-- ++(0,-.3) -- ++(-.9,-.6) [sharp corners]-- cycle;
\draw[dashed, ultra thin] (-.9,-1.2) +(0:.6 and .18) arc (0:180:.6 and .18);
\draw[dashed, ultra thin] (.9,-1.2) +(0:.6 and .18) arc (0:180:.6 and .18);
\draw[dashed, ultra thin] (1.8,0) +(0:.6 and .18) arc (0:180:.6 and .18);
\draw[dashed, ultra thin] (.9,1.2) +(0:.6 and .18) arc (0:180:.6 and .18);
\draw (.9,2.4) +(0:.6 and .18) arc (0:180:.6 and .18);
\draw (.9,1.2) +(0:.6 and .18) arc (0:-180:.6 and .18);
\node at (.9,1.2+.5) {$\scriptstyle X$};
\node at (.03,0) {$\scriptstyle V$};
}$};
\node (H) at (6,2) {$\tikzmath[scale=.4]{\fill[gray!50] (1.8,2.4) circle (.6 and .18);
\filldraw[fill=gray!20, line join=bevel](-.6,0) arc (-180:0:.6 and .18) arc (180:0:.3 and .18) arc (-180:0:.6 and .18) arc (180:0:.3 and .18) arc (-180:0:.6 and .18)[rounded corners=4.8]-- ++(-.1,.4) -- ++(-1.6,.4) [sharp corners]-- ++(-.1,.4) -- ++(0,1.2) arc (0:-180:.6 and .18) -- ++(0,-1.2)[rounded corners=4.8]-- ++(-.1,-.4) -- ++(-1.6,-.4)  [sharp corners]-- cycle;
\draw[dashed, ultra thin] (0,0) +(0:.6 and .18) arc (0:180:.6 and .18);
\draw[dashed, ultra thin] (1.8,0) +(0:.6 and .18) arc (0:180:.6 and .18);
\draw[dashed, ultra thin] (3.6,0) +(0:.6 and .18) arc (0:180:.6 and .18);
\draw[dashed, ultra thin] (1.8,1.2) +(0:.6 and .18) arc (0:180:.6 and .18);
\draw (1.8,2.4) +(0:.6 and .18) arc (0:180:.6 and .18);
\draw (1.8,1.2) +(0:.6 and .18) arc (0:-180:.6 and .18);
\node at (1.8,1.2+.5) {$\scriptstyle X$};
\node at (1.8,.57) {$\scriptstyle V$};}
$};
\node (X) at (3,2) {$\tikzmath[scale=.4]{\fill[gray!50] (1.8,2.4) circle (.6 and .18);
\filldraw[fill=gray!20, line join=bevel](-.6,0) arc (-180:0:.6 and .18) arc (180:0:.3 and .18) arc (-180:0:.6 and .18) arc (180:0:.3 and .18) arc (-180:0:.6 and .18)[rounded corners=4.8]-- ++(-.1,.4) -- ++(-1.6,.4) [sharp corners]-- ++(-.1,.4) -- ++(0,1.2) arc (0:-180:.6 and .18) -- ++(0,-1.2)[rounded corners=4.8]-- ++(-.1,-.4) -- ++(-1.6,-.4)  [sharp corners]-- cycle;
\draw[dashed, ultra thin] (0,0) +(0:.6 and .18) arc (0:180:.6 and .18);
\draw[dashed, ultra thin] (1.8,0) +(0:.6 and .18) arc (0:180:.6 and .18);
\draw[dashed, ultra thin] (3.6,0) +(0:.6 and .18) arc (0:180:.6 and .18);
\draw[dashed, ultra thin] (1.8,1.2) +(0:.6 and .18) arc (0:180:.6 and .18);
\draw (1.8,2.4) +(0:.6 and .18) arc (0:180:.6 and .18);
\draw (1.8,1.2) +(0:.6 and .18) arc (0:-180:.6 and .18);
\node at (1.8,1.2+.5) {$\scriptstyle X$};
\draw[line join=bevel] (.9,.18) -- (1.8,1.02) -- (2.7,.18);
\node at (.4,.28) {$\scriptstyle V$};
\node at (1.8,.28) {$\scriptstyle V$};
\node at (3.35,.28) {$\scriptstyle V$};}
$};
\draw (A) -- (B);
\draw (A) -- (D);
\draw (D) -- (F);
\draw (F) -- (G);
\draw (B) -- (X);
\draw (G) -- (H);
\draw (X) -- (H);
}
\medskip
\end{equation}
Note that the above pictures are somewhat misleading: these constructions involve not just a single Connes fusion over the bottom circle of the annulus labeled $X$, but rather many fusions, each one of which uses at most an interval in this circle.

By choosing $X$ carefully,
we can arrange that the union of $X$ and of a soccer 
ball decomposition of $S^1\times [0,1]$ (which again needs to be picked carefully, namely so that the left half of the top $S^1$ is covered by a single edge) is an
open soccer ball decomposition, call it $Y$, of the manifold 
\[
\tikzmath[scale=.4]{
\fill[gray!50] (.9,2.4) circle (.6 and .18);
\fill[gray!50] (.955,1.02) arc (0:-180:.06 and .84);
\filldraw[fill=gray!20] (.9,.18) arc (90:0:.3 and .18) arc (-180:0:.6 and .18)
[rounded corners=2.4]-- ++(0,.3) -- ++(-.9,.6) [sharp corners]-- ++(0,.3) 
-- ++(0,1.2) arc (0:-180:.6 and .18) -- ++(0,-1.2) arc (-180:-85:.6 and .18) arc (0:-90:.06 and .84) coordinate (x);
\draw (x) arc (-90:-180:.06 and .84);
\draw[dashed, very thin] (1.8,0) +(0:.6 and .18) arc (0:180:.6 and .18);
\draw[dashed, very thin] (.9,1.2) +(180:.6 and .18) arc (180:95:.6 and .18) -- ++(0,-2*.18);
\draw (.9,2.4) +(0:.6 and .18) arc (0:180:.6 and .18);}
=
(S^1\times[0,1])\sqcup(S^1\times[0,1])/\sim,
\]
where the equivalence relation identifies $(z,0)$ in the first copy of $S^1\times[0,1]$
with $(z,1)$ in the second copy of $S^1\times[0,1]$, for every $z\in S^1_\dashv$.

With these preliminaries in place, the commutativity of \eqref{7gon of 3-pants'} follows from this commutative diagram:
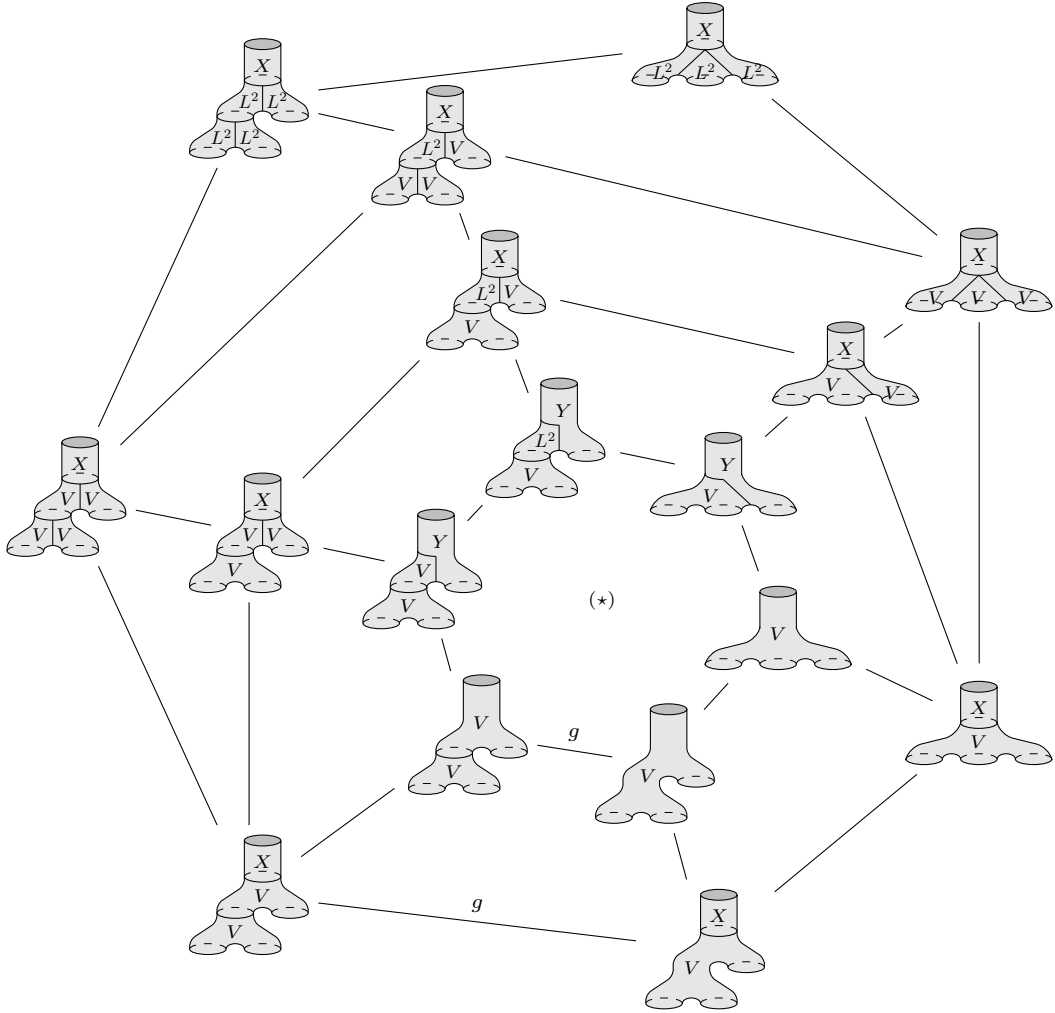
\[
\hspace{-.85cm}
\tikzmath
{
\pgftransformscale{1.2}
\node (A) at (2,9.4) {$\tikzmath[scale=.4]{
\fill[gray!50] (.9,2.4) circle (.6 and .18);
\filldraw[fill=gray!20]
(-1.5,-1.2) arc (-180:0:.6 and .18) arc (180:0:.3 and .18) arc (-180:0:.6 and .18) 
[rounded corners=2.4]-- ++(0,.3) -- ++(-.9,.6) [sharp corners]-- ++(0,.3) 
arc (180:0:.3 and .18) arc (-180:0:.6 and .18)
[rounded corners=2.4]-- ++(0,.3) -- ++(-.9,.6) [sharp corners]-- ++(0,.3) 
-- ++(0,1.2) arc (0:-180:.6 and .18) -- ++(0,-1.2)
[rounded corners=2.4]-- ++(0,-.3) -- ++(-.9,-.6) [sharp corners]-- ++(0,-.3)
[rounded corners=2.4]-- ++(0,-.3) -- ++(-.9,-.6) [sharp corners]-- cycle;
\draw[dashed, ultra thin] (-.9,-1.2) +(0:.6 and .18) arc (0:180:.6 and .18);
\draw[dashed, ultra thin] (.9,-1.2) +(0:.6 and .18) arc (0:180:.6 and .18);
\draw[dashed, ultra thin] (0,0) +(0:.6 and .18) arc (0:180:.6 and .18);
\draw (0,0) +(0:.6 and .18) arc (0:-180:.6 and .18);
\draw[dashed, ultra thin] (1.8,0) +(0:.6 and .18) arc (0:180:.6 and .18);
\draw[dashed, ultra thin] (.9,1.2) +(0:.6 and .18) arc (0:180:.6 and .18);
\draw (.9,1.2) +(0:.6 and .18) arc (0:-180:.6 and .18);
\node at (.9,1.2+.5) {$\scriptstyle X$};
\draw (.9,2.4) +(0:.6 and .18) arc (0:180:.6 and .18);
\draw(.9,.18) -- (.9,1.02)(0,.18-1.2) -- (0,1.02-1.2);
\node[scale=.7] at (.45,.53) {$L^2$};
\node[scale=.7] at (1.38,.53) {$L^2$};
\node[scale=.7] at (.45-.9,.53-1.2) {$L^2$};
\node[scale=.7] at (1.38-.9,.53-1.2) {$L^2$};}
$};
\node (B) at (7,10) {$\tikzmath[scale=.4]{\fill[gray!50] (1.8,2.4) circle (.6 and .18);
\filldraw[fill=gray!20, line join=bevel](-.6,0) arc (-180:0:.6 and .18) arc (180:0:.3 and .18) arc (-180:0:.6 and .18) arc (180:0:.3 and .18) arc (-180:0:.6 and .18)[rounded corners=4.8]-- ++(-.1,.4) -- ++(-1.6,.4) [sharp corners]-- ++(-.1,.4) -- ++(0,1.2) arc (0:-180:.6 and .18) -- ++(0,-1.2)[rounded corners=4.8]-- ++(-.1,-.4) -- ++(-1.6,-.4)  [sharp corners]-- cycle;
\draw[dashed, ultra thin] (0,0) +(0:.6 and .18) arc (0:180:.6 and .18);
\draw[dashed, ultra thin] (1.8,0) +(0:.6 and .18) arc (0:180:.6 and .18);
\draw[dashed, ultra thin] (3.6,0) +(0:.6 and .18) arc (0:180:.6 and .18);
\draw[dashed, ultra thin] (1.8,1.2) +(0:.6 and .18) arc (0:180:.6 and .18);
\draw (1.8,2.4) +(0:.6 and .18) arc (0:180:.6 and .18);
\draw (1.8,1.2) +(0:.6 and .18) arc (0:-180:.6 and .18);
\node at (1.8,1.2+.5) {$\scriptstyle X$};
\draw[line join=bevel] (.9,.18) -- (1.8,1.02) -- (2.7,.18);
\node[scale=.7] at (.4,.3) {$L^2$};
\node[scale=.7] at (1.8,.3) {$L^2$};
\node[scale=.7] at (3.35,.3) {$L^2$};}
$};
\node (C) at (4,8.95) {$\tikzmath[scale=.4]{\useasboundingbox (-1.5,-1.4) rectangle (2.4,3);
\fill[gray!50] (.9,2.4) circle (.6 and .18);
\filldraw[fill=gray!20]
(-1.5,-1.2) arc (-180:0:.6 and .18) arc (180:0:.3 and .18) arc (-180:0:.6 and .18) 
[rounded corners=2.4]-- ++(0,.3) -- ++(-.9,.6) [sharp corners]-- ++(0,.3) 
arc (180:0:.3 and .18) arc (-180:0:.6 and .18)
[rounded corners=2.4]-- ++(0,.3) -- ++(-.9,.6) [sharp corners]-- ++(0,.3) 
-- ++(0,1.2) arc (0:-180:.6 and .18) -- ++(0,-1.2)[rounded corners=2.4]-- ++(0,-.3) -- ++(-.9,-.6) [sharp corners]-- ++(0,-.3)
[rounded corners=2.4]-- ++(0,-.3) -- ++(-.9,-.6) [sharp corners]-- cycle;
\draw[dashed, ultra thin] (-.9,-1.2) +(0:.6 and .18) arc (0:180:.6 and .18);
\draw[dashed, ultra thin] (.9,-1.2) +(0:.6 and .18) arc (0:180:.6 and .18);
\draw[dashed, ultra thin] (0,0) +(0:.6 and .18) arc (0:180:.6 and .18);
\draw (0,0) +(0:.6 and .18) arc (0:-180:.6 and .18);
\draw[dashed, ultra thin] (1.8,0) +(0:.6 and .18) arc (0:180:.6 and .18);
\draw[dashed, ultra thin] (.9,1.2) +(0:.6 and .18) arc (0:180:.6 and .18);
\draw (.9,2.4) +(0:.6 and .18) arc (0:180:.6 and .18);
\draw (.9,1.2) +(0:.6 and .18) arc (0:-180:.6 and .18);
\node at (.9,1.2+.5) {$\scriptstyle X$};
\draw(.9,.18) -- (.9,1.02)(0,.18-1.2) -- (0,1.02-1.2);
\node at (1.3,.53) {$\scriptstyle V$};
\node at (.5-.9,.53-1.2) {$\scriptstyle V$};
\node at (1.3-.9,.53-1.2) {$\scriptstyle V$};
\node[scale=.7] at (.45,.53) {$L^2$};}
$};
\node (D) at (0,5) {$\tikzmath[scale=.4]{
\fill[gray!50] (.9,2.4) circle (.6 and .18);
\filldraw[fill=gray!20]
(-1.5,-1.2) arc (-180:0:.6 and .18) arc (180:0:.3 and .18) arc (-180:0:.6 and .18) 
[rounded corners=2.4]-- ++(0,.3) -- ++(-.9,.6) [sharp corners]-- ++(0,.3) 
arc (180:0:.3 and .18) arc (-180:0:.6 and .18)
[rounded corners=2.4]-- ++(0,.3) -- ++(-.9,.6) [sharp corners]-- ++(0,.3) 
-- ++(0,1.2) arc (0:-180:.6 and .18) -- ++(0,-1.2)[rounded corners=2.4]-- ++(0,-.3) -- ++(-.9,-.6) [sharp corners]-- ++(0,-.3)
[rounded corners=2.4]-- ++(0,-.3) -- ++(-.9,-.6) [sharp corners]-- cycle;
\draw[dashed, ultra thin] (-.9,-1.2) +(0:.6 and .18) arc (0:180:.6 and .18);
\draw[dashed, ultra thin] (.9,-1.2) +(0:.6 and .18) arc (0:180:.6 and .18);
\draw[dashed, ultra thin] (0,0) +(0:.6 and .18) arc (0:180:.6 and .18);
\draw (0,0) +(0:.6 and .18) arc (0:-180:.6 and .18);
\draw[dashed, ultra thin] (1.8,0) +(0:.6 and .18) arc (0:180:.6 and .18);
\draw[dashed, ultra thin] (.9,1.2) +(0:.6 and .18) arc (0:180:.6 and .18);
\draw (.9,2.4) +(0:.6 and .18) arc (0:180:.6 and .18);
\draw (.9,1.2) +(0:.6 and .18) arc (0:-180:.6 and .18);
\node at (.9,1.2+.5) {$\scriptstyle X$};
\draw(.9,.18) -- (.9,1.02)(0,.18-1.2) -- (0,1.02-1.2);
\node at (.5,.53) {$\scriptstyle V$};
\node at (1.3,.53) {$\scriptstyle V$};
\node at (.5-.9,.53-1.2) {$\scriptstyle V$};
\node at (1.3-.9,.53-1.2) {$\scriptstyle V$};}
$};
\node (E) at (4.6,7.35) {$\tikzmath[scale=.4]{\useasboundingbox (-1.5,-1.4) rectangle (2.4,3);
\fill[gray!50] (.9,2.4) circle (.6 and .18);
\filldraw[fill=gray!20]
(-1.5,-1.2) arc (-180:0:.6 and .18) arc (180:0:.3 and .18) arc (-180:0:.6 and .18) 
[rounded corners=2.4]-- ++(0,.3) -- ++(-.9,.6) [sharp corners]-- ++(0,.3) 
arc (180:0:.3 and .18) arc (-180:0:.6 and .18)
[rounded corners=2.4]-- ++(0,.3) -- ++(-.9,.6) [sharp corners]-- ++(0,.3) 
-- ++(0,1.2) arc (0:-180:.6 and .18) -- ++(0,-1.2)[rounded corners=2.4]-- ++(0,-.3) -- ++(-.9,-.6) [sharp corners]-- ++(0,-.3)
[rounded corners=2.4]-- ++(0,-.3) -- ++(-.9,-.6) [sharp corners]-- cycle;
\draw[dashed, ultra thin] (-.9,-1.2) +(0:.6 and .18) arc (0:180:.6 and .18);
\draw[dashed, ultra thin] (.9,-1.2) +(0:.6 and .18) arc (0:180:.6 and .18);
\draw[dashed, ultra thin] (0,0) +(0:.6 and .18) arc (0:180:.6 and .18);
\draw (0,0) +(0:.6 and .18) arc (0:-180:.6 and .18);
\draw[dashed, ultra thin] (1.8,0) +(0:.6 and .18) arc (0:180:.6 and .18);
\draw[dashed, ultra thin] (.9,1.2) +(0:.6 and .18) arc (0:180:.6 and .18);
\draw (.9,2.4) +(0:.6 and .18) arc (0:180:.6 and .18);
\draw (.9,1.2) +(0:.6 and .18) arc (0:-180:.6 and .18);
\node at (.9,1.2+.5) {$\scriptstyle X$};
\draw(.9,.18) -- (.9,1.02);
\node[scale=.7] at (.45,.53) {$L^2$};
\node at (1.3,.53) {$\scriptstyle V$};
\node at (0,.57-1.2) {$\scriptstyle V$};}
$};
\node (F) at (10,7.5) {$\tikzmath[scale=.4]{\fill[gray!50] (1.8,2.4) circle (.6 and .18);
\filldraw[fill=gray!20, line join=bevel](-.6,0) arc (-180:0:.6 and .18) arc (180:0:.3 and .18) arc (-180:0:.6 and .18) arc (180:0:.3 and .18) arc (-180:0:.6 and .18)[rounded corners=4.8]-- ++(-.1,.4) -- ++(-1.6,.4) [sharp corners]-- ++(-.1,.4) -- ++(0,1.2) arc (0:-180:.6 and .18) -- ++(0,-1.2)[rounded corners=4.8]-- ++(-.1,-.4) -- ++(-1.6,-.4)  [sharp corners]-- cycle;
\draw[dashed, ultra thin] (0,0) +(0:.6 and .18) arc (0:180:.6 and .18);
\draw[dashed, ultra thin] (1.8,0) +(0:.6 and .18) arc (0:180:.6 and .18);
\draw[dashed, ultra thin] (3.6,0) +(0:.6 and .18) arc (0:180:.6 and .18);
\draw[dashed, ultra thin] (1.8,1.2) +(0:.6 and .18) arc (0:180:.6 and .18);
\draw (1.8,2.4) +(0:.6 and .18) arc (0:180:.6 and .18);
\draw (1.8,1.2) +(0:.6 and .18) arc (0:-180:.6 and .18);
\node at (1.8,1.2+.5) {$\scriptstyle X$};
\draw[line join=bevel] (.9,.18) -- (1.8,1.02) -- (2.7,.18);
\node at (.4,.28) {$\scriptstyle V$};
\node at (1.8,.28) {$\scriptstyle V$};
\node at (3.35,.28) {$\scriptstyle V$};}
$};
\node (G) at (2,4.6) {$\tikzmath[scale=.4]{
\fill[gray!50] (.9,2.4) circle (.6 and .18);
\filldraw[fill=gray!20]
(-1.5,-1.2) arc (-180:0:.6 and .18) arc (180:0:.3 and .18) arc (-180:0:.6 and .18) 
[rounded corners=2.4]-- ++(0,.3) -- ++(-.9,.6) [sharp corners]-- ++(0,.3) 
arc (180:0:.3 and .18) arc (-180:0:.6 and .18)
[rounded corners=2.4]-- ++(0,.3) -- ++(-.9,.6) [sharp corners]-- ++(0,.3) 
-- ++(0,1.2) arc (0:-180:.6 and .18) -- ++(0,-1.2)[rounded corners=2.4]-- ++(0,-.3) -- ++(-.9,-.6) [sharp corners]-- ++(0,-.3)
[rounded corners=2.4]-- ++(0,-.3) -- ++(-.9,-.6) [sharp corners]-- cycle;
\draw[dashed, ultra thin] (-.9,-1.2) +(0:.6 and .18) arc (0:180:.6 and .18);
\draw[dashed, ultra thin] (.9,-1.2) +(0:.6 and .18) arc (0:180:.6 and .18);
\draw[dashed, ultra thin] (0,0) +(0:.6 and .18) arc (0:180:.6 and .18);
\draw (0,0) +(0:.6 and .18) arc (0:-180:.6 and .18);
\draw[dashed, ultra thin] (1.8,0) +(0:.6 and .18) arc (0:180:.6 and .18);
\draw[dashed, ultra thin] (.9,1.2) +(0:.6 and .18) arc (0:180:.6 and .18);
\draw (.9,2.4) +(0:.6 and .18) arc (0:180:.6 and .18);
\draw (.9,1.2) +(0:.6 and .18) arc (0:-180:.6 and .18);
\node at (.9,1.2+.5) {$\scriptstyle X$};
\draw(.9,.18) -- (.9,1.02);
\node at (.5,.53) {$\scriptstyle V$};
\node at (1.3,.53) {$\scriptstyle V$};
\node at (0,.57-1.2) {$\scriptstyle V$};}
$};
\node (H) at (3.9,4.2) {$\tikzmath[scale=.4]{
\fill[gray!50] (.9,2.4) circle (.6 and .18);
\filldraw[fill=gray!20]
(-1.5,-1.2) arc (-180:0:.6 and .18) arc (180:0:.3 and .18) arc (-180:0:.6 and .18) 
[rounded corners=2.4]-- ++(0,.3) -- ++(-.9,.6) [sharp corners]-- ++(0,.3) 
arc (180:0:.3 and .18) arc (-180:0:.6 and .18)
[rounded corners=2.4]-- ++(0,.3) -- ++(-.9,.6) [sharp corners]-- ++(0,.3) 
-- ++(0,1.2) arc (0:-180:.6 and .18) -- ++(0,-1.2)[rounded corners=2.4]-- ++(0,-.3) -- ++(-.9,-.6) [sharp corners]-- ++(0,-.3)
[rounded corners=2.4]-- ++(0,-.3) -- ++(-.9,-.6) [sharp corners]-- cycle;
\draw[dashed, ultra thin] (-.9,-1.2) +(0:.6 and .18) arc (0:180:.6 and .18);
\draw[dashed, ultra thin] (.9,-1.2) +(0:.6 and .18) arc (0:180:.6 and .18);
\draw[dashed, ultra thin] (0,0) +(0:.6 and .18) arc (0:180:.6 and .18);
\draw (0,0) +(0:.6 and .18) arc (0:-180:.6 and .18);
\draw[dashed, ultra thin] (1.8,0) +(0:.6 and .18) arc (0:180:.6 and .18);
\draw (.9,2.4) +(0:.6 and .18) arc (0:180:.6 and .18);
\draw (.9,1.2) +(-180:.6 and .18) arc (-180:-90:.6 and .18);
\node at (1.05,1.5) {$\scriptstyle Y$};
\draw(.9,.18) -- (.9,1.02);
\node at (.5,.53) {$\scriptstyle V$};
\node at (0,.57-1.2) {$\scriptstyle V$};}
$};
\node (I) at (5.25,5.65) {$\tikzmath[scale=.4]{
\fill[gray!50] (.9,2.4) circle (.6 and .18);
\filldraw[fill=gray!20]
(-1.5,-1.2) arc (-180:0:.6 and .18) arc (180:0:.3 and .18) arc (-180:0:.6 and .18) 
[rounded corners=2.4]-- ++(0,.3) -- ++(-.9,.6) [sharp corners]-- ++(0,.3) 
arc (180:0:.3 and .18) arc (-180:0:.6 and .18)
[rounded corners=2.4]-- ++(0,.3) -- ++(-.9,.6) [sharp corners]-- ++(0,.3) 
-- ++(0,1.2) arc (0:-180:.6 and .18) -- ++(0,-1.2)[rounded corners=2.4]-- ++(0,-.3) -- ++(-.9,-.6) [sharp corners]-- ++(0,-.3)
[rounded corners=2.4]-- ++(0,-.3) -- ++(-.9,-.6) [sharp corners]-- cycle;
\draw[dashed, ultra thin] (-.9,-1.2) +(0:.6 and .18) arc (0:180:.6 and .18);
\draw[dashed, ultra thin] (.9,-1.2) +(0:.6 and .18) arc (0:180:.6 and .18);
\draw[dashed, ultra thin] (0,0) +(0:.6 and .18) arc (0:180:.6 and .18);
\draw (0,0) +(0:.6 and .18) arc (0:-180:.6 and .18);
\draw[dashed, ultra thin] (1.8,0) +(0:.6 and .18) arc (0:180:.6 and .18);
\draw (.9,2.4) +(0:.6 and .18) arc (0:180:.6 and .18);
\draw (.9,1.2) +(-180:.6 and .18) arc (-180:-90:.6 and .18);
\node at (1.05,1.5) {$\scriptstyle Y$};
\draw(.9,.18) -- (.9,1.02);
\node[scale=.7] at (.45,.53) {$L^2$};
\node at (0,.57-1.2) {$\scriptstyle V$};}
$};
\node (J) at (7.2,5.25) {$\tikzmath[scale=.4]{\fill[gray!50] (1.8,2.4) circle (.6 and .18);
\filldraw[fill=gray!20, line join=bevel](-.6,0) arc (-180:0:.6 and .18) arc (180:0:.3 and .18) arc (-180:0:.6 and .18) arc (180:0:.3 and .18) arc (-180:0:.6 and .18)[rounded corners=4.8]-- ++(-.1,.4) -- ++(-1.6,.4) [sharp corners]-- ++(-.1,.4) -- ++(0,1.2) arc (0:-180:.6 and .18) -- ++(0,-1.2)[rounded corners=4.8]-- ++(-.1,-.4) -- ++(-1.6,-.4)  [sharp corners]-- cycle;
\draw[dashed, ultra thin] (0,0) +(0:.6 and .18) arc (0:180:.6 and .18);
\draw[dashed, ultra thin] (1.8,0) +(0:.6 and .18) arc (0:180:.6 and .18);
\draw[dashed, ultra thin] (3.6,0) +(0:.6 and .18) arc (0:180:.6 and .18);
\draw (1.8,2.4) +(0:.6 and .18) arc (0:180:.6 and .18);
\draw (1.8,1.2) +(-180:.6 and .18) arc (-180:-90:.6 and .18);
\node at (1.95,1.5) {$\scriptstyle Y$};
\draw (1.8,1.02) -- (2.7,.18);
\node at (1.4,.5) {$\scriptstyle V$};}
$};
\node (K) at (8.54,6.46) {$\tikzmath[scale=.4]{\fill[gray!50] (1.8,2.4) circle (.6 and .18);
\filldraw[fill=gray!20, line join=bevel](-.6,0) arc (-180:0:.6 and .18) arc (180:0:.3 and .18) arc (-180:0:.6 and .18) arc (180:0:.3 and .18) arc (-180:0:.6 and .18)[rounded corners=4.8]-- ++(-.1,.4) -- ++(-1.6,.4) [sharp corners]-- ++(-.1,.4) -- ++(0,1.2) arc (0:-180:.6 and .18) -- ++(0,-1.2)[rounded corners=4.8]-- ++(-.1,-.4) -- ++(-1.6,-.4)  [sharp corners]-- cycle;
\draw[dashed, ultra thin] (0,0) +(0:.6 and .18) arc (0:180:.6 and .18);
\draw[dashed, ultra thin] (1.8,0) +(0:.6 and .18) arc (0:180:.6 and .18);
\draw[dashed, ultra thin] (3.6,0) +(0:.6 and .18) arc (0:180:.6 and .18);
\draw[dashed, ultra thin] (1.8,1.2) +(0:.6 and .18) arc (0:180:.6 and .18);
\draw (1.8,2.4) +(0:.6 and .18) arc (0:180:.6 and .18);
\draw (1.8,1.2) +(0:.6 and .18) arc (0:-180:.6 and .18);
\node at (1.8,1.2+.5) {$\scriptstyle X$};
\draw (1.8,1.02) -- (2.7,.18);
\node at (1.4,.5) {$\scriptstyle V$};
\node at (3.35,.25) {$\scriptstyle V$};}
$};
\node (L) at (2,.6) {$\tikzmath[scale=.4]{
\fill[gray!50] (.9,2.4) circle (.6 and .18);
\filldraw[fill=gray!20]
(-1.5,-1.2) arc (-180:0:.6 and .18) arc (180:0:.3 and .18) arc (-180:0:.6 and .18) 
[rounded corners=2.4]-- ++(0,.3) -- ++(-.9,.6) [sharp corners]-- ++(0,.3) 
arc (180:0:.3 and .18) arc (-180:0:.6 and .18)
[rounded corners=2.4]-- ++(0,.3) -- ++(-.9,.6) [sharp corners]-- ++(0,.3) 
-- ++(0,1.2) arc (0:-180:.6 and .18) -- ++(0,-1.2)[rounded corners=2.4]-- ++(0,-.3) -- ++(-.9,-.6) [sharp corners]-- ++(0,-.3)
[rounded corners=2.4]-- ++(0,-.3) -- ++(-.9,-.6) [sharp corners]-- cycle;
\draw[dashed, ultra thin] (-.9,-1.2) +(0:.6 and .18) arc (0:180:.6 and .18);
\draw[dashed, ultra thin] (.9,-1.2) +(0:.6 and .18) arc (0:180:.6 and .18);
\draw[dashed, ultra thin] (0,0) +(0:.6 and .18) arc (0:180:.6 and .18);
\draw (0,0) +(0:.6 and .18) arc (0:-180:.6 and .18);
\draw[dashed, ultra thin] (1.8,0) +(0:.6 and .18) arc (0:180:.6 and .18);
\draw[dashed, ultra thin] (.9,1.2) +(0:.6 and .18) arc (0:180:.6 and .18);
\draw (.9,2.4) +(0:.6 and .18) arc (0:180:.6 and .18);
\draw (.9,1.2) +(0:.6 and .18) arc (0:-180:.6 and .18);
\node at (.9,1.2+.5) {$\scriptstyle X$};
\node at (.9,.57) {$\scriptstyle V$};
\node at (0,.57-1.2) {$\scriptstyle V$};}
$};
\node[inner ysep=1] (M) at (4.4,2.37) {$\tikzmath[scale=.4]{
\fill[gray!50] (.9,2.4) circle (.6 and .18);
\filldraw[fill=gray!20]
(-1.5,-1.2) arc (-180:0:.6 and .18) arc (180:0:.3 and .18) arc (-180:0:.6 and .18) 
[rounded corners=2.4]-- ++(0,.3) -- ++(-.9,.6) [sharp corners]-- ++(0,.3) 
arc (180:0:.3 and .18) arc (-180:0:.6 and .18)
[rounded corners=2.4]-- ++(0,.3) -- ++(-.9,.6) [sharp corners]-- ++(0,.3) 
-- ++(0,1.2) arc (0:-180:.6 and .18) -- ++(0,-1.2)[rounded corners=2.4]-- ++(0,-.3) -- ++(-.9,-.6) [sharp corners]-- ++(0,-.3)
[rounded corners=2.4]-- ++(0,-.3) -- ++(-.9,-.6) [sharp corners]-- cycle;
\draw[dashed, ultra thin] (-.9,-1.2) +(0:.6 and .18) arc (0:180:.6 and .18);
\draw[dashed, ultra thin] (.9,-1.2) +(0:.6 and .18) arc (0:180:.6 and .18);
\draw[dashed, ultra thin] (0,0) +(0:.6 and .18) arc (0:180:.6 and .18);
\draw (0,0) +(0:.6 and .18) arc (0:-180:.6 and .18);
\draw[dashed, ultra thin] (1.8,0) +(0:.6 and .18) arc (0:180:.6 and .18);
\draw (.9,2.4) +(0:.6 and .18) arc (0:180:.6 and .18);
\node at (.9,1) {$\scriptstyle V$};
\node at (0,.57-1.2) {$\scriptstyle V$};}
$};
\node (N) at (6.45,2.05) {$\tikzmath[scale=.4]{
\fill[gray!50] (.9,2.4) circle (.6 and .18);
\filldraw[fill=gray!20]
(-1.5,-1.2) arc (-180:0:.6 and .18) arc (180:0:.3 and .18) arc (-180:0:.6 and .18) 
[rounded corners=2.4]-- ++(0,.3) -- ++(-.9,.6) [sharp corners]-- ++(0,.3) 
arc (180:0:.3 and .18) arc (-180:0:.6 and .18)
[rounded corners=2.4]-- ++(0,.3) -- ++(-.9,.6) [sharp corners]-- ++(0,.3) 
-- ++(0,1.2) arc (0:-180:.6 and .18) -- ++(0,-1.2)[rounded corners=2.4]-- ++(0,-.3) -- ++(-.9,-.6) [sharp corners]-- ++(0,-.3)
[rounded corners=2.4]-- ++(0,-.3) -- ++(-.9,-.6) [sharp corners]-- cycle;
\draw[dashed, ultra thin] (-.9,-1.2) +(0:.6 and .18) arc (0:180:.6 and .18);
\draw[dashed, ultra thin] (.9,-1.2) +(0:.6 and .18) arc (0:180:.6 and .18);
\draw[dashed, ultra thin] (1.8,0) +(0:.6 and .18) arc (0:180:.6 and .18);
\draw (.9,2.4) +(0:.6 and .18) arc (0:180:.6 and .18);
\node at (.23,.2) {$\scriptstyle V$};
}$};
\node (O) at (7,0) {$\tikzmath[scale=.4]{
\fill[gray!50] (.9,2.4) circle (.6 and .18);
\filldraw[fill=gray!20]
(-1.5,-1.2) arc (-180:0:.6 and .18) arc (180:0:.3 and .18) arc (-180:0:.6 and .18) 
[rounded corners=2.4]-- ++(0,.3) -- ++(-.9,.6) [sharp corners]-- ++(0,.3) 
arc (180:0:.3 and .18) arc (-180:0:.6 and .18)
[rounded corners=2.4]-- ++(0,.3) -- ++(-.9,.6) [sharp corners]-- ++(0,.3) 
-- ++(0,1.2) arc (0:-180:.6 and .18) -- ++(0,-1.2)[rounded corners=2.4]-- ++(0,-.3) -- ++(-.9,-.6) [sharp corners]-- ++(0,-.3)
[rounded corners=2.4]-- ++(0,-.3) -- ++(-.9,-.6) [sharp corners]-- cycle;
\draw[dashed, ultra thin] (-.9,-1.2) +(0:.6 and .18) arc (0:180:.6 and .18);
\draw[dashed, ultra thin] (.9,-1.2) +(0:.6 and .18) arc (0:180:.6 and .18);
\draw[dashed, ultra thin] (1.8,0) +(0:.6 and .18) arc (0:180:.6 and .18);
\draw[dashed, ultra thin] (.9,1.2) +(0:.6 and .18) arc (0:180:.6 and .18);
\draw (.9,2.4) +(0:.6 and .18) arc (0:180:.6 and .18);
\draw (.9,1.2) +(0:.6 and .18) arc (0:-180:.6 and .18);
\node at (.9,1.2+.5) {$\scriptstyle X$};
\node at (.03,0) {$\scriptstyle V$};
}$};
\node (P) at (7.8,3.55) {$\tikzmath[scale=.4]{\fill[gray!50] (1.8,2.4) circle (.6 and .18);
\filldraw[fill=gray!20, line join=bevel](-.6,0) arc (-180:0:.6 and .18) arc (180:0:.3 and .18) arc (-180:0:.6 and .18) arc (180:0:.3 and .18) arc (-180:0:.6 and .18)[rounded corners=4.8]-- ++(-.1,.4) -- ++(-1.6,.4) [sharp corners]-- ++(-.1,.4) -- ++(0,1.2) arc (0:-180:.6 and .18) -- ++(0,-1.2)[rounded corners=4.8]-- ++(-.1,-.4) -- ++(-1.6,-.4)  [sharp corners]-- cycle;
\draw[dashed, ultra thin] (0,0) +(0:.6 and .18) arc (0:180:.6 and .18);
\draw[dashed, ultra thin] (1.8,0) +(0:.6 and .18) arc (0:180:.6 and .18);
\draw[dashed, ultra thin] (3.6,0) +(0:.6 and .18) arc (0:180:.6 and .18);
\draw (1.8,2.4) +(0:.6 and .18) arc (0:180:.6 and .18);
\node at (1.8,1) {$\scriptstyle V$};}
$};
\node (Q) at (10,2.5) {$\tikzmath[scale=.4]{\fill[gray!50] (1.8,2.4) circle (.6 and .18);
\filldraw[fill=gray!20, line join=bevel](-.6,0) arc (-180:0:.6 and .18) arc (180:0:.3 and .18) arc (-180:0:.6 and .18) arc (180:0:.3 and .18) arc (-180:0:.6 and .18)[rounded corners=4.8]-- ++(-.1,.4) -- ++(-1.6,.4) [sharp corners]-- ++(-.1,.4) -- ++(0,1.2) arc (0:-180:.6 and .18) -- ++(0,-1.2)[rounded corners=4.8]-- ++(-.1,-.4) -- ++(-1.6,-.4)  [sharp corners]-- cycle;
\draw[dashed, ultra thin] (0,0) +(0:.6 and .18) arc (0:180:.6 and .18);
\draw[dashed, ultra thin] (1.8,0) +(0:.6 and .18) arc (0:180:.6 and .18);
\draw[dashed, ultra thin] (3.6,0) +(0:.6 and .18) arc (0:180:.6 and .18);
\draw[dashed, ultra thin] (1.8,1.2) +(0:.6 and .18) arc (0:180:.6 and .18);
\draw (1.8,2.4) +(0:.6 and .18) arc (0:180:.6 and .18);
\draw (1.8,1.2) +(0:.6 and .18) arc (0:-180:.6 and .18);
\node at (1.8,1.2+.5) {$\scriptstyle X$};
\node at (1.8,.57) {$\scriptstyle V$};}
$};
\node at (5.87,3.85) {$\scriptstyle (\star)$};
\draw (D) -- (A) -- (B) (D) -- (L) -- (G) (Q) -- (O) --node[above, scale=1.1]{$\scriptstyle g$} (L) (N) -- (O) (H) -- (M);
\draw[shorten >=-10, shorten <=2] (E) -- (C);
\draw[shorten >=-12] (A) -- (C);
\draw[shorten >=-8] (C.south west)+(0,.02) -- (D);
\draw[shorten >=-8] (E.south west)+(.02,-.01) -- (G);
\draw[shorten >=-12] (D) -- (G);
\draw[shorten <=2, shorten >=-12] (G) -- (H);
\draw[shorten <=8, shorten >=-10] (K) -- (F);
\draw[shorten >=-10] (E) -- (I);
\draw[shorten >=2, shorten <=-10] (F) -- (C);
\draw[shorten >=2, shorten <=-10] (F) -- (B);
\draw[shorten <=2, shorten >=-12] (E) -- (K);
\draw[shorten <=2, shorten >=-15] (P) -- (Q);
\draw[shorten >=-10] (M) --node[above, pos=.75, scale=1.1]{$\scriptstyle g$} (N);
\draw[shorten >=2] (J) -- (P);
\draw[shorten <=2, shorten >=-8] (P) -- (N);
\draw[shorten <=3, shorten >=-7] (H) -- (I);
\draw[shorten <=2, shorten >=-13] (I) -- (J);
\draw[shorten >=-8] (K) -- (J);
\draw[shorten >=2, shorten <=-8] (L) -- (M);
\draw[shorten >=3] (K) -- (Q);
\draw[shorten >=3] (F) -- (Q);
}
\]
Here, the central hexagon \raisebox{.5pt}{$\scriptstyle (\star)$} commutes because the gluing isomorphism $g$ is equal to the rightmost composition in \eqref{eq: rightmost composition}, precomposed by the diffeormorphism
$\tikzmath[scale=.2]{
\fill[gray!50] (.9,2.4) circle (.6 and .18);
\filldraw[fill=gray!20]
(-1.5,-1.2) arc (-180:0:.6 and .18) arc (180:0:.3 and .18) arc (-180:0:.6 and .18) 
[rounded corners=1.2]-- ++(0,.3) -- ++(-.9,.6) [sharp corners]-- ++(0,.3) 
arc (180:0:.3 and .18) arc (-180:0:.6 and .18)
[rounded corners=1.2]-- ++(0,.3) -- ++(-.9,.6) [sharp corners]-- ++(0,.3) 
-- ++(0,1.2) arc (0:-180:.6 and .18) -- ++(0,-1.2)[rounded corners=1.2]-- ++(0,-.3) -- ++(-.9,-.6) [sharp corners]-- ++(0,-.3)
[rounded corners=1.2]-- ++(0,-.3) -- ++(-.9,-.6) [sharp corners]-- cycle;
\draw[dashed, ultra thin] (-.9,-1.2) +(0:.6 and .18) arc (0:180:.6 and .18);
\draw[dashed, ultra thin] (.9,-1.2) +(0:.6 and .18) arc (0:180:.6 and .18);
\draw[dashed, ultra thin] (1.8,0) +(0:.6 and .18) arc (0:180:.6 and .18);
\draw (.9,2.4) +(0:.6 and .18) arc (0:180:.6 and .18);
\node[scale=.7] at (.23,.2) {$\scriptstyle V$};
}
\cong
\tikzmath[scale=.2]{\fill[gray!50] (1.8,2.4) circle (.6 and .18);
\filldraw[fill=gray!20, line join=bevel](-.6,0) arc (-180:0:.6 and .18) arc (180:0:.3 and .18) arc (-180:0:.6 and .18) arc (180:0:.3 and .18) arc (-180:0:.6 and .18)[rounded corners=2.4]-- ++(-.1,.4) -- ++(-1.6,.4) [sharp corners]-- ++(-.1,.4) -- ++(0,1.2) arc (0:-180:.6 and .18) -- ++(0,-1.2)[rounded corners=2.4]-- ++(-.1,-.4) -- ++(-1.6,-.4)  [sharp corners]-- cycle;
\draw[dashed, ultra thin] (0,0) +(0:.6 and .18) arc (0:180:.6 and .18);
\draw[dashed, ultra thin] (1.8,0) +(0:.6 and .18) arc (0:180:.6 and .18);
\draw[dashed, ultra thin] (3.6,0) +(0:.6 and .18) arc (0:180:.6 and .18);
\draw (1.8,2.4) +(0:.6 and .18) arc (0:180:.6 and .18);
\node[scale=.7] at (1.8,1) {$\scriptstyle V$};}
$.
\end{proof}

\subsection{The braiding}

\def\underbeta{
\tikzmath{\useasboundingbox (-.08,-.12) rectangle (.07,.12);
\node{$\beta$};\draw (-.06,-.13) -- (.07,-.13);}}
\def\undergamma{
\tikzmath{\node{$\gamma$};\draw (-.08,-.15) -- (.06,-.15);}}
\def\underdelta{
\tikzmath{\node{$\delta$};\draw (-.08,-.16) -- (.06,-.16);}}

The extended mapping class group of the standard pair-of-pants contains a special element $\underbeta:P\to P$, called the \emph{braiding}, that exchanges the
two boundary circles $S_1$ and $S_2$, and fixes $S_3$.
Using a flat depiction of $P$, we illustrate the homeomorphism $\underbeta$ by means of what it does to some internal lines:
\begin{equation}\label{eq: picture of beta}
P\,:\,\,\tikzmath[scale=.8]{\filldraw[fill=gray!30] (0,0) circle (1); 
\filldraw[fill=white] (-.4,0) circle (.25)(.4,0) circle (.25);
\node at (-.4,-.01) {$\scriptstyle S_1$};
\node at (.4,-.01) {$\scriptstyle S_2$};
\node at (.9,-.85) {$\scriptstyle S_3$};
}
\qquad\qquad
\underbeta:\,\,
\tikzmath[scale=.8]{\filldraw[fill=gray!30] (0,0) circle (1); 
\draw[ultra thin](-1,0) -- (1,0) (-.4,.92) -- (-.4,-.92)(.4,.92) -- (.4,-.92);
\filldraw[fill=white] (-.4,0) circle (.25)(.4,0) circle (.25);
}
\,\,\,\mapsto\,\,\,
\tikzmath[scale=.8]{\filldraw[fill=gray!30] (0,0) circle (1); 
\draw[ultra thin, rounded corners=1.6](-1,0) -- (-.9,0) [rounded corners=4]-- (-.9,.3) -- (-.6,.6) -- (-.2,.6) -- (.1,.3) [rounded corners=1.2]-- (.1,0) -- (.7,0) [rounded corners=2.4]-- (.7,-.2) [rounded corners=2.8]-- (.5,-.4) [rounded corners=3.2]-- (.2,-.4) -- (0,-.2) -- (0,0)
[rounded corners=1.6]
(-.4,.92) -- (-.4,.8) [rounded corners=4.8]-- (0,.79) -- (.4,.4) [rounded corners=1.2]-- (.4,-.3) [rounded corners=2.4]-- (.2,-.3) -- (.05,-.1) [rounded corners=3.2]-- (.05,.25) -- (-.2,.5) -- (-.55,.5) -- (-.8,.25) [rounded corners=4.8]-- (-.8,-.4) [rounded corners=2.4]-- (-.4,-.8) -- (-.4,-.92);
\pgftransformrotate {180}
\draw[ultra thin, rounded corners=1.6](-1,0) -- (-.9,0) [rounded corners=4]-- (-.9,.3) -- (-.6,.6) -- (-.2,.6) -- (.1,.3) [rounded corners=1.2]-- (.1,0) -- (.7,0) [rounded corners=2.4]-- (.7,-.2) [rounded corners=2.8]-- (.5,-.4) [rounded corners=3.2]-- (.2,-.4) -- (0,-.2) -- (0,0)
[rounded corners=1.6]
(-.4,.92) -- (-.4,.8) [rounded corners=4.8]-- (0,.79) -- (.4,.4) [rounded corners=1.2]-- (.4,-.3) [rounded corners=2.4]-- (.2,-.3) -- (.05,-.1) [rounded corners=3.2]-- (.05,.25) -- (-.2,.5) -- (-.55,.5) -- (-.8,.25) [rounded corners=4.8]-- (-.8,-.4) [rounded corners=2.4]-- (-.4,-.8) -- (-.4,-.92);
\pgftransformrotate {180}
\filldraw[fill=white] (-.4,0) circle (.25)(.4,0) circle (.25);
}
\end{equation}
By Theorem \ref{Mthm: conformal blocks}, there is an associated unitary
\[
\,\,\beta:=V(\,\underbeta\,):\,H_P\to H_P\,,
\]
well defined up to phase.
Moreover, letting $\tau:H_0\otimes H_0\to H_0\otimes H_0$ be the switch map, we can fix the phase of $\beta$ by requiring that the composite
\begin{equation}\label{eq: phase-fixing}
\begin{split}
H_0\,\cong\, H_0\,\boxtimes^\mathsf{h}H_0 \,&\cong \,(H_0\otimes H_0)\boxtimes_{\cala(S_1\cup S_2)} H_P\\
&\xrightarrow{\tau\boxtimes\beta}
(H_0\otimes H_0)\boxtimes_{\cala(S_1\cup S_2)}H_P \,\cong\, H_0\,\boxtimes^\mathsf{h}H_0\,\cong\, H_0
\end{split}
\end{equation}
be the identity on $H_0$.

\begin{theorem}
The transformation
\[
\begin{split}
\beta_{H,K}:\,H\,\boxtimes^\mathsf{h}K\,&\cong \,(H\otimes K)\boxtimes_{\cala(S_1\cup S_2)} H_P\\
&\xrightarrow{\tau\boxtimes\beta}
(K\otimes H)\boxtimes_{\cala(S_1\cup S_2)}H_P \,\cong\, K\,\boxtimes^\mathsf{h}H
\end{split}
\]
equips the category $(\Rep(\cala),\boxtimes^\mathsf{h})$ with the structure of a braided tensor category.
\end{theorem}

\begin{proof}
We need to show that the isomorphisms $\beta_{H,K}$ are natural in $H$ and $K$, and that they satisfy the two hexagon axioms.
Naturality is obvious from the definition.
The two hexagon axioms are
\[
\tikzmath{
\node (1) at (-1.95,1.4) {$H\,\boxtimes^\mathsf{h} (K\,\boxtimes^\mathsf{h} L)$};
\node (2) at (1.95,1.4) {$(K\,\boxtimes^\mathsf{h} L)\,\boxtimes^\mathsf{h} H$};
\node (3) at (-3.7,0) {$(H\,\boxtimes^\mathsf{h} K)\,\boxtimes^\mathsf{h} L$};
\node (4) at (3.7,0) {$K\,\boxtimes^\mathsf{h} (L\,\boxtimes^\mathsf{h} H)$};
\node (5) at (-1.95,-1.4) {$(K\,\boxtimes^\mathsf{h} H)\,\boxtimes^\mathsf{h} L$};
\node (6) at (1.95,-1.4) {$K\,\boxtimes^\mathsf{h} (H\,\boxtimes^\mathsf{h} L)$};
\draw[->] (1)--node[above]{$\scriptstyle \beta_{H,K\boxtimes L}$}(2);\draw[->] (3)--node[left, pos=.6]{$\scriptstyle \alpha\,\,$}(1);\draw[->] (3)--node[left]{$\scriptstyle \beta_{H,K}\boxtimes L\,\,$}(5);\draw[->] (5)--node[above]{$\scriptstyle \alpha$}(6);\draw[->] (2)--node[right, pos=.4]{$\scriptstyle \,\,\alpha$}(4);\draw[->] (6)--node[right]{$\scriptstyle \,\,K\boxtimes\beta_{H,L}$}(4);
}
\]
and
\[
\tikzmath{
\node (1) at (-1.95,1.4) {$(H\,\boxtimes^\mathsf{h} K)\,\boxtimes^\mathsf{h} L$};
\node (2) at (1.95,1.4) {$L\,\boxtimes^\mathsf{h}(H\,\boxtimes^\mathsf{h} K)$};
\node (3) at (-3.7,0) {$H\,\boxtimes^\mathsf{h} (K\,\boxtimes^\mathsf{h} L)$};
\node (4) at (3.7,0) {$(L\,\boxtimes^\mathsf{h} H)\,\boxtimes^\mathsf{h} K$};
\node (5) at (-1.95,-1.4) {$H\,\boxtimes^\mathsf{h} (L\,\boxtimes^\mathsf{h} K)$};
\node (6) at (1.95,-1.4) {$(H\,\boxtimes^\mathsf{h} L)\,\boxtimes^\mathsf{h} K$};
\draw[->] (1)--node[above]{$\scriptstyle \beta_{H\boxtimes K, L}$}(2);\draw[->] (3)--node[left, pos=.6]{$\scriptstyle \alpha^{-1}\,$}(1);\draw[->] (3)--node[left]{$\scriptstyle H\boxtimes \beta_{K,L}\,$}(5);\draw[->] (5)--node[above]{$\scriptstyle \alpha^{-1}$}(6);\draw[->] (2)--node[right, pos=.4]{$\scriptstyle \,\,\alpha^{-1}$}(4);\draw[->] (6)--node[right]{$\scriptstyle \,\,\beta_{H,L}\boxtimes K$}(4);
}
\medskip\]
and we shall only treat the first one here.

The following diagram is commutative in the groupoid $\mathsf{2MAN}$
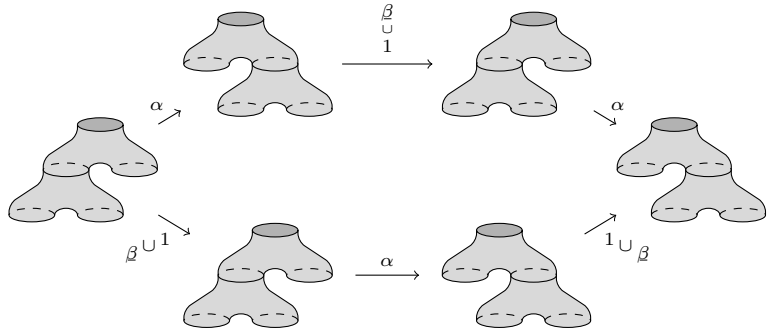
\[
\def\surfaceA{
\tikzmath[scale=.5]{
\fill[gray!60] (.9,1.2) circle (.6 and .18);
\filldraw[fill=gray!30]
(-1.5,-1.2) arc (-180:0:.6 and .18) arc (180:0:.3 and .18) arc (-180:0:.6 and .18) 
[rounded corners=3]-- ++(0,.3) -- ++(-.9,.6) [sharp corners]-- ++(0,.3) 
arc (180:0:.3 and .18) arc (-180:0:.6 and .18)
[rounded corners=3]-- ++(0,.3) -- ++(-.9,.6) [sharp corners]-- ++(0,.3) 
arc (0:-180:.6 and .18)
[rounded corners=3]-- ++(0,-.3) -- ++(-.9,-.6) [sharp corners]-- ++(0,-.3)
[rounded corners=3]-- ++(0,-.3) -- ++(-.9,-.6) [sharp corners]-- cycle;
\draw[dashed] (-.9,-1.2) +(0:.6 and .18) arc (0:180:.6 and .18);
\draw[dashed] (.9,-1.2) +(0:.6 and .18) arc (0:180:.6 and .18);
\draw[dashed] (0,0) +(0:.6 and .18) arc (0:180:.6 and .18);
\draw (0,0) +(0:.6 and .18) arc (0:-180:.6 and .18);
\draw[dashed] (1.8,0) +(0:.6 and .18) arc (0:180:.6 and .18);
\draw (.9,1.2) +(0:.6 and .18) arc (0:180:.6 and .18);
}}
\def\surfaceB{
\tikzmath[scale=.5]{
\pgftransformxscale{-1}
\fill[gray!60] (.9,1.2) circle (.6 and .18);
\filldraw[fill=gray!30]
(-1.5,-1.2) arc (-180:0:.6 and .18) arc (180:0:.3 and .18) arc (-180:0:.6 and .18) 
[rounded corners=3]-- ++(0,.3) -- ++(-.9,.6) [sharp corners]-- ++(0,.3) 
arc (180:0:.3 and .18) arc (-180:0:.6 and .18)
[rounded corners=3]-- ++(0,.3) -- ++(-.9,.6) [sharp corners]-- ++(0,.3) 
arc (0:-180:.6 and .18)
[rounded corners=3]-- ++(0,-.3) -- ++(-.9,-.6) [sharp corners]-- ++(0,-.3)
[rounded corners=3]-- ++(0,-.3) -- ++(-.9,-.6) [sharp corners]-- cycle;
\draw[dashed] (-.9,-1.2) +(0:.6 and .18) arc (0:180:.6 and .18);
\draw[dashed] (.9,-1.2) +(0:.6 and .18) arc (0:180:.6 and .18);
\draw[dashed] (0,0) +(0:.6 and .18) arc (0:180:.6 and .18);
\draw (0,0) +(0:.6 and .18) arc (0:-180:.6 and .18);
\draw[dashed] (1.8,0) +(0:.6 and .18) arc (0:180:.6 and .18);
\draw (.9,1.2) +(0:.6 and .18) arc (0:180:.6 and .18);
}}
\tikzmath{
\node (1) at (-1.7,1.4) {$\surfaceB$};
\node (2) at (1.7,1.4) {$\surfaceA$};
\node[inner sep=0] (3) at (-4,0) {$\surfaceA$};
\node[inner sep=0] (4) at (4,0) {$\surfaceB$};
\node (5) at (-1.7,-1.4) {$\surfaceA$};
\node (6) at (1.7,-1.4) {$\surfaceB$};
\draw[->] (1)--node[above]{$\scriptstyle \underset{\overset {\cup}{\scriptstyle \phantom{O^i}1\phantom{O^i}}}{\beta}$}(2);
\draw[->,shorten >=-3] (3)--node[left, pos=1.9]{$\scriptstyle \alpha\,\,\,\,$}(1);
\draw[->, shorten >=-7] (3)--node[left, pos=2.9, xshift=-2, yshift=-2]{$\scriptstyle \raisebox{-1mm}{$\scriptstyle\beta$}\,\cup \,\raisebox{.8mm}{$\scriptstyle 1$}\,\,\,\,\,$}(5);
\draw[->, shorten >=5, shorten <=5] (5)--node[above]{$\scriptstyle \alpha$}(6);
\draw[->,shorten <=-3] (2)--node[right, pos=-.9]{$\scriptstyle \,\,\,\,\alpha$}(4);
\draw[->,shorten <=-7] (6)--node[right, pos=-1.9, xshift=2, yshift=-2]{$\scriptstyle \,\,\,\,\,\raisebox{.8mm}{$\scriptstyle1$}\,\cup\, \raisebox{-1mm}{$\scriptstyle\beta$}$}(4);
\draw[line width=.25] (-.045,2) -- +(.1,0);
\draw[line width=.25] (3.32,-1.2) -- +(.1,0);
\draw[line width=.25] (-3.41,-1.2) -- +(.1,0);
}
\medskip
\] 
By Theorem \ref{Mthm: conformal blocks} and Proposition \ref{prop: al = al}, the corresponding diagram of Hilbert spaces 
\begin{equation}\label{eq: Hex of manifolds}
\tikzmath{
\node (1) at (-2.2,1.4) {$H_P\,\mbox{${}_2\boxtimes_3$}\, H_P$};
\node (2) at (1.6,1.5) {$H_P\,\mbox{${}_1\boxtimes_3$}\, H_P\stackrel\tau\cong H_P\,\mbox{${}_3\boxtimes_1$}\, H_P$};
\node (3) at (-3.7,0) {$H_P\,\mbox{${}_3\boxtimes_1$}\, H_P$};
\node (4) at (3.9,0) {$H_P\,\mbox{${}_2\boxtimes_3$}\, H_P$};
\node (5) at (-1.9,-1.4) {$H_P\,\mbox{${}_3\boxtimes_1$}\, H_P$};
\node (6) at (2.2,-1.4) {$H_P\,\mbox{${}_2\boxtimes_3$}\, H_P$};
\draw[->] (1)--node[above]{$\scriptstyle \beta\boxtimes 1$}(2.west|-1);
\draw[->] (3)--node[left, pos=.6]{$\scriptstyle \alpha\,\,\,\,$}(1);
\draw[->] (3)--node[left, pos=.6]{$\scriptstyle \beta\boxtimes 1\,\,\,\,\,$}(5);
\draw[->] (5)--node[above]{$\scriptstyle \alpha$}(6);
\draw[->] (2.south)+(1.2,0)--node[right, pos=.4]{$\scriptstyle \,\,\,\,\alpha$}(4);
\draw[->] (6)--node[right, pos=.4]{$\scriptstyle \,\,\,\,\,1\boxtimes \beta$}(4);
}
\end{equation}
therefore commutes up to phase.
If we fuse $H_0\otimes H_0\otimes H_0$ onto \eqref{eq: Hex of manifolds}, then the resulting diagram
\[
\tikzmath{
\node[scale=.8] (1) at (-2.95,1) {$(H_0\otimes H_0\otimes H_0)\boxtimes_{\cala(S^1\sqcup S^1\sqcup S^1)}(H_P\,\mbox{${}_2\boxtimes_3$}\, H_P)$};
\node[scale=.8] (2) at (2.95,1) {$(H_0\otimes H_0\otimes H_0)\boxtimes_{\cala(S^1\sqcup S^1\sqcup S^1)}(H_P\,\mbox{${}_3\boxtimes_1$}\, H_P)$};
\node[scale=.8] (3) at (-3.6,0) {$(H_0\otimes H_0\otimes H_0)\boxtimes_{\cala(S^1\sqcup S^1\sqcup S^1)}(H_P\,\mbox{${}_3\boxtimes_1$}\, H_P)$};
\node[scale=.8] (4) at (3.6,0) {$(H_0\otimes H_0\otimes H_0)\boxtimes_{\cala(S^1\sqcup S^1\sqcup S^1)}(H_P\,\mbox{${}_2\boxtimes_3$}\, H_P)$};
\node[scale=.8] (5) at (-2.95,-1) {$(H_0\otimes H_0\otimes H_0)\boxtimes_{\cala(S^1\sqcup S^1\sqcup S^1)}(H_P\,\mbox{${}_3\boxtimes_1$}\, H_P)$};
\node[scale=.8] (6) at (2.95,-1) {$(H_0\otimes H_0\otimes H_0)\boxtimes_{\cala(S^1\sqcup S^1\sqcup S^1)}(H_P\,\mbox{${}_2\boxtimes_3$}\, H_P)$};
\draw[->] (1)--(2);
\draw[->] (3.north)+(.1,0)--(1.south);
\draw[->] (3.south)+(.1,0)--(5.north);
\draw[->] (5)--(6);
\draw[<-] (4.north)+(-.1,0)--(2.south);
\draw[<-] (4.south)+(-.1,0)--(6.north);
}
\]
is simply
\[
\tikzmath{
\node (1) at (-2.05,1.4) {$H_0\,\boxtimes^\mathsf{h} (H_0\,\boxtimes^\mathsf{h} H_0)$};
\node (2) at (2.05,1.4) {$(H_0\,\boxtimes^\mathsf{h} H_0)\,\boxtimes^\mathsf{h} H_0$};
\node (3) at (-3.7,0) {$(H_0\,\boxtimes^\mathsf{h} H_0)\,\boxtimes^\mathsf{h} H_0$};
\node (4) at (3.7,0) {$H_0\,\boxtimes^\mathsf{h} (H_0\,\boxtimes^\mathsf{h} H_0)$};
\node (5) at (-2.05,-1.4) {$(H_0\,\boxtimes^\mathsf{h} H_0)\,\boxtimes^\mathsf{h} H_0$};
\node (6) at (2.05,-1.4) {$H_0\,\boxtimes^\mathsf{h} (H_0\,\boxtimes^\mathsf{h} H_0)$};
\draw[->] (1)--node[above]{$\scriptstyle \beta$}(2);\draw[->] (3)--node[left, pos=.6]{$\scriptstyle \alpha\,\,$}(1);\draw[->] (3)--node[left, pos=.6]{$\scriptstyle \beta\boxtimes 1\,\,$}(5);\draw[->] (5)--node[above]{$\scriptstyle \alpha$}(6);\draw[->] (2)--node[right, pos=.4]{$\scriptstyle \,\,\alpha$}(4);\draw[->] (6)--node[right, pos=.4]{$\scriptstyle \,\,1\boxtimes\beta$}(4);
}
\]
and it follows from \eqref{eq: phase-fixing} that it is commutative, not just up to phase.
The diagram \eqref{eq: Hex of manifolds} was therefore also commutative, not just up to phase.
\end{proof}

It would be interesting to compare the braiding $\beta$ defined above with the one introduced in \cite{Gabbiani-Froehlich(OperatorAlg-CFT)}.
We conjecture that those two braidings are equal.

\subsection{Modularity of the category of sectors}
Based on the technology that we developed in the previous sections, we will now present a new proof of a famous result of Kawahigashi--Longo--M\"uger 
\cite[Cor.~37]{Kawahigashi-Longo-Mueger(2001multi-interval)} about the modularity of the representation category of conformal nets.

Recall that an object $T$ in a braided tensor category is called \emph{transparent} if the equality
\[
\beta_{X,T}=\beta_{T,X}^{-1}:\,X\boxtimes T\to T\boxtimes X
\]
holds for every object $X$ in the category.
By M\"uger's theorem (first proven by Rehren in \cite{Rehren(Braid-group-statistics-and-their-superselection-rules)} in the language of low-dimensional algebraic quantum field theory
and then by M\"uger in an abstract categorical setup \cite[Cor.~7.11]{Mueger(Subfactors-to-categories-II)}---see also \cite{Beliakova-Blanchet(mod-cats-types-BCD)} for a proof that does not require unitarity),
one of the equivalent definitions of a \emph{modular tensor category}
is a braided fusion category in which the only transparent objects are the multiples of the identity.

\begin{theorem}
Let $\cala$ be a conformal net with finite index.
Then the braided tensor category $\Rep(\cala)$ is modular.
\end{theorem}

\begin{remark}\label{rmk: modular functor+}
As explained in Remark \ref{rmk: modular functor}, the system of spaces of conformal blocks we have constructed together with the factorization property has as a special case the data of the projective version of a modular functor.
If we knew that the central extension of the bordism 2-category of surfaces with parametrized boundary (objects=unions of circles, morphisms=cobordisms, 2-morphisms=homeomorphisms modulo isotopy rel boundary) described in~\cite[Def.~5.7.6]{Bakalov-Kirillov(Lect-tens-cat+mod-func)} is a \emph{universal} central extension,
then we could deduce that we have a modular functor in the sense of \cite[Def.~5.7.10]{Bakalov-Kirillov(Lect-tens-cat+mod-func)}.
The results of Bakalov--Kirillov \cite[Thm.~5.7.11]{Bakalov-Kirillov(Lect-tens-cat+mod-func)} would then provide yet another proof of the above modularity theorem.
\end{remark}

\begin{proof}
Given an irreducible transparent object $T\in\Rep(\cala)$, we need to show that $T$ is isomorphic to $H_0$.

We shall draw the standard pair of pants $P$ as in \eqref{eq: picture of beta}.
Consider the following submanifolds of $P$
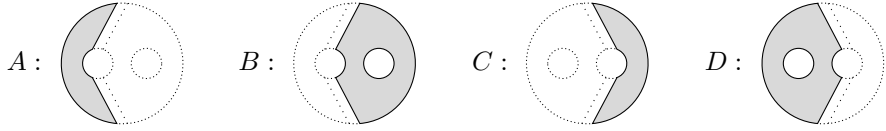
\[
A:\,\,\tikzmath[scale=.8]{\filldraw[fill=gray!30] (-.4,0)+(360-111:.25) arc (360-111:111:.25) -- (95:1) arc (95:360-95:1) -- cycle; \draw[densely dotted](95:1) arc (95:-95:1)(-.4,0)+(-111:.25) arc (-111:111:.25)(.4,0) circle (.25);\draw[dotted] (-.4,0)+(69:.25) -- (85:1)(-.4,0)+(-69:.25) -- (-85:1);}
\qquad B:\,\,\tikzmath[scale=.8]{\filldraw[fill=gray!30] (-.4,0)+(-69:.25) arc (-69:69:.25) -- (85:1) arc (85:-85:1) -- cycle;\filldraw[fill=white] (.4,0) circle (.25); \draw[densely dotted](85:1) arc (85:360-85:1)(-.4,0)+(-69:.25) arc (-69:69-360:.25);\draw[dotted](-.4,0)+(111:.25) -- (95:1)(-.4,0)+(-111:.25) -- (-95:1);}
\qquad C:\,\,\tikzmath[scale=.8]{\pgftransformxscale{-1}\filldraw[fill=gray!30] (-.4,0)+(360-111:.25) arc (360-111:111:.25) -- (95:1) arc (95:360-95:1) -- cycle; \draw[densely dotted](95:1) arc (95:-95:1)(-.4,0)+(-111:.25) arc (-111:111:.25)(.4,0) circle (.25);\draw[dotted](-.4,0)+(69:.25) -- (85:1)(-.4,0)+(-69:.25) -- (-85:1);}
\qquad D:\,\,\tikzmath[scale=.8]{\pgftransformxscale{-1}\filldraw[fill=gray!30] (-.4,0)+(-69:.25) arc (-69:69:.25) -- (85:1) arc (85:-85:1) -- cycle;\filldraw[fill=white] (.4,0) circle (.25); \draw[densely dotted](85:1) arc (85:360-85:1)(-.4,0)+(-69:.25) arc (-69:69-360:.25);\draw[dotted](-.4,0)+(111:.25) -- (95:1)(-.4,0)+(-111:.25) -- (-95:1);}
\smallskip
\]
and let $i_A:A\hookrightarrow P,\ldots,i_D:D\hookrightarrow P$ be the inclusion maps.

Let us also consider the manifolds
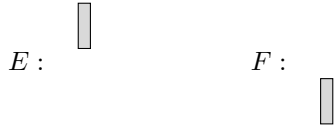
\[
E:\quad\tikzmath[scale=.8]{\useasboundingbox (-.1,1) rectangle (.1,-1);\filldraw[fill=gray!30] (-.1,1) rectangle (.1,.25);}
\qquad\qquad\qquad F:\quad\tikzmath[scale=.8]{\useasboundingbox (-.1,1) rectangle (.1,-1);\filldraw[fill=gray!30] (-.1,-1) rectangle (.1,-.25);}
\medskip\]
along with embeddings $f,g:E\to P$ and $h,k:F\to P$
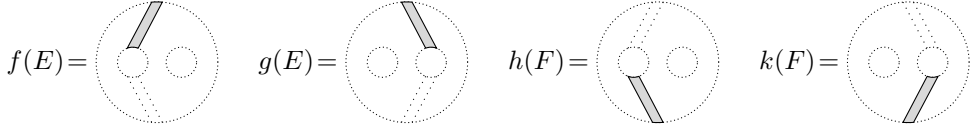
\[
f(E)\!=\tikzmath[scale=.8]{\draw[densely dotted] (0,0) circle (1) (-.4,0) circle (.25)(.4,0) circle (.25);\draw[fill=gray!30] (-.4,0) +(111:.25) arc (111:69:.25) -- (85:1) arc (85:95:1) -- cycle;
\draw[dotted] (-.4,0)+(-111:.25) -- (-95:1)(-.4,0)+(-69:.25) -- (-85:1);}
\,\,\,\quad g(E)\!=\tikzmath[scale=.8]{\draw[densely dotted] (0,0) circle (1) (-.4,0) circle (.25)(.4,0) circle (.25);\pgftransformxscale{-1}\draw[fill=gray!30] (-.4,0) +(111:.25) arc (111:69:.25) -- (85:1) arc (85:95:1) -- cycle;\draw[dotted] (-.4,0)+(-111:.25) -- (-95:1)(-.4,0)+(-69:.25) -- (-85:1);}
\,\,\,\quad h(F)\!=\tikzmath[scale=.8]{\draw[densely dotted] (0,0) circle (1) (-.4,0) circle (.25)(.4,0) circle (.25);\pgftransformyscale{-1}\draw[fill=gray!30] (-.4,0) +(111:.25) arc (111:69:.25) -- (85:1) arc (85:95:1) -- cycle;\draw[dotted] (-.4,0)+(-111:.25) -- (-95:1)(-.4,0)+(-69:.25) -- (-85:1);}
\,\,\,\quad k(F)\!=\tikzmath[scale=.8]{\draw[densely dotted] (0,0) circle (1) (-.4,0) circle (.25)(.4,0) circle (.25);\pgftransformyscale{-1}\pgftransformxscale{-1}\draw[fill=gray!30] (-.4,0) +(111:.25) arc (111:69:.25) -- (85:1) arc (85:95:1) -- cycle;\draw[dotted] (-.4,0)+(-111:.25) -- (-95:1)(-.4,0)+(-69:.25) -- (-85:1);}
\]
whose images in $P$ are such that
\[
A\cup f(F)\cup B \cup h(E)=P\qquad \text{and}\qquad C\cup g(E)\cup D \cup k(F)=P.
\]
Recall from Theorem \ref{Mthm: conformal blocks} that $V(A),\ldots, V(F)$ are only well defined up to canonical-up-to-phase unitary isomorphism.
Let us fix, by means of some arbitrary choice, Hilbert spaces $V(A),\ldots,V(F)$ within their ``up-to-phase isomorphism class''.

Let us write $V\Big(
\tikzmath[scale=.35]{\filldraw[fill=gray!30] (0,0) circle (1); 
\filldraw[fill=white] (-.4,0) circle (.25)(.4,0) circle (.25);}\Big)$
for the Hilbert space $H_P$ defined in \eqref{eq: def: H_P}.
By Theorem~\ref{maintheorem gluing}, there are unitary isomorphisms
\[
\begin{split}
\kappa_l:\,\,\tikzmath{\node (a) at (0,0) {$V(A)\boxtimes V(F)\boxtimes V(B) \boxtimes V(E)\,\boxtimes$};\def\dd{.4}\def\ll{.25}\def\rr{.25}
\draw[dashed, rounded corners = 5] (a.east) -- ++(\rr,0) -- ++(0,-\dd) -- ($(a.west) + (-\ll,-\dd)$) -- +(0,\dd) -- (a.west);} 
\,\rightarrow\, V\bigg(\tikzmath[scale=.5]{\filldraw[fill=gray!30] (0,0) circle (1); \filldraw[fill=white] (-.4,0) circle (.25)(.4,0) circle (.25);}\bigg) 
\\\kappa_r:\,\,\tikzmath{\node (a) at (0,0) {$V(C)\boxtimes V(E)\boxtimes V(D) \boxtimes V(F)\,\boxtimes$};\def\dd{.4}\def\ll{.25}\def\rr{.25}
\draw[dashed, rounded corners = 5] (a.east) -- ++(\rr,0) -- ++(0,-\dd) -- ($(a.west) + (-\ll,-\dd)$) -- +(0,\dd) -- (a.west);} 
\,\rightarrow\, V\bigg(\tikzmath[scale=.5]{\filldraw[fill=gray!30] (0,0) circle (1); \filldraw[fill=white] (-.4,0) circle (.25)(.4,0) circle (.25);}\bigg)
\end{split}\medskip
\]
that correspond to the homeomorphisms
\[
\begin{split}   i_A\cup h\cup i_B \cup f\,:&\,\,\,\,\,A\cup F\cup B \cup E\,\to\, P\\
i_C\cup g\cup i_D \cup k\,:&\,\,\,\,\,C\cup E\cup D \cup F\,\to\, P,   \end{split}
\]
(the algebras over which the fusion product are taken are omitted from the notation), and that are well defined up to phase.
Fix choices of $\kappa_l$ and $\kappa_r$ within their equivalence classes.

Recall the definition of $\underbeta:P\to P$ from \eqref{eq: picture of beta}, and let
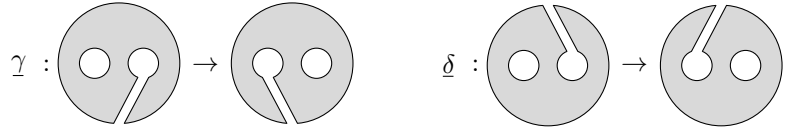
\[
\undergamma:\tikzmath[scale=.8]{\filldraw[fill=gray!30] (0,0) circle (1); 
\filldraw[fill=white] (-.4,0) circle (.25)(.4,0) circle (.25);
\pgftransformxscale{-1}\pgftransformyscale{-1}
\fill[white] (95:1) -- ($(-.4,0)+(111:.25)$) -- +(0,-.05) -- +(.175,-.05) -- ($(-.4,0)+(69:.25)$) -- (85:1) -- +(0,.05) -- +(-.175,.05);
\draw[line cap=round] (-.4,0) +(111:.25) -- (95:1) (-.4,0) +(69:.25) -- (85:1);}
\,\to\,%
\tikzmath[scale=.8]{\filldraw[fill=gray!30] (0,0) circle (1); 
\filldraw[fill=white] (-.4,0) circle (.25)(.4,0) circle (.25);
\pgftransformyscale{-1}
\fill[white] (95:1) -- ($(-.4,0)+(111:.25)$) -- +(0,-.05) -- +(.175,-.05) -- ($(-.4,0)+(69:.25)$) -- (85:1) -- +(0,.05) -- +(-.175,.05);
\draw[line cap=round] (-.4,0) +(111:.25) -- (95:1) (-.4,0) +(69:.25) -- (85:1);}
\quad\qquad%
\underdelta:
\tikzmath[scale=.8]{\filldraw[fill=gray!30] (0,0) circle (1); 
\filldraw[fill=white] (-.4,0) circle (.25)(.4,0) circle (.25);
\pgftransformxscale{-1}
\fill[white] (95:1) -- ($(-.4,0)+(111:.25)$) -- +(0,-.05) -- +(.175,-.05) -- ($(-.4,0)+(69:.25)$) -- (85:1) -- +(0,.05) -- +(-.175,.05);
\draw[line cap=round] (-.4,0) +(111:.25) -- (95:1) (-.4,0) +(69:.25) -- (85:1);}
\,\to\,%
\tikzmath[scale=.8]{\filldraw[fill=gray!30] (0,0) circle (1); 
\filldraw[fill=white] (-.4,0) circle (.25)(.4,0) circle (.25);
\fill[white] (95:1) -- ($(-.4,0)+(111:.25)$) -- +(0,-.05) -- +(.175,-.05) -- ($(-.4,0)+(69:.25)$) -- (85:1) -- +(0,.05) -- +(-.175,.05);
\draw[line cap=round] (-.4,0) +(111:.25) -- (95:1) (-.4,0) +(69:.25) -- (85:1);}
\medskip\]
be homeomorphisms such that
$\undergamma\!\cup (h\circ k^{-1}) = \,\underbeta\,$ and $\underdelta\!\cup (f\circ g^{-1}) = \,\underbeta\,{}^{-1}$ in the groupoid $\mathsf{2MAN}$ (assume that $f$, $g$, $h$, $k$ are chosen so that this is possible).
Let then
\begin{equation}\label{eq: gam en del}
\begin{split}
\gamma:V(C)\boxtimes V(E)\boxtimes V(D)&\,\to V(B)\boxtimes V(E)\boxtimes V(A)\\
\delta:V(D)\boxtimes V(F)\boxtimes V(C)&\,\to V(A)\boxtimes V(F)\boxtimes V(B)
\end{split}
\end{equation}
be the unique representatives $\gamma\in V(\!\undergamma\!)$, $\delta\in V(\!\underdelta\!)$ for which
\[
\tikzmath{
\node (a) at (0,0) {$\gamma \boxtimes 1_{V(F)}\,\boxtimes$};
\def\dd{.35}
\def\ll{.25}
\def\rr{.22}
\draw[dashed, rounded corners = 5] (a.east) -- ++(\rr,0) -- ++(0,-\dd) -- ($(a.west) + (-\ll,-\dd)$) -- +(0,\dd) -- (a.west);
}
\,=\kappa_l^{-1}\circ\beta\circ\kappa_r
\,\,\,\quad\text{and}\,\,\,\quad
\tikzmath{
\node (a) at (0,0) {$\delta \boxtimes 1_{V(E)}\,\boxtimes$};
\def\dd{.35}
\def\ll{.25}
\def\rr{.22}
\draw[dashed, rounded corners = 5] (a.east) -- ++(\rr,0) -- ++(0,-\dd) -- ($(a.west) + (-\ll,-\dd)$) -- +(0,\dd) -- (a.west);
}
\,=\kappa_l^{-1}\circ\beta^{-1}\circ\kappa_r.
\smallskip
\]
Let us agree that
$
V\Big(
\tikzmath[scale=.35]{\filldraw[fill=gray!30] (0,0) circle (1); 
\filldraw[fill=white] (-.4,0) circle (.25)(.4,0) circle (.25);
\pgftransformxscale{-1}\pgftransformyscale{-1}
\fill[white] (95:1) -- ($(-.4,0)+(111:.25)$) -- +(0,-.05) -- +(.175,-.05) -- ($(-.4,0)+(69:.25)$) -- (85:1) -- +(0,.05) -- +(-.175,.05);
\draw[line cap=round] (-.4,0) +(111:.25) -- (95:1) (-.4,0) +(69:.25) -- (85:1);}
\Big)
$, $
V\Big(
\tikzmath[scale=.35]{\filldraw[fill=gray!30] (0,0) circle (1); 
\filldraw[fill=white] (-.4,0) circle (.25)(.4,0) circle (.25);
\pgftransformyscale{-1}
\fill[white] (95:1) -- ($(-.4,0)+(111:.25)$) -- +(0,-.05) -- +(.175,-.05) -- ($(-.4,0)+(69:.25)$) -- (85:1) -- +(0,.05) -- +(-.175,.05);
\draw[line cap=round] (-.4,0) +(111:.25) -- (95:1) (-.4,0) +(69:.25) -- (85:1);}
\Big)
$, $
V\Big(
\tikzmath[scale=.35]{\filldraw[fill=gray!30] (0,0) circle (1); 
\filldraw[fill=white] (-.4,0) circle (.25)(.4,0) circle (.25);
\pgftransformxscale{-1}
\fill[white] (95:1) -- ($(-.4,0)+(111:.25)$) -- +(0,-.05) -- +(.175,-.05) -- ($(-.4,0)+(69:.25)$) -- (85:1) -- +(0,.05) -- +(-.175,.05);
\draw[line cap=round] (-.4,0) +(111:.25) -- (95:1) (-.4,0) +(69:.25) -- (85:1);}
\Big)
$, $
V\Big(
\tikzmath[scale=.35]{\filldraw[fill=gray!30] (0,0) circle (1); 
\filldraw[fill=white] (-.4,0) circle (.25)(.4,0) circle (.25);
\fill[white] (95:1) -- ($(-.4,0)+(111:.25)$) -- +(0,-.05) -- +(.175,-.05) -- ($(-.4,0)+(69:.25)$) -- (85:1) -- +(0,.05) -- +(-.175,.05);
\draw[line cap=round] (-.4,0) +(111:.25) -- (95:1) (-.4,0) +(69:.25) -- (85:1);}
\Big)
$
stand for the Hilbert spaces 
$V(C)\boxtimes V(E)\boxtimes V(D)$,
$V(B)\boxtimes V(E)\boxtimes V(A)$,
$V(D)\boxtimes V(F)\boxtimes V(C)$, and
$V(A)\boxtimes V(F)\boxtimes V(B)$, so we can rewrite \eqref{eq: gam en del} as
\[
\gamma:\,V\bigg(
\tikzmath[scale=.5]{\filldraw[fill=gray!30] (0,0) circle (1); 
\filldraw[fill=white] (-.4,0) circle (.25)(.4,0) circle (.25);
\pgftransformxscale{-1}\pgftransformyscale{-1}
\fill[white] (95:1) -- ($(-.4,0)+(111:.25)$) -- +(0,-.05) -- +(.175,-.05) -- ($(-.4,0)+(69:.25)$) -- (85:1) -- +(0,.05) -- +(-.175,.05);
\draw[line cap=round] (-.4,0) +(111:.25) -- (95:1) (-.4,0) +(69:.25) -- (85:1);}
\bigg)
\,\to\,
V\bigg(
\tikzmath[scale=.5]{\filldraw[fill=gray!30] (0,0) circle (1); 
\filldraw[fill=white] (-.4,0) circle (.25)(.4,0) circle (.25);
\pgftransformyscale{-1}
\fill[white] (95:1) -- ($(-.4,0)+(111:.25)$) -- +(0,-.05) -- +(.175,-.05) -- ($(-.4,0)+(69:.25)$) -- (85:1) -- +(0,.05) -- +(-.175,.05);
\draw[line cap=round] (-.4,0) +(111:.25) -- (95:1) (-.4,0) +(69:.25) -- (85:1);}
\bigg),\qquad
\delta:\,V\bigg(
\tikzmath[scale=.5]{\filldraw[fill=gray!30] (0,0) circle (1); 
\filldraw[fill=white] (-.4,0) circle (.25)(.4,0) circle (.25);
\pgftransformxscale{-1}
\fill[white] (95:1) -- ($(-.4,0)+(111:.25)$) -- +(0,-.05) -- +(.175,-.05) -- ($(-.4,0)+(69:.25)$) -- (85:1) -- +(0,.05) -- +(-.175,.05);
\draw[line cap=round] (-.4,0) +(111:.25) -- (95:1) (-.4,0) +(69:.25) -- (85:1);}
\bigg)
\,\to\,
V\bigg(
\tikzmath[scale=.5]{\filldraw[fill=gray!30] (0,0) circle (1); 
\filldraw[fill=white] (-.4,0) circle (.25)(.4,0) circle (.25);
\fill[white] (95:1) -- ($(-.4,0)+(111:.25)$) -- +(0,-.05) -- +(.175,-.05) -- ($(-.4,0)+(69:.25)$) -- (85:1) -- +(0,.05) -- +(-.175,.05);
\draw[line cap=round] (-.4,0) +(111:.25) -- (95:1) (-.4,0) +(69:.25) -- (85:1);}
\bigg).
\]

Let now $T\in\Rep(\cala)$ be irreducible and transparent.
By definition of transparent, the two natural transformations $\beta_{T,-}$ and $\beta_{-,T}^{-1}$ from $T\,\boxtimes^\mathsf{h}-$ to $-\,\boxtimes^\mathsf{h}T$
are equal to each other.
Equivalently, the two maps $1_T\boxtimes \beta$ and $1_T\boxtimes \beta^{-1}$ from $T\boxtimes_{\cala(S_1)} H_P$ to $T\boxtimes_{\cala(S_2)} H_P$ are equal to each other.
Let us define
\[
\begin{split}
\beta_T:=1_T\boxtimes \beta = 1_T\boxtimes \beta^{-1}:\,\,\,&\,\,\,
T\boxtimes_{\cala(S_1)} 
\!V\Big(\tikzmath[scale=.35]{\filldraw[fill=gray!30] (0,0) circle (1); 
\filldraw[fill=white] (-.4,0) circle (.25)(.4,0) circle (.25);}\Big)\to
T\boxtimes_{\cala(S_2)} 
\!V\Big(\tikzmath[scale=.35]{\filldraw[fill=gray!30] (0,0) circle (1); 
\filldraw[fill=white] (-.4,0) circle (.25)(.4,0) circle (.25);}\Big)
\\
\gamma_T:=1_T\boxtimes\gamma\quad:\quad&\,\,\,
T\boxtimes_{\cala(S_1)} 
\!V\Big(
\tikzmath[scale=.35]{\filldraw[fill=gray!30] (0,0) circle (1); 
\filldraw[fill=white] (-.4,0) circle (.25)(.4,0) circle (.25);
\pgftransformxscale{-1}\pgftransformyscale{-1}
\fill[white] (95:1) -- ($(-.4,0)+(111:.25)$) -- +(0,-.05) -- +(.175,-.05) -- ($(-.4,0)+(69:.25)$) -- (85:1) -- +(0,.05) -- +(-.175,.05);
\draw[line cap=round] (-.4,0) +(111:.25) -- (95:1) (-.4,0) +(69:.25) -- (85:1);}
\Big)
\to T\boxtimes_{\cala(S_2)}
\!V\Big(
\tikzmath[scale=.35]{\filldraw[fill=gray!30] (0,0) circle (1); 
\filldraw[fill=white] (-.4,0) circle (.25)(.4,0) circle (.25);
\pgftransformyscale{-1}
\fill[white] (95:1) -- ($(-.4,0)+(111:.25)$) -- +(0,-.05) -- +(.175,-.05) -- ($(-.4,0)+(69:.25)$) -- (85:1) -- +(0,.05) -- +(-.175,.05);
\draw[line cap=round] (-.4,0) +(111:.25) -- (95:1) (-.4,0) +(69:.25) -- (85:1);}
\Big)
\\
\delta_T:=1_T\boxtimes\delta\quad:\quad&\,\,\,
T\boxtimes_{\cala(S_1)}
\!V\Big(\tikzmath[scale=.35]{\filldraw[fill=gray!30] (0,0) circle (1); 
\filldraw[fill=white] (-.4,0) circle (.25)(.4,0) circle (.25);
\pgftransformxscale{-1}
\fill[white] (95:1) -- ($(-.4,0)+(111:.25)$) -- +(0,-.05) -- +(.175,-.05) -- ($(-.4,0)+(69:.25)$) -- (85:1) -- +(0,.05) -- +(-.175,.05);
\draw[line cap=round] (-.4,0) +(111:.25) -- (95:1) (-.4,0) +(69:.25) -- (85:1);}
\Big)
\to
T\boxtimes_{\cala(S_2)}
\!V\Big(\tikzmath[scale=.35]{\filldraw[fill=gray!30] (0,0) circle (1); 
\filldraw[fill=white] (-.4,0) circle (.25)(.4,0) circle (.25);
\fill[white] (95:1) -- ($(-.4,0)+(111:.25)$) -- +(0,-.05) -- +(.175,-.05) -- ($(-.4,0)+(69:.25)$) -- (85:1) -- +(0,.05) -- +(-.175,.05);
\draw[line cap=round] (-.4,0) +(111:.25) -- (95:1) (-.4,0) +(69:.25) -- (85:1);}
\Big).
\end{split}
\]
From now on, we will simplify our notation even further, and denote the maps $\beta_T$, $\gamma_T$, $\delta_T$ in the following way:\smallskip
\[
\beta_T\!:\,
\tikzmath[scale=.6]{\filldraw[fill=gray!30] (0,0) circle (1); 
\filldraw[fill=gray!40, dashed] (-.4,0) circle (.25);\filldraw[fill=white](.4,0) circle (.25);
\node[scale=.9] at (-.4,-.01) {$\scriptstyle T$};
\pgftransformxscale{-1}\pgftransformyscale{-1}
}\,\rightarrow\,
\tikzmath[scale=.6]{\filldraw[fill=gray!30] (0,0) circle (1); 
\filldraw[fill=white] (-.4,0) circle (.25);
\filldraw[fill=gray!40, dashed] (.4,0) circle (.25);
\node[scale=.9] at (.4,-.01) {$\scriptstyle T$};}\,,
\quad
\gamma_T\!:\,
\tikzmath[scale=.6]{\filldraw[fill=gray!30] (0,0) circle (1); 
\filldraw[fill=gray!40, dashed] (-.4,0) circle (.25);\filldraw[fill=white](.4,0) circle (.25);
\node[scale=.9] at (-.4,-.01) {$\scriptstyle T$};
\pgftransformxscale{-1}\pgftransformyscale{-1}
\fill[white] (95:1) -- ($(-.4,0)+(111:.25)$) -- +(0,-.05) -- +(.175,-.05) -- ($(-.4,0)+(69:.25)$) -- (85:1) -- +(0,.05) -- +(-.175,.05);
\draw[line cap=round] (-.4,0) +(111:.25) -- (95:1) (-.4,0) +(69:.25) -- (85:1);}
\,\rightarrow\,
\tikzmath[scale=.6]{\filldraw[fill=gray!30] (0,0) circle (1); 
\filldraw[fill=white] (-.4,0) circle (.25);
\filldraw[fill=gray!40, dashed] (.4,0) circle (.25);
\node[scale=.9] at (.4,-.01) {$\scriptstyle T$};
\pgftransformyscale{-1}
\fill[white] (95:1) -- ($(-.4,0)+(111:.25)$) -- +(0,-.05) -- +(.175,-.05) -- ($(-.4,0)+(69:.25)$) -- (85:1) -- +(0,.05) -- +(-.175,.05);
\draw[line cap=round] (-.4,0) +(111:.25) -- (95:1) (-.4,0) +(69:.25) -- (85:1);}\,,
\quad
\delta_T\!:\,
\tikzmath[scale=.6]{\filldraw[fill=gray!30] (0,0) circle (1); 
\filldraw[fill=gray!40, dashed] (-.4,0) circle (.25);\filldraw[fill=white](.4,0) circle (.25);
\node[scale=.9] at (-.4,-.01) {$\scriptstyle T$};
\pgftransformxscale{-1}
\fill[white] (95:1) -- ($(-.4,0)+(111:.25)$) -- +(0,-.05) -- +(.175,-.05) -- ($(-.4,0)+(69:.25)$) -- (85:1) -- +(0,.05) -- +(-.175,.05);
\draw[line cap=round] (-.4,0) +(111:.25) -- (95:1) (-.4,0) +(69:.25) -- (85:1);}
\,\rightarrow\,
\tikzmath[scale=.6]{\filldraw[fill=gray!30] (0,0) circle (1); 
\filldraw[fill=white] (-.4,0) circle (.25);
\filldraw[fill=gray!40, dashed] (.4,0) circle (.25);
\node[scale=.9] at (.4,-.01) {$\scriptstyle T$};
\fill[white] (95:1) -- ($(-.4,0)+(111:.25)$) -- +(0,-.05) -- +(.175,-.05) -- ($(-.4,0)+(69:.25)$) -- (85:1) -- +(0,.05) -- +(-.175,.05);
\draw[line cap=round] (-.4,0) +(111:.25) -- (95:1) (-.4,0) +(69:.25) -- (85:1);}\,.
\smallskip
\]

At this point, it becomes convenient to name some of the submanifolds of $\partial E$ and of $\partial F$:
\[
M:\quad\tikzmath[scale=.8]{
\useasboundingbox (-.1,1) rectangle (.1,-1);
\draw[line width=.7] (-.1,1) -- (-.1,.25)(.1,1) -- (.1,.25);
\draw[densely dotted] (-.1,1) -- (.1,1)(-.1,.25) -- (.1,.25);
\draw[->] (-.1,.58) -- (-.1,.57);\draw[->] (.1,.67) -- (.1,.68);}
\qquad
M':\quad\tikzmath[scale=.8]{
\useasboundingbox (-.1,1) rectangle (.1,-1);
\draw[densely dotted] (-.1,1) -- (-.1,.25)(.1,1) -- (.1,.25);
\draw[line width=.7] (-.1,1) -- (.1,1)(-.1,.25) -- (.1,.25);
\draw[->] (-.04,1) -- (-.05,1);\draw[->] (.04,.25) -- (.05,.25);}
\qquad\qquad
N:\quad\tikzmath[scale=.8]{
\useasboundingbox (-.1,1) rectangle (.1,-1);
\draw[line width=.7] (-.1,-1) -- (-.1,-.25)(.1,-1) -- (.1,-.25);
\draw[densely dotted] (-.1,-1) -- (.1,-1)(-.1,-.25) -- (.1,-.25);
\draw[->] (.1,-.58) -- (.1,-.57);\draw[->] (-.1,-.67) -- (-.1,-.68);}
\qquad
N':\quad\tikzmath[scale=.8]{
\useasboundingbox (-.1,1) rectangle (.1,-1);
\draw[densely dotted] (-.1,-1) -- (-.1,-.25)(.1,-1) -- (.1,-.25);
\draw[line width=.7] (-.1,-1) -- (.1,-1)(-.1,-.25) -- (.1,-.25);
\draw[->] (-.04,-.25) -- (-.05,-.25);\draw[->] (.04,-1) -- (.05,-1);}
\qquad
\medskip\]
The maps $\beta_T$, $\gamma_T$, $\delta_T$,  then satisfy the relations
\[
\gamma_T\boxtimes_{\cala(N)} {V(F)} \,=\, \delta_T\boxtimes_{\cala(M)} {V(E)} \,=\, \beta_T.
\]
Applying the functor $-\boxtimes_{\cala(M')\bar\otimes\cala(N')} (\overline{V(E)}\otimes \overline{V(F)})$ to the above equality,
we get an equation in the space of maps from\smallskip
\[
\begin{split}
\tikzmath[scale=.8]{\pgftransformyscale{.9}\pgftransformxscale{-1.8}
\useasboundingbox (-1,-1) rectangle (1,1.1);
\coordinate (a) at ($(-.4,0)+(111:.25)$);\coordinate (b) at ($(-.4,0)+(69:.25)$);\coordinate (c) at ($(-.4,0)+(-69:.25)$);\coordinate (d) at ($(-.4,0)+(-111:.25)$);
\coordinate (A) at (95:1);\coordinate (B) at (85:1);\coordinate (C) at (-85:1);\coordinate (D) at (-95:1);
\fill[fill=gray!30, even odd rule] (0,0) circle (1)(-.4,0) circle (.25);
\draw[line join=bevel] (-.4,0)+(360-111:.25) arc (360-111:111:.25) (95:1) arc (95:360-95:1);
\draw[line join=bevel] (-.4,0)+(-69:.25) arc (-69:69:.25) (85:1) arc (85:-85:1);
\draw[ultra thin] (-.4,0)+(111:.25) -- (95:1) (360-95:1) -- ($(-.4,0)+(-111:.25)$);
\draw[ultra thin] (-.4,0)+(69:.25) -- (85:1) (-85:1) -- ($(-.4,0)+(-69:.25)$);
\fill[fill=gray!40] (.4,0) circle (.25);\draw[dashed] (.4,0) circle (.25);\node at (.4,.01) {$\scriptstyle T$};
\draw[fill=gray!70] (B) to[in=0, out=90, looseness=.57] ($(b)+(.29,1.1)$) -- ($(a)+(.29,1.1)$) to[out=0, in=90, looseness=.57] (A) arc (95:85:1); 
\draw[fill=gray!30] (a) to[in=200, out=93, looseness=.47] ($(a)+(.29,1.1)$) -- ($(b)+(.29,1.1)$) to[out=200, in=93, looseness=.47] (b) arc (69:111:.25);
\draw[fill=gray!70] (d) to[in=180, out=90, looseness=.57] ($(D)+(-.29,1.1)$) -- ($(C)+(-.29,1.1)$) to[out=180, in=90, looseness=.57] (c) arc (-69:-111:.25); 
\draw[fill=gray!30] (C) to[in=-20, out=87, looseness=.47] ($(C)+(-.29,1.1)$) -- ($(D)+(-.29,1.1)$) to[out=-20, in=87, looseness=.47] (D) arc (-95:-85:1);
}
\,\,:=\,\,T\,\boxtimes_{\cala(S_1)}&
V\Big(\tikzmath[scale=.35]{\filldraw[fill=gray!30] (0,0) circle (1); 
\filldraw[fill=white] (-.4,0) circle (.25)(.4,0) circle (.25);}\Big)\boxtimes_{\cala(M')\bar\otimes\cala(N')} (\overline{V(E)}\otimes \overline{V(F)})\\
\cong\,\,\,
\tikzmath[scale=.6]{\filldraw[fill=gray!30] (0,0) circle (1); 
\filldraw[fill=gray!40, dashed] (-.4,0) circle (.25);\filldraw[fill=white](.4,0) circle (.25);
\node[scale=.9] at (-.4,-.01) {$\scriptstyle T$};
\pgftransformxscale{-1}
\fill[white] (95:1) -- ($(-.4,0)+(111:.25)$) -- +(0,-.05) -- +(.175,-.05) -- ($(-.4,0)+(69:.25)$) -- (85:1) -- +(0,.05) -- +(-.175,.05);
\draw[line cap=round] (-.4,0) +(111:.25) -- (95:1) (-.4,0) +(69:.25) -- (85:1);
\pgftransformyscale{-1}
\fill[white] (95:1) -- ($(-.4,0)+(111:.25)$) -- +(0,-.05) -- +(.175,-.05) -- ($(-.4,0)+(69:.25)$) -- (85:1) -- +(0,.05) -- +(-.175,.05);
\draw[line cap=round] (-.4,0) +(111:.25) -- (95:1) (-.4,0) +(69:.25) -- (85:1);}\,\,
\boxtimes_{\cala(M)} \Big(&{V(E)\boxtimes _{\cala(M')}\overline{V(E)}}\Big)
\boxtimes_{\cala(N)} \Big({V(F)\boxtimes _{\cala(N')}\overline{V(F)}}\Big)
\end{split}
\]
to\smallskip
\[
\begin{split}
\tikzmath[scale=.8]{\pgftransformyscale{.9}\pgftransformxscale{1.8}
\useasboundingbox (-1,-1) rectangle (1,1.1);
\coordinate (a) at ($(-.4,0)+(111:.25)$);\coordinate (b) at ($(-.4,0)+(69:.25)$);\coordinate (c) at ($(-.4,0)+(-69:.25)$);\coordinate (d) at ($(-.4,0)+(-111:.25)$);
\coordinate (A) at (95:1);\coordinate (B) at (85:1);\coordinate (C) at (-85:1);\coordinate (D) at (-95:1);
\fill[fill=gray!30, even odd rule] (0,0) circle (1)(-.4,0) circle (.25);
\draw[line join=bevel] (-.4,0)+(360-111:.25) arc (360-111:111:.25) (95:1) arc (95:360-95:1);
\draw[line join=bevel] (-.4,0)+(-69:.25) arc (-69:69:.25) (85:1) arc (85:-85:1);
\draw[ultra thin] (-.4,0)+(111:.25) -- (95:1) (360-95:1) -- ($(-.4,0)+(-111:.25)$);
\draw[ultra thin] (-.4,0)+(69:.25) -- (85:1) (-85:1) -- ($(-.4,0)+(-69:.25)$);
\fill[fill=gray!40] (.4,0) circle (.25);\draw[dashed] (.4,0) circle (.25);\node at (.4,.01) {$\scriptstyle T$};
\draw[fill=gray!70] (B) to[in=0, out=90, looseness=.57] ($(b)+(.29,1.1)$) -- ($(a)+(.29,1.1)$) to[out=0, in=90, looseness=.57] (A) arc (95:85:1); 
\draw[fill=gray!30] (a) to[in=200, out=93, looseness=.47] ($(a)+(.29,1.1)$) -- ($(b)+(.29,1.1)$) to[out=200, in=93, looseness=.47] (b) arc (69:111:.25);
\draw[fill=gray!70] (d) to[in=180, out=90, looseness=.57] ($(D)+(-.29,1.1)$) -- ($(C)+(-.29,1.1)$) to[out=180, in=90, looseness=.57] (c) arc (-69:-111:.25); 
\draw[fill=gray!30] (C) to[in=-20, out=87, looseness=.47] ($(C)+(-.29,1.1)$) -- ($(D)+(-.29,1.1)$) to[out=-20, in=87, looseness=.47] (D) arc (-95:-85:1);
}
\,\,:=\,\,T\boxtimes_{\cala(S_2)}&
V\Big(\tikzmath[scale=.35]{\filldraw[fill=gray!30] (0,0) circle (1); 
\filldraw[fill=white] (-.4,0) circle (.25)(.4,0) circle (.25);}\Big)\boxtimes_{\cala(M')\bar\otimes\cala(N')} (\overline{V(E)}\otimes \overline{V(F)})\\
\cong\,\,\,
\tikzmath[scale=.6]{\filldraw[fill=gray!30] (0,0) circle (1); 
\filldraw[fill=gray!40, dashed] (.4,0) circle (.25);\filldraw[fill=white](-.4,0) circle (.25);
\node[scale=.9] at (.4,-.01) {$\scriptstyle T$};
\fill[white] (95:1) -- ($(-.4,0)+(111:.25)$) -- +(0,-.05) -- +(.175,-.05) -- ($(-.4,0)+(69:.25)$) -- (85:1) -- +(0,.05) -- +(-.175,.05);
\draw[line cap=round] (-.4,0) +(111:.25) -- (95:1) (-.4,0) +(69:.25) -- (85:1);
\pgftransformyscale{-1}
\fill[white] (95:1) -- ($(-.4,0)+(111:.25)$) -- +(0,-.05) -- +(.175,-.05) -- ($(-.4,0)+(69:.25)$) -- (85:1) -- +(0,.05) -- +(-.175,.05);
\draw[line cap=round] (-.4,0) +(111:.25) -- (95:1) (-.4,0) +(69:.25) -- (85:1);}
\boxtimes_{\cala(M)} \Big(&{V(E)\boxtimes _{\cala(M')}\overline{V(E)}}\Big)
\boxtimes_{\cala(N)} \Big({V(F)\boxtimes _{\cala(N')}\overline{V(F)}}\Big).
\end{split}
\]
More specifically, we learn that the two maps
\[
\begin{split}
\Gamma:=\,\,&\Big(\gamma_T\boxtimes_{\cala(N)} {V(F)}\Big)\boxtimes _{\cala(M')\,\bar\otimes\,\cala(N')} {\overline{V(E)}\otimes\overline{V(F)}}
\\
&\,\,=\Big(\gamma_T\boxtimes_{\cala(M')}{\overline{V(E)}}\Big)\boxtimes_{\cala(N)} \Big({V(F)\boxtimes _{\cala(N')}\overline{V(F)}}\Big)
\end{split}
\]
and
\[
\begin{split}
\Delta:=\,\,&\Big(\delta_T\boxtimes_{\cala(M)} {V(E)}\Big)\boxtimes _{\cala(M')\,\bar\otimes\,\cala(N')} {\overline{V(E)}\otimes\overline{V(F)}}
\\
&\,\,=\Big(\delta_T\boxtimes_{\cala(N')}{\overline{V(F)}}\Big)\boxtimes_{\cala(N)} \Big({V(E)\boxtimes _{\cala(M')}\overline{V(E)}}\Big)
\end{split}
\]
are both equal to
\[
B:=\beta_T\boxtimes _{\cala(M')\,\bar\otimes\,\cala(N')} {\overline{V(E)}\otimes\overline{V(F)}}.
\medskip
\]

Summarizing, we have three equal maps (actually isomorphisms)\medskip
\begin{equation}\label{eq: X}
\Gamma=\Delta=B\,:\,\,\,\,\,
\tikzmath[scale=.8]{\pgftransformyscale{.9}\pgftransformxscale{-1.8}
\useasboundingbox (-1,-1) rectangle (1,1.1);
\coordinate (a) at ($(-.4,0)+(111:.25)$);\coordinate (b) at ($(-.4,0)+(69:.25)$);\coordinate (c) at ($(-.4,0)+(-69:.25)$);\coordinate (d) at ($(-.4,0)+(-111:.25)$);
\coordinate (A) at (95:1);\coordinate (B) at (85:1);\coordinate (C) at (-85:1);\coordinate (D) at (-95:1);
\fill[fill=gray!30, even odd rule] (0,0) circle (1)(-.4,0) circle (.25);
\draw[line join=bevel] (-.4,0)+(360-111:.25) arc (360-111:111:.25) (95:1) arc (95:360-95:1);
\draw[line join=bevel] (-.4,0)+(-69:.25) arc (-69:69:.25) (85:1) arc (85:-85:1);
\draw[ultra thin] (-.4,0)+(111:.25) -- (95:1) (360-95:1) -- ($(-.4,0)+(-111:.25)$);
\draw[ultra thin] (-.4,0)+(69:.25) -- (85:1) (-85:1) -- ($(-.4,0)+(-69:.25)$);
\fill[fill=gray!40] (.4,0) circle (.25);\draw[dashed] (.4,0) circle (.25);\node at (.4,.01) {$\scriptstyle T$};
\draw[fill=gray!70] (B) to[in=0, out=90, looseness=.57] ($(b)+(.29,1.1)$) -- ($(a)+(.29,1.1)$) to[out=0, in=90, looseness=.57] (A) arc (95:85:1); 
\draw[fill=gray!30] (a) to[in=200, out=93, looseness=.47] ($(a)+(.29,1.1)$) -- ($(b)+(.29,1.1)$) to[out=200, in=93, looseness=.47] (b) arc (69:111:.25);
\draw[fill=gray!70] (d) to[in=180, out=90, looseness=.57] ($(D)+(-.29,1.1)$) -- ($(C)+(-.29,1.1)$) to[out=180, in=90, looseness=.57] (c) arc (-69:-111:.25); 
\draw[fill=gray!30] (C) to[in=-20, out=87, looseness=.47] ($(C)+(-.29,1.1)$) -- ($(D)+(-.29,1.1)$) to[out=-20, in=87, looseness=.47] (D) arc (-95:-85:1);
}
\,\,\,\,\to\,\,\,\,
\tikzmath[scale=.8]{\pgftransformyscale{.9}\pgftransformxscale{1.8}
\useasboundingbox (-1,-1) rectangle (1,1.1);
\coordinate (a) at ($(-.4,0)+(111:.25)$);\coordinate (b) at ($(-.4,0)+(69:.25)$);\coordinate (c) at ($(-.4,0)+(-69:.25)$);\coordinate (d) at ($(-.4,0)+(-111:.25)$);
\coordinate (A) at (95:1);\coordinate (B) at (85:1);\coordinate (C) at (-85:1);\coordinate (D) at (-95:1);
\fill[fill=gray!30, even odd rule] (0,0) circle (1)(-.4,0) circle (.25);
\draw[line join=bevel] (-.4,0)+(360-111:.25) arc (360-111:111:.25) (95:1) arc (95:360-95:1);
\draw[line join=bevel] (-.4,0)+(-69:.25) arc (-69:69:.25) (85:1) arc (85:-85:1);
\draw[ultra thin] (-.4,0)+(111:.25) -- (95:1) (360-95:1) -- ($(-.4,0)+(-111:.25)$);
\draw[ultra thin] (-.4,0)+(69:.25) -- (85:1) (-85:1) -- ($(-.4,0)+(-69:.25)$);
\fill[fill=gray!40] (.4,0) circle (.25);\draw[dashed] (.4,0) circle (.25);\node at (.4,.01) {$\scriptstyle T$};
\draw[fill=gray!70] (B) to[in=0, out=90, looseness=.57] ($(b)+(.29,1.1)$) -- ($(a)+(.29,1.1)$) to[out=0, in=90, looseness=.57] (A) arc (95:85:1); 
\draw[fill=gray!30] (a) to[in=200, out=93, looseness=.47] ($(a)+(.29,1.1)$) -- ($(b)+(.29,1.1)$) to[out=200, in=93, looseness=.47] (b) arc (69:111:.25);
\draw[fill=gray!70] (d) to[in=180, out=90, looseness=.57] ($(D)+(-.29,1.1)$) -- ($(C)+(-.29,1.1)$) to[out=180, in=90, looseness=.57] (c) arc (-69:-111:.25); 
\draw[fill=gray!30] (C) to[in=-20, out=87, looseness=.47] ($(C)+(-.29,1.1)$) -- ($(D)+(-.29,1.1)$) to[out=-20, in=87, looseness=.47] (D) arc (-95:-85:1);
}\,\,,\smallskip
\end{equation}
equivariant for
$
\cala\big(\tikzmath[scale=.3]{\draw (0,0) circle (1); 
\filldraw[fill=white] (0,0) circle (.35);
\fill[white] (97:1) -- ($(0,0)+(111:.35)$) -- +(0,-.05) -- +(.25,-.05) -- ($(0,0)+(69:.35)$) -- (83:1) -- +(0,.05) -- +(-.25,.05);
\draw[line cap=round] (0,0) +(111:.35) -- (97:1) (0,0) +(69:.35) -- (83:1);
\pgftransformyscale{-1}
\fill[white] (97:1) -- ($(0,0)+(111:.35)$) -- +(0,-.05) -- +(.25,-.05) -- ($(0,0)+(69:.35)$) -- (83:1) -- +(0,.05) -- +(-.25,.05);
\draw[line cap=round] (0,0) +(111:.35) -- (97:1) (0,0) +(69:.35) -- (83:1);}
\big)
$
(where the actions of that algebra come from appropriate diffeomorphisms with the boundaries of the two $2$-manifolds in \eqref{eq: X}).

Since ${}_{\cala(M)}V(E)_{\cala(M')}$ and ${}_{\cala(N)}V(F)_{\cala(N')}$ are dualizable bimodules, we have inclusions
$L^2\cala(M)\subset
{V(E)\boxtimes _{\cala(M')}\overline{V(E)}}
$ and $L^2\cala(N)\subset
{V(F)\boxtimes _{\cala(N')}\overline{V(F)}}$.
It follows that there is a diagram of inclusions
\begin{equation}\label{eq: Y}
\tikzmath{
\node (Q) at (0,0) {$
\tikzmath[scale=.8]{\pgftransformyscale{.9}\pgftransformxscale{1.8}
\useasboundingbox (-1,-1) rectangle (1,1.1);
\coordinate (a) at ($(-.4,0)+(111:.25)$);\coordinate (b) at ($(-.4,0)+(69:.25)$);\coordinate (c) at ($(-.4,0)+(-69:.25)$);\coordinate (d) at ($(-.4,0)+(-111:.25)$);
\coordinate (A) at (95:1);\coordinate (B) at (85:1);\coordinate (C) at (-85:1);\coordinate (D) at (-95:1);
\fill[fill=gray!30, even odd rule] (0,0) circle (1)(-.4,0) circle (.25);
\draw[line join=bevel] (-.4,0)+(360-111:.25) arc (360-111:111:.25) (95:1) arc (95:360-95:1);
\draw[line join=bevel] (-.4,0)+(-69:.25) arc (-69:69:.25) (85:1) arc (85:-85:1);
\fill[white] (95:1) -- ($(-.4,0)+(111:.25)$) -- +(0,-.05) -- +(.175,-.05) -- ($(-.4,0)+(69:.25)$) -- (85:1) -- +(0,.05) -- +(-.175,.05);
\fill[white] (-95:1) -- ($(-.4,0)+(-111:.25)$) -- +(0,.05) -- +(.175,.05) -- ($(-.4,0)+(-69:.25)$) -- (-85:1) -- +(0,-.05) -- +(-.175,-.05);
\draw (-.4,0)+(111:.25) -- (95:1);
\draw (360-95:1) -- ($(-.4,0)+(-111:.25)$);
\draw (-.4,0)+(69:.25) -- (85:1);
\draw (-85:1) -- ($(-.4,0)+(-69:.25)$);
\fill[fill=gray!40] (.4,0) circle (.25);\draw[dashed] (.4,0) circle (.25);\node at (.4,.01) {$\scriptstyle T$};
}$};
\node (R) at (0,-3) {$
\tikzmath[scale=.8]{\pgftransformyscale{.9}\pgftransformxscale{1.8}
\useasboundingbox (-1,-1) rectangle (1,1.1);
\coordinate (a) at ($(-.4,0)+(111:.25)$);\coordinate (b) at ($(-.4,0)+(69:.25)$);\coordinate (c) at ($(-.4,0)+(-69:.25)$);\coordinate (d) at ($(-.4,0)+(-111:.25)$);
\coordinate (A) at (95:1);\coordinate (B) at (85:1);\coordinate (C) at (-85:1);\coordinate (D) at (-95:1);
\fill[fill=gray!30, even odd rule] (0,0) circle (1)(-.4,0) circle (.25);
\fill[white] (95:1) -- ($(-.4,0)+(111:.25)$) -- +(0,-.05) -- +(.175,-.05) -- ($(-.4,0)+(69:.25)$) -- (85:1) -- +(0,.05) -- +(-.175,.05);
\draw[line join=bevel] (-.4,0)+(360-111:.25) arc (360-111:111:.25) (95:1) arc (95:360-95:1);
\draw[line join=bevel] (-.4,0)+(-69:.25) arc (-69:69:.25) (85:1) arc (85:-85:1);
\draw (-.4,0)+(111:.25) -- (95:1);
\draw[ultra thin] (360-95:1) -- ($(-.4,0)+(-111:.25)$);
\draw (-.4,0)+(69:.25) -- (85:1);
\draw[ultra thin] (-85:1) -- ($(-.4,0)+(-69:.25)$);
\fill[fill=gray!40] (.4,0) circle (.25);\draw[dashed] (.4,0) circle (.25);\node at (.4,.01) {$\scriptstyle T$};
\draw[fill=gray!70] (d) to[in=180, out=90, looseness=.57] ($(D)+(-.29,1.1)$) -- ($(C)+(-.29,1.1)$) to[out=180, in=90, looseness=.57] (c) arc (-69:-111:.25); 
\draw[fill=gray!30] (C) to[in=-20, out=87, looseness=.47] ($(C)+(-.29,1.1)$) -- ($(D)+(-.29,1.1)$) to[out=-20, in=87, looseness=.47] (D) arc (-95:-85:1);
}$};
\node (S) at (5,0) {$
\tikzmath[scale=.8]{\pgftransformyscale{.9}\pgftransformxscale{1.8}
\useasboundingbox (-1,-1) rectangle (1,1.1);
\coordinate (a) at ($(-.4,0)+(111:.25)$);\coordinate (b) at ($(-.4,0)+(69:.25)$);\coordinate (c) at ($(-.4,0)+(-69:.25)$);\coordinate (d) at ($(-.4,0)+(-111:.25)$);
\coordinate (A) at (95:1);\coordinate (B) at (85:1);\coordinate (C) at (-85:1);\coordinate (D) at (-95:1);
\fill[fill=gray!30, even odd rule] (0,0) circle (1)(-.4,0) circle (.25);
\draw[line join=bevel] (-.4,0)+(360-111:.25) arc (360-111:111:.25) (95:1) arc (95:360-95:1);
\draw[line join=bevel] (-.4,0)+(-69:.25) arc (-69:69:.25) (85:1) arc (85:-85:1);
\fill[white] (-95:1) -- ($(-.4,0)+(-111:.25)$) -- +(0,.05) -- +(.175,.05) -- ($(-.4,0)+(-69:.25)$) -- (-85:1) -- +(0,-.05) -- +(-.175,-.05);
\draw[ultra thin] (-.4,0)+(111:.25) -- (95:1);
\draw (360-95:1) -- ($(-.4,0)+(-111:.25)$);
\draw[ultra thin] (-.4,0)+(69:.25) -- (85:1);
\draw (-85:1) -- ($(-.4,0)+(-69:.25)$);
\fill[fill=gray!40] (.4,0) circle (.25);\draw[dashed] (.4,0) circle (.25);\node at (.4,.01) {$\scriptstyle T$};
\draw[fill=gray!70] (B) to[in=0, out=90, looseness=.57] ($(b)+(.29,1.1)$) -- ($(a)+(.29,1.1)$) to[out=0, in=90, looseness=.57] (A) arc (95:85:1); 
\draw[fill=gray!30] (a) to[in=200, out=93, looseness=.47] ($(a)+(.29,1.1)$) -- ($(b)+(.29,1.1)$) to[out=200, in=93, looseness=.47] (b) arc (69:111:.25);
}$};
\node (T) at (5,-3) {$
\tikzmath[scale=.8]{\pgftransformyscale{.9}\pgftransformxscale{1.8}
\useasboundingbox (-1,-1) rectangle (1,1.1);
\coordinate (a) at ($(-.4,0)+(111:.25)$);\coordinate (b) at ($(-.4,0)+(69:.25)$);\coordinate (c) at ($(-.4,0)+(-69:.25)$);\coordinate (d) at ($(-.4,0)+(-111:.25)$);
\coordinate (A) at (95:1);\coordinate (B) at (85:1);\coordinate (C) at (-85:1);\coordinate (D) at (-95:1);
\fill[fill=gray!30, even odd rule] (0,0) circle (1)(-.4,0) circle (.25);
\draw[line join=bevel] (-.4,0)+(360-111:.25) arc (360-111:111:.25) (95:1) arc (95:360-95:1);
\draw[line join=bevel] (-.4,0)+(-69:.25) arc (-69:69:.25) (85:1) arc (85:-85:1);
\draw[ultra thin] (-.4,0)+(111:.25) -- (95:1) (360-95:1) -- ($(-.4,0)+(-111:.25)$);
\draw[ultra thin] (-.4,0)+(69:.25) -- (85:1) (-85:1) -- ($(-.4,0)+(-69:.25)$);
\fill[fill=gray!40] (.4,0) circle (.25);\draw[dashed] (.4,0) circle (.25);\node at (.4,.01) {$\scriptstyle T$};
\draw[fill=gray!70] (B) to[in=0, out=90, looseness=.57] ($(b)+(.29,1.1)$) -- ($(a)+(.29,1.1)$) to[out=0, in=90, looseness=.57] (A) arc (95:85:1); 
\draw[fill=gray!30] (a) to[in=200, out=93, looseness=.47] ($(a)+(.29,1.1)$) -- ($(b)+(.29,1.1)$) to[out=200, in=93, looseness=.47] (b) arc (69:111:.25);
\draw[fill=gray!70] (d) to[in=180, out=90, looseness=.57] ($(D)+(-.29,1.1)$) -- ($(C)+(-.29,1.1)$) to[out=180, in=90, looseness=.57] (c) arc (-69:-111:.25); 
\draw[fill=gray!30] (C) to[in=-20, out=87, looseness=.47] ($(C)+(-.29,1.1)$) -- ($(D)+(-.29,1.1)$) to[out=-20, in=87, looseness=.47] (D) arc (-95:-85:1);
}$};
\draw[<-, shorten <=4] (R) -- ($(Q.south)+(0,-.25)$) arc (0:180:.1);
\draw[<-, shorten <=4] (S) -- ($(Q.east)+(.25,0)$) arc (-90:-270:.1);
\draw[<-, shorten <=4] (T) -- ($(R.east)+(.25,0)$) arc (-90:-270:.1);
\draw[<-, shorten <=7] (T) -- ($(S.south)+(0,-.25)$) arc (0:180:.1);
}
\bigskip
\end{equation}
and that
\begin{equation}\label{eq: final3}
\tikzmath[scale=.6]{\pgftransformyscale{.9}\pgftransformxscale{1.8}
\useasboundingbox (-1,-1) rectangle (1,1.1);
\coordinate (a) at ($(-.4,0)+(111:.25)$);\coordinate (b) at ($(-.4,0)+(69:.25)$);\coordinate (c) at ($(-.4,0)+(-69:.25)$);\coordinate (d) at ($(-.4,0)+(-111:.25)$);
\coordinate (A) at (95:1);\coordinate (B) at (85:1);\coordinate (C) at (-85:1);\coordinate (D) at (-95:1);
\fill[fill=gray!30, even odd rule] (0,0) circle (1)(-.4,0) circle (.25);
\fill[white] (95:1) -- ($(-.4,0)+(111:.25)$) -- +(0,-.05) -- +(.175,-.05) -- ($(-.4,0)+(69:.25)$) -- (85:1) -- +(0,.05) -- +(-.175,.05);
\draw[line join=bevel] (-.4,0)+(360-111:.25) arc (360-111:111:.25) (95:1) arc (95:360-95:1);
\draw[line join=bevel] (-.4,0)+(-69:.25) arc (-69:69:.25) (85:1) arc (85:-85:1);
\draw (-.4,0)+(111:.25) -- (95:1);
\draw[ultra thin] (360-95:1) -- ($(-.4,0)+(-111:.25)$);
\draw (-.4,0)+(69:.25) -- (85:1);
\draw[ultra thin] (-85:1) -- ($(-.4,0)+(-69:.25)$);
\fill[fill=gray!40] (.4,0) circle (.25);\draw[dashed] (.4,0) circle (.25);\node at (.4,.01) {$\scriptstyle T$};
\draw[fill=gray!70] (d) to[in=180, out=90, looseness=.57] ($(D)+(-.29,1.1)$) -- ($(C)+(-.29,1.1)$) to[out=180, in=90, looseness=.57] (c) arc (-69:-111:.25); 
\draw[fill=gray!30] (C) to[in=-20, out=87, looseness=.47] ($(C)+(-.29,1.1)$) -- ($(D)+(-.29,1.1)$) to[out=-20, in=87, looseness=.47] (D) arc (-95:-85:1);
}
\,\,\cap\,\,
\tikzmath[scale=.6]{\pgftransformyscale{.9}\pgftransformxscale{1.8}
\useasboundingbox (-1,-1) rectangle (1,1.1);
\coordinate (a) at ($(-.4,0)+(111:.25)$);\coordinate (b) at ($(-.4,0)+(69:.25)$);\coordinate (c) at ($(-.4,0)+(-69:.25)$);\coordinate (d) at ($(-.4,0)+(-111:.25)$);
\coordinate (A) at (95:1);\coordinate (B) at (85:1);\coordinate (C) at (-85:1);\coordinate (D) at (-95:1);
\fill[fill=gray!30, even odd rule] (0,0) circle (1)(-.4,0) circle (.25);
\draw[line join=bevel] (-.4,0)+(360-111:.25) arc (360-111:111:.25) (95:1) arc (95:360-95:1);
\draw[line join=bevel] (-.4,0)+(-69:.25) arc (-69:69:.25) (85:1) arc (85:-85:1);
\fill[white] (-95:1) -- ($(-.4,0)+(-111:.25)$) -- +(0,.05) -- +(.175,.05) -- ($(-.4,0)+(-69:.25)$) -- (-85:1) -- +(0,-.05) -- +(-.175,-.05);
\draw[ultra thin] (-.4,0)+(111:.25) -- (95:1);
\draw (360-95:1) -- ($(-.4,0)+(-111:.25)$);
\draw[ultra thin] (-.4,0)+(69:.25) -- (85:1);
\draw (-85:1) -- ($(-.4,0)+(-69:.25)$);
\fill[fill=gray!40] (.4,0) circle (.25);\draw[dashed] (.4,0) circle (.25);\node at (.4,.01) {$\scriptstyle T$};
\draw[fill=gray!70] (B) to[in=0, out=90, looseness=.57] ($(b)+(.29,1.1)$) -- ($(a)+(.29,1.1)$) to[out=0, in=90, looseness=.57] (A) arc (95:85:1); 
\draw[fill=gray!30] (a) to[in=200, out=93, looseness=.47] ($(a)+(.29,1.1)$) -- ($(b)+(.29,1.1)$) to[out=200, in=93, looseness=.47] (b) arc (69:111:.25);
}
\,\,=\,\,
\tikzmath[scale=.6]{\pgftransformyscale{.9}\pgftransformxscale{1.8}
\useasboundingbox (-1,-1) rectangle (1,1.1);
\coordinate (a) at ($(-.4,0)+(111:.25)$);\coordinate (b) at ($(-.4,0)+(69:.25)$);\coordinate (c) at ($(-.4,0)+(-69:.25)$);\coordinate (d) at ($(-.4,0)+(-111:.25)$);
\coordinate (A) at (95:1);\coordinate (B) at (85:1);\coordinate (C) at (-85:1);\coordinate (D) at (-95:1);
\fill[fill=gray!30, even odd rule] (0,0) circle (1)(-.4,0) circle (.25);
\draw[line join=bevel] (-.4,0)+(360-111:.25) arc (360-111:111:.25) (95:1) arc (95:360-95:1);
\draw[line join=bevel] (-.4,0)+(-69:.25) arc (-69:69:.25) (85:1) arc (85:-85:1);
\fill[white] (95:1) -- ($(-.4,0)+(111:.25)$) -- +(0,-.05) -- +(.175,-.05) -- ($(-.4,0)+(69:.25)$) -- (85:1) -- +(0,.05) -- +(-.175,.05);
\fill[white] (-95:1) -- ($(-.4,0)+(-111:.25)$) -- +(0,.05) -- +(.175,.05) -- ($(-.4,0)+(-69:.25)$) -- (-85:1) -- +(0,-.05) -- +(-.175,-.05);
\draw (-.4,0)+(111:.25) -- (95:1);
\draw (360-95:1) -- ($(-.4,0)+(-111:.25)$);
\draw (-.4,0)+(69:.25) -- (85:1);
\draw (-85:1) -- ($(-.4,0)+(-69:.25)$);
\fill[fill=gray!40] (.4,0) circle (.25);\draw[dashed] (.4,0) circle (.25);\node at (.4,.01) {$\scriptstyle T$};
}
\medskip
\end{equation}
Note that the inclusions \eqref{eq: Y} are all compatible with the actions of
$
\cala\big(\tikzmath[scale=.3]{\draw (0,0) circle (1); 
\filldraw[fill=white] (0,0) circle (.35);
\fill[white] (97:1) -- ($(0,0)+(111:.35)$) -- +(0,-.05) -- +(.25,-.05) -- ($(0,0)+(69:.35)$) -- (83:1) -- +(0,.05) -- +(-.25,.05);
\draw[line cap=round] (0,0) +(111:.35) -- (97:1) (0,0) +(69:.35) -- (83:1);
\pgftransformyscale{-1}
\fill[white] (97:1) -- ($(0,0)+(111:.35)$) -- +(0,-.05) -- +(.25,-.05) -- ($(0,0)+(69:.35)$) -- (83:1) -- +(0,.05) -- +(-.25,.05);
\draw[line cap=round] (0,0) +(111:.35) -- (97:1) (0,0) +(69:.35) -- (83:1);}
\big)
$
(where again the actions come from appropriate diffeomorphisms with the boundaries of the above $2$-manifolds).

Let us now look at the image of\,
\(
\tikzmath[scale=.6]{\pgftransformyscale{.9}\pgftransformxscale{-1.8}
\useasboundingbox (-1,-1) rectangle (1,1.1);
\coordinate (a) at ($(-.4,0)+(111:.25)$);\coordinate (b) at ($(-.4,0)+(69:.25)$);\coordinate (c) at ($(-.4,0)+(-69:.25)$);\coordinate (d) at ($(-.4,0)+(-111:.25)$);
\coordinate (A) at (95:1);\coordinate (B) at (85:1);\coordinate (C) at (-85:1);\coordinate (D) at (-95:1);
\fill[fill=gray!30, even odd rule] (0,0) circle (1)(-.4,0) circle (.25);
\draw[line join=bevel] (-.4,0)+(360-111:.25) arc (360-111:111:.25) (95:1) arc (95:360-95:1);
\draw[line join=bevel] (-.4,0)+(-69:.25) arc (-69:69:.25) (85:1) arc (85:-85:1);
\fill[white] (95:1) -- ($(-.4,0)+(111:.25)$) -- +(0,-.05) -- +(.175,-.05) -- ($(-.4,0)+(69:.25)$) -- (85:1) -- +(0,.05) -- +(-.175,.05);
\fill[white] (-95:1) -- ($(-.4,0)+(-111:.25)$) -- +(0,.05) -- +(.175,.05) -- ($(-.4,0)+(-69:.25)$) -- (-85:1) -- +(0,-.05) -- +(-.175,-.05);
\draw (-.4,0)+(111:.25) -- (95:1);
\draw (360-95:1) -- ($(-.4,0)+(-111:.25)$);
\draw (-.4,0)+(69:.25) -- (85:1);
\draw (-85:1) -- ($(-.4,0)+(-69:.25)$);
\fill[fill=gray!40] (.4,0) circle (.25);\draw[dashed] (.4,0) circle (.25);\node at (.4,.01) {$\scriptstyle T$};
}
\,\subset\,
\tikzmath[scale=.6]{\pgftransformyscale{.9}\pgftransformxscale{-1.8}
\useasboundingbox (-1,-1) rectangle (1,1.1);
\coordinate (a) at ($(-.4,0)+(111:.25)$);\coordinate (b) at ($(-.4,0)+(69:.25)$);\coordinate (c) at ($(-.4,0)+(-69:.25)$);\coordinate (d) at ($(-.4,0)+(-111:.25)$);
\coordinate (A) at (95:1);\coordinate (B) at (85:1);\coordinate (C) at (-85:1);\coordinate (D) at (-95:1);
\fill[fill=gray!30, even odd rule] (0,0) circle (1)(-.4,0) circle (.25);
\draw[line join=bevel] (-.4,0)+(360-111:.25) arc (360-111:111:.25) (95:1) arc (95:360-95:1);
\draw[line join=bevel] (-.4,0)+(-69:.25) arc (-69:69:.25) (85:1) arc (85:-85:1);
\draw[ultra thin] (-.4,0)+(111:.25) -- (95:1) (360-95:1) -- ($(-.4,0)+(-111:.25)$);
\draw[ultra thin] (-.4,0)+(69:.25) -- (85:1) (-85:1) -- ($(-.4,0)+(-69:.25)$);
\fill[fill=gray!40] (.4,0) circle (.25);\draw[dashed] (.4,0) circle (.25);\node at (.4,.01) {$\scriptstyle T$};
\draw[fill=gray!70] (B) to[in=0, out=90, looseness=.57] ($(b)+(.29,1.1)$) -- ($(a)+(.29,1.1)$) to[out=0, in=90, looseness=.57] (A) arc (95:85:1); 
\draw[fill=gray!30] (a) to[in=200, out=93, looseness=.47] ($(a)+(.29,1.1)$) -- ($(b)+(.29,1.1)$) to[out=200, in=93, looseness=.47] (b) arc (69:111:.25);
\draw[fill=gray!70] (d) to[in=180, out=90, looseness=.57] ($(D)+(-.29,1.1)$) -- ($(C)+(-.29,1.1)$) to[out=180, in=90, looseness=.57] (c) arc (-69:-111:.25); 
\draw[fill=gray!30] (C) to[in=-20, out=87, looseness=.47] ($(C)+(-.29,1.1)$) -- ($(D)+(-.29,1.1)$) to[out=-20, in=87, looseness=.47] (D) arc (-95:-85:1);
}
\)
\,under the map~$B$.
Since $B=\Gamma$ is the result of applying $-\boxtimes_{\cala(N)} ({V(F)\boxtimes _{\cala(N')}\overline{V(F)}})$ to some map, it follows that
\begin{equation}\label{eq: final1}
B\bigg(
\tikzmath[scale=.6]{\pgftransformyscale{.9}\pgftransformxscale{-1.8}
\useasboundingbox (-1,-1) rectangle (1,1);
\coordinate (a) at ($(-.4,0)+(111:.25)$);\coordinate (b) at ($(-.4,0)+(69:.25)$);\coordinate (c) at ($(-.4,0)+(-69:.25)$);\coordinate (d) at ($(-.4,0)+(-111:.25)$);
\coordinate (A) at (95:1);\coordinate (B) at (85:1);\coordinate (C) at (-85:1);\coordinate (D) at (-95:1);
\fill[fill=gray!30, even odd rule] (0,0) circle (1)(-.4,0) circle (.25);
\draw[line join=bevel] (-.4,0)+(360-111:.25) arc (360-111:111:.25) (95:1) arc (95:360-95:1);
\draw[line join=bevel] (-.4,0)+(-69:.25) arc (-69:69:.25) (85:1) arc (85:-85:1);
\fill[white] (95:1) -- ($(-.4,0)+(111:.25)$) -- +(0,-.05) -- +(.175,-.05) -- ($(-.4,0)+(69:.25)$) -- (85:1) -- +(0,.05) -- +(-.175,.05);
\fill[white] (-95:1) -- ($(-.4,0)+(-111:.25)$) -- +(0,.05) -- +(.175,.05) -- ($(-.4,0)+(-69:.25)$) -- (-85:1) -- +(0,-.05) -- +(-.175,-.05);
\draw (-.4,0)+(111:.25) -- (95:1);
\draw (360-95:1) -- ($(-.4,0)+(-111:.25)$);
\draw (-.4,0)+(69:.25) -- (85:1);
\draw (-85:1) -- ($(-.4,0)+(-69:.25)$);
\fill[fill=gray!40] (.4,0) circle (.25);\draw[dashed] (.4,0) circle (.25);\node at (.4,.01) {$\scriptstyle T$};
}
\bigg)
\,\,\subset\,\,
B\bigg(\tikzmath[scale=.6]{\pgftransformyscale{.9}\pgftransformxscale{-1.8}
\useasboundingbox (-1,-1) rectangle (1,1.2);
\coordinate (a) at ($(-.4,0)+(111:.25)$);\coordinate (b) at ($(-.4,0)+(69:.25)$);\coordinate (c) at ($(-.4,0)+(-69:.25)$);\coordinate (d) at ($(-.4,0)+(-111:.25)$);
\coordinate (A) at (95:1);\coordinate (B) at (85:1);\coordinate (C) at (-85:1);\coordinate (D) at (-95:1);
\fill[fill=gray!30, even odd rule] (0,0) circle (1)(-.4,0) circle (.25);
\draw[line join=bevel] (-.4,0)+(360-111:.25) arc (360-111:111:.25) (95:1) arc (95:360-95:1);
\draw[line join=bevel] (-.4,0)+(-69:.25) arc (-69:69:.25) (85:1) arc (85:-85:1);
\fill[white] (-95:1) -- ($(-.4,0)+(-111:.25)$) -- +(0,.05) -- +(.175,.05) -- ($(-.4,0)+(-69:.25)$) -- (-85:1) -- +(0,-.05) -- +(-.175,-.05);
\draw[ultra thin] (-.4,0)+(111:.25) -- (95:1);
\draw (360-95:1) -- ($(-.4,0)+(-111:.25)$);
\draw[ultra thin] (-.4,0)+(69:.25) -- (85:1);
\draw (-85:1) -- ($(-.4,0)+(-69:.25)$);
\fill[fill=gray!40] (.4,0) circle (.25);\draw[dashed] (.4,0) circle (.25);\node at (.4,.01) {$\scriptstyle T$};
\draw[fill=gray!70] (B) to[in=0, out=90, looseness=.57] ($(b)+(.29,1.1)$) -- ($(a)+(.29,1.1)$) to[out=0, in=90, looseness=.57] (A) arc (95:85:1); 
\draw[fill=gray!30] (a) to[in=200, out=93, looseness=.47] ($(a)+(.29,1.1)$) -- ($(b)+(.29,1.1)$) to[out=200, in=93, looseness=.47] (b) arc (69:111:.25);
}
\bigg)
\,\,=\,\,
\tikzmath[scale=.6]{\pgftransformyscale{.9}\pgftransformxscale{1.8}
\useasboundingbox (-1,-1) rectangle (1,1.2);
\coordinate (a) at ($(-.4,0)+(111:.25)$);\coordinate (b) at ($(-.4,0)+(69:.25)$);\coordinate (c) at ($(-.4,0)+(-69:.25)$);\coordinate (d) at ($(-.4,0)+(-111:.25)$);
\coordinate (A) at (95:1);\coordinate (B) at (85:1);\coordinate (C) at (-85:1);\coordinate (D) at (-95:1);
\fill[fill=gray!30, even odd rule] (0,0) circle (1)(-.4,0) circle (.25);
\draw[line join=bevel] (-.4,0)+(360-111:.25) arc (360-111:111:.25) (95:1) arc (95:360-95:1);
\draw[line join=bevel] (-.4,0)+(-69:.25) arc (-69:69:.25) (85:1) arc (85:-85:1);
\fill[white] (-95:1) -- ($(-.4,0)+(-111:.25)$) -- +(0,.05) -- +(.175,.05) -- ($(-.4,0)+(-69:.25)$) -- (-85:1) -- +(0,-.05) -- +(-.175,-.05);
\draw[ultra thin] (-.4,0)+(111:.25) -- (95:1);
\draw (360-95:1) -- ($(-.4,0)+(-111:.25)$);
\draw[ultra thin] (-.4,0)+(69:.25) -- (85:1);
\draw (-85:1) -- ($(-.4,0)+(-69:.25)$);
\fill[fill=gray!40] (.4,0) circle (.25);\draw[dashed] (.4,0) circle (.25);\node at (.4,.01) {$\scriptstyle T$};
\draw[fill=gray!70] (B) to[in=0, out=90, looseness=.57] ($(b)+(.29,1.1)$) -- ($(a)+(.29,1.1)$) to[out=0, in=90, looseness=.57] (A) arc (95:85:1); 
\draw[fill=gray!30] (a) to[in=200, out=93, looseness=.47] ($(a)+(.29,1.1)$) -- ($(b)+(.29,1.1)$) to[out=200, in=93, looseness=.47] (b) arc (69:111:.25);
}
\end{equation}
Similarly, since $B=\Delta$ is the result of applying $-\boxtimes_{\cala(M)} ({V(E)\boxtimes _{\cala(M')}\overline{V(E)}})$ to some map, it follows that
\begin{equation}\label{eq: final2}
B\bigg(
\tikzmath[scale=.6]{\pgftransformyscale{.9}\pgftransformxscale{-1.8}
\useasboundingbox (-1,-1) rectangle (1,1);
\coordinate (a) at ($(-.4,0)+(111:.25)$);\coordinate (b) at ($(-.4,0)+(69:.25)$);\coordinate (c) at ($(-.4,0)+(-69:.25)$);\coordinate (d) at ($(-.4,0)+(-111:.25)$);
\coordinate (A) at (95:1);\coordinate (B) at (85:1);\coordinate (C) at (-85:1);\coordinate (D) at (-95:1);
\fill[fill=gray!30, even odd rule] (0,0) circle (1)(-.4,0) circle (.25);
\draw[line join=bevel] (-.4,0)+(360-111:.25) arc (360-111:111:.25) (95:1) arc (95:360-95:1);
\draw[line join=bevel] (-.4,0)+(-69:.25) arc (-69:69:.25) (85:1) arc (85:-85:1);
\fill[white] (95:1) -- ($(-.4,0)+(111:.25)$) -- +(0,-.05) -- +(.175,-.05) -- ($(-.4,0)+(69:.25)$) -- (85:1) -- +(0,.05) -- +(-.175,.05);
\fill[white] (-95:1) -- ($(-.4,0)+(-111:.25)$) -- +(0,.05) -- +(.175,.05) -- ($(-.4,0)+(-69:.25)$) -- (-85:1) -- +(0,-.05) -- +(-.175,-.05);
\draw (-.4,0)+(111:.25) -- (95:1);
\draw (360-95:1) -- ($(-.4,0)+(-111:.25)$);
\draw (-.4,0)+(69:.25) -- (85:1);
\draw (-85:1) -- ($(-.4,0)+(-69:.25)$);
\fill[fill=gray!40] (.4,0) circle (.25);\draw[dashed] (.4,0) circle (.25);\node at (.4,.01) {$\scriptstyle T$};
}
\bigg)
\,\,\subset\,\,
B\bigg(
\tikzmath[scale=.6]{\pgftransformyscale{.9}\pgftransformxscale{1.8}
\useasboundingbox (-1,-1) rectangle (1,1.1);
\coordinate (a) at ($(-.4,0)+(111:.25)$);\coordinate (b) at ($(-.4,0)+(69:.25)$);\coordinate (c) at ($(-.4,0)+(-69:.25)$);\coordinate (d) at ($(-.4,0)+(-111:.25)$);
\coordinate (A) at (95:1);\coordinate (B) at (85:1);\coordinate (C) at (-85:1);\coordinate (D) at (-95:1);
\fill[fill=gray!30, even odd rule] (0,0) circle (1)(-.4,0) circle (.25);
\fill[white] (95:1) -- ($(-.4,0)+(111:.25)$) -- +(0,-.05) -- +(.175,-.05) -- ($(-.4,0)+(69:.25)$) -- (85:1) -- +(0,.05) -- +(-.175,.05);
\draw[line join=bevel] (-.4,0)+(360-111:.25) arc (360-111:111:.25) (95:1) arc (95:360-95:1);
\draw[line join=bevel] (-.4,0)+(-69:.25) arc (-69:69:.25) (85:1) arc (85:-85:1);
\draw (-.4,0)+(111:.25) -- (95:1);
\draw[ultra thin] (360-95:1) -- ($(-.4,0)+(-111:.25)$);
\draw (-.4,0)+(69:.25) -- (85:1);
\draw[ultra thin] (-85:1) -- ($(-.4,0)+(-69:.25)$);
\fill[fill=gray!40] (.4,0) circle (.25);\draw[dashed] (.4,0) circle (.25);\node at (.4,.01) {$\scriptstyle T$};
\draw[fill=gray!70] (d) to[in=180, out=90, looseness=.57] ($(D)+(-.29,1.1)$) -- ($(C)+(-.29,1.1)$) to[out=180, in=90, looseness=.57] (c) arc (-69:-111:.25); 
\draw[fill=gray!30] (C) to[in=-20, out=87, looseness=.47] ($(C)+(-.29,1.1)$) -- ($(D)+(-.29,1.1)$) to[out=-20, in=87, looseness=.47] (D) arc (-95:-85:1);
}
\bigg)
\,\,=\,\,
\tikzmath[scale=.6]{\pgftransformyscale{.9}\pgftransformxscale{1.8}
\useasboundingbox (-1,-1) rectangle (1,1.1);
\coordinate (a) at ($(-.4,0)+(111:.25)$);\coordinate (b) at ($(-.4,0)+(69:.25)$);\coordinate (c) at ($(-.4,0)+(-69:.25)$);\coordinate (d) at ($(-.4,0)+(-111:.25)$);
\coordinate (A) at (95:1);\coordinate (B) at (85:1);\coordinate (C) at (-85:1);\coordinate (D) at (-95:1);
\fill[fill=gray!30, even odd rule] (0,0) circle (1)(-.4,0) circle (.25);
\fill[white] (95:1) -- ($(-.4,0)+(111:.25)$) -- +(0,-.05) -- +(.175,-.05) -- ($(-.4,0)+(69:.25)$) -- (85:1) -- +(0,.05) -- +(-.175,.05);
\draw[line join=bevel] (-.4,0)+(360-111:.25) arc (360-111:111:.25) (95:1) arc (95:360-95:1);
\draw[line join=bevel] (-.4,0)+(-69:.25) arc (-69:69:.25) (85:1) arc (85:-85:1);
\draw (-.4,0)+(111:.25) -- (95:1);
\draw[ultra thin] (360-95:1) -- ($(-.4,0)+(-111:.25)$);
\draw (-.4,0)+(69:.25) -- (85:1);
\draw[ultra thin] (-85:1) -- ($(-.4,0)+(-69:.25)$);
\fill[fill=gray!40] (.4,0) circle (.25);\draw[dashed] (.4,0) circle (.25);\node at (.4,.01) {$\scriptstyle T$};
\draw[fill=gray!70] (d) to[in=180, out=90, looseness=.57] ($(D)+(-.29,1.1)$) -- ($(C)+(-.29,1.1)$) to[out=180, in=90, looseness=.57] (c) arc (-69:-111:.25); 
\draw[fill=gray!30] (C) to[in=-20, out=87, looseness=.47] ($(C)+(-.29,1.1)$) -- ($(D)+(-.29,1.1)$) to[out=-20, in=87, looseness=.47] (D) arc (-95:-85:1);
}
\end{equation}
Combining \eqref{eq: final1}, \eqref{eq: final2} and \eqref{eq: final3}, 
we learn that the map $B$ sends\, 
$\tikzmath[scale=.5]{\filldraw[fill=gray!30] (0,0) circle (1); 
\filldraw[fill=gray!40, dashed] (-.4,0) circle (.25);\filldraw[fill=white](.4,0) circle (.25);
\node[scale=.8] at (-.4,-.01) {$\scriptstyle T$};
\pgftransformxscale{-1}
\fill[white] (95:1) -- ($(-.4,0)+(111:.25)$) -- +(0,-.05) -- +(.175,-.05) -- ($(-.4,0)+(69:.25)$) -- (85:1) -- +(0,.05) -- +(-.175,.05);
\draw[line cap=round] (-.4,0) +(111:.25) -- (95:1) (-.4,0) +(69:.25) -- (85:1);
\pgftransformyscale{-1}
\fill[white] (95:1) -- ($(-.4,0)+(111:.25)$) -- +(0,-.05) -- +(.175,-.05) -- ($(-.4,0)+(69:.25)$) -- (85:1) -- +(0,.05) -- +(-.175,.05);
\draw[line cap=round] (-.4,0) +(111:.25) -- (95:1) (-.4,0) +(69:.25) -- (85:1);}$ 
to
$\tikzmath[scale=.5]{\filldraw[fill=gray!30] (0,0) circle (1); 
\filldraw[fill=gray!40, dashed] (.4,0) circle (.25);\filldraw[fill=white](-.4,0) circle (.25);
\node[scale=.8] at (.4,-.01) {$\scriptstyle T$};
\fill[white] (95:1) -- ($(-.4,0)+(111:.25)$) -- +(0,-.05) -- +(.175,-.05) -- ($(-.4,0)+(69:.25)$) -- (85:1) -- +(0,.05) -- +(-.175,.05);
\draw[line cap=round] (-.4,0) +(111:.25) -- (95:1) (-.4,0) +(69:.25) -- (85:1);
\pgftransformyscale{-1}
\fill[white] (95:1) -- ($(-.4,0)+(111:.25)$) -- +(0,-.05) -- +(.175,-.05) -- ($(-.4,0)+(69:.25)$) -- (85:1) -- +(0,.05) -- +(-.175,.05);
\draw[line cap=round] (-.4,0) +(111:.25) -- (95:1) (-.4,0) +(69:.25) -- (85:1);}$\,.
Recall that $B$ intertwines the actions of 
$
\cala\big(\tikzmath[scale=.3]{\draw (0,0) circle (1); 
\filldraw[fill=white] (0,0) circle (.35);
\fill[white] (97:1) -- ($(0,0)+(111:.35)$) -- +(0,-.05) -- +(.25,-.05) -- ($(0,0)+(69:.35)$) -- (83:1) -- +(0,.05) -- +(-.25,.05);
\draw[line cap=round] (0,0) +(111:.35) -- (97:1) (0,0) +(69:.35) -- (83:1);
\pgftransformyscale{-1}
\fill[white] (97:1) -- ($(0,0)+(111:.35)$) -- +(0,-.05) -- +(.25,-.05) -- ($(0,0)+(69:.35)$) -- (83:1) -- +(0,.05) -- +(-.25,.05);
\draw[line cap=round] (0,0) +(111:.35) -- (97:1) (0,0) +(69:.35) -- (83:1);}
\big).
$
The representations $T\otimes H_0$ and $H_0\otimes T$ of that algebra are irreducible;
$B$ therefore induces an isomorphism $T\otimes H_0\cong H_0\otimes T$.
It follows that $T\cong H_0$ (and $H_0\cong T$).
\end{proof}


\appendix

\renewcommand{\thesubsection}{\Alph{section}.{\Roman{subsection}}}

\section{Cyclic fusion and fusion along graphs}
\label{app:Cyclic fusion}

Let $\Gamma$ be a finite oriented graph\footnote{Here, by a graph we mean {\it combinatorial graph}, unlike the prevalent usage of the word \emph{graph} in the body of the article, where it refers to a particular kind of topological space.} for which the two vertices of each edge are distinct.
Suppose that we are given, for every edge $e\in \Gamma_1$ a von Neumann algebra $A_e$, 
and for every vertex $v\in \Gamma_0$ a Hilbert space $H_v$.
Assume furthermore that each $H_v$ is equipped with a left action of $A_e$ for every incoming edge $e$,
and a right action of $A_e$ for every outgoing edge.
Moreover, the actions on $H_v$ are required to be split, meaning that they extend to an action of the spatial tensor product
\[
\Bigg(\raisebox{.1cm}{$\displaystyle\bar{\bigotimes_{\substack{\text{incoming}\\\text{edges } e}}}$} A_e\Bigg)
\,\,\bar\otimes\,\,
\Bigg(\raisebox{.1cm}{$\displaystyle\bar{\bigotimes_{\substack{\text{outgoing}\\\text{edges } e}}}$} A_e^\op\Bigg).
\]
The \emph{graph fusion}, denoted
\(
\bigboxtimes_{\{A_e\}}\{H_v\}_{e,v\in \Gamma}
\) (or simply 
\(
\bigboxtimes\,\{H_v\}_{v\in \Gamma}
\),
when the algebras $A_e$ are obvious from the context),
is the Connes fusion of all the Hilbert spaces $H_v$ along all the algebras $A_e$.

Let us write $e\triangleright v$ to denote that $e$ in an incoming edge for $v$,
and $e\triangleleft v$  to denote that $e$ is an outgoing edge for $v$.

\begin{definition}\label{def: graph fusion}
Let $\Gamma$, $\{H_v\}_{v\in\Gamma_0}$, $\{A_e\}_{e\in\Gamma_1}$ be as above.
The graph fusion
\[
\bigboxtimes_{\{A_e\}}\{H_v\}_{e,v\in \Gamma}
\]
of the Hilbert spaces $H_v$ over the algebras $A_e$ is the completion of the vector space
\[
\bigotimes_{v\in \Gamma_0}\,\,\,\hom_{\big(\underset{e\triangleright v}{\bar\bigotimes}\, A_e\,\bar\otimes\,\underset{e\triangleleft v}{\bar\bigotimes}\, A_e^\op\big)}
\left(
\,\underset{e\triangleright v}{\bar\bigotimes}\, L^2A_e
\,\,\bar\otimes\,\,
\,\underset{e\triangleleft v}{\bar\bigotimes}\, L^2A_e
\,,\,H_v\right)\otimes \bigotimes_{e\in\Gamma_1} L^2A_e
\]
with respect to the pre-inner product
\[
\left\langle\textstyle\underset{v\in\Gamma_0}\bigotimes\varphi_v\otimes\underset{e\in\Gamma_1}\bigotimes\xi_e\,,
\underset{v\in\Gamma_0}\bigotimes\psi_v\otimes\underset{e\in\Gamma_1}\bigotimes\eta_e
\!\right\rangle\,:=\,\left\langle\!\Big(\,\textstyle\underset{v\in\Gamma_0}\prod(\psi_v^*\varphi_v)\Big)\Big(\underset{e\in\Gamma_1}\bigotimes\xi_e\Big)\,,\underset{e\in\Gamma_1}\bigotimes\eta_e\!\right\rangle_{\!\bigotimes_{e\in\Gamma_1} \!L^2A_e}
\]
Note that it is important for $\Gamma$ not to have any loops since, otherwise, the action of the element 
$\psi_v^*\varphi_v \in{\bar\bigotimes_{e\triangleright v}}\, A_e\,\bar\otimes\,{\bar\bigotimes_{e\triangleleft v}}\, A_e^\op$
on the Hilbert space $\bigotimes_{e\in\Gamma_1}\! L^2A_e$ is ill-defined,
as the left and right actions of an algebra $A$ on its $L^2$-space are typically not split.
This is essentially the same issue as in Warning \ref{warn: cyc fus with one} below.
\end{definition}

Given a subset $S\subset \Gamma_1$ of edges, consider the graph $\tilde\Gamma$ obtained by reversing the orientations of all the edges in $S$.
Given algebras $A_e$ indexed by $e\in\Gamma_1$, we can form new ones $\tilde A_e$ by letting 
$\tilde A_e:=A_e^\op$ if $e\in S$ and $\tilde A_e:=A_e$ if $e\not\in S$.
It is then immediate from the definition that
\begin{equation}\label{eq: orientations in graph fusion}
\bigboxtimes_{\{\tilde A_e\}}\{H_v\}_{e, v\in \tilde\Gamma}\,\,=\,\,\bigboxtimes_{\{A_e\}}\{H_v\}_{e, v\in \Gamma}
\end{equation}

Now consider an equivalence relation $\calr$ on the set $\Gamma_0$ of vertices of $\Gamma$.
The quotient graph $\Gamma/\calr$ is defined as
\[
(\Gamma/\calr)_0:=\Gamma_0/\calr\qquad\text{and}\qquad (\Gamma/\calr)_1:=\Big\{e\in\Gamma_1\Big| \parbox{4.5cm}{\center$e$\, connects vertices from\\ different equivalence classes}\Big\}.
\]
For each equivalence class $V\in\Gamma_0/\calr$, let $\Gamma_V$ denote the corresponding full subgraph of $\Gamma$.
Given algebras $\{A_e\}_{e\in\Gamma_1}$ and Hilbert spaces $\{H_v\}_{v\in\Gamma_0}$, the graph fusion satisfies the following version of associativity:
\begin{equation}\label{eq: ass for graph fusion}
\bigboxtimes_{\{A_E\}}\Big\{\bigboxtimes_{\{A_e\}}\big\{H_v\big\}_{e,v\in\Gamma_V}\Big\}_{E, V\in \Gamma/\calr}
=\, \bigboxtimes_{\{A_e\}}\big\{H_v\big\}_{v\in \Gamma}
\end{equation}
The following result is the simplest special case of the above equation:

\begin{lemma}\label{lem: HKK}
Let $A_1$ and $A_2$ be von Neuman algebras. Let $K_i$ be left $A_i$-modules, and let $H$ be a right $A_1\bar\otimes A_2$-module.
Then there is a canonical isomorphism
\[
H\boxtimes_{A_1\bar\otimes A_2} (K_1\otimes K_2)\cong K_1\boxtimes_{A_1^\op}\! H \,\boxtimes_{A_2^{\phantom{\op}}}\!\! K_2.
\]
\end{lemma}

There is an important special case of graph fusion, called \emph{cyclic fusion}, which is the case when the graph $\Gamma$ is a cycle.
Let $n\ge 2$ be some number.
For each $i\in\{1,\ldots,n\}$, let $A_i$ be a von Neumann algebra, and let $H_i$ be an $A_i^{\phantom{\op}}\!\!\!\bar\otimes\, A_{i+1}^\op$-module (cyclic numbering).
Then for each $i,j\in\{1,\ldots,n\}$, we can form the fusion of
$H_i\boxtimes_{A_{i+1}}\ldots\boxtimes_{A_{j-1}}H_{j-1}$
(cyclic numbering) with
$H_j\boxtimes_{A_{j+1}}\ldots\boxtimes_{A_{i-1}}H_{i-1}$
over the algebra $A_i^\op\,\bar\otimes\, A_j^{\phantom{\op}}\!\!\!$.
Under the above conditions, the Hilbert space
\[
\big(H_i\boxtimes_{A_{i+1}}\ldots\boxtimes_{A_{j-1}}H_{j-1}\big)
\underset{A_i^\op\,\bar\otimes\, A_j^{\phantom{\op}}\!\!\!}\boxtimes
\big(H_j\boxtimes_{A_{j+1}}\ldots\boxtimes_{A_{i-1}}H_{i-1}\big)
\]
is isomorphic to the graph fusion of the $H_i$'s and therefore independent, up to canonical unitary isomorphism, of the choices of $i$ and $j$.
We call the above Hilbert space the cyclic fusion of the $H_i$'s, and denote it by
\begin{equation}\label{eq: cyc fus}
\tikzmath{
\node (a) at (0,0) {$H_1\,\boxtimes_{A_2}\cdots\,\boxtimes_{A_n}\!H_n\,\,\boxtimes_{A_1}$};
\def\dd{.4}
\def\ll{.35}
\def\rr{.25}
\draw[dashed, rounded corners = 5] (a.east) -- ++(\rr,0) -- ++(0,-\dd) -- ($(a.west) + (-\ll,-\dd)$) -- +(0,\dd) -- (a.west);
}
\end{equation}

\begin{warning}\label{warn: cyc fus with one}
For the space \eqref{eq: cyc fus} to be well defined, it is crucial to have $n\ge 2$.
In other words, given an $A\,\bar\otimes\,A^\op$ module $H$, the expression
$\tikzmath{\useasboundingbox (-.8,-.2) rectangle (.85,.2);
\node (a) at (0,0) {$H\,\boxtimes_{A}$};
\def\dd{.3}\def\ll{.2}\def\rr{.15}
\draw[dashed, rounded corners = 5] (a.east) -- ++(\rr,0) -- ++(0,-\dd) -- ($(a.west) + (-\ll,-\dd)$) -- +(0,\dd) -- (a.west);
}
$\smallskip is ill-defined.
This can be seen by analyzing the example $H:={}_AL^2A\otimes_{\IC} L^2A_A$.
In that case, one might expect to find $\tikzmath{\useasboundingbox (-.8,-.2) rectangle (.75,.2);
\node (a) at (0,0) {$H\,\boxtimes_{A}$};
\def\dd{.3}\def\ll{.2}\def\rr{.15}
\draw[dashed, rounded corners = 5] (a.east) -- ++(\rr,0) -- ++(0,-\dd) -- ($(a.west) + (-\ll,-\dd)$) -- +(0,\dd) -- (a.west);
}= L^2A
$\smallskip.
However, there is in general no meaningful way of letting
$\mathrm{End}_{A\,\bar\otimes\,A^\op}(H)$ act on that Hilbert space.

There is one exception to this warning: if $A$ is a direct sum of type $I$ factors, then
$\tikzmath{\useasboundingbox (-.8,-.3) rectangle (.75,.2);
\node (a) at (0,0) {$H\,\boxtimes_{A}$};
\def\dd{.3}\def\ll{.2}\def\rr{.15}
\draw[dashed, rounded corners = 5] (a.east) -- ++(\rr,0) -- ++(0,-\dd) -- ($(a.west) + (-\ll,-\dd)$) -- +(0,\dd) -- (a.west);
}
$ still makes sense.
It can be defined as $H\boxtimes_{A\,\bar\otimes\,A^\op} L^2(A)$,
which is now meaningful because $A\,\bar\otimes\,A^\op$ does act on $L^2(A)$.
\end{warning}

\section{The Hilbert space associated to an annulus}

Let $S$ be  circle, and let $\Sigma=S\times [0,1]$.
We have seen in \eqref{eq: def H-sigma} how to associate, non canonically, a Hilbert space $H_\Sigma\in\Rep_{S\sqcup \bar S}(\cala)$ to this annulus.
Later, in \eqref{eq:   KLM  }, we learned that there is a non-canonical unitary isomorphism of $S\sqcup \bar S$-sectors
\[
H_\Sigma \,\cong\,\, \bigoplus_{\lambda\in\Delta} H_\lambda\big(S\big) \otimes H_{\bar \lambda}\big(\bar S\big).
\]
Note that, by \eqref{eq: hatcalaS}, this can be reinterpreted as an isomorphism
\begin{equation}\label{H_Sigm=L^2hatA(S)}
H_\Sigma \,\cong\,\, L^2\cala(S).
\end{equation}
The goal of this section is to redefine $H_\Sigma$ in such a way that it becomes well defined up to \emph{canonical} unitary isomorphism,
and to upgrade \eqref{H_Sigm=L^2hatA(S)} to a \emph{canonical} unitary isomorphism.

Let $\cali=\{I_1,\ldots,I_n\}$ be a $c$-cover of $S$.
The intervals $I_j$ are arranged so that each intersection $p_i:=I_{i-1}\cap I_i$ (cyclic numbering) is a single point (or two points if $n=2$).
By definition, $H_\Sigma$ is then the fusion of all the vacuum sectors $H_i:=H_0(\partial I_i\times [0,1])$ along all the algebras $A_i:=\cala(\{p_i\}\times [0,1])$.

The first reason $H_\Sigma$ isn't canonically defined
is that the Hilbert spaces $H_i$ themselves are only well defined up to non-canonical unitary isomorphism.
One can fix that issue by being more specific:
from now on, $H_i$ will denote the vacuum sector of $\cala$ associated to the circle $S_i:=\partial I_i\times [0,1]$, its upper half 
$S_i^\top:=S_i\cap (I_i\times [1/2,1])$, and the involution $j:S_i\to S_i$, $j(x,t)=(x,1-t)$, as described near the beginning of Section \ref{sec:nets}.

Our next task is to show that $H_\Sigma$ is independent of the $c$-cover $\cali$.
In order to do so, it is useful to introduce a notation that stresses the dependence:
\begin{equation}\label{eq:cyc fus}
H_\Sigma^{(\cali)}:=\tikzmath{  
\node (a) at (0,0) {$H_1\,\boxtimes_{A_2}\cdots\,\boxtimes_{A_n}\!H_n\,\,\boxtimes_{A_1}$};
\def\dd{.4}
\def\ll{.35}
\def\rr{.25}
\draw[dashed, rounded corners = 5] (a.east) -- ++(\rr,0) -- ++(0,-\dd) -- ($(a.west) + (-\ll,-\dd)$) -- +(0,\dd) -- (a.west);
}
\end{equation}
The dashed line denotes the operation of \emph{cyclic fusion}, described in 
Appendix~\ref{app:Cyclic fusion}. 

Given two $c$-covers $\cali$ and $\calj$ of $S$,
we need to construct a unitary isomorphism between $H_\Sigma^{(\cali)}$ and $H_\Sigma^{(\calj)}$.
As in the proof of Lemma \ref{lem: H_Sigma^(cali)= H_Sigma^(calj)}, note that
one can go from any cover to any other one by subdividing and recombining intervals.
It is therefore enough to treat the case when $\cali=\{I_1,I_2,\ldots,I_n\}$ and $\calj=\{I_1',I_1'',I_2,\ldots,I_n\}$ with $I_1'\cup I_1''=I_1$.
Our first goal is therefore to upgrade the non-canonical isomorphism \eqref{hoAho=Ho} to a \emph{canonical} unitary isomorphism.
This will be the content of Lemma \ref{lem: u:HH->H}.

Let $I$ be an interval and let $\{I_1,I_2\}$ be a $c$-cover of $I$.
Let $\{p,q,r\}=\partial I_1\cup \partial I_2$, with $\partial I_1=\{p,q\}$, $\partial I_2=\{q,r\}$, $\partial I=\{p,r\}$, and
let $S_1=\partial (I_1\times [0,1])$, $S_2=\partial (I_2\times [0,1])$, $S_3=\partial (I\times [0,1])$.
Write $j$ for the involution $j(x,t)=(x,1-t)$ on $I\times [0,1]$.
Finally, let $S_a^\top=S_a\cap (I\times [\frac12,1])$ for $a\in\{1,2,3\}$,
and let $K={q}\times [0,1]$.
We orient $S_1$, $S_2$, $S_3$, $S_1^\top$, $S_2^\top$, $S_3^\top$, $K$ as follows:
\begin{equation}\label{eq: Theta-graph [box]}
\begin{split}
\tikzmath[scale=.05]{
\useasboundingbox (-14,-25) rectangle (28,27);
\coordinate (a1) at (-12,15);\coordinate (a2) at ($(a1)+(0,-30)$);\coordinate (a0) at ($(a1)+(0,-15)$);
\coordinate (b1) at (7,15);\coordinate (b2) at ($(b1)+(0,-30)$);\coordinate (b0) at ($(b1)+(0,-15)$);
\coordinate (c1) at (25,17);\coordinate (c2) at ($(c1)+(0,-30)$);\coordinate (c0) at ($(c1)+(0,-15)$);
\draw[line width=.7] (a1)node[left]{$I\,:\,\,\,$}node[above,yshift=1]{$\scriptstyle p$} to[bend left=25] (b1)node[above,yshift=1]{$\scriptstyle q$} to[bend right=25] (c1)node[above,yshift=-1]{$\scriptstyle r$};\fill (a1) circle (.6)(a0) circle (.6);\fill (b1) circle (.6)(b0) circle (.6)(b2) circle (.6);\fill (c1) circle (.6)(c2) circle (.6);
\draw[line width=.7] (a0)node[left]{$I_1:\,\,$} to[bend left=25] (b0);\draw[densely dotted] (b0) to[bend right=25] (c0);
\draw[densely dotted] (a2)node[left]{$I_2:\,\,$} to[bend left=25] (b2);\draw[line width=.7] (b2) to[bend right=25] (c2);
} 
\,\,\,\,\qquad \tikzmath[scale=.037]{ \useasboundingbox (-14,-28) rectangle (28,26); 
\coordinate (a1) at (-12,10);\coordinate (a2) at ($(a1)+(0,-20)$);\coordinate (a0) at ($(a1)+(0,-10)$);
\coordinate (b1) at (7,10);\coordinate (b2) at ($(b1)+(0,-20)$);\coordinate (b0) at ($(b1)+(0,-10)$);
\coordinate (c1) at (25,12);\coordinate (c2) at ($(c1)+(0,-20)$);\coordinate (c0) at ($(c1)+(0,-10)$);
\draw[line width=.7] (a1) to[bend left=25] (b1) -- (b2) to[bend right=25] (a2) -- cycle;
\draw[densely dotted] (b1) to[bend right=25] (c1) -- (c2) to[bend left=25] (b2);\node[scale=.9] at (7,-21) {$S_1$};\draw[->] (-3.8,12.35) -- +(-.1,0);
} 
\,\,\,\,\quad 
\tikzmath[scale=.037]{ \useasboundingbox (-14,-28) rectangle (28,26); 
\draw[densely dotted] (b2) to[bend right=25] (a2) -- (a1) to[bend left=25] (b1);
\draw[line width=.7] (b1) to[bend right=25] (c1) -- (c2) to[bend left=25] (b2) -- cycle;\node[scale=.9] at (7,-21) {$S_2$};\draw[->] (14.2,8.7) -- +(-.1,0);
} 
\,\,\,\,\quad 
\tikzmath[scale=.037]{ \useasboundingbox (-14,-28) rectangle (35,26); 
\draw[densely dotted] (b1) -- (b2);
\draw[line width=.7] (a1) to[bend left=25] (b1) to[bend right=25] (c1) -- (c2) to[bend left=25] (b2) to[bend right=25] (a2) -- cycle;\node[scale=.9] at (7,-21) {$S_3$};\draw[->] (-3.8,12.35) -- +(-.1,0);
\draw[<->] (38,-5) --node[right, scale=1.2]{$\scriptstyle j$} (38,7);
} 
\,\,\,\\
\,\,\,\,\qquad \tikzmath[scale=.037]{ \useasboundingbox (-14,-28) rectangle (28,23); 
\draw[line width=.7] (a0) -- (a1) to[bend left=25] (b1) -- (b0);
\draw[densely dotted] (b1) to[bend right=25] (c1) -- (c2) to[bend left=25] (b2) to[bend right=25] (a2) -- (a0)(b2) -- (b0);\node[scale=.9] at (7,-21) {$S_1^\top$};\draw[->] (-3.8,12.35) -- +(-.1,0);
} 
\,\,\,\,\quad 
\tikzmath[scale=.037]{ \useasboundingbox (-14,-28) rectangle (28,23); 
\draw[densely dotted] (b0) -- (b2)(c0) -- (c2) to[bend left=25] (b2) to[bend right=25] (a2) -- (a1) to[bend left=25] (b1);
\draw[line width=.7] (b0) -- (b1) to[bend right=25] (c1) -- (c0);\node[scale=.9] at (7,-21) {$S_2^\top$};\draw[->] (14.2,8.7) -- +(-.1,0);
} 
\,\,\,\,\quad 
\tikzmath[scale=.037]{ \useasboundingbox (-14,-28) rectangle (28,23); 
\draw[densely dotted] (b1) -- (b2)(c0) -- (c2) to[bend left=25] (b2) to[bend right=25] (a2) -- (a0);
\draw[line width=.7] (a0) -- (a1) to[bend left=25] (b1) to[bend right=25] (c1) -- (c0);\node[scale=.9] at (7,-21) {$S_3^\top$};\draw[->] (-3.8,12.35) -- +(-.1,0);
} 
\,\,\,\,\quad 
\tikzmath[scale=.037]{ \useasboundingbox (-14,-28) rectangle (28,23); 
\draw[line width=.7] (b1) -- (b2);
\draw[densely dotted] (a1) to[bend left=25] (b1) to[bend right=25] (c1) -- (c2) to[bend left=25] (b2) to[bend right=25] (a2) -- cycle;\node[scale=.9] at (7,-21) {$K$};\draw[->] (7,-.5) -- +(0,-.1);
} 
\hspace{1cm}\end{split}
\end{equation}

\begin{lemma}\label{lem: u:HH->H}
With $S_1$, $S_2$, $S_3$, $j$, $S_1^\top$, $S_2^\top$, $S_3^\top$, $K$ as above, let us define
\[
\qquad\qquad H_a:=L^2\cala(S_a^\top)\in\Rep_{S_a}(\cala),\qquad \tikzmath{\node[scale=.9]{$a\in\{1,2,3\}$};}
\]
to be the vacuum sector of $\cala$ associated to $S_a$, $S_a^\top$, and $j$.
Then there is a canonical unitary isomorphism of $S_3$-sectors
\begin{equation}\label{eq: u:HH->H}
u: H_1 \boxtimes_K H_2 \,\stackrel{\scriptscriptstyle\cong}{\longrightarrow}\, H_3.
\end{equation}
\end{lemma}

\begin{proof}
Let $S^1$ be the standard circle, let
\[
S^1_\top:=\{z\in S^1| \Im\mathrm{m}(z)\ge 0\},
\,\,S^1_\vdash:=\{z\in S^1| \Re\mathrm{e}(z)\le 0\},
\,\,S^1_\dashv:=\{z\in S^1| \Re\mathrm{e}(z)\ge 0\},
\]
and let $\psi_\vdash:[0,1]\to S^1_\vdash$ and $\psi_\dashv:[0,1]\to S^1_\dashv$ be the diffeomorphisms given by $\psi_\vdash(t)=-\sin(\pi t)-i\cos(\pi t)$ and
$\psi_\dashv(t)=\sin(\pi t)-i\cos(\pi t)$.

Let $p,q,r\in I$ be as above, and let us pick diffeomorphisms $f_1:S_1\to S^1$, $f_2:S_2\to S^1$ so that
\[
\begin{split}
\tikzmath{\node[scale=.9]{$\exists\, \varepsilon>0:\,\,\forall t\in \textstyle[\frac12-\varepsilon,\frac12+\varepsilon]:$};}&
\,\,\, f_1(p,t)=\psi_\vdash(t),
\,\,\, f_2(r,t)=\psi_\dashv(t),\\
\quad\tikzmath{\node[scale=.9]{$\forall t\in \textstyle[0,1]:$};}&
\,\,\, f_1(q,t)=\psi_\dashv(t),
\,\,\, f_2(q,t)=\psi_\vdash(t),
\end{split}
\]
and $f_a(j(x))=\overline{f_a(x)}$.
Finally, let $f_3:=f_1|_{S_1\cap S_3}\cup f_2|_{S_2\cap S_3}:S_3\to S^1$.

Recall from~\cite[\thmVaccumSector]{BDH(nets)} that there is an 
$\cala$-sector $H_0=H_0(S^1,\cala)$ that is canonically associated to 
the standard circle $S^1$.
It is equipped, among others, with 
isomorphisms of $S^1$-sectors $v_\top:H_0\to L^2(\cala(S^1_\top))$ and $v_\vdash:H_0\to L^2(\cala(S^1_\vdash))$.
The isomorphism \eqref{eq: u:HH->H} is then the composite
\begin{equation}\label{eq: u:HH->H: the def.}
\begin{split}
u\,:\,H_1 \boxtimes_K H_2\,=\,\,&L^2\cala(S_1^\top) \boxtimes_K L^2\cala(S_2^\top)
\xrightarrow{L^2\cala(f_1)\boxtimes L^2\cala(f_2)}\\
\to\,&L^2\cala(S^1_\top) \boxtimes_K L^2\cala(S^1_\top)
\xrightarrow{1\boxtimes v_\top^{-1}} L^2\cala(S^1_\top) \boxtimes_K H_0\\
=\,\,&L^2\cala(S^1_\top) \boxtimes_{\cala(S^1_\vdash)} H_0
\xrightarrow{1\boxtimes v_\vdash^{\phantom{-1}}\!\!} L^2\cala(S^1_\top) \boxtimes_{\cala(S^1_\vdash)}L^2\cala(S^1_\vdash)\\
\cong\,\,& L^2\cala(S^1_\top)\xrightarrow{L^2\cala(f_3)^{-1}}L^2\cala(S_3^\top)=H_3.
\end{split}
\end{equation}

We still need to show that $u$ is independent of the choices made.
For that, we pick new diffeomorphisms $\tilde f_a:S_a\to S^1$ with the same properties as the maps $f_a$,
and define
\[
\tilde u:=\Big(L^2\cala(\tilde f_3)^{-1}\Big)  \circ  \Big(1\boxtimes v_\vdash v_\top^{-1}\Big)  \circ  \Big(L^2\cala(\tilde f_1)\boxtimes L^2\cala(\tilde f_2)\Big)
\]
as in \eqref{eq: u:HH->H: the def.}.
Let $g_a:=(\tilde f_a \circ f_a^{-1})|_{S^1_\top}$, and note that $g_3=g_1\circ g_2$.
By the inner covariance axiom, there exist unitaries $w_a\in \cala(S^1_\top)$ so that $\cala(g_a)=\ad(w_a)$.
Moreover, we can chose $w_3=w_1 w_2$.
Let $j_0:S^1\to S^1$ denote complex conjugation,
and let $W_a$ be the operator given by multiplication by $w_1$ followed by multiplication by $\cala(j_0)(w_1^*)$.
To show that $\tilde u=u$, we need to argue that following diagrams are commutative:
\[
\tikzmath{ \matrix [matrix of math nodes,column sep=3.5cm,row sep=1cm]
{ |(a)| L^2\cala(S_1^\top) \boxtimes_{K} L^2\cala(S_2^\top) \pgfmatrixnextcell |(b)| L^2\cala(S^1_\top) \boxtimes_{\cala(S^1_\vdash)} L^2\cala(S^1_\top)\\
|(c)|  {L^2\cala(S_1^\top)} \boxtimes_{K} {L^2\cala(S_2^\top)} \pgfmatrixnextcell |(d)| L^2\cala(S^1_\top) \boxtimes_{\cala(S^1_\vdash)} L^2\cala(S^1_\top)\\ }; 
\draw[->] (a) --node[above]{$\scriptstyle L^2\cala(f_1)\,\boxtimes\, L^2\cala(f_2)$} (b);
\draw[->] (c) --node[above]{$\scriptstyle L^2\cala(\tilde f_1)\,\boxtimes\, L^2\cala(\tilde f_2)$} (d);
\draw[double, double distance=1pt] (a) -- (c);
\draw[->] (b) -- node [right] {$\scriptstyle W_1\,\boxtimes\, W_2$} (d);
\node[yshift=2] at ($(a)!.5!(d)$) {\squared1};
}
\]
\[
\tikzmath{ \matrix [matrix of math nodes,column sep=1.8cm,row sep=1cm]
{ |(a)| L^2\cala(S^1_\top) \boxtimes_{\cala(S^1_\vdash)} L^2\cala(S^1_\top) \pgfmatrixnextcell |(b)| L^2\cala(S^1_\top) \boxtimes_{\cala(S^1_\vdash)}L^2\cala(S^1_\vdash)\\
|(a')| L^2\cala(S^1_\top) \boxtimes_{\cala(S^1_\vdash)} L^2\cala(S^1_\top) \pgfmatrixnextcell |(b')| L^2\cala(S^1_\top) \boxtimes_{\cala(S^1_\vdash)}L^2\cala(S^1_\vdash)\\ };
\draw[->] (a) --node[above]{$\scriptstyle 1\boxtimes v_\vdash v_\top^{-1}$} (b);
\draw[->] (a') --node[above]{$\scriptstyle 1\boxtimes v_\vdash v_\top^{-1}$} (b');
\draw[->] (a) -- node [right] {$\scriptstyle W_1\,\boxtimes\, W_2$} (a');
\draw[->] (b) -- node [right] {$\scriptstyle W_1\,\boxtimes\, W_2$} (b');
\node[yshift=2] at ($(a)!.5!(b')$) {\squared2};
}
\]
\[
\tikzmath{ \matrix [matrix of math nodes,column sep=2cm,row sep=1cm]
{|(z)| L^2\cala(S^1_\top) \boxtimes_{\cala(S^1_\vdash)}L^2\cala(S^1_\vdash) \pgfmatrixnextcell[-1cm] |(a)| L^2\cala(S^1_\top)  \pgfmatrixnextcell  |(b)| L^2\cala(S_3^\top)\\ 
|(z')| L^2\cala(S^1_\top) \boxtimes_{\cala(S^1_\vdash)}L^2\cala(S^1_\vdash) \pgfmatrixnextcell |(a')| L^2\cala(S^1_\top)  \pgfmatrixnextcell  |(b')| L^2\cala(S_3^\top)\\ };
\draw[->] (z) -- node [right] {$\scriptstyle W_1\,\boxtimes\, W_2$} (z');
\draw[->] (a) --node[above]{$\scriptstyle L^2\cala(f_3)^{-1}$} (b);
\draw[->] (a') --node[above]{$\scriptstyle L^2\cala(\tilde f_3)^{-1}$} (b');
\draw[->] (a) -- node [right] {$\scriptstyle W_3$} (a');
\draw[double, double distance=1pt] (b) -- (b');
\draw[->] (z) --node[above]{$\scriptstyle \cong$} (a);
\draw[->] (z') --node[above]{$\scriptstyle \cong$} (a');
\node[yshift=1, xshift=5] at ($(z)!.5!(a')$) {\squared3};
\node[yshift=1] at ($(a)!.5!(b')$) {\squared4};
}
\]
The commutativity of \squared2 and \squared3 is obvious.
Finally, the commutativity of \squared1 and \squared4 follows from the following general fact:
given a von Neumann algebra $A$, a unitary $u\in \U(A)$ and a vector $\xi\in L^2(A)$, one always has $L^2(\ad(u))(\xi)=u\xi u^*$.
\end{proof}

\noindent Note that despite its apparent asymmetry, the definition \eqref{eq: u:HH->H: the def.} is left-right symmetric. 
Indeed, we could have used
\[
L^2\cala(S^1_\top) \boxtimes_{\cala(S^1_\vdash)} L^2\cala(S^1_\top) \xrightarrow{v_\vdash v_\top^{-1}\boxtimes 1}
L^2\cala(S^1_\vdash) \boxtimes_{\cala(S^1_\vdash)}L^2\cala(S^1_\top)\cong\, L^2\cala(S^1_\top)
\]
instead of
\[
L^2\cala(S^1_\top) \boxtimes_{\cala(S^1_\vdash)} L^2\cala(S^1_\top) \xrightarrow{1\boxtimes v_\vdash v_\top^{-1}}
L^2\cala(S^1_\top) \boxtimes_{\cala(S^1_\vdash)}L^2\cala(S^1_\vdash)\cong\, L^2\cala(S^1_\top)
\]
in the middle of \eqref{eq: u:HH->H: the def.} as both are equal to
\[
\begin{split}
L^2\cala(S^1_\top) \boxtimes_{\cala(S^1_\vdash)} L^2\cala(S^1_\top) \xrightarrow{v_\vdash v_\top^{-1}\boxtimes v_\vdash v_\top^{-1}}
&\,L^2\cala(S^1_\vdash) \boxtimes_{\cala(S^1_\vdash)}L^2\cala(S^1_\vdash)\cong\\
\cong\, L^2\cala(S^1_\vdash) \xrightarrow{v_\top v^{-1}_\vdash} &\,L^2\cala(S^1_\top).
\end{split}
\]

Let $S_1$, $S_2$, $S_3$, $H_1$, $H_2$, $H_3$ be as in Lemma \ref{lem: u:HH->H}.
Letting $J_a$ denote the modular conjugation on $H_a=L^2(\cala(S_a^\top))$, $a\in\{1,2,3\}$,
we expect the isomorphism \eqref{eq: u:HH->H} to satisfy
\[
u \circ (J_1\boxtimes_{\alpha} J_2) = J_3\circ u,
\]
where $\alpha:\cala(K)\to\cala(K)$ is the anti-linear homomorphism given by $\alpha(x)=\cala(j)(x^*)$.
We only know how to prove the above equation up to sign:

\begin{lemma}\label{lem: ubarJJ=Ju}
Let $u$, $J_a$, and $\alpha$ be as above.
Then we have
\begin{equation}\label{eq: ubarJJ=Ju}
u \circ (J_1\boxtimes_{\alpha} J_2) = \pm J_3\circ u.
\end{equation}
\end{lemma}

\begin{proof}
Recall from~\cite[\thmVaccumSector]{BDH(nets)} that 
$\varphi\mapsto H_0(\varphi)$ is a representation of the group of 
M\"obius transformations of $S^1$ on the vacuum sector $H_0=H_0(S^1)$.
The operator $H_0(\varphi)$ is unitary for $\varphi$ orientation preserving, and antiunitary for $\varphi$ orientation reversing.
Let $j_0:S^1\to S^1$ denote complex conjugation, and let $\beta:\cala(S^1_\vdash)\to \cala(S^1_\vdash)$ be given by $\beta(x)=\cala(j_0)(x^*)$.
Recall the definition \eqref{eq: u:HH->H: the def.} of the isomorphism $u$.
In order to prove \eqref{eq: ubarJJ=Ju}, it is enough to show that the following squares labelled \squared1\hspace{.02cm}, \squared2\hspace{.02cm}, \squared3\hspace{.02cm}, \squared4 commute up to sign
\begin{gather*}
\!\!\tikzmath{ \matrix [matrix of math nodes,column sep=4.6cm,row sep=1cm]
{ |(a)| L^2(\cala(S_1^\top)) \boxtimes_{\cala(K)} L^2(\cala(S_2^\top)) \pgfmatrixnextcell |(c)| H_0 \boxtimes_{\cala(S^1_\vdash)} H_0\\
|(b)|  {L^2(\cala(S_1^\top))} \boxtimes_{{\cala(K)}} {L^2(\cala(S_2^\top))} \pgfmatrixnextcell |(d)| {H_0} \boxtimes_{{\cala(S^1_\vdash)}} {H_0}\\ }; 
\draw[->] (a) -- node [left]	{$\scriptstyle J_1\boxtimes_{\alpha} J_2$} (b); \draw[->] (c) -- node [fill=white, inner sep=2] {$\scriptstyle H_0(j_0)\boxtimes_{\beta}H_0(j_0)$} (d);
\draw[->] (a) --node[above]{$\scriptstyle v_\top^{-1} L^2(\cala(f_1))\,\boxtimes\, v_\top^{-1} L^2(\cala(f_2))$} (c); \draw[->] (b) --node[above]{$\scriptstyle {v_\top^{-1} L^2(\cala(f_1))}\,\boxtimes\, {v_\top^{-1} L^2(\cala(f_2))}$} (d);
\node[yshift=2] at ($(a)!.5!(d)$) {\squared1};
}
\\
\!\!\tikzmath{ \matrix [matrix of math nodes,column sep=1.8cm,row sep=1cm]
{ |(a)| H_0 \boxtimes_{\cala(S^1_\vdash)} H_0 \pgfmatrixnextcell[-.5cm] |(b)| H_0 \boxtimes_{\cala(S^1_\vdash)} L^2(\cala(S^1_\vdash)) \pgfmatrixnextcell[-1cm] |(c)| H_0\pgfmatrixnextcell[.3cm] |(cc)| L^2(\cala(S_3^\top))\\
|(d)| {H_0} \boxtimes_{{\cala(S^1_\vdash)}} {H_0} \pgfmatrixnextcell |(e)| {H_0} \boxtimes_{{\cala(S^1_\vdash)}} {L^2(\cala(S^1_\vdash))} \pgfmatrixnextcell |(f)| {H_0} \pgfmatrixnextcell |(ff)| {L^2(\cala(S_3^\top))}\\}; 
\draw[->] (a) -- node [above]	{$\scriptstyle 1\,\boxtimes\, v_\vdash$} (b); \draw[->] (b) -- node [above]{$\scriptstyle \cong$} (c);
\draw[->] (d) -- node [above]	{$\scriptstyle 1\,\boxtimes\, {v}_\vdash$} (e); \draw[->] (e) -- node [above]{$\scriptstyle \cong$} (f);
\draw[->] (a) --node [fill=white, inner sep=2] {$\scriptstyle H_0(j_0)\boxtimes_{\beta}H_0(j_0)$} (d);
\draw[->] (b) --node [xshift=-2, fill=white, inner sep=2] {$\scriptstyle H_0(j_0)\boxtimes_{\beta}L^2(\beta)$} (e);
\draw[->] (c) --node [fill=white, inner sep=2]{$\scriptstyle H_0(j_0)$} (f);
\draw[->] (cc) --node [right] {$\scriptstyle J_3$} (ff);
\draw[->] (c) -- node [above]{$\scriptstyle L^2(\cala(f_3))^{-1}v_\top$} (cc);
\draw[->] (f) -- node [above]{$\scriptstyle {L^2(\cala(f_3))}^{-1}v_\top$} (ff);
\node[yshift=2, xshift=-4] at ($(a)!.5!(e)$) {\squared2};
\node[yshift=2, xshift=7] at ($(b)!.5!(f)$) {\squared3};
\node[yshift=2] at ($(c)!.5!(ff)$) {\squared4};
}
\end{gather*}
where $f_1$, $f_2$, $f_3$ are as in the proof of Lemma \ref{lem: u:HH->H}.
The squares \squared1 and \squared4 commute 
by~\cite[\thmVaccumSector~\propvacuumsectorL]{BDH(nets)}, 
and \squared3 is easily seen to be commutative.
So the only square that remains is \squared2\hspace{.03cm}, 
which is equivalent to
\begin{equation}\label{eq: remaining square}
\tikzmath{ \matrix [matrix of math nodes,column sep=1.8cm,row sep=1cm]
{
|(a)| H_0 \pgfmatrixnextcell |(b)| L^2(\cala(S^1_\vdash)) \\
|(d)| {H_0} \pgfmatrixnextcell |(e)| {L^2(\cala(S^1_\vdash))} \\}; 
\draw[->] (a) -- node [above, pos=.6]	{$\scriptstyle v_\vdash$} (b);
\draw[->] (d) -- node [above, pos=.6]	{$\scriptstyle {v}_\vdash$} (e);
\draw[->] (a) --node [left] {$\scriptstyle H_0(j_0)$} (d);
\draw[->] (b) --node [right] {$\scriptstyle L^2(\beta)$} (e);
}\,\,\,
\end{equation}
Note that 
by~\cite[\thmVaccumSector~\propvacuumsectorcovariant]{BDH(nets)} 
(see also \cite[\propprojectiveimplementationdiffeo]{BDH(nets)}) and 
\cite[\lemLconformalimplementaionOP]{BDH(nets)} respectively,
both $H_0(j_0)$ and $L^2(\beta)=L^2(\cala(j_0))\circ J$ implement $j_0$ 
(\cite[\defimplements]{BDH(nets)}).
Since both of them are antilinear involutions that  implement $j_0$, by 
Schur's lemma, they are equal up to sign.
\end{proof}

The isomorphism $u: H_1 \boxtimes_K H_2 \,\cong\, H_3$ from \eqref{eq: u:HH->H} satisfies a certain version of associativity, which we now describe.
Let $\{I_1,I_2, I_3\}$ be a $c$-cover of $I$.
We call the boundary points $p$, $q$, $r$, $s$ and assume that they are arranged as follows:
\[
\tikzmath{
\coordinate (a) at (0,0);\coordinate (b) at (3,0);\coordinate (c) at (6,0);\coordinate (d) at (9,0);
\draw[line width=.7] (a)node[left, xshift=-15, yshift=2]{$I\,:$}
      node[above,yshift=1, scale=1.1]{$\scriptstyle p$} to[bend left=2]node[above, scale=1.3] {$\scriptstyle I_1$} 
(b)node[above,yshift=1, scale=1.1]{$\scriptstyle q$} to[bend right=2]node[above, yshift=1, scale=1.3] {$\scriptstyle I_2$} 
(c)node[above,yshift=1, scale=1.1] {$\scriptstyle r$} to[bend left=2]node[above, scale=1.3] {$\scriptstyle I_3$} (d)node[above,yshift=1] {$\scriptstyle s$};
\fill (a) circle (.05) (b) circle (.05) (c) circle (.05) (d) circle (.05);
}
\]
Consider the circles
\[\begin{split}
S_1:=\partial (I_1\times [0,1]),\quad S_2&:=\partial (I_2\times [0,1]),\quad S_3:=\partial (I_3\times [0,1]),\\
S_{12}:=\partial ((I_1\cup I_2)\times [0,1]),\quad S_{23}&:=\partial ((I_2\cup I_3)\times [0,1]),\quad S_{123}:=\partial (I\times [0,1]),
\end{split}\]
and let
\[\begin{split}
H_1:=H_0(S_1),\quad H_2&:=H_0(S_2),\quad\, H_3:=H_0(S_3),\\ 
H_{12}:=H_0(S_{12}),\quad H_{23}&:=H_0(S_{23}),\quad H_{123}:=H_0(S_{123})
\end{split}\]
be the vacuum sectors associated the upper halves and to the involutions $j:(x,t)\mapsto(x,1-t)$.
Let also $K:=\{q\}\times [0,1]$ and $L:=\{r\}\times [0,1]$.
We then have the following four instances of the isomorphism \eqref{eq: u:HH->H}:
\begin{equation*}
\begin{split}
u_{12}: H_1 \boxtimes_K H_2 \,\cong\, H_{12},&\qquad
u_{23}: H_2 \boxtimes_L H_3 \,\cong\, H_{23},\\
u_{12,3}: H_{12} \boxtimes_L H_3 \,\cong\, H_{123},&\qquad
u_{1,23}: H_1 \boxtimes_K H_{23} \,\cong\, H_{123}.
\end{split}
\end{equation*}

\begin{lemma}\label{lem: associativity of u}
The above maps fit into a commutative diagram
\[
\tikzmath{ \matrix [matrix of math nodes,column sep=2.6cm,row sep=1cm]
{|(a)| H_1 \boxtimes_K H_2 \boxtimes_L H_3 \pgfmatrixnextcell |(b)| H_{12} \boxtimes_L H_3\\
|(c)| H_1 \boxtimes_K H_{23}  \pgfmatrixnextcell |(d)| H_{123}\\ }; 
\draw[->] (a) -- node [left]	{$\scriptstyle 1\boxtimes u_{23}$} (c); \draw[->] (b) -- node [right] {$\scriptstyle u_{12,3}$} (d);
\draw[->] (a) --node[above]{$\scriptstyle u_{12}\boxtimes\, 1$} (b); \draw[->] (c) --node[above]{$\scriptstyle u_{1,23}$} (d);
}
\]
\end{lemma}

\begin{proof}
Letting $\psi_\vdash:[0,1]\to S^1_\vdash$ and $\psi_\dashv:[0,1]\to S^1_\dashv$ be as in the proof of Lemma~\ref{lem: u:HH->H},
and given diffeomorphisms $f:S_1\to S^1$, $g:S_3\to S^1$ subject to
\[
\begin{split}
\tikzmath{\node[scale=.9]{$\exists\, \varepsilon>0:\,\,\forall t\in \textstyle[\frac12-\varepsilon,\frac12+\varepsilon]:$};}&
\,\,\, f(p,t)=\psi_\vdash(t),
\,\,\, g(s,t)=\psi_\dashv(t),\\
\quad\tikzmath{\node[scale=.9]{$\forall t\in \textstyle[0,1]:$};}&
\,\,\, f(q,t)=\psi_\dashv(t),
\,\,\, g(r,t)=\psi_\vdash(t),\qquad
\end{split}
\]
one can write $u_{12,3}\circ (u_{12}\boxtimes 1)$ as 
\[
\begin{split}
H_1 \boxtimes_K H_2 \boxtimes_L H_3
\xrightarrow{v_\vdash v_\top^{-1}L^2\cala(f) \,\boxtimes\, 1\,\boxtimes\, 1}
L^2\cala(S^1_\vdash) \boxtimes_{\cala(S^1_\vdash)} H_2 \boxtimes_L H_3\\
\cong H_2 \boxtimes_L H_3
\xrightarrow{L^2\cala((f^{-1}\circ\, \psi_\vdash)\cup\, \mathrm{id}) \,\boxtimes\, 1} H_{12} \boxtimes_L H_3
\xrightarrow{1\,\boxtimes\, v_\vdash v_\top^{-1} L^2\cala(g)}
\\\to H_{12} \boxtimes_{\cala(S^1_\vdash)} L^2\cala(S^1_\vdash)
\cong H_{12}
\xrightarrow{L^2\cala(\mathrm{id}\,\cup (g^{-1}\circ\,\psi_\dashv))} H_{123}
\end{split}
\]
and $u_{1,23}\circ (1\boxtimes u_{23})$ as
\[
\begin{split}
\,\,H_1 \boxtimes_K H_2 \boxtimes_L H_3
\xrightarrow{1 \,\boxtimes\, 1\,\boxtimes\, v_\vdash v_\top^{-1} L^2\cala(g)}
H_1 \boxtimes_K H_2 \boxtimes_{\cala(S^1_\vdash)} L^2\cala(S^1_\vdash)\:\\
\cong H_1 \boxtimes_K H_2
\xrightarrow{1 \,\boxtimes\, L^2\cala(\mathrm{id}\,\cup (g^{-1}\circ\,\psi_\dashv))} H_1 \boxtimes_K H_{23}
\xrightarrow{v_\vdash v_\top^{-1} L^2\cala(f)\,\boxtimes\, 1} \:\\
\to L^2\cala(S^1_\vdash) \boxtimes_{\cala(S^1_\vdash)} H_{23}
\cong H_{23}
\xrightarrow{L^2\cala((f^{-1}\circ\, \psi_\vdash)\cup\, \mathrm{id})} H_{123}.
\end{split}
\]
Both composites are equal to
\begin{gather*}
H_1 \boxtimes_K H_2 \boxtimes_L H_3
\xrightarrow{v_\vdash v_\top^{-1}L^2\cala(f)\, \boxtimes\, 1\, \boxtimes\, v_\vdash v_\top^{-1}L^2\cala(g)}
\\L^2\cala(S^1_\vdash) \boxtimes_{\cala(S^1_\vdash)} H_2 \boxtimes_{\cala(S^1_\vdash)} L^2\cala(S^1_\vdash)
\cong H_2
\xrightarrow{L^2\cala((f^{-1}\circ\, \psi_\vdash)\cup\,\mathrm{id}\,\cup (g^{-1}\circ\,\psi_\dashv))}
H_{123}
\end{gather*}
\end{proof}

We can now prove that the Hilbert space associated to $\Sigma=S\times[0,1]$ is well defined up to unique unitary isomorphism.
Given a $c$-cover $\cali$ of $S$, let 
\[
\qquad\quad H_\Sigma^{(\cali)}:=\tikzmath{  
\node (a) at (0,0) {$H_1\,\boxtimes_{A_2}\cdots\,\boxtimes_{A_n}\!H_n\,\,\boxtimes_{A_1}$};
\def\dd{.4}
\def\ll{.35}
\def\rr{.25}
\draw[dashed, rounded corners = 5] (a.east) -- ++(\rr,0) -- ++(0,-\dd) -- ($(a.west) + (-\ll,-\dd)$) -- +(0,\dd) -- (a.west);
}
\,\in\,\Rep_{S\sqcup \bar S}(\cala)
\]
be the Hilbert space defined in \eqref{eq:cyc fus}.

\begin{proposition}\label{prop: H_sig is canonical}
Given two $c$-covers $\cali_1$ and $\cali_2$ of a circle $S$, there is a canonical unitary isomorphism of $S\sqcup \bar S$-sectors
\[
u^{(\cali_1,\cali_2)}:\,H_\Sigma^{(\cali_1)}\longrightarrow H_\Sigma^{(\cali_2)}.
\]
Moreover, given three $c$-covers $\cali_1$, $\cali_2$, $\cali_3$ of $S$, the following diagram commutes:
\begin{equation}\label{eq:: H_sig is canonical}
\tikzmath{\node (1) at (-.08,0) {$H_\Sigma^{(\cali_1)}$};\node (2) at (2.4,.8) {$H_\Sigma^{(\cali_2)}$};\node (3) at (4.9,0) {$H_\Sigma^{(\cali_3)}$.};\draw[->] (1) --node[above, xshift=-.5, yshift=1]{$\scriptstyle u^{(\cali_1,\cali_2)}$} (2);\draw[->] (2) --node[above, xshift=4, yshift=1]{$\scriptstyle u^{(\cali_2,\cali_3)}$} (3);\draw[->] (1) --node[below, yshift=-0]{$\scriptstyle u^{(\cali_1,\cali_3)}$} (3);}
\end{equation}
\end{proposition}

\begin{proof}
Let $\cali_1$ and $\cali_2$ be two $c$-covers of $S$.
If $\cali_1$ is a refinement of $\cali_2$, then we may pick a sequence of $c$-covers
$\cali_1=\calj_1\prec \calj_2\prec \ldots\prec \calj_n=\cali_2$
such that each $\calj_n$ is obtained from the next one $\calj_{n+1}$
by subdividing some interval in two.
Using \eqref{eq: u:HH->H} in \eqref{eq: from cali to calj} produces a canonical
isomorphisms of $S\sqcup \bar S$-sectors $v_i:H_\Sigma^{(\calj_i)}\to H_\Sigma^{(\calj_{i+1})}$.
By Lemma \ref{lem: associativity of u}, the composite
\[
v^{(\cali_1,\cali_2)}:=v_{n-1}\circ \ldots \circ v_2\circ v_1\,\,\,:\,\,\,\,H_\Sigma^{(\cali_1)}\longrightarrow H_\Sigma^{(\cali_2)}
\]
is independent of the choice of intermediate $c$-covers and it is straightforward to check that
the above isomorphisms satisfy $v^{(\cali_2,\cali_3)}\circ v^{(\cali_1,\cali_2)}=v^{(\cali_1,\cali_3)}$.

Now given two arbitrary $c$-covers $\cali_1$ and $\cali_2$, we proceed as in Lemma~\ref{lem: H_Sigma^(cali)= H_Sigma^(calj)}.
Chose $c$-covers $\calj$, $\cali'$, $\calj'$ so that $\cali_1\succ \calj\prec \cali'\succ \calj'\prec \cali_2$
and set
\begin{equation*}
u^{(\cali_1,\cali_2)}:=v^{(\calj',\cali_2)} (v^{(\calj',\cali')})^{-1} v^{(\calj,\cali')} (v^{(\calj,\cali_1)})^{-1}\,\,:\,\,\,H_\Sigma^{(\cali_1)}\longrightarrow H_\Sigma^{(\cali_2)}.
\end{equation*}
It is then fairly easy to verify that $u^{(\cali_1,\cali_2)}$ is independent of the choice of intermediate $c$-covers $\calj$, $\cali'$, $\calj'$,
and that it satisfies $u^{(\cali_2,\cali_3)}\circ u^{(\cali_1,\cali_2)}=u^{(\cali_1,\cali_3)}$.
\end{proof}

Given an interval $I$ with local coordinates around its endpoints, let us denote by $\partial(I\times [0,1])_\top:=(I\times [{\textstyle \frac12},1])\cap \partial (I\times [0,1])$, the upper half of $\partial(I\times [0,1])$.
As a corollary, we have the following result:

\begin{theorem}\label{thm: any cover of the circle}
Let $S$ be a circle. Then the Hilbert space
$H_\Sigma\in\Rep_{S\sqcup \bar S}(\cala)$
associated to the annulus $\Sigma=S\times[0,1]$ is well defined up to canonical unitary isomorphism.

For every $c$-cover $\cali=\{I_1,\ldots, I_n\}$ (Definition \ref{def: c-cover}) of the circle $S$,
there is a canonical unitary isomorphism of $S\sqcup \bar S$-sectors
\smallskip
\begin{equation}\label{eq:  u^(cali) }
u^{(\cali)}:\,H_\Sigma\longrightarrow H_\Sigma^{(\cali)}=\tikzmath{  
\node (a) at (0,0) {$H_1\,\boxtimes_{A_2}\cdots\,\boxtimes_{A_n}\!H_n\,\,\boxtimes_{A_1}$};
\def\dd{.4}
\def\ll{.35}
\def\rr{.25}
\draw[dashed, rounded corners = 5] (a.east) -- ++(\rr,0) -- ++(0,-\dd) -- ($(a.west) + (-\ll,-\dd)$) -- +(0,\dd) -- (a.west);
}
\smallskip
\end{equation}
where
\begin{gather*}
H_i=L^2\cala\big(\partial (I_i\times [0,1])_\top\big)\in\Rep_{\partial I_i\times [0,1]}(\cala)\\
\text{and}\qquad A_i=\cala((I_{i-1}\cap I_i)\times [0,1]) 
\end{gather*}
are as in \eqref{eq:cyc fus}.

Moreover, given two $c$-covers $\cali_1$, $\cali_2$, the composite $u^{(\cali_2)}\circ (u^{(\cali_1)})^{-1}$ is the map $u^{(\cali_1,\cali_2)}$ from Proposition \ref{prop: H_sig is canonical}.
\end{theorem}

\begin{proof}
Let us define an element $\xi\in H_\Sigma$ to be a family of vectors $\{\xi^{(\cali)}\in H_\Sigma^{(\cali)}\}_\cali$ indexed by all $c$-covers $\cali$ of $S$,
subject to the condition
\begin{equation}\label{eq: condition u^(cali_1,cali_2)(xi^(cali_1))}
\qquad\quad u^{(\cali_1,\cali_2)}(\xi^{(\cali_1)})=\xi^{(\cali_2)}\qquad \forall\,\, \cali_1, \cali_2.
\end{equation}
Addition and scalar multiplication are defined pointwise.
The inner product is given by $\langle\{\xi^{(\cali)}\},\{\eta^{(\cali)}\}\rangle:=\langle\xi^{(\cali_0)},\eta^{(\cali_0)}\rangle$ for any $c$-cover $\cali_0$,
and is well defined because the maps $u^{(\cali_1,\cali_2)}$ are unitary.
Finally, the map $u^{(\calj)}$ sends a family $\{\xi^{(\cali)}\}_\cali$ to its $\calj$-th element~$\xi^{(\calj)}$, and the inverse map sends $\eta\in H_\Sigma^{(\calj)}$ to the family $\{u^{(\calj,\cali)}(\eta)\}_\cali$. The latter sign satisfies \eqref{eq: condition u^(cali_1,cali_2)(xi^(cali_1))} by~\eqref{eq:: H_sig is canonical}.
\end{proof}

As a consequence of Lemma \ref{lem: ubarJJ=Ju}, the Hilbert space $H_\Sigma$ is equipped with an antilinear involution
$J_\Sigma:H_\Sigma\to H_\Sigma$
given by
\begin{equation}\label{eq:J_Sigma}
\begin{split}
J_\Sigma\,:\,\,H_\Sigma\xrightarrow{u^{(\cali)}} \,\,&\tikzmath{  
\node (a) at (0,0) {$H_1\,\boxtimes_{A_2}\cdots\,\boxtimes_{A_n}\!H_n\,\,\boxtimes_{A_1}$};
\def\dd{.35}
\def\ll{.25}
\def\rr{.2}
\draw[dashed, rounded corners = 5] (a.east) -- ++(\rr,0) -- ++(0,-\dd) -- ($(a.west) + (-\ll,-\dd)$) -- +(0,\dd) -- (a.west);
}\\
\,\,&\xrightarrow{\epsilon^n\cdot J_1\boxtimes\ldots\boxtimes J_n}\,\,
\tikzmath{  
\node (a) at (0,0) {$H_1\,\boxtimes_{A_2}\cdots\,\boxtimes_{A_n}\!H_n\,\,\boxtimes_{A_1}$};
\def\dd{.35}
\def\ll{.25}
\def\rr{.2}
\draw[dashed, rounded corners = 5] (a.east) -- ++(\rr,0) -- ++(0,-\dd) -- ($(a.west) + (-\ll,-\dd)$) -- +(0,\dd) -- (a.west);
}
\xrightarrow{(u^{(\cali)})^*}H_\Sigma,
\end{split}
\end{equation}
where $J_i$ is the modular conjugation on $H_i=L^2(\cala(S_i^\top))$,
and $\epsilon\in\{\pm1\}$ is the sign that appears in Lemma \ref{lem: ubarJJ=Ju} (the latter only depends on the conformal net $\cala$, and is conjecturally equal to $1$).

Our next task is to construct a canonical isomorphism \eqref{H_Sigm=L^2hatA(S)} between $H_\Sigma$ and $L^2\cala(S)$.
This will occupy us for the remainder of this appendix.
Recall that $\Sigma=S\times [0,1]$, and that $H_\Sigma$ is the associated Hilbert space, constructed in \ref{thm: any cover of the circle}.

\begin{theorem}\label{thm: H_Sigma == L^2 cala(S)}
There is a canonical isomorphism of $S\sqcup \bar S$-sectors
\begin{equation}\label{H_Sigm=L^2hatA(S)+}
w: L^2\cala(S)\to H_\Sigma 
\end{equation}
that intertwines the modular conjugation on $L^2\cala(S)$ and the involution $J_\Sigma$ on $H_\Sigma$.
\end{theorem}

\begin{proof}
Recall from Theorem \ref{thm:compute-bfB(net)} that there is a canonical isomorphism
$\cala(S) \cong \bigoplus_{\lambda\in\Delta} \bfB(H_\lambda(S))$.
After applying the functor $L^2$, this becomes an isomorphism
\[
L^2\cala(S) \,\,\cong\,\, \bigoplus_{\lambda} \mathbf{HS}(H_\lambda(S))\,\cong\,\,\bigoplus_{\lambda} H_\lambda(S)\otimes \overline{H_\lambda(S)},
\]
where $\mathbf{HS}$ stands for the Hilbert space of Hilbert-Schmidt operators.
Note that, by Lemma \ref{lem: dual of H_lambda}, the right hand side is non-canonically isomorphic to $\bigoplus_{\lambda} H_\lambda(S)\otimes {H_{\bar \lambda}(\bar S)}$.

Let us simplify the notation and write $H_\lambda$ instead of $H_\lambda(S)$.
We therefore have a canonical isomorphism $L^2\cala(S)\cong \bigoplus_{\lambda}H_\lambda\otimes \overline{H_\lambda}$.
We know from \eqref{eq:   KLM  } that the $S\sqcup \bar S$-sectors $H_\Sigma$ and $L^2\cala(S)$ are isomorphic.
Therefore, in order to make the isomorphism canonical, it is enough to identify
$\hom(H_\lambda\otimes \overline{H_\lambda},H_\Sigma)$ and $\IC$ for every $\lambda\in\Delta$.
Once canonical isomorphisms 
\begin{equation}\label{eq: varpi}
\varpi_\lambda\,:\,\hom\big(H_\lambda\otimes \overline{H_\lambda},H_\Sigma\big) \,\rightarrow\, \IC
\end{equation}
are constructed, we may consider the isometries
\[
w_\lambda:={\textstyle\frac{\varpi_\lambda^{-1}(1)}{\|\varpi_\lambda^{-1}(1)\|}}:H_\lambda\otimes \overline{H_\lambda}\to H_\Sigma.
\]
Taking the sum over $\lambda$ will then give us the desired unitary isomorphism
\begin{equation}\label{eq: (+)H_l bar H_l --> H_Sig}
w:={\scriptstyle \bigoplus_\lambda} w_\lambda:L^2\cala(S) \,\cong\,\,\bigoplus_{\lambda} H_\lambda\otimes \overline{H_\lambda} \longrightarrow  H_\Sigma.
\end{equation}

As a first step towards \eqref{eq: varpi}, we construct an isomorphism
\begin{equation}\label{eq: map w_I}
\varpi_I=\varpi_{\lambda,I}\,:\,\hom\big(H_\lambda\otimes \overline{H_\lambda},H_\Sigma\big) \,\rightarrow\, \IC
\end{equation}
that depends on the choice of an interval $I\subset S$ and on the choice of local coordinates at the two endpoints of $I$.
We will show later that this map is in fact independent of the choice of interval and of local coordinates.
Let $I'$ be the closure of $S\setminus I$, and let
\[
S_I:=\partial(I\times [0,1])\quad
S_{I'}:=\partial(I'\times [0,1])\quad
S_I^+:=\partial(I\times [{\textstyle\frac12},1])\quad
S_I^-:=\partial(I\times [0,{\textstyle\frac12}])
\]
be as in the following picture:\medskip
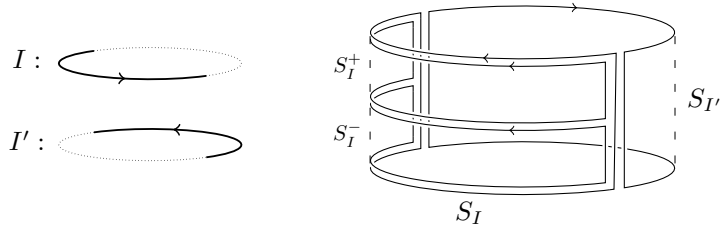
\[
\tikzmath[scale=.6]{\useasboundingbox (-3.05,-2.4) rectangle (2,0); \draw[densely dotted, very thin](-52:2 and .35) arc (-52:128:2 and .35);\draw[line width=.7](0,-1.8)+(-52:2 and .35) arc (-52:128:2 and .35); \draw[densely dotted, very thin](0,-1.8)+(128:2 and .35) arc (128:360-52:2 and .35);\draw[line width=.7](128:2 and .35) arc (128:360-52:2 and .35); \draw[->,line width=.7] (74:2 and .35)+(0,-1.8) arc (74:75:2 and .35);\draw[->,line width=.7] (-107:2 and .35) arc (-107:-106:2 and .35); \node at (-2.7,.1) {$I:$};\node at (-2.7,-1.7) {$I':$};}\, 
\qquad\quad \tikzmath{\draw(-48:2 and .35) ++(0,-1.8) arc (-48:128:2 and .35) -- ++ (0,1.8) arc (128:-48:2 and .35)(-228:2 and .35) -- ++(0,-1.8) arc (-228:-52:2 and .35) (-224:2 and .35) ++(0,-1.7) -- ++(0,.75)(-224:2 and .35) ++(0,-.1) -- ++(0,-.75);\draw[double, draw=white, double=black, double distance = .4, line width = .9] (-224:2 and .35) ++(0,-1.7) ++(0,.08) -- ++(0,-.08) arc (-224:-57:2 and .35) -- ++(0,.75) arc (-57:-224:2 and .35) -- ++(0,-.08); \draw[double, draw=white, double=black, double distance = .4, line width = .9](-224:2 and .35) ++(0,-.1) ++(0,-.08) -- ++(0,.08) arc (-224:-57:2 and .35) -- ++(0,-.75) arc (-57:-224:2 and .35) -- ++(0,.08); \draw[double, draw=white, double=black, double distance = .4, line width = .9] (-45:2 and .35) arc (-45:-48:2 and .35) --  +(0,-1.8) arc (-48:-45:2 and .35) (-55:2 and .35) +(0,-1.8) arc (-54:-52:2 and .35) -- (-52:2 and .35) arc (-52:-228:2 and .35) -- +(0,-.08); \draw[very thin, loosely dashed] (2.005,-.05) -- (2.005,-1.8)(-2.005,-.05) -- (-2.005,-1.8); \node[scale=.8] at (-2.3,-.45) {$S_I^+$}; \node[scale=.8] at (-2.3,-1.35) {$S_I^-$};\node[scale=1] at (2.4,-.9) {$S_{I'}$}; \node[scale=1] at (-.7,-2.4) {$S_I$};\draw[->] (-104:2 and .35) arc (-104:-105:2 and .35);\draw[->] (70:2 and .35) arc (70:69:2 and .35);\draw[->] (0,-.1) ++(-94:2 and .35) arc (-94:-95:2 and .35);\draw[->] (0,-.95) ++(-94:2 and .35) arc (-94:-95:2 and .35);} 
\]
Recall that $\partial(I\times [0,1])_\top$ denotes the upper half of $S_I$, and
let us write $\partial(I\times [0,1])_\bot$ for the corresponding lower half.

Let also $H_I:=H_0(S_I)$,  $H_I':=H_0(S_{I'})$, $H_I^+:=H_0(S_I^+)$, $H_I^-:=H_0(S_I^-)$ be the correspdonding vacuum sectors.
More precisely, we take $H_I$ to be the vacuum sector associated to the circle $S_I$, its upper half $\partial(I\times [0,1])_\top$, and the involution $j:(x,t)\mapsto(x,1-t)$.
Similarly, we let $H_I'$ be the vacuum sector associated to $S_{I'}$, its upper half, and $j$.
The vacuum sector $H_I^+\in\Rep_{S_I^+}(\cala)$ is chosen arbitrarily in its isomorphism class.
Finally, we take $H_I^-:=j^*(H_I^+)$.\footnote{The functor $j^*:\Rep_{S_I^+}(\cala)\to\Rep_{S_I^-}(\cala)$ 
    is defined in~\cite[\eqfunctorphi]{BDH(nets)}.}
Let also
\[
\begin{split}
A:=\cala(I'\times\{0,1\}),\quad
B:=\cala(I\times\{0,1\}),\quad
C:=\cala(\partial I\times[0,1]),\\
B_0:=\cala(I\times\{0\}),\quad\,
B_1:=\cala(I\times\{1\}),\quad
B_{1/2}:=\cala(I\times\{1/2\}),\\
D_0:=\cala\big(\partial (I\times [0,1])_\bot\big),\quad
D_1:=\cala\big(\partial (I\times [0,1])_\top\big)\hspace{1cm}&
\end{split}
\]
with orientations as indicated here:
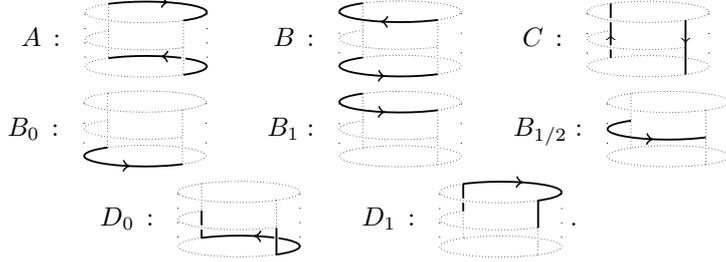
\[
\begin{split} A\,:\,\,
\tikzmath[scale=.4]{\draw[line width=.7](-52:2 and .35) arc (-52:128:2 and .35);\draw[line width=.7](0,-1.8)+(-52:2 and .35) arc (-52:128:2 and .35);\draw[densely dotted, very thin](0,-1.8)+(128:2 and .35) arc (128:360-52:2 and .35);\draw[densely dotted, very thin](128:2 and .35) -- ++(0,-.9) ++(0,-.9);\draw[densely dotted, very thin](128:2 and .35) ++(0,-.9) -- ++(0,-.9);\fill[white](-1.23,-1.177) ++(.07,-.08) -- ++ (0,.14) -- ++ (-.14,.02) -- ++ (0,-.14) -- cycle(-1.23,-.277) ++(.07,-.08) -- ++ (0,.14) -- ++ (-.14,.02) -- ++ (0,-.14) -- cycle(1.23,.277-1.8) ++(.07,-.08) -- ++ (0,.14) -- ++ (-.14,.02) -- ++ (0,-.14) -- cycle; \draw[densely dotted, very thin](128:2 and .35) arc (128:360-52:2 and .35);\draw[densely dotted, very thin](0,-.9)+(128:2 and .35) arc (128:360-52:2 and .35);\draw[densely dotted, very thin] (-52:2 and .35) -- ++(0,-.9) ++(0,-.9);\draw[densely dotted, very thin] (-52:2 and .35) ++(0,-.9) -- ++(0,-.9);\draw[very thin, dash pattern=on .4pt off 5pt] (2.005,-.05) -- (2.005,-1.8)(-2.005,-.05) -- (-2.005,-1.8);\draw[->,line width=.7] (70:2 and .35) arc (70:69:2 and .35);\draw[->,line width=.7] (74:2 and .35)+(0,-1.8) arc (74:75:2 and .35);}\, 
\qquad B\,&:\,\,
\tikzmath[scale=.4]{\draw[densely dotted, very thin](-52:2 and .35) arc (-52:128:2 and .35);\draw[densely dotted, very thin](0,-1.8)+(-52:2 and .35) arc (-52:128:2 and .35); \draw[line width=.7](0,-1.8)+(128:2 and .35) arc (128:360-52:2 and .35);\draw[densely dotted, very thin](128:2 and .35) -- ++(0,-1.8); \fill[white](-1.23,-1.177) ++(.07,-.08) -- ++ (0,.14) -- ++ (-.14,.02) -- ++ (0,-.14) -- cycle(-1.23,-.277) ++(.07,-.08) -- ++ (0,.14) -- ++ (-.14,.02) -- ++ (0,-.14) -- cycle(1.23,.277-1.8) ++(.07,-.08) -- ++ (0,.14) -- ++ (-.14,.02) -- ++ (0,-.14) -- cycle; \draw[line width=.7](128:2 and .35) arc (128:360-52:2 and .35); \draw[densely dotted, very thin](0,-.9)+(128:2 and .35) arc (128:360-52:2 and .35);\draw[densely dotted, very thin](-52:2 and .35) -- ++(0,-1.8);\draw[very thin, dash pattern=on .4pt off 7pt] (2.005,-.05) -- (2.005,-1.8)(-2.005,-.05) -- (-2.005,-1.8);\draw[->,line width=.7] (-107:2 and .35) arc (-107:-108:2 and .35);\draw[->,line width=.7] (-106:2 and .35)+(0,-1.8) arc (-106:-105:2 and .35);} 
\qquad C\,:\,\,
\tikzmath[scale=.4]{\draw[densely dotted, very thin](-52:2 and .35) arc (-52:128:2 and .35);\draw[densely dotted, very thin](0,-1.8)+(-52:2 and .35) arc (-52:128:2 and .35);\draw[densely dotted, very thin](0,-1.8)+(128:2 and .35) arc (128:360-52:2 and .35);\draw[line width=.7](128:2 and .35) -- ++(0,-1.8);\fill[white](-1.23,-1.177) ++(.07,-.08) -- ++ (0,.14) -- ++ (-.14,.02) -- ++ (0,-.14) -- cycle(-1.23,-.277) ++(.07,-.08) -- ++ (0,.14) -- ++ (-.14,.02) -- ++ (0,-.14) -- cycle(1.23,.277-1.8) ++(.07,-.08) -- ++ (0,.14) -- ++ (-.14,.02) -- ++ (0,-.14) -- cycle; \draw[densely dotted, very thin](128:2 and .35) arc (128:360-52:2 and .35);\draw[densely dotted, very thin](0,-.9)+(128:2 and .35) arc (128:360-52:2 and .35);\draw[line width=.7] (-52:2 and .35) -- ++(0,-1.8);\draw[very thin, dash pattern=on .4pt off 7pt] (2.005,-.05) -- (2.005,-1.8)(-2.005,-.05) -- (-2.005,-1.8);\draw[->] (-52:2 and .35) ++(0,-.65) -- ++(0,-.1);\draw[->] (128:2 and .35) ++(0,-1.15) -- ++(0,.1);}\\  
B_0\,:\,\,
\tikzmath[scale=.4]{\draw[densely dotted, very thin](-52:2 and .35) arc (-52:128:2 and .35);\draw[densely dotted, very thin](0,-1.8)+(-52:2 and .35) arc (-52:128:2 and .35);\draw[line width=.7](0,-1.8)+(128:2 and .35) arc (128:360-52:2 and .35);\draw[densely dotted, very thin](128:2 and .35) -- ++(0,-1.8);\fill[white](-1.23,-1.177) ++(.07,-.08) -- ++ (0,.14) -- ++ (-.14,.02) -- ++ (0,-.14) -- cycle(-1.23,-.277) ++(.07,-.08) -- ++ (0,.14) -- ++ (-.14,.02) -- ++ (0,-.14) -- cycle(1.23,.277-1.8) ++(.07,-.08) -- ++ (0,.14) -- ++ (-.14,.02) -- ++ (0,-.14) -- cycle; \draw[densely dotted, very thin](128:2 and .35) arc (128:360-52:2 and .35);\draw[densely dotted, very thin](0,-.9)+(128:2 and .35) arc (128:360-52:2 and .35);\draw[densely dotted, very thin](-52:2 and .35) -- ++(0,-1.8);\draw[very thin, dash pattern=on .4pt off 7pt] (2.005,-.05) -- (2.005,-1.8)(-2.005,-.05) -- (-2.005,-1.8);\draw[->,line width=.7] (-106:2 and .35)+(0,-1.8) arc (-106:-105:2 and .35);} 
\qquad B_1&:\,\,
\tikzmath[scale=.4]{\draw[densely dotted, very thin](-52:2 and .35) arc (-52:128:2 and .35);\draw[densely dotted, very thin](0,-1.8)+(-52:2 and .35) arc (-52:128:2 and .35);\draw[densely dotted, very thin](0,-1.8)+(128:2 and .35) arc (128:360-52:2 and .35);\draw[densely dotted, very thin](128:2 and .35) -- ++(0,-1.8);\fill[white](-1.23,-1.177) ++(.07,-.08) -- ++ (0,.14) -- ++ (-.14,.02) -- ++ (0,-.14) -- cycle(-1.23,-.277) ++(.07,-.08) -- ++ (0,.14) -- ++ (-.14,.02) -- ++ (0,-.14) -- cycle(1.23,.277-1.8) ++(.07,-.08) -- ++ (0,.14) -- ++ (-.14,.02) -- ++ (0,-.14) -- cycle; \draw[line width=.7](128:2 and .35) arc (128:360-52:2 and .35);\draw[densely dotted, very thin](0,-.9)+(128:2 and .35) arc (128:360-52:2 and .35);\draw[densely dotted, very thin](-52:2 and .35) -- ++(0,-1.8);\draw[very thin, dash pattern=on .4pt off 7pt] (2.005,-.05) -- (2.005,-1.8)(-2.005,-.05) -- (-2.005,-1.8);\draw[->,line width=.7] (-106:2 and .35) arc (-106:-105:2 and .35);} 
\qquad\!\! B_{1/2}:\,\,
\tikzmath[scale=.4]{\draw[densely dotted, very thin](-52:2 and .35) arc (-52:128:2 and .35);\draw[densely dotted, very thin](0,-1.8)+(-52:2 and .35) arc (-52:128:2 and .35);\draw[densely dotted, very thin](0,-1.8)+(128:2 and .35) arc (128:360-52:2 and .35);\draw[densely dotted, very thin](128:2 and .35) -- ++(0,-1.8);\fill[white](-1.23,-1.177) ++(.07,-.08) -- ++ (0,.14) -- ++ (-.14,.02) -- ++ (0,-.14) -- cycle(-1.23,-.277) ++(.07,-.08) -- ++ (0,.14) -- ++ (-.14,.02) -- ++ (0,-.14) -- cycle(1.23,.277-1.8) ++(.07,-.08) -- ++ (0,.14) -- ++ (-.14,.02) -- ++ (0,-.14) -- cycle; \draw[densely dotted, very thin](128:2 and .35) arc (128:360-52:2 and .35);\draw[line width=.7](0,-.9)+(128:2 and .35) arc (128:360-52:2 and .35);\draw[densely dotted, very thin](-52:2 and .35) -- ++(0,-1.8);\draw[very thin, dash pattern=on .4pt off 7pt] (2.005,-.05) -- (2.005,-1.8)(-2.005,-.05) -- (-2.005,-1.8);\draw[->,line width=.7] (-106:2 and .35)+(0,-.9) arc (-106:-105:2 and .35);} 
\\D_0\,:\,\,
\tikzmath[scale=.4]{\draw[densely dotted, very thin] (-52:2 and .35) arc (-52:128:2 and .35);\draw[line width=.7, line join = bevel](0,-1.8)++(-52:2 and .35) arc (-52:128:2 and .35) -- ++(0,.9);\draw[densely dotted, very thin](0,-1.8)+(128:2 and .35) arc (128:360-52:2 and .35);\draw[densely dotted, very thin](128:2 and .35) -- ++(0,-.9) ++(0,-.9);\fill[white](-1.23,-1.177) ++(.07,-.08) -- ++ (0,.14) -- ++ (-.14,.02) -- ++ (0,-.14) -- cycle(-1.23,-.277) ++(.07,-.08) -- ++ (0,.14) -- ++ (-.14,.02) -- ++ (0,-.14) -- cycle(1.23,.277-1.8) ++(.07,-.08) -- ++ (0,.14) -- ++ (-.14,.02) -- ++ (0,-.14) -- cycle; \draw[densely dotted, very thin](128:2 and .35) arc (128:360-52:2 and .35);\draw[densely dotted, very thin](0,-.9)+(128:2 and .35) arc (128:360-52:2 and .35);\draw[densely dotted, very thin] (-52:2 and .35) -- ++(0,-.9) ++(0,-.9);\draw[line width=.7, line cap=round] (-52:2 and .35) ++(0,-1.8) -- ++(0,.86);\draw[very thin, dash pattern=on .4pt off 5pt] (2.005,-.05) -- (2.005,-1.8)(-2.005,-.05) -- (-2.005,-1.8);\draw[->,line width=.7] (74:2 and .35)+(0,-1.8) arc (74:75:2 and .35);}\!\!&  
\qquad\,\,D_1\,:\,\,
\tikzmath[scale=.4]{\draw[line width=.7, line join = bevel](-52:2 and .35)++(0,-.9) -- ++(0,.9) arc (-52:128:2 and .35) -- ++(0,-.9);\draw[densely dotted, very thin](0,-1.8)+(-52:2 and .35) arc (-52:128:2 and .35);\draw[densely dotted, very thin](0,-1.8)+(128:2 and .35) arc (128:360-52:2 and .35);\draw[densely dotted, very thin](128:2 and .35) ++(0,-.9) -- ++(0,-.9);\fill[white](-1.23,-1.177) ++(.07,-.08) -- ++ (0,.14) -- ++ (-.14,.02) -- ++ (0,-.14) -- cycle(-1.23,-.277) ++(.07,-.08) -- ++ (0,.14) -- ++ (-.14,.02) -- ++ (0,-.14) -- cycle(1.23,.277-1.8) ++(.07,-.08) -- ++ (0,.14) -- ++ (-.14,.02) -- ++ (0,-.14) -- cycle;\draw[densely dotted, very thin](128:2 and .35) arc (128:360-52:2 and .35);\draw[densely dotted, very thin](0,-.9)+(128:2 and .35) arc (128:360-52:2 and .35);\draw[densely dotted, very thin] (-52:2 and .35) ++(0,-.9) -- ++(0,-.9);\draw[very thin, dash pattern=on .4pt off 5pt] (2.005,-.05) -- (2.005,-1.8)(-2.005,-.05) -- (-2.005,-1.8);\draw[->,line width=.7] (70:2 and .35) arc (70:69:2 and .35);}.\end{split} 
\]
Finally, let us denote by $D$ the algebra $\cala\big([0,\frac12]\cup I'\cup[0,\frac12]\big)$, which we identify with $D_1$ in the obvious way.
With all those preliminaries in place, we can now define the map \eqref{eq: map w_I} as the composite of the following isomorphisms:
\begin{equation}\label{eq: defin. of w_I}
\begin{split}
\hspace{-.7cm}\varpi_I:\,\hom\big(H_\lambda\otimes \overline{H_\lambda},H_\Sigma\big)
&\to\hom_{A{\scriptscriptstyle\vee}B}\big(H_\lambda\otimes \overline{H_\lambda},H_I\boxtimes_C H'_I\big)\hspace{-1.2cm}\\\hspace{-1.2cm}
&\to\hom_{A{\scriptscriptstyle\vee}C}\big((H_\lambda\otimes \overline{H_\lambda})\boxtimes_B \overline{H_I},H'_I\big)\\\hspace{-1.2cm}
&\to\hom_{A{\scriptscriptstyle\vee}C}\big((H_\lambda\otimes \overline{H_\lambda})\boxtimes_B H_I,H'_I\big)\\\hspace{-1.2cm}
&\to\hom_{A{\scriptscriptstyle\vee}C}\big(H_\lambda\boxtimes_{B_0} H_I\boxtimes_{B_1} \overline{H_\lambda},H'_I\big)\\\hspace{-1.2cm}
&\to\hom_{A{\scriptscriptstyle\vee}C}\big(H_\lambda\boxtimes_{B_0} H_I^-\boxtimes_{B_{1/2}} H_I^+\boxtimes_{B_1} \overline{H_\lambda},H'_I\big)\\\hspace{-1.2cm}
&=\hom_{D_0{\scriptscriptstyle\vee}D_1}\!\big(H_\lambda\boxtimes_{B_0} H_I^-\boxtimes_{B_{1/2}} H_I^+\boxtimes_{B_1} \overline{H_\lambda},H'_I\big)\\\hspace{-1.2cm}
&=\hom_{D,D}\!\big(H_\lambda\boxtimes_{B_0} H_I^-\boxtimes_{B_{1/2}} H_I^+\boxtimes_{B_1} \overline{H_\lambda},L^2(D)\big)\\\hspace{-1.2cm}
&\to\IC.
\end{split}
\end{equation}
The first map is the isomorphism \eqref{eq:  u^(cali) } associated to the $c$-cover $\{I,I'\}$ of $S$.
The second map is a duality isomorphism, using that $\cala$ has finite index.
The third map is the modular conjugation of $H_I=L^2\cala(\partial(I\times [0,1])_\top)$.
The fourth map is an instance Lemma \ref{lem: HKK}.
The fifth map is the counit of the duality between $H_I^-$ and $H_I^+$. 
Finally, the sixth and last map sends the counit of the duality between $H_\lambda\boxtimes_{B_0} H_I^-$ and $H_I^+\boxtimes_{B_1} \overline{H_\lambda}$
to the complex number $1\in \IC$.

At this point, we can check that \eqref{eq: (+)H_l bar H_l --> H_Sig} intertwines the modular conjugation $J$ on $L^2\cala(S)$ and the involution $J_\Sigma$ on $H_\Sigma$.
Note that the restriction of $J$ to $H_\lambda\otimes\overline{H_\lambda}$ is the map that exchanges the two factors.
We need to show that $\varpi_{\lambda}^{-1}(1):H_\lambda\otimes\overline{H_\lambda}\to H_\Sigma$ intertwines that involution with $J_\Sigma$.
In other words, we need to show that $\varpi_\lambda^{-1}(1)$ is invariant under the natural involution on $\hom(H_\lambda\otimes \overline{H_\lambda},H_\Sigma)$.
All the spaces in \eqref{eq: defin. of w_I} are equipped with their own antilinear involution, and all the maps are compatible with the involutions.
The invariance of $\varpi_\lambda^{-1}(1)$ then follows from the invariance of $1\in \mathbb C$ under complex conjugation.

Given two intervals $I_1,I_2\subset S$, each one with local coordinates at their boundary, we still need to show that $\varpi_{I_1}=\varpi_{I_2}$.
Without loss of generality, we may assume that $I_1\subset I_2$, that the two intervals share a common boundary point $p$, and that the local coordinates agree at that point.
Let $\varphi\in\Diff_+(S)$ be a diffeomorphism that sends $I_1$ to $I_2$ and that is compatible with the local coordinates.
In particular, $\varphi$ fixes a neighborhood of the point $p$.
Pick an interval $J$ that contains $\supp(\varphi)$ in its interior and that misses the point $p$,
and let $v\in\cala(J)$ be a unitary that implements~$\varphi$, that is, such that $\ad(v)=\cala(\varphi)$.

Let $H_1:=H_{I_1}$, $H_1':=H_{I_1}'$, $H_1^+:=H_{I_1}^+$, $H_1^-:=H_{I_1}^-$ be the Hilbert spaces that enter in the definition of $\varpi_{I_1}$,
and let $H_2$, $H_2'$, $H_2^+$, $H_2^-$ be the corresponding Hilbert spaces for $\varpi_{I_2}$.
By evaluating $L^2(\cala(-))$ on the diffeomorphisms that $\varphi$ induces between the upper halves of the appropriate circles,
we get unitary isomorphisms $a:H_1 \to H_2$ and $a':H_1' \to H_2'$.
Let $b^+:H_1^+ \to \varphi^*(H_2^+)$ be an arbitrary $S_{I_1}^+$-sector isomorphism, and let $b^-:=j^*(b^+):H_1^- \to \varphi^*(H_2^-)$.
Finally, let $c:H_\lambda\to H_\lambda$ and $\bar c:\overline{H_\lambda}\to \overline{H_\lambda}$ be the maps induced by multiplication by $v$.
We then get a diagram
\[
\tikzmath{
\node[anchor=east] (a0) at (0,10) {$\hom\big(H_\lambda\otimes \overline{H_\lambda},H_\Sigma\big)$};
\node[anchor=east] (a1) at (0,9) {$\hom\big(H_\lambda\otimes \overline{H_\lambda},H_1\boxtimes H'_1\big)$};
\node[anchor=east] (a2) at (0,8) {$\hom\big((H_\lambda\otimes \overline{H_\lambda})\boxtimes \overline{H_1},H'_1\big)$};
\node[anchor=east] (a3) at (0,7) {$\hom\big((H_\lambda\otimes \overline{H_\lambda})\boxtimes H_1,H'_1\big)$};
\node[anchor=east] (a4) at (0,6) {$\hom\big(H_\lambda\boxtimes H_1\boxtimes \overline{H_\lambda},H'_1\big)$};
\node[anchor=east] (a5) at (0,5) {$\hom\big(H_\lambda\boxtimes H_1^-\boxtimes H_1^+\boxtimes \overline{H_\lambda},H'_1\big)$};
\node[anchor=east] (a6) at (-1.5,4) {$\IC$};
\node[anchor=west] (b0) at (1,10) {$\hom\big(H_\lambda\otimes \overline{H_\lambda},H_\Sigma\big)$};
\node[anchor=west] (b1) at (1,9) {$\hom\big(H_\lambda\otimes \overline{H_\lambda},H_2\boxtimes H'_2\big)$};
\node[anchor=west] (b2) at (1,8) {$\hom\big((H_\lambda\otimes \overline{H_\lambda})\boxtimes \overline{H_2},H'_2\big)$};
\node[anchor=west] (b3) at (1,7) {$\hom\big((H_\lambda\otimes \overline{H_\lambda})\boxtimes H_2,H'_2\big)$};
\node[anchor=west] (b4) at (1,6) {$\hom\big(H_\lambda\boxtimes H_2\boxtimes \overline{H_\lambda},H'_2\big)$};
\node[anchor=west] (b5) at (1,5) {$\hom\big(H_\lambda\boxtimes H_2^-\boxtimes H_2^+\boxtimes \overline{H_\lambda},H'_2\big)$};
\node[anchor=west] (b6) at (2.5,4) {$\IC$};
\foreach \x in {1,...,5}{\draw[->] (a\x) -- (b\x);}
\foreach \x in {0,6}{\draw[double, double distance=1pt] (a\x) -- (b\x);}
\foreach \x/\y in {0/1,1/2,2/3,3/4,4/5,5/6}{
\draw[->] (a\x.south -| a6) -- (a\y.north -| a6);
\draw[->] (b\x.south -| b6) -- (b\y.north -| b6);
\node at ($(a\x.east)+(.5,-.5)$) {\squared\y};}
\draw[->, shorten >=1.5] (a5.south -| a6) -- (a6.north -| a6);
\draw[->, shorten >=1.5] (b5.south -| b6) -- (b6.north -| b6);
}
\]
where various subscripts are left implicit, the vertical arrows compose to $\varpi_{I_1}$ and $\varpi_{I_2}$,
and the horizontal arrows are induced by $a$, $a'$, $b^+$, $b^-$, $c$, $\bar c$.
The squares $\squared2,\, \squared3, \ldots, \squared6$ are easily seen to commute.
Finally, the commutativity of $\squared1$ is the content of Lemma \ref{lem: finish proof of annulus sector}.
\end{proof}

Recall that $H_1=H_{I_1}$, $H_2=H_{I_2}$, $H_1'=H_{I_1}'$, $H_2'=H_{I_2}'$ are the vacuum sectors associated to the boundaries of $I_1\times [0,1]$,  $I_2\times [0,1]$, $I_1'\times [0,1]$, and $I_2'\times [0,1]$.

\begin{lemma}\label{lem: finish proof of annulus sector}
Let $H_\lambda$, $H_\Sigma$, $H_{I_1}$, $H_{I_1}'$, $H_{I_2}$, $H_{I_2}'$, $a$, $a'$, $c$, $\bar c$ be as above,
and let
$
u_1:H_\Sigma\to \tikzmath{  
\node (A) at (0,0) {$H_{I_1}\boxtimes H_{I_1}'\boxtimes$};
\def\dd{.3}
\def\ll{.1}
\def\rr{.15}
\draw[dashed, rounded corners = 5] (A.east) -- ++(\rr,0) -- ++(0,-\dd) -- ($(A.west) + (-\ll,-\dd)$) -- +(0,\dd) -- (A.west);
}
$ and
$u_2:H_\Sigma\to \tikzmath{  
\node (A) at (0,0) {$H_{I_2}\boxtimes H_{I_2}'\boxtimes$};
\def\dd{.3}
\def\ll{.1}
\def\rr{.15}
\draw[dashed, rounded corners = 5] (A.east) -- ++(\rr,0) -- ++(0,-\dd) -- ($(A.west) + (-\ll,-\dd)$) -- +(0,\dd) -- (A.west);
}
$
be instances of \eqref{eq:  u^(cali) }.
Then the triangle
\[
\!\tikzmath{
\node at (-2.4,.8) {$\scriptstyle u_1\circ -$};
\node at (2.4,.8) {$\scriptstyle u_2\circ -$};
\node at (0,.25) {$\scriptstyle (a\boxtimes a')\circ - \circ(c\otimes \bar c)^*$};
\node (a) at (0,1.1) {$\hom\big(H_\lambda\otimes \overline{H_\lambda},H_\Sigma\big)$};
\node (b) at (-3.85,0) {$\hom\big(H_\lambda\otimes \overline{H_\lambda},
\tikzmath{  
\node (A) at (0,0) {$H_{I_1}\boxtimes H_{I_1}'\boxtimes$};
\def\dd{.3}
\def\ll{.1}
\def\rr{.15}
\draw[dashed, rounded corners = 5] (A.east) -- ++(\rr,0) -- ++(0,-\dd) -- ($(A.west) + (-\ll,-\dd)$) -- +(0,\dd) -- (A.west);
}
\,\big)$};
\node (c) at (3.85,0) {$\hom\big(H_\lambda\otimes \overline{H_\lambda},
\tikzmath{  
\node (A) at (0,0) {$H_{I_2}\boxtimes H_{I_2}'\boxtimes$};
\def\dd{.3}
\def\ll{.1}
\def\rr{.15}
\draw[dashed, rounded corners = 5] (A.east) -- ++(\rr,0) -- ++(0,-\dd) -- ($(A.west) + (-\ll,-\dd)$) -- +(0,\dd) -- (A.west);
}\,
\big)$};
\draw[->](a)--(b);\draw[->](a)--(c);\draw[->](b)--(c);
}
\]
is commutative, where the $\hom$ are taken in the category of $S\sqcup\bar S$-sectors.
\end{lemma}

\begin{proof}
Given a morphism of $S\sqcup\bar S$-sectors $f:H_\lambda\otimes \overline{H_\lambda}\to H_\Sigma$,
we need to show that $(a\boxtimes a')\circ u_1\circ f \circ(c\otimes \bar c)^*=u_2\circ f$.
Recall that $c$ is the operator of multiplication by the unitary $v$, defined above.
The operator
\[
c\otimes \bar c: H_\lambda\otimes \overline{H_\lambda}\to H_\lambda\otimes \overline{H_\lambda}
\]
is therefore given by multiplication by $v\otimes \bar v$.
Since $u_1\circ f$ commutes with the action of that element, we have
\[
(a\boxtimes a')\circ u_1\circ f \circ(c\otimes \bar c)^*=(a\boxtimes a')\circ(v\otimes \bar v)^*\circ u_1\circ f,
\]
and so it is enough to show that $(a\boxtimes a')(v\otimes \bar v)^*u_1=u_2$.
In other words, we need to show that the following triangle commutes:
\begin{equation}\label{eq: triangle that must commute to finish proof of annular sector}
\,\tikzmath{
\node at (-1.2,.8) {$\scriptstyle u_1$};\node at (1.2,.8) {$\scriptstyle u_2$};\node at (.02,.25) {$\scriptstyle (a\boxtimes a')(v\otimes \bar v)^*$};\node (a) at (0,1) {$H_\Sigma$};\node (b) at (-2.8,0) {$
\tikzmath{\node (A) at (0,0) {$H_{I_1}\boxtimes H_{I_1}'\boxtimes$};\def\dd{.3}\def\ll{.1}\def\rr{.15}
\draw[dashed, rounded corners = 5] (A.east) -- ++(\rr,0) -- ++(0,-\dd) -- ($(A.west) + (-\ll,-\dd)$) -- +(0,\dd) -- (A.west);}$};
\node (c) at (2.8,0) {$\tikzmath{\node (A) at (0,0) {$H_{I_2}\boxtimes H_{I_2}'\boxtimes$};\def\dd{.3}\def\ll{.1}\def\rr{.15}
\draw[dashed, rounded corners = 5] (A.east) -- ++(\rr,0) -- ++(0,-\dd) -- ($(A.west) + (-\ll,-\dd)$) -- +(0,\dd) -- (A.west);}$};
\draw[->, shorten >=-2](a)--(b);\draw[->, shorten >=-2](a)--(c);\draw[->,shorten <=2, shorten >=2](b)--(c);
}
\end{equation}

Recall that $I_1$ is contained in $I_2$, and that $p$ is their common boundary point.
Recall that $\varphi$ is a diffeomorphism that satisfies $\varphi(I_1)=I_2$, and that it is the identity in a neighborhood of $p$.
Finally, recall that $\ad(v)=\cala(\varphi)$.
Let $q\in I_2'$ be such that the interval $K:=[q,p]\subset S$ is disjoint from $\supp(\varphi)$, and let $J_1$ and $J_2$ be the closures of $I_1'\setminus K$ and $I_2'\setminus K$.
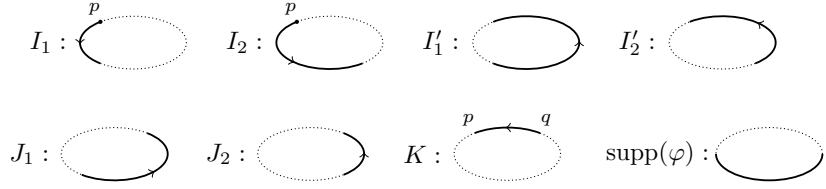
\[
\tikzmath[scale=.7]{
\node at (-.73,.65) {$\scriptstyle p$};\fill (128:1 and .5) circle (0.04);
\draw[line width=.7] (360-128:1 and .5) arc (360-128:128:1 and .5);\draw[densely dotted] (-128:1 and .5) arc (-128:128:1 and .5);\draw[->] (182:1 and .5) -- +(0,-.01);\node[scale=.95] at (-1.6,0) {$I_1:$};\pgftransformxshift{105}
\node at (-.73,.65) {$\scriptstyle p$};\fill (128:1 and .5) circle (0.04);
\draw[line width=.7] (360-52:1 and .5) arc (360-52:128:1 and .5);\draw[densely dotted] (-52:1 and .5) arc (-52:128:1 and .5);\draw[->] (182+45:1 and .5) -- +(.02,-.01);\node[scale=.95] at (-1.6,0) {$I_2:$};\pgftransformxshift{105}
\draw[densely dotted] (360-128:1 and .5) arc (360-128:128:1 and .5);\draw[line width=.7] (-128:1 and .5) arc (-128:128:1 and .5);\draw[->] (2:1 and .5) -- +(0,.01);\node[scale=.95] at (-1.6,0) {$I_1':$};
\pgftransformxshift{105}
\draw[densely dotted] (360-52:1 and .5) arc (360-52:128:1 and .5);\draw[line width=.7] (-52:1 and .5) arc (-52:128:1 and .5);\draw[->] (47:1 and .5) -- +(-.02,.01);\node[scale=.95] at (-1.6,0) {$I_2':$};
\pgftransformxshift{-325}\pgftransformyshift{-60}
\draw[line width=.7] (52:1 and .5) arc (52:-128:1 and .5);\draw[densely dotted] (52:1 and .5) arc (52:360-128:1 and .5);\draw[->] (182+135:1 and .5) -- +(.02,.01);\node[scale=.95] at (-1.6,0) {$J_1:$};
\pgftransformxshift{105}
\draw[line width=.7] (52:1 and .5) arc (52:-52:1 and .5);\draw[densely dotted] (52:1 and .5) arc (52:360-52:1 and .5);\draw[->] (2:1 and .5) -- +(0,.01);\node[scale=.95] at (-1.6,0) {$J_2:$};
\pgftransformxshift{105}
\node at (-.73,.65) {$\scriptstyle p$};
\node at (.73,.65) {$\scriptstyle q$};
\draw[densely dotted] (52:1 and .5) arc (52:128-360:1 and .5);\draw[line width=.7] (52:1 and .5) arc (52:128:1 and .5);\draw[->] (91:1 and .5) -- +(-.01,0);\node[scale=.95] at (-1.6,0) {$K:$};
\pgftransformxshift{140}
\draw[densely dotted](0:1 and .5) arc (0:180:1 and .5);
\draw[line width=.7] (0:1 and .5) arc (0:-180:1 and .5);\node[scale=.95] at (-2.1,0) {$\supp(\varphi):$};}
\]
Chose a local coordinate at the point $q$, and
let $H_{J_1}$, $H_{J_2}$, $H_{K}$ be the vacuum sectors associated to the boundaries of $J_1\times [0,1]$,  $J_2\times [0,1]$, $K\times [0,1]$, their upper halves, and the involution $j$.
We then have the following two commutative triangles
\begin{equation}\label{eq: two triangles}
\!\tikzmath{
\node (a1) at (-2.75,1) {$H_\Sigma$};
\node (a2) at (3.69,1) {$H_\Sigma$};
\draw[->] (a1) -- +(1,-.5);
\draw[->] (a1) -- +(-1,-.5);
\draw[->] (a2) -- +(1,-.5);
\draw[->] (a2) -- +(-1,-.5);
\draw[->] (a1)+(-.1,-.9) --node[above, scale=.9]{$\scriptstyle k_1$} +(.21,-.9);
\draw[->] (a2)+(-.1,-.9) --node[above, scale=.9]{$\scriptstyle k_2$} +(.21,-.9);
\node (b) at (0,.1) {$
\tikzmath{
\node (a) at (0,0) {$H_{I_1} \boxtimes H_{J_1} \boxtimes H_K \boxtimes$};
\def\dd{.3}
\def\ll{.1}
\def\rr{.1}
\draw[dashed, rounded corners = 5] (a.east) -- ++(\rr,0) -- ++(0,-\dd) -- ($(a.west) + (-\ll,-\dd)$) -- +(0,\dd) -- (a.west);
}
\hspace{.43cm}
\tikzmath{\node (A) at (0,0) {$H_{I_1}\boxtimes H_{I_1}'\boxtimes$};
\def\dd{.3}
\def\ll{.1}
\def\rr{.1}
\draw[dashed, rounded corners = 5] (A.east) -- ++(\rr,0) -- ++(0,-\dd) -- ($(A.west) + (-\ll,-\dd)$) -- +(0,\dd) -- (A.west);}
\,,\,\,
\tikzmath{  
\node (a) at (0,0) {$H_{I_2} \boxtimes H_{J_2} \boxtimes H_K \boxtimes$};
\def\dd{.3}
\def\ll{.1}
\def\rr{.1}
\draw[dashed, rounded corners = 5] (a.east) -- ++(\rr,0) -- ++(0,-\dd) -- ($(a.west) + (-\ll,-\dd)$) -- +(0,\dd) -- (a.west);
}\hspace{.43cm}
\tikzmath{\node (A) at (0,0) {$H_{I_2}\boxtimes H_{I_2}'\boxtimes$};
\def\dd{.3}
\def\ll{.1}
\def\rr{.1}
\draw[dashed, rounded corners = 5] (A.east) -- ++(\rr,0) -- ++(0,-\dd) -- ($(A.west) + (-\ll,-\dd)$) -- +(0,\dd) -- (A.west);}$};
}
\end{equation}
where the horizontal maps $k_1$ and $k_2$ are induced by \eqref{eq: u:HH->H},
and the maps from $H_\Sigma$ are instances of \eqref{eq:  u^(cali) }.
There is also a commutative square
\begin{equation}\label{eq: two triangles+}
\tikzmath{ \matrix [matrix of math nodes,column sep=2.5cm,row sep=.5cm]
{ |(a)| \tikzmath{\node (A) at (0,0) {$H_{I_1}\boxtimes H_{I_1}'\boxtimes$};
\def\dd{.3}
\def\ll{.1}
\def\rr{.1}
\draw[dashed, rounded corners = 5] (A.east) -- ++(\rr,0) -- ++(0,-\dd) -- ($(A.west) + (-\ll,-\dd)$) -- +(0,\dd) -- (A.west);
}
\pgfmatrixnextcell |(b)| 
\tikzmath{\node (A) at (0,0) {$H_{I_2}\boxtimes H_{I_2}'\boxtimes$};
\def\dd{.3}
\def\ll{.1}
\def\rr{.1}
\draw[dashed, rounded corners = 5] (A.east) -- ++(\rr,0) -- ++(0,-\dd) -- ($(A.west) + (-\ll,-\dd)$) -- +(0,\dd) -- (A.west);
}\\
|(c)|  \tikzmath{
\node (A) at (0,0) {$H_{I_1} \boxtimes H_{J_1} \boxtimes H_K \boxtimes$};
\def\dd{.3}
\def\ll{.1}
\def\rr{.1}
\draw[dashed, rounded corners = 5] (A.east) -- ++(\rr,0) -- ++(0,-\dd) -- ($(A.west) + (-\ll,-\dd)$) -- +(0,\dd) -- (A.west);
} \pgfmatrixnextcell |(d)| 
\tikzmath{
\node (A) at (0,0) {$H_{I_2} \boxtimes H_{J_2} \boxtimes H_K \boxtimes$};
\def\dd{.3}
\def\ll{.1}
\def\rr{.1}
\draw[dashed, rounded corners = 5] (A.east) -- ++(\rr,0) -- ++(0,-\dd) -- ($(A.west) + (-\ll,-\dd)$) -- +(0,\dd) -- (A.west);
}\\ }; 
\draw[->] (a) --node[left]{$\scriptstyle k_1^{-1}$} (c); \draw[->] (b) --node[right]{$\scriptstyle k_2^{-1}$} (d);
\draw[->] (a) --node[above]{$\scriptstyle (a\boxtimes a')(v\otimes \bar v)^*$} (b); \draw[->] (c) --node[above]{$\scriptstyle (a\boxtimes a''\boxtimes 1)(v\otimes \bar v)^*$} (d); }
\end{equation}
where $a'':H_{J_1}\to H_{J_2}$ is the value of $L^2(\cala(-))$ on the diffeomorphism that $\varphi$ induces from the upper half of $\partial(J_1\times [0,1])$ to the upper half of $\partial(J_2\times [0,1])$.
In view of \eqref{eq: two triangles} and \eqref{eq: two triangles+}, the commutativity of \eqref{eq: triangle that must commute to finish proof of annular sector} is equivalent to that of 
\begin{equation}
\label{eq: triangle that must commute to finish proof of annular sector ++}
\tikzmath{
\node at (0,.25) {$\scriptstyle (a\boxtimes a''\boxtimes 1)(v\otimes \bar v)^*$};\node (a) at (0,1) {$H_\Sigma$};
\node (b) at (-3,0) {$
\tikzmath{\node (A) at (0,0) {$H_{I_1}\boxtimes H_{J_1}\boxtimes H_K$};\def\dd{.3}\def\ll{.1}\def\rr{.15}
\draw[dashed, rounded corners = 5] (A.east) -- ++(\rr,0) -- ++(0,-\dd) -- ($(A.west) + (-\ll,-\dd)$) -- +(0,\dd) -- (A.west);}$};
\node (c) at (3,0) {$\tikzmath{\node (A) at (0,0) {$H_{I_2}\boxtimes H_{J_2}\boxtimes H_K$};\def\dd{.3}\def\ll{.1}\def\rr{.15}
\draw[dashed, rounded corners = 5] (A.east) -- ++(\rr,0) -- ++(0,-\dd) -- ($(A.west) + (-\ll,-\dd)$) -- +(0,\dd) -- (A.west);}$};
\draw[->, shorten >=-2](a)--(b);\draw[->, shorten >=-2](a)--(c);\draw[->,shorten <=2, shorten >=2](b)--(c);
}\hspace{-.81cm}
\end{equation}
where the maps down from $H_\Sigma$ are again instances of \eqref{eq:  u^(cali) }.
Consider now the following two commutative triangles
\begin{equation}\label{eq: two triangles next}
\!\tikzmath{
\node (a) at (0,1) {$H_\Sigma$};
\node (b1) at (-3,0) {$
\tikzmath{
\node (A) at (0,0) {$H_{I_1} \boxtimes H_{J_1} \boxtimes H_K \boxtimes$};
\def\dd{.3}
\def\ll{.1}
\def\rr{.1}
\draw[dashed, rounded corners = 5] (A.east) -- ++(\rr,0) -- ++(0,-\dd) -- ($(A.west) + (-\ll,-\dd)$) -- +(0,\dd) -- (A.west);
}$};
\node (b2) at (0,-1.15) {$
\tikzmath{\node (A) at (0,0) {$H_{K'}\boxtimes H_{K}\boxtimes$};
\def\dd{.3}
\def\ll{.1}
\def\rr{.1}
\draw[dashed, rounded corners = 5] (A.east) -- ++(\rr,0) -- ++(0,-\dd) -- ($(A.west) + (-\ll,-\dd)$) -- +(0,\dd) -- (A.west);
}$};
\node (b3) at (3,0) {$
\tikzmath{  
\node (A) at (0,0) {$H_{I_2} \boxtimes H_{J_2} \boxtimes H_K \boxtimes$};
\def\dd{.3}
\def\ll{.1}
\def\rr{.1}
\draw[dashed, rounded corners = 5] (A.east) -- ++(\rr,0) -- ++(0,-\dd) -- ($(A.west) + (-\ll,-\dd)$) -- +(0,\dd) -- (A.west);
}$};
\draw[->,shorten >=-4] (a) -- (b1);\draw[->] (a) -- (b2);\draw[->,shorten >=-4] (a) -- (b3);
\draw[->, shorten >=-5] (b1) -- (b2);
\draw[->, shorten >=-5] (b3) -- (b2);
\node at (-1.8,-.8) {$\scriptstyle \omega_1\boxtimes 1$};
\node at (1.8,-.8) {$\scriptstyle \omega_2\boxtimes 1$};
}
\end{equation}
where $\omega_1:H_{I_1}\boxtimes H_{J_1}\to H_{K'}$ and $\omega_2:H_{I_2}\boxtimes H_{J_2}\to H_{K'}$ are instances of \eqref{eq: u:HH->H},
and the maps from $H_\Sigma$ are as in \eqref{eq:  u^(cali) }.
Using \eqref{eq: two triangles next},
we can further reduce \eqref{eq: triangle that must commute to finish proof of annular sector ++}
to the commutativity of this triangle:
\begin{equation}
\label{eq: triangle that must commute to finish proof of annular sector ++ ++}
\,\tikzmath{
\node at (-1.9,-.8) {$\scriptstyle \omega_1\boxtimes 1$};
\node at (1.9,-.8) {$\scriptstyle \omega_2\boxtimes 1$};
\node at (0,.25) {$\scriptstyle (a\boxtimes a''\boxtimes 1)(v\otimes \bar v)^*$};
\node (a) at (0,-1.15) {$
\tikzmath{\node (A) at (0,0) {$H_{K'}\boxtimes H_{K}\boxtimes$};
\def\dd{.3}
\def\ll{.1}
\def\rr{.1}
\draw[dashed, rounded corners = 5] (A.east) -- ++(\rr,0) -- ++(0,-\dd) -- ($(A.west) + (-\ll,-\dd)$) -- +(0,\dd) -- (A.west);}$};
\node (b) at (-3.2,0) {$
\tikzmath{\node (A) at (0,0) {$H_{I_1}\boxtimes H_{J_1}\boxtimes H_K$};\def\dd{.3}\def\ll{.1}\def\rr{.15}
\draw[dashed, rounded corners = 5] (A.east) -- ++(\rr,0) -- ++(0,-\dd) -- ($(A.west) + (-\ll,-\dd)$) -- +(0,\dd) -- (A.west);}$};
\node (c) at (3.2,0) {$\tikzmath{\node (A) at (0,0) {$H_{I_2}\boxtimes H_{J_2}\boxtimes H_K$};\def\dd{.3}\def\ll{.1}\def\rr{.15}
\draw[dashed, rounded corners = 5] (A.east) -- ++(\rr,0) -- ++(0,-\dd) -- ($(A.west) + (-\ll,-\dd)$) -- +(0,\dd) -- (A.west);}$};
\draw[<-, shorten <=-5](a)--(b);\draw[<-, shorten <=-5](a)--(c);\draw[->,shorten <=2, shorten >=2](b)--(c);
}
\end{equation}
Recall that the support of $\varphi$ is disjoint from $K$.
It follows that $v\in \cala(K')$, and that the operator $v\otimes \bar v$ only acts on $H_{I_1}\boxtimes H_{J_1}$.
The commutativity of \eqref{eq: triangle that must commute to finish proof of annular sector ++ ++} therefore boils down to the commutativity of this diagram:
\begin{equation}
\label{eq: triangle that must commute to finish proof of annular sector ++ ++ ++}
\,\tikzmath{
\node at (-1.23,-.77) {$\scriptstyle \omega_1$};
\node at (1.23,-.77) {$\scriptstyle \omega_2$};
\node at (0,.25) {$\scriptstyle (a\boxtimes a'')(v\otimes \bar v)^*$};
\node (a) at (0,-1.15) {$H_{K'}$};
\node (b) at (-2,0) {$H_{I_1}\boxtimes H_{J_1}$};
\node (c) at (2,0) {$H_{I_2}\boxtimes H_{J_2}$};
\draw[<-, shorten <=-2](a)--(b);\draw[<-, shorten <=-2](a)--(c);\draw[->](b)--(c);
}
\end{equation}
Write
\[
\begin{split}
\varphi_I:\partial(I_1\times [0,1])_\top \to \partial(I_2\times [0,1])_\top\qquad
\varphi_J:\partial(J_1\times [0,1])_\top &\to \partial(J_2\times [0,1])_\top\\
\varphi_K:\partial(K'\times [0,1])_\top \to \partial(K'\times [0,1])_\top\,&
\end{split}
\]
for the maps induced by $\varphi$,
and recall that $a=L^2\cala(\varphi_I)$ and $a''=L^2\cala(\varphi_J)$.
Pick
\[
f_I\!:\partial(I_1\!\times\! [0,1])_\top \to S^1_\top\quad\,\,\,\,\,
f_J\!:\partial(J_1\!\times\! [0,1])_\top \to S^1_\top\quad\,\,\,\,\,
f_K\!:\partial(K'\!\times\! [0,1])_\top \to S^1_\top
\]
with the same properties as the maps $f_1$, $f_2$, $f_3$ that enter \eqref{eq: u:HH->H: the def.},
and let 
\[
g_I\!:\partial(I_2\!\times\! [0,1])_\top \to S^1_\top\quad\,\,\,\,\,
g_J\!:\partial(J_2\!\times\! [0,1])_\top \to S^1_\top\quad\,\,\,\,\,
g_K\!:\partial(K'\!\times\! [0,1])_\top \to S^1_\top
\]
be given by $g_\alpha=f_\alpha\circ \varphi_\alpha^{-1}$, $\alpha\in\{I,J,K\}$.
By definition, we have
\[
\begin{split}
\omega_1&=\big(L^2\cala(f_K)^{-1}\big)    \big(1\boxtimes v_\vdash v_\top^{-1}\big)    \big(L^2\cala(f_I)\boxtimes L^2\cala(f_J)\big)\\
\omega_2&=\big(L^2\cala(g_K)^{-1}\big)    \big(1\boxtimes v_\vdash v_\top^{-1}\big)    \big(L^2\cala(g_I)\boxtimes L^2\cala(g_J)\big).
\end{split}
\]
For $w:=\cala(f_K)(v)$, it follows from $\ad(v)=\cala(\varphi_K)$ that $\ad(w)=\cala(f_K\circ g_K^{-1})$.
Writing $H_0$ for $L^2\cala(S^1_\top)$, the triangle \eqref{eq: triangle that must commute to finish proof of annular sector ++ ++ ++} can then be decomposed as
\[
\tikzmath{
\node (a) at (0,-3.5) {$H_{K'}$};
\node (b) at (-3.2,0) {$H_{I_1}\boxtimes H_{J_1}$};
\node (c) at (3.2,0) {$H_{I_2}\boxtimes H_{J_2}$};
\node (bc) at ($(b)!.5!(c)$) {$H_{I_1}\boxtimes H_{J_1}$};
\node (bba) at ($(b)!.5!(a)$) {$H_0\boxtimes H_0$};
\node (baa) at ($(b)!.75!(a)$) {$H_0$};
\node (cca) at ($(c)!.5!(a)$) {$H_0\boxtimes H_0$};
\node (caa) at ($(c)!.75!(a)$) {$H_0$};
\draw[->] (b) --node[left, yshift=1]{$\scriptstyle L^2\cala(f_I)\boxtimes L^2\cala(f_J)$} (bba);\draw[->] (bba) --node[left, yshift=-2]{$\scriptstyle 1\boxtimes v_\dashv v_\top^{-1}$} (baa);\draw[->] (baa) --node[left, yshift=-1]{$\scriptstyle L^2\cala(f_K)^{-1}$} (a);
\draw[->] (c) --node[right, yshift=1]{$\scriptstyle L^2\cala(g_I)\boxtimes L^2\cala(g_J)$} (cca);\draw[->] (cca) --node[right, yshift=-2]{$\scriptstyle 1\boxtimes v_\dashv v_\top^{-1}$} (caa);\draw[->] (caa) --node[right, yshift=-1]{$\scriptstyle L^2\cala(g_K)^{-1}$} (a);
\draw[->] (caa) --node[above]{$\scriptstyle w\otimes \bar w$} (baa);
\draw[->] (cca) --node[above]{$\scriptstyle w\otimes \bar w$} (bba);
\draw[->] (bc) --node[left, fill=white, xshift=12, yshift=1, inner sep=1]{$\scriptstyle L^2\cala(f_I)\boxtimes L^2\cala(f_J)$} (cca);
\draw[->] (bc) --node[above]{$\scriptstyle v\otimes \bar v$} (b);
\draw[->] (bc) --node[above, yshift=4]{$\scriptstyle L^2\cala(\varphi_I)\boxtimes L^2\cala(\varphi_J)$} (c);
}
\]
and it is now clear that each one of the pieces commutes.
\end{proof}


\bibliographystyle{abbrv}
\bibliography{../Files/db-cn3}

\begin{thebibliography}{10}

\bibitem{Andersen-Ueno(Abelian-CFT+det-bundles)}
J.~E. Andersen and K.~Ueno.
\newblock Abelian conformal field theory and determinant bundles.
\newblock {\em Internat. J. Math.}, 18(8):919--993, 2007.

\bibitem{Andersen-Ueno(Geometric-constr-modular-functors)}
J.~E. Andersen and K.~Ueno.
\newblock Geometric construction of modular functors from conformal field
  theory.
\newblock {\em J. Knot Theory Ramifications}, 16(2):127--202, 2007.

\bibitem{Andersen-Ueno(Modular-functors-determined-genus-zero)}
J.~E. Andersen and K.~Ueno.
\newblock Modular functors are determined by their genus zero data.
\newblock {\em Quantum Topol.}, 3(3-4):255--291, 2012.

\bibitem{Andersen-Ueno(Reshetikhin-Turaev-from-CFT)}
J.~E. Andersen and K.~Ueno.
\newblock Construction of the {R}eshetikhin-{T}uraev {TQFT} from conformal
  field theory.
\newblock {\em Inventiones Mathematicae}, 2014.
\newblock arXiv:1110.5027.

\bibitem{Axelrod-DellaPietra-Witten}
S.~Axelrod, S.~Della~Pietra, and E.~Witten.
\newblock Geometric quantization of {C}hern-{S}imons gauge theory.
\newblock {\em J. Differential Geom.}, 33(3):787--902, 1991.

\bibitem{Bakalov-Kirillov(Lect-tens-cat+mod-func)}
B.~Bakalov and A.~Kirillov, Jr.
\newblock {\em Lectures on tensor categories and modular functors}, volume~21
  of {\em University Lecture Series}.
\newblock American Mathematical Society, Providence, RI, 2001.

\bibitem{BDH(Dualizability+Index-of-subfactors)}
A.~Bartels, C.~L. Douglas, and A.~Henriques.
\newblock Dualizability and index of subfactors.
\newblock {\em Quantum Topology}, 5:289--345, 2014.
\newblock arXiv:1110.5671.

\bibitem{BDH(nets)}
A.~Bartels, C.~L. Douglas, and A.~Henriques.
\newblock Conformal nets {I}: {C}oordinate-free nets.
\newblock {\em Int. Math. Res. Not.}, 13:4975--5052, 2015.
\newblock arXiv:1302.2604v2.

\bibitem{BDH(dualizability)}
A.~Bartels, C.~L. Douglas, and A.~Henriques.
\newblock Conformal nets {V}: dualizability.
\newblock In preparation, 2015.

\bibitem{Beauville-Laszlo(Conformal-blocks-and-generalized-theta-functions)}
A.~Beauville and Y.~Laszlo.
\newblock Conformal blocks and generalized theta functions.
\newblock {\em Comm. Math. Phys.}, 164(2):385--419, 1994.

\bibitem{Belavin-Polyakov-Zamolodchikov}
A.~A. Belavin, A.~M. Polyakov, and A.~B. Zamolodchikov.
\newblock Infinite conformal symmetry in two-dimensional quantum field theory.
\newblock {\em Nuclear Phys. B}, 241(2):333--380, 1984.

\bibitem{Beliakova-Blanchet(mod-cats-types-BCD)}
A.~Beliakova and C.~Blanchet.
\newblock Modular categories of types {$B$}, {$C$} and {$D$}.
\newblock {\em Comment. Math. Helv.}, 76(3):467--500, 2001.

\bibitem{Bisch(Bimodules)}
D.~Bisch.
\newblock Bimodules, higher relative commutants and the fusion algebra
  associated to a subfactor.
\newblock In {\em Operator algebras and their applications ({W}aterloo, {ON},
  1994/1995)}, volume~13 of {\em Fields Inst. Commun.}, pages 13--63. Amer.
  Math. Soc., Providence, RI, 1997.

\bibitem{Blanchet(Hecke-algebras-modular-categories-and-3-manifolds-quantum-invariants)}
C.~Blanchet.
\newblock Hecke algebras, modular categories and {$3$}-manifolds quantum
  invariants.
\newblock {\em Topology}, 39(1):193--223, 2000.

\bibitem{Brunetti-Guido-Longo(1993modular+duality-in-CQFT)}
R.~Brunetti, D.~Guido, and R.~Longo.
\newblock Modular structure and duality in conformal quantum field theory.
\newblock {\em Comm. Math. Phys.}, 156(1):201--219, 1993.

\bibitem{cchw}
S.~Carpi, R.~Conti, R.~Hillier, and M.~Weiner.
\newblock Representations of conformal nets, universal {$\rm C^*$}-algebras and
  {K}-theory.
\newblock {\em Comm. Math. Phys.}, 320(1):275--300, 2013.

\bibitem{cklw}
S.~Carpi, Y.~Kawahigashi, R.~Longo, and M.~Weiner.
\newblock From vertex operator algebras to conformal nets and back.
\newblock arXiv:1503.01260, 2015.

\bibitem{Connes(Non-commutative-geometry)}
A.~Connes.
\newblock {\em Noncommutative geometry}.
\newblock Academic Press Inc., San Diego, CA, 1994.

\bibitem{Dong-Lin(Unitary-vertex-operator-algebras)}
C.~Dong and X.~Lin.
\newblock Unitary vertex operator algebras.
\newblock arXiv:1308.2361, 2013.

\bibitem{Faltings(A-proof-for-the-Verlinde-formula)}
G.~Faltings.
\newblock A proof for the {V}erlinde formula.
\newblock {\em J. Algebraic Geom.}, 3(2):347--374, 1994.

\bibitem{Frenkel--Ben-Zvi}
E.~Frenkel and D.~Ben-Zvi.
\newblock {\em Vertex algebras and algebraic curves}, volume~88 of {\em
  Mathematical Surveys and Monographs}.
\newblock American Mathematical Society, Providence, RI, second edition, 2004.

\bibitem{Friedan-Shenker(The-analytic-geometry-of-two-dimensional-conformal-field-theory)}
D.~Friedan and S.~Shenker.
\newblock The analytic geometry of two-dimensional conformal field theory.
\newblock {\em Nuclear Phys. B}, 281(3-4):509--545, 1987.

\bibitem{Gabbiani-Froehlich(OperatorAlg-CFT)}
F.~Gabbiani and J.~Fr{\"o}hlich.
\newblock Operator algebras and conformal field theory.
\newblock {\em Comm. Math. Phys.}, 155(3):569--640, 1993.

\bibitem{Hitchin(Flat-connections-and-geometric-quantization)}
N.~J. Hitchin.
\newblock Flat connections and geometric quantization.
\newblock {\em Comm. Math. Phys.}, 131(2):347--380, 1990.

\bibitem{Huang-Lepowsky(Tensor-categories-and-the-mathematics-of-rational-and-logarithmic-conformal-field-theory)}
Y.-Z. Huang and J.~Lepowsky.
\newblock Tensor categories and the mathematics of rational and logarithmic
  conformal field theory.
\newblock arXiv:1304.7556, 2013.

\bibitem{Kawahigashi-Longo-Mueger(2001multi-interval)}
Y.~Kawahigashi, R.~Longo, and M.~M{\"u}ger.
\newblock Multi-interval subfactors and modularity of representations in
  conformal field theory.
\newblock {\em Comm. Math. Phys.}, 219(3):631--669, 2001.

\bibitem{Kumar-Narasimhan-Ramanathan(Infinite-Grassmannians-and-moduli-spaces-of-G-bundles)}
S.~Kumar, M.~S. Narasimhan, and A.~Ramanathan.
\newblock Infinite {G}rassmannians and moduli spaces of {$G$}-bundles.
\newblock {\em Math. Ann.}, 300(1):41--75, 1994.

\bibitem{Laszlo(Hitchin's-and-WZW-connections-are-the-same)}
Y.~Laszlo.
\newblock Hitchin's and {WZW} connections are the same.
\newblock {\em J. Differential Geom.}, 49(3):547--576, 1998.

\bibitem{Longo(Lectures-on-Nets-II)}
R.~Longo.
\newblock Lectures on conformal nets {II}.
\newblock http://www.mat.uniroma2.it/\~{}longo/Lecture\%20Notes.html, 2008.

\bibitem{Looijenga(Unitarity-of-SL(2)-conformal-blocks-in-genus-zero)}
E.~Looijenga.
\newblock Unitarity of {${\rm SL}(2)$}-conformal blocks in genus zero.
\newblock {\em J. Geom. Phys.}, 59(5):654--662, 2009.

\bibitem{Moore-Seiberg(classical+quantum-conf-field-theory)}
G.~Moore and N.~Seiberg.
\newblock Classical and quantum conformal field theory.
\newblock {\em Comm. Math. Phys.}, 123(2):177--254, 1989.

\bibitem{Moore-Seiberg(Naturality-in-conformal-field-theory)}
G.~Moore and N.~Seiberg.
\newblock Naturality in conformal field theory.
\newblock {\em Nuclear Phys. B}, 313(1):16--40, 1989.

\bibitem{Mueger(Subfactors-to-categories-II)}
M.~M{\"u}ger.
\newblock From subfactors to categories and topology. {II}. {T}he quantum
  double of tensor categories and subfactors.
\newblock {\em J. Pure Appl. Algebra}, 180(1-2):159--219, 2003.

\bibitem{Posthuma(The-Heisenberg-group-and-conformal-field-theory)}
H.~Posthuma.
\newblock The {H}eisenberg group and conformal field theory.
\newblock {\em Q. J. Math.}, 63(2):423--465, 2012.

\bibitem{Ramadas(The-Harder-Narasimhan-trace-and-unitarity-of-the-KZ/Hitchin-connection:-genus-0)}
T.~R. Ramadas.
\newblock The ``{H}arder-{N}arasimhan trace'' and unitarity of the
  {KZ}/{H}itchin connection: genus 0.
\newblock {\em Ann. of Math. (2)}, 169(1):1--39, 2009.

\bibitem{Rehren(Braid-group-statistics-and-their-superselection-rules)}
K.-H. Rehren.
\newblock Braid group statistics and their superselection rules.
\newblock In {\em The algebraic theory of superselection sectors ({P}alermo,
  1989)}, pages 333--355. World Sci. Publ., River Edge, NJ, 1990.

\bibitem{Rowell(From-quantum-groups-to-UMTC)}
E.~C. Rowell.
\newblock From quantum groups to unitary modular tensor categories.
\newblock In {\em Representations of algebraic groups, quantum groups, and
  {L}ie algebras}, volume 413 of {\em Contemp. Math.}, pages 215--230. Amer.
  Math. Soc., Providence, RI, 2006.

\bibitem{Sauvageot(Sur-le-produit-tensoriel-relatif)}
J.-L. Sauvageot.
\newblock Sur le produit tensoriel relatif d'espaces de {H}ilbert.
\newblock {\em J. Operator Theory}, 9(2):237--252, 1983.

\bibitem{Segal(Def-CFT)}
G.~Segal.
\newblock The definition of conformal field theory.
\newblock In {\em Topology, geometry and quantum field theory}, volume 308 of
  {\em London Math. Soc. Lecture Note Ser.}, pages 421--577. Cambridge Univ.
  Press, Cambridge, 2004.

\bibitem{Teleman(Borel-Weil-Bott)}
C.~Teleman.
\newblock Borel-{W}eil-{B}ott theory on the moduli stack of {$G$}-bundles over
  a curve.
\newblock {\em Invent. Math.}, 134(1):1--57, 1998.

\bibitem{Tsuchiya-Ueno-Yamada(CFT-universal-family)}
A.~Tsuchiya, K.~Ueno, and Y.~Yamada.
\newblock Conformal field theory on universal family of stable curves with
  gauge symmetries.
\newblock In {\em Integrable systems in quantum field theory and statistical
  mechanics}, volume~19 of {\em Adv. Stud. Pure Math.}, pages 459--566.
  Academic Press, Boston, MA, 1989.

\bibitem{Wassermann(Operator-algebras-and-conformal-field-theory)}
A.~Wassermann.
\newblock Operator algebras and conformal field theory. {III}. {F}usion of
  positive energy representations of {${\rm LSU}(N)$} using bounded operators.
\newblock {\em Invent. Math.}, 133(3):467--538, 1998.

\bibitem{wenzl-cstar}
H.~Wenzl.
\newblock {$C^*$} tensor categories from quantum groups.
\newblock {\em J. Amer. Math. Soc.}, 11(2):261--282, 1998.

\bibitem{Xu(Jones-Wassermann-subfactors)}
F.~Xu.
\newblock Jones-{W}assermann subfactors for disconnected intervals.
\newblock {\em Commun. Contemp. Math.}, 2(3):307--347, 2000.

\end{thebibliography}

\end{document}